\documentclass[10pt,a4paper]{article}
\pdfoutput=1
\usepackage[lmargin=1.2in,rmargin=1.2in,tmargin=1.3in,bmargin=1.9in]{geometry}
\usepackage{amsthm,amsfonts,amssymb,amsmath,comment,slashed,mathtools,bbold}
\usepackage[T1]{fontenc} %
\usepackage[utf8]{inputenc}
\usepackage{bbm}
\usepackage[affil-it]{authblk} %
\usepackage{hyperref}    
\hypersetup{hidelinks}
\usepackage{graphicx} 
\usepackage{subcaption}
\usepackage{thm-restate} %
\usepackage{xargs}          
\usepackage{here}
\usepackage{mathrsfs}
\usepackage{xcolor}  %
\usepackage{here} %
\usepackage[capitalize]{cleveref} %
\usepackage{tikz-cd} %
\usepackage{color}
\usepackage{transparent}
\usepackage{bm}
\usepackage[backend=biber,bibencoding=ascii,giveninits=true,doi=false,isbn=false,url=false,sorting=nyt,maxbibnames=99]{biblatex} 
\renewbibmacro{in:}{} %
\bibliography{kerr_ads_interior.bib}
\theoremstyle{plain} 
\newtheorem{definition}{Definition} 
\newtheorem{prop}{Proposition}

\newtheorem{lemma}{Lemma}
\newtheorem{cor}{Corollary}
\newtheorem{rmk}{Remark}

\newtheorem{conjecture}{Conjecture}

\newtheorem{thmrough}{Theorem}
\declaretheorem[name=Theorem]{theorem}
\renewcommand{\d}{\mathrm{d}}

\newcommand{\Hp}{\mathcal{H}}

\newcommand{\Ch}{\mathcal{CH}}

\newcommand{\dvol}{\mathrm{dvol}}

\newcommand{\uhl}{{u_{\mathcal{H}_L}}}
\newcommand{\uhr}{{u_{\mathcal{H}_R}}}
\newcommand{\uchl}{{u_{\mathcal{CH}_L}}}
\newcommand{\uchr}{{u_{\mathcal{CH}_R}}}

\newcommand{\uhplus}{	u_{\mathcal{H}^+} }
\newcommand{\uhminus}{	u_{\mathcal{H}^-} }
\newcommand{\Ai}{{\mathrm{Ai}}}
\newcommand{\Bi}{{\mathrm{Bi}}}
\newcommand{\ecut}{{\epsilon_{\textup{cut}}} }
\newcommand{\emphb}[1]{\textbf{\emph{#1}}}

\makeatletter
\renewcommand{\paragraph}[1]{%
	\par %
	\addvspace{\medskipamount}%
	\textbf{\textit{#1\@addpunct{.}}}\enspace\ignorespaces
}
\makeatother
\numberwithin{equation}{section}
\numberwithin{prop}{section}
\numberwithin{rmk}{section}
\numberwithin{lemma}{section}
\numberwithin{definition}{section}

\title{Diophantine approximation as  Cosmic Censor \\
	for Kerr--AdS black holes}
\author[1,2]{Christoph Kehle\thanks{christoph.kehle@eth-its.ethz.ch}}
\affil[1]{\small  Department of Pure Mathematics and Mathematical Statistics, University~of~Cambridge,~Wilberforce~Road,~Cambridge CB3 0WB, United Kingdom \vskip.1pc \ }
\affil[2]{Institute for Theoretical Studies, ETH Zürich, Clausiusstrasse 47, 8092 Zürich, Switzerland \vskip.1pc\ }

\date{June 30, 2021}
\begin{document}
	
\maketitle
\thispagestyle{empty}
\begin{abstract}
The purpose of this paper is to show an unexpected connection between Diophantine approximation and the behavior of waves on black hole interiors with negative cosmological constant $\Lambda<0$ and explore the consequences of this for the Strong Cosmic Censorship conjecture in general relativity.

We study linear scalar perturbations $\psi$ of Kerr--AdS solving $\Box_g\psi-\frac{2}{3}\Lambda \psi=0$ with reflecting boundary conditions imposed at infinity. Understanding the behavior of $\psi$ at the Cauchy horizon corresponds to a linear analog of the problem of Strong Cosmic Censorship. Our main result shows that if the dimensionless black hole parameters mass $\mathfrak m = M \sqrt{-\Lambda}$ and angular momentum $\mathfrak a = a \sqrt{-\Lambda}$ satisfy a certain non-Diophantine condition, then  perturbations $\psi$ arising from generic smooth initial data blow up $|\psi|\to+\infty$ at the Cauchy horizon. The proof crucially relies on a novel resonance phenomenon between stable trapping on the black hole exterior and the poles of the interior scattering operator that gives rise to a small divisors problem. 
Our result is in stark contrast to the result on Reissner--Nordstr\"om--AdS \cite{CKehle2019} as well as to previous work on the analogous problem for $\Lambda \geq 0$---in both cases such linear scalar perturbations were shown to remain bounded.

As a result of the non-Diophantine condition, the set of parameters $\mathfrak m, \mathfrak a$ for which we show blow-up forms a Baire-generic but Lebesgue-exceptional subset of all parameters below the Hawking--Reall bound. On the other hand, we conjecture that for a set of parameters $\mathfrak m, \mathfrak a $ which is Baire-exceptional but Lebesgue-generic, all linear scalar perturbations remain bounded at the Cauchy horizon $|\psi|\leq C$. This suggests that the validity of the $C^0$-formulation of Strong Cosmic Censorship for $\Lambda <0$ may change in a spectacular way according to the notion of genericity imposed.
\end{abstract}
\newpage
\thispagestyle{empty}
\tableofcontents
\thispagestyle{empty}
\newpage
\section{Introduction}
The Kerr--Anti-de~Sitter (Kerr--AdS) black hole spacetimes $(\mathcal M, g)$ constitute a 2-parameter family of solutions to the celebrated Einstein equations 
\begin{align}
\mathrm{Ric}_{\mu\nu}(g) - \frac 12\mathrm{R} g_{\mu\nu} +\Lambda g_{\mu\nu} = 8 \pi T_{\mu\nu}
\label{eq:Einstein}
\end{align} 
in vacuum $(T_{\mu\nu} =0)$ and with negative cosmological constant $\Lambda <0$.
The family (see already \eqref{eq:kerradsmetric} for the metric) is parameterized by the black hole mass $M>0$, and specific angular momentum $a\neq 0$.
The Kerr--AdS black holes posses a smooth Cauchy horizon beyond which the spacetime has infinitely many smooth extensions---thus violating determinism. Regular Cauchy horizons are thought, however, to be generically unstable, which is the content of the Strong Cosmic Censorship conjecture due to Roger Penrose \cite{penrose1974gravitational}. Its strongest formulation, the $C^0$-formulation~\cite{chr1} (see already \cref{conj:c0}), states that for generic initial data for \eqref{eq:Einstein},  the metric cannot be continuously extended beyond a Cauchy horizon,  in this sense saving determinism within classical general relativity.
Unfortunately, for $\Lambda =0$ and $\Lambda >0$, this formulation was disproved by Dafermos--Luk \cite{dafermos2017interior}. However, a weaker formulation put forward by Christodoulou is still expected to be true  (see already \cref{conj:2}). Refer to \cref{sec:intscc} for a more detailed discussion. For $\Lambda <0$, the question of the validity of the $C^0$-formulation of Strong Cosmic Censorship has until today remained open. 

Motivated by the above, we study linear scalar perturbations $\psi$ of subextremal Kerr--AdS black holes solving the conformal scalar wave equation
\begin{align}
\Box_{g}\psi - \frac{2}{3}  \Lambda \psi =0 \label{eq:wavekerr}
\end{align}
which arise from smooth and compactly supported initial data posed on a spacelike hypersurface and which satisfy reflecting boundary conditions at infinity. 
We further assume that the black hole parameters satisfy the Hawking--Reall bound \cite{hawking_reall}, see already \eqref{eq:hawkingreallbound}. One can view \eqref{eq:wavekerr} as a linear scalar analog of \eqref{eq:Einstein}, and so the linear scalar analog of the $C^0$-formulation of Strong Cosmic Censorship  is the statement that for \underline{generic} black hole parameters, linear scalar perturbations  $\psi$, arising from \underline{generic} initial data for \eqref{eq:wavekerr}, \emph{fail} to  be continuous  at the Cauchy horizon (see already \cref{conj:linearana}).

Our main result \cref{thm:rough} shows that if the dimensionless Kerr--AdS parameters mass $\mathfrak m = M \sqrt{-\Lambda}$ and angular momentum $\mathfrak a = a \sqrt{-\Lambda}$   satisfy a certain \emphb{non-Diophantine condition}, then linear scalar perturbations  $\psi$ solving  \eqref{eq:wavekerr} and arising from generic initial data \emphb{blow up}
\begin{align}
|\psi| \to +\infty
\label{eq:blow-up}
\end{align}   
at the Cauchy horizon. We show that the set of such parameters is \emphb{Baire-generic} (but \emphb{Lebesgue-exceptional}). 
 
Hence,  our main result provides an---unexpected---positive resolution of  the linear scalar analog of the $C^0$-formulation of the Strong Cosmic Censorship conjecture for $\Lambda<0$, provided that the genericity of the set of parameters is taken in the Baire-generic sense.
 
\cref{thm:rough} is in sharp contrast to the result  on   Reissner--Nordström--AdS black holes \cite{CKehle2019} and to previous work on Strong Cosmic Censorship for $\Lambda \geq 0$---in both cases such perturbations  $\psi$ were shown to remain bounded and to extend continuously across the Cauchy horizon. 

The instability result \eqref{eq:blow-up} of \cref{thm:rough} is not associated to superradiance (since the parameters satisfy the Hawking--Reall bound) and, more surprisingly, is also not a consequence of the well-known blue-shift instability \cite{Penrose:1968ar} at the Cauchy horizon.  Instead, \cref{thm:rough} is a manifestation of the occurrence of \emphb{small divisors} originating from a new resonance phenomenon between, on the one hand, high frequencies associated to \emphb{stable trapping} on the exterior  \cite{gustav,quasimodes} and, on the other hand, the \emphb{poles} of the interior scattering operator which are characteristic frequencies with respect to the Killing generator of the Cauchy horizon  \cite{kehle2018scattering}. For this, it is fundamental that Kerr--AdS is rotating, as it is only in this case that stably trapped high frequency waves can, at the same time, be characteristic frequency waves with respect to the Killing generator of the Cauchy horizon. If now $\mathfrak m$, $\mathfrak a$ satisfy the non-Diophantine condition, then the resonance will be sufficiently strong (and the occurring divisors will be sufficiently small) so as to cause the instability \eqref{eq:blow-up}. 
 
\emph{Thus, in the case $\Lambda <0$, surprisingly, Diophantine approximation may turn out to be the elusive ``Cosmic Censor'' which Penrose was searching for in order to protect determinism in general relativity \cite{penrose1974gravitational}.}
 
The story, however, has yet another level of complexity. We also conjecture that, if the dimensionless black hole parameters  $\mathfrak m= M \sqrt{-\Lambda}$ and $\mathfrak a = a \sqrt{-\Lambda}$ satisfy a Diophantine condition, then linear scalar perturbations $\psi$ remain \emphb{bounded} $|\psi|\leq C$ at the Cauchy horizon. This would then hold for \emphb{Lebesgue-generic} but \emphb{Baire-exceptional} black hole parameters. If true, this would provide a negative resolution of the linear scalar analog of the $C^0$-formulation of Strong Cosmic Censorship provided that genericity of the parameters is now taken in the Lebesgue-generic sense. 
  
Returning to the fully nonlinear $C^0$-formulation of Strong Cosmic Censorship, the black hole parameters are themselves dynamic in evolution under \eqref{eq:Einstein}. Thus, the above competing notions of genericity \emph{for the parameters} may now be reflected in different formulations of the genericity assumption \emph{imposed on initial data}  in the statement of the conjecture. This could mean that the validity of Strong Cosmic Censorship is not only sensitive to the regularity of the extension but may also become highly sensitive to the precise notion of genericity imposed on the initial data.

\subsubsection*{Outline of the introduction}
We begin in \cref{sec:intscc} with a presentation of the  $C^0$-formulation (\cref{conj:c0}) and Christo\-doulou's reformulation (\cref{conj:2})  of the Strong Cosmic Censorship conjecture. We also introduce  their respective linear scalar analogs \cref{conj:linearana} and \cref{conj:linearanachristo} and review the relevant previous work and difficulties for $\Lambda \geq 0$ and $\Lambda <0$.
Then, turning to the Kerr--AdS case ($\Lambda <0)$, we will first outline  in \cref{sec:introex} the behavior of linear scalar perturbations on the black hole exterior before we focus on the interior in \cref{sec:introin}, see \cref{fig:intextprop}.
In \cref{sec:introdioph} we put both insights together and we will see, at least on a heuristic level, how small divisors and Diophantine approximation arise. This will lead to a new expectation that transcends \cref{conj:linearana} and \cref{conj:linearanachristo} and which we formulate in  \cref{sec:introconjectures} as \cref{con:2a} and \cref{con:2b}. In \cref{sec:intromainthm} we state our main result \cref{thm:rough}, which resolves \cref{con:2a} in the affirmative.  Then, in \cref{sec:outlookonres} we give an outlook on  \cref{con:2b}. In \cref{sec:breifdesc} we  describe how we turn our heuristics of \cref{sec:introdioph} into a proof of \cref{thm:rough}.  Finally, we give a brief outline of the paper in \cref{sec:outlineofc3}.
\begin{figure}[!h]
	\centering
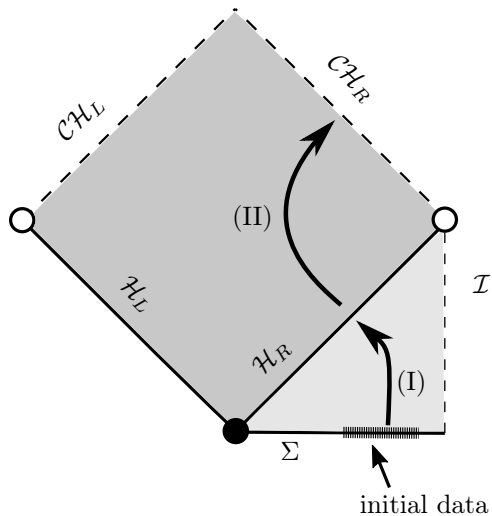
	\caption{(I): Exterior propagation, 
		(II): Interior propagation}
	\label{fig:intextprop}
\end{figure}

\subsection{Strong Cosmic Censorship: Conjectures 1--4}
\label{sec:intscc}
Recall from our previous discussion that our main motivation for studying linear perturbations on black hole interiors is to shed light on one of the most fundamental problems in general   relativity: the existence of smooth Cauchy horizons.

In general, a  Cauchy horizon $\mathcal{CH}$ defines the boundary beyond which initial data on a spacelike hypersurface (together with boundary conditions at infinity in the asymptotically AdS case)  no longer uniquely determine the spacetime as a solution of the Einstein equations  \eqref{eq:Einstein}. 
 The Kerr(--de~Sitter or --Anti-de~Sitter) black holes share the property that they indeed posses a smooth Cauchy horizon $\mathcal{CH}$ in their interiors.
In particular, these spacetimes admit infinitely many smooth extensions beyond their Cauchy horizons solving \eqref{eq:Einstein}, and in this sense violating determinism and the predictability of the theory. From a PDE point of view, this corresponds to a lack of global uniqueness for \eqref{eq:Einstein}. However, the existence of regular Cauchy horizons is conjectured to be an artifact of the high degree of symmetry in those explicit spacetimes and generically  it is expected that some sort of singularity ought to form at or before a Cauchy horizon. The original mechanism which was invoked to support this expectation is a blue-shift instability associated to Cauchy horizons \cite{Penrose:1968ar}. The emergence of such a  singularity at or before a Cauchy horizon is paradoxically ``good'' because---if sufficiently strong---it can be argued that this restores determinism, as the fate of any observer approaching the singularity, though bleak, is uniquely determined.  Making this precise gives rise to various formulations of what is known as the \emph{Strong Cosmic Censorship} (SCC) conjecture \cite{penrose1974gravitational,christo}.

We begin with the $C^0$-formulation of the SCC conjecture which can be seen as the strongest, most desirable, inextendibility statement in this context. 
\begin{conjecture}[$C^0$-formulation of Strong Cosmic Censorship]
	\label{conj:c0}
	For generic compact, asymptotically flat or asymptotically Anti-de Sitter vacuum initial data, the maximal Cauchy development of \eqref{eq:Einstein} is inextendible as a Lorentzian manifold with $C^0$ metric.
\end{conjecture}
This formulation is related to the statement that observers  are torn apart by infinite tidal deformations before they have the chance to cross a Cauchy horizon \cite{ori91,dafermos2017interior}.

Surprisingly, the $C^0$-formulation (\cref{conj:c0}) was recently proved to be false for both cases $\Lambda =0$ and $\Lambda>0$ \cite{dafermos2017interior} (see discussion later). The reason is that it turns out that the blue-shift instability is not sufficiently strong to destroy the metric itself, only derivatives of the metric. However, the following weaker, yet still well-motivated, formulation introduced by Christodoulou in \cite{christo} is still expected to hold true (though for the $\Lambda >0$ case see the discussion later). 
\begin{conjecture}[Christodoulou's reformulation of Strong Cosmic Censorship]\label{conj:2}
	For generic compact, asymptotically flat or asymptotically Anti-de Sitter vacuum initial data, the maximal Cauchy development of \eqref{eq:Einstein} is inextendible as a Lorentzian manifold with $C^0$ metric and locally square integrable Christoffel symbols.
\end{conjecture}
Unlike the $C^0$-formulation in \cref{conj:c0}, the statement of \cref{conj:2} does not guarantee  the complete destruction of observers approaching Cauchy horizons. However, it restores determinism in the sense that even just weak solutions must break down at  Cauchy horizons. Nonetheless, one may remain uneasy as to whether the standard notion of weak solution to \eqref{eq:Einstein} is finally the correct one \cite{penrosewaves,gravwaves,gravwaves2}. In this sense it is a pity that \cref{conj:c0} turned out to be false in the $\Lambda \geq 0$ cases,  as it would have provided a much more definitive resolution of the spirit of the Strong Cosmic Censorship conjecture. Hence, it is of interest to know whether the situation is better in the $\Lambda <0$ case!
\subsubsection*{Linear scalar analog of the Strong Cosmic Censorship conjecture} The aforementioned formulations of SCC have linear scalar analogs on the level of \eqref{eq:wavekerr}. Indeed, under the identification $\psi \sim g$, the linear scalar wave equation \eqref{eq:wavekerr} can be seen as a naive linearization of the Einstein equations \eqref{eq:Einstein}  after neglecting the nonlinearities and the tensorial structure. Moreover, many phenomena and difficulties for the full Einstein equations \eqref{eq:Einstein}  are already present at the level of \eqref{eq:wavekerr}.

The linear scalar analog of \cref{conj:c0} in a neighborhood of Kerr and Kerr--(Anti-)de~Sitter corresponds to the statement  that for generic black hole parameters, linear scalar perturbations $\psi$ arising from generic data on a spacelike hypersurface solving \eqref{eq:wavekerr} blow up in amplitude at the Cauchy horizon. 
\begin{conjecture}[Linear scalar analog of the $C^0$-formulation of SCC (\cref{conj:c0})]For \underline{generic} Kerr--(dS/AdS) black hole parameters, linear scalar perturbations $\psi$ solving \eqref{eq:wavekerr}, arising from \underline{generic} initial data, blow up in amplitude
\begin{align}|\psi|\to +\infty \end{align}
at the Cauchy horizon.
\label{conj:linearana}
\end{conjecture} 

 The reformulation due to Christodoulou (\cref{conj:2}) finds its linear scalar analog in the $H^1_\textup{loc}$ blow up of $\psi$ at the Cauchy horizon in view of the identification $\partial \psi \sim \Gamma$. 
\begin{conjecture}[Linear scalar analog of Christodoulou's reformulation of SCC (\cref{conj:2})]
For \underline{generic} Kerr--(dS/AdS) black hole parameters, linear scalar perturbations $\psi$ solving \eqref{eq:wavekerr}, arising from \underline{generic} initial data, blow up in local energy
	\begin{align}\|\psi\|_{H^1_\textup{loc}} = +\infty \end{align}
	at the Cauchy horizon.\label{conj:linearanachristo}
\end{conjecture} 

The word \emph{generic} appears twice in the above formulations, both in the context of the parameters and in the context of the perturbation. This is because in the fully nonlinear \cref{conj:c0} and \cref{conj:2}, the background parameters are themselves dynamic in evolution under \eqref{eq:Einstein} and thus both would be encompassed in the genericity of the initial data. 

\paragraph{Genericity of the black hole parameters}
As we will show in the present paper, for the Kerr--AdS case, the validity of \cref{conj:linearana} and \cref{conj:linearanachristo} will depend in a crucial way on the notion of genericity (Baire-generic or Lebesgue-generic) imposed on the parameters. This will eventually lead us to refine the above statements of \cref{conj:linearana} and \cref{conj:linearanachristo} (see already \cref{sec:introconjectures}).

\paragraph{Genericity of the initial data}
We will assume that the initial data lie in the class of smooth functions of compact support. 
Regarding  genericity within that class, note that just finding one single solution for which the blow-up statement is true already yields a natural notion of genericity.
 Indeed, since \eqref{eq:wavekerr} is linear, it would then follow that data for which the arising solution does not blow up satisfy a co-dimension 1 property (see already \cref{rmk:linearisenough}) and thus, would be exceptional. It is this notion of genericity of the initial data which we will consider later in \cref{sec:introconjectures}. Note that we will also consider a more refined notion of genericity of initial data in \cref{rmk:othernotionsofgenericity}.

\par 
\addvspace{\medskipamount}
In the above discussion, one may also consider the Reissner--Nordström(--dS/AdS) spacetimes (see e.g.~\cite{kehle2020diophantine}) which are spherically symmetric electrovacuum solutions to \eqref{eq:Einstein}. Reissner--Nord\-ström(--dS/AdS) spacetimes  are often studied as a toy model for Kerr(--dS/AdS) and the above \cref{conj:linearana} and \cref{conj:linearanachristo} can also be formulated replacing Kerr(--dS/AdS)  with Reissner--Nordström(--dS/AdS).

Before we bring our discussion of SCC to asymptotically AdS black holes ($\Lambda<0$),
we will first review the state of the art of the SCC conjecture for the cases $\Lambda=0$ and $\Lambda >0$.

\subsubsection*{SCC for $\Lambda=0$ and $\Lambda >0$} 
\indent \paragraph{Linear level}
The definitive negative resolution of the fully nonlinear \cref{conj:c0} in \cite{dafermos2017interior} for both $\Lambda=0$ and $\Lambda >0$ was preceded by the negative resolution of the linear \cref{conj:linearana} in \cite{anneboundedness,annekerr,hintz_dS} for $\Lambda =0$ and in \cite{hintz_vasy,annedesitter} for $\Lambda >0$. It was shown that solutions of \eqref{eq:wavekerr} arising from   regular and localized data on a spacelike hypersurface remain continuous and uniformly bounded $|\psi|\leq C$ at the Cauchy horizon  for all subextremal Kerr black hole interiors $(\Lambda =0)$, and very slowly rotating subextremal Kerr--dS black hole interiors---hence disproving \cref{conj:linearana} for $\Lambda=0$ and $\Lambda >0$. (For the extremal case see \cite{gajic1,gajic2} and for the Schwarzschild case see \cite{interiorschwarzschild}.) The key ingredient in showing boundedness at the Cauchy horizon is a sufficiently fast decay (polynomial with rate $v^{-p}$ with $p>1$ for $\Lambda =0$ and exponential for $\Lambda >0$) of linear scalar perturbations along the event horizon. Using suitable energy estimates associated to the red-shift vector field  introduced in \cite{redshift} and the vector field $S=|u|^p \partial_u + |v|^p \partial_v$, this decay is then propagated into the black hole all the way up to the Cauchy horizon $\mathcal{CH}$, where it is sufficient to conclude uniform boundedness. We remark already that this method manifestly fails for asymptotically AdS black holes, where linear scalar perturbations decay merely at a logarithmic rate along the event horizon \cite{gustav,quasimodes}.

While \cref{conj:linearana} is false for $\Lambda =0$, as remarked above, at least the weaker formulation  \cref{conj:linearanachristo} holds true: It was proved that the (non-degenerate) local energy at the Cauchy horizon blows up, $\|\psi\|_{H^1_{\textup{loc}}} = +\infty$, for a generic set of solutions $\psi$ on Reissner--Nordström \cite{Lukreissner} and Kerr \cite{timetranslation} black holes in the full subextremal range of parameters.
 A similar blow-up behavior was obtained for  Kerr  in \cite{luk2016kerr} assuming lower bounds (which were shown later in \cite{hintz2020sharp} to indeed hold generically) on the energy decay rate of a solution along the event horizon.  
These results thus also support the validity of the fully nonlinear \cref{conj:2} for $\Lambda =0$.

On the other hand, in the $\Lambda >0$ case, the exponential convergence of perturbations along the event horizon of a Kerr--de~Sitter black hole is in direct competition with the exponential blue-shift instability near the Cauchy horizon. Thus, the question of the validity of \cref{conj:linearanachristo} becomes even more subtle for $\Lambda >0$ and has received lots of attention in the recent literature. We refer to the conjecture in \cite{nospacelike},    the survey article \cite{reallsurvey},  the recent results \cite{dafermos2018rough,dias2018strong,dias2018strong2,MR3882684,annedesitter} and the works  \cite{Hollands_2020,hollands2020quantum} taking also quantum effects into account.

Another related  result, which will turn out to be important for the paper at hand, is proved in work of the author and Shlapentokh-Rothman \cite{kehle2018scattering}: The main theorem establishes a finite energy scattering theory for solutions of \eqref{eq:wavekerr}  on the interior of Reissner--Nordström. In this scattering theory, a linear isomorphism between the degenerate energy spaces (associated to the Killing field $T$) corresponding to the event and Cauchy horizon is established. The problem reduces to showing uniform bounds for the transmission and reflection coefficients $\mathfrak T(\omega,\ell)$ and $\mathfrak R(\omega,\ell)$ for fixed frequency solutions. Formally, for an incoming wave at the right event horizon $\mathcal{H}_R$, the transmission and reflection coefficients correspond to the amount of $T$-energy scattered to the left and right Cauchy horizon $\Ch_L$ and $\Ch_R$, respectively. Indeed, the theory also carries over to non-zero cosmological constant $\Lambda \neq 0$ \emph{except} for the characteristic frequency ($\omega=0$) associated to $T$, thought of now as the generator of the Cauchy horizon. (Note that these results are compatible with the blow-up of the local energy at the Cauchy horizon \cite{Lukreissner} because of the degeneracy of the $T$-energy.) These insights will turn out to be important for the interior analysis of the present paper, see already \cref{sec:introin}.

\paragraph{Nonlinear level} 
Turning to the nonlinear problem  of  \eqref{eq:Einstein},  Dafermos--Luk proved the full nonlinear $C^0$-stability of the Kerr Cauchy horizon in \cite{dafermos2017interior}. Their work definitively disproves \cref{conj:c0} for $\Lambda=0$ (subject only to the completion of a proof of the nonlinear stability of the Kerr exterior).
Mutatis mutandis, their proof of   $C^0$-stability  also applies to Kerr--de~Sitter Cauchy horizons, where the exterior has been shown to be nonlinearly stable in the very slowly rotating case \cite{hintz2016global}. This unconditionally disproves \cref{conj:c0} for $\Lambda >0$.

Nonlinear inextendibility results at   Cauchy horizons have been proved only in spherical symmetry: For the Einstein--Maxwell-scalar field system, the  Cauchy horizon  is shown to be $C^2$ unstable \cite{dafermos_cauchy_horizon,luk2017strong,luk2017strong2}  for a generic set of spherically symmetric initial data. See also the pioneering work in \cite{internal90,ori91} and the more general results on the Einstein--Maxwell--charged scalar field system in \cite{maxime,van2019breakdown,van2020mass,kehle2021strong}.  This proves the $C^2$-formulation of SCC, and by very recent work \cite{sbierski_c01}, the $C^{0,1}$-formulation  (but not yet \cref{conj:2}) in spherical symmetry. For work in the $\Lambda>0$ case see \cite{costa2,costa3}. The question of any type of nonlinear instability of the Cauchy horizon without symmetry assumptions and the validity of \cref{conj:2} (even restricted to  a neighborhood of Kerr) have yet to be understood.

We shall also mention that for $T^3$-Gowdy spacetimes the $C^2$-formulation of SCC has been shown in \cite{ringa,ringb}. Further, in the context of  Bianchi systems  \cite{bianchi} (which can be formulated as   finite dimensional dynamical systems \cite{bianchia,bianchib}), a $C^2$-formulation of SCC has been shown for  generic  data in \cite{ringstroem1,ringstroem2} for Bianchi A and  in \cite{rademacher1} for Bianchi B systems. In the dynamical system formulation, Baire-genericity has been crucial to the argument, see e.g.\ \cite[Section~1.4]{rademacher1}.

\subsubsection*{SCC for asymptotically AdS spacetimes \texorpdfstring{$\Lambda <0$}{lambda<0}}

As we shall see in the present paper, the situation for asymptotically AdS black holes with $\Lambda<0$ will turn out to be radically different.

\begin{figure}[h]
	\centering
	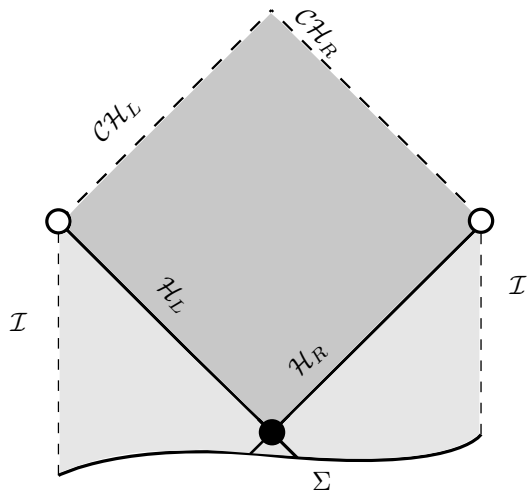
	\caption{Penrose diagram of the maximal Cauchy development of  Kerr--AdS  data on a spacelike surface $\Sigma$ with Dirichlet (reflecting) boundary conditions prescribed on null infinity $\mathcal{I}$.}
	\label{fig:adsintro}
\end{figure} 

First, in view of the lack of global hyperbolicity of asymptotically AdS spacetimes, one needs to specify additional boundary conditions at infinity (at $\mathcal I$) to guarantee well-posedness of \eqref{eq:Einstein} and \eqref{eq:wavekerr}, see \cite{FRIEDRICH1995125,fournodavlos2019initial,wellposed,warnick_massive_wave,gannot_wrochna_2020}. The most natural in this context are reflecting (Dirichlet) boundary conditions \cite{FRIEDRICH1995125}. In what follows we will  assume such Dirichlet boundary conditions. (Refer to \cref{sec:introconjectures} and \cref{sec:intromainthm} later for remarks on more general boundary conditions.)

We first discuss  linear scalar perturbations solving \eqref{eq:wavekerr} arising from data posed on a spacelike hypersurface on  asymptotically AdS black holes. In contrast to $\Lambda \geq 0$, where linear scalar perturbations $\psi$ decay at a polynomial ($\Lambda=0$) and exponential ($\Lambda>0$) rate,  linear scalar perturbations $\psi$ of Kerr--AdS (and Reissner--Nordström--AdS) decay merely at a logarithmic rate on the exterior as proved in \cite{gustav,quasimodes}.\footnote{Recall that we restrict our attention to Kerr black holes below the Hawking--Reall bound \cite{hawking_reall} as otherwise growing modes are shown to exist \cite{dold}.} The origin of this slow decay is a stable trapping phenomenon of high-frequency waves traveling along stably trapped null geodesics which repeatedly bounce off null infinity $\mathcal I$. (Contrast this with the work \cite{dias2019btz,pandya2020rotating,bhattacharjee2020mass} in 2+1 dimensions and refer to \cite{ori2} for a discussion of the Ori model for Reissner--Nordström--AdS.) For 5D asymptotically flat black holes, a similar $\log$-decay result was shown in \cite{benomio}, which also relies on the existence of stably trapped null geodesics.

With the logarithmic decay on the exterior in hand, we first recall from the discussion above that in the $\Lambda \geq 0$ cases  \cref{conj:linearana} is false (and, in fact, so is the fully nonlinear \cref{conj:c0}), yet at least in the $\Lambda =0$ case \cref{conj:linearanachristo} is true (and, hopefully, \cref{conj:2} as well). Indeed, our methods in principle also show \cref{conj:linearanachristo}   for $\Lambda <0$. However, in view of the slower decay in the case $\Lambda <0$, one might suspect a stronger instability  at the Cauchy horizon in this case. This raises the attractive possibility that \cref{conj:c0} and   \cref{conj:linearana} might actually be \emph{true} for $\Lambda <0$, which would give a more definitive resolution to the issue of Strong Cosmic Censorship than the weaker \cref{conj:2} and   \cref{conj:linearanachristo}.

For the Reissner--Nordström--AdS spacetime, which is often considered as a toy model of Kerr--AdS, this question was first taken up in \cite{CKehle2019}. For that case, however, the hopes expressed in the above paragraph were not fulfilled! It was shown in \cite{CKehle2019} that,  \textit{despite the slow decay on the exterior}, all linear scalar perturbations $\psi$ on Reissner--Nordström--AdS (in the full subextremal range) remain  uniformly bounded, ${|\psi|\leq C}$, on the interior and  extend continuously across the Cauchy horizon. Thus, the Reissner--Nordström analog of \cref{conj:linearana} is false. To understand the additional phenomenon which was exploited to prove boundedness, let us decompose a linear scalar perturbation $\psi$ into frequencies $\omega, m,\ell$ associated to the separation of variables. On the exterior, it is the high frequency part (i.e.~$|\omega|,|m|,\ell$ large) of $\psi$ which is exposed to stable trapping and decays slowly, whereas the low frequency part ($|\omega|,|m|,\ell$ small) decays superpolynomially.  
In the interior, however, the main obstruction to boundedness is the interior scattering pole which is located at the characteristic frequency $\omega=0$ with respect to $T$, now thought of as the Killing generator of the Cauchy horizon. (Refer also to the discussion in \cite[Section~3.6]{kehle2018scattering}.)
Thus, for Reissner--Nordström--AdS, the slowly decaying part of $\psi$ is decoupled in frequency space from the part susceptible to the interior scattering pole at $\omega =0$.
(See already \cref{fig:kerradscoupling}.)
The above result on Reissner--Nordström--AdS may suggest that, just as in the cases of $\Lambda \geq 0$, \cref{conj:linearana} is false for $\Lambda <0$, albeit for different reasons.

The present paper on Kerr--AdS, however, provides an unexpected positive resolution of \cref{conj:linearana} for $\Lambda<0$.
We show in \cref{thm:rough} that there exists a set $\mathscr P_{\textup{Blow-up}}$ of dimensionless Kerr--AdS parameters $\mathfrak m := M\sqrt{-\Lambda}$ and  $\mathfrak a:= a\sqrt{-\Lambda}$ which is  \emph{Baire-generic} but \emph{Lebesgue-exceptional}, such that on all Kerr--AdS black hole whose parameters lie in $\mathscr P_{\textup{Blow-up}}$, generic linear scalar perturbations $\psi$ blow up ${|\psi|\to + \infty}$ at the Cauchy horizon. \emphb{Thus, our main result \cref{thm:rough} shows that \cref{conj:linearana} is \underline{true} if Baire-genericity is imposed on the  Kerr--AdS  parameters.}

This set of parameters is defined through a non-Diophantine condition. This condition arises from small divisors originating from a resonance phenomenon  between, on the one hand, specific high frequencies associated to stable trapping on the exterior and, on the other hand,  poles of the interior scattering operator  which are characteristic frequencies with respect to the Killing generator of the Cauchy horizon. 
This resonance phenomenon is possible because the characteristic frequencies of the Cauchy horizon are now the frequencies 
\begin{align}\label{eq:characteristicfrequency}
\omega- \omega_- m =0,\end{align}
 where \begin{align}\omega_-= \frac{a(1+a^2\Lambda/3) }{r_-^2 a^2}\end{align} is the frequency at which the Cauchy horizon rotates.  In contrast to Reissner--Nordström, \eqref{eq:characteristicfrequency} can now be satisfied for frequencies $|\omega|, |m|,\ell$ which are large. It is not all high frequencies, however, but only specific high frequencies, so-called quasimode (on the real axis) or quasinormal mode (in the complex plane) frequencies $(\omega_{n}, m_n, \ell_n)_{n\in \mathbb N}$ which are responsible for the slow decay on the exterior.  (See already \cref{sec:introex}.)
 This resonance phenomenon will lead to small divisors of the form $\frac{1}{\omega_{n} - \omega_- m_n}$. Now, if the specific quasinormal mode frequencies approximate the characteristic frequencies $\omega = \omega_- m$ sufficiently well, i.e.~if $|\omega_n - \omega_- m_n|$ is sufficiently small for infinitely many $n \in \mathbb N$, then we will show that generic linear scalar perturbations $\psi$ of Kerr--AdS blow-up $|\psi|\to \infty$ at the Cauchy horizon. This naturally leads to a non-Diophantine condition on the black hole parameters $\mathfrak m, \mathfrak a$ which, as we will show, holds true for a set of parameters $\mathfrak m, \mathfrak a$ which is  Baire-generic but Lebesgue-exceptional.

The above result is not the last word on \cref{conj:linearana} on Kerr--AdS black holes. We also complement our main result with the conjecture that if the parameters $\mathfrak m, \mathfrak a$ satisfy a complementary Diophantine condition, then the resonance phenomenon outlined above is ``weak'' and linear scalar perturbations $\psi$ remain bounded $|\psi|\leq C $ at the Cauchy horizon. 
 This would then hold for black hole parameters which lie in a set $\mathscr{P}_{\textup{Bounded}}$ which is \emph{Baire-exceptional} but \emph{Lebesgue-generic}.
\emphb{Thus, we expect \cref{conj:linearana} to be \underline{false} if Lebesgue-genericity is imposed on the  Kerr--AdS parameters.}

Since the parameters are dynamic in the full nonlinear \eqref{eq:Einstein}, this suggests that for $\Lambda <0$ the validity of the $C^0$-formulation of Strong Cosmic Censorship (\cref{conj:c0})  may change in a spectacular way according to the notion of genericity imposed.
\paragraph{Instability of asymptotically AdS spacetimes?}
If we accept to interpret the above results as supporting \cref{conj:c0}, they leave determinism  in better shape for $\Lambda <0$ compared to  the $\Lambda \geq 0$ cases. 
However, turning to the fully nonlinear dynamics governed by \eqref{eq:Einstein}, there is yet another scenario which could happen. While Minkowski space ($\Lambda =0$) and de~Sitter space ($\Lambda >0$) have been proved to be nonlinearly stable \cite{Friedrich1986,christodoulou2014global}, Anti-de~Sitter space ($\Lambda <0$) is expected to be nonlinearly unstable with Dirichlet conditions imposed at infinity. This was recently proved by Moschidis \cite{moschidis2017einstein,moschidis2017proof,moschidis2,moschidis3} for appropriate matter models. See also the original conjecture in \cite{dafermosholzegelconj} and the numerical results in \cite{bizon}. Similarly, for Kerr--AdS (or Reissner--Nordström--AdS), the slow logarithmic decay on the linear level proved in \cite{quasimodes} could in fact give rise to nonlinear instabilities in the exterior. (Note that in contrast, nonlinear stability for spherically symmetric perturbations of Schwarzschild--AdS was shown for Einstein--Klein--Gordon systems \cite{schwarzschild_ads_stable}.) If indeed the exterior of Kerr--AdS was nonlinearly unstable, the linear analysis on the level of \eqref{eq:wavekerr} could not serve as a model for \eqref{eq:Einstein} and the question of the validity of Strong Cosmic Censorship would be thrown even more open!

\subsection{Exterior: \texorpdfstring{$\log$}{log}-decay, quasi(normal) modes and semi-classical analysis}\label{sec:introex}
We recall   the result of Holzegel--Smulevici \cite{gustav,quasimodes} that linear scalar perturbations $\psi$ solving \eqref{eq:wavekerr} decay at a sharp inverse logarithmic rate \begin{align}\label{eq:decayofpsi}
|\psi|\leq \frac{C}{\log(t)}
\end{align}
on the Kerr--AdS exterior. (For smooth initial data, the decay in \eqref{eq:decayofpsi} can be slightly improved to $|\psi|\leq \frac{C_n}{\log^n(t)}$ for $n \in \mathbb N$.)
The reason for the slow decay is the stable trapping phenomenon near infinity discussed earlier. One manifestation of this phenomenon is the existence of so-called \emph{quasimodes} and \emph{quasinormal modes} which are ``converging exponentially fast'' to the real axis.
Note already that in the proof of \cref{thm:rough} we will work with quasimode frequencies but we will not make use of a quasinormal mode construction or decomposition. However, quasinormal modes provide perhaps the simplest route to obtain some intuition---paired with the interior analysis in \cref{sec:introin}---for how the relation to Diophantine approximation arises. 
   Our discussion of quasi(normal) modes starts with the   property that \eqref{eq:wavekerr} is formally separable \cite{carter}.
\paragraph{Separation of Variables}
With the fixed-frequency ansatz
\begin{align}\label{eq:qnms}
	\psi = \frac{u(r)}{\sqrt{r^2+a^2}} S_{m\ell}(a\omega, \cos \theta) e^{i m \phi} e^{-i \omega t},
\end{align}
the wave equation \eqref{eq:wavekerr} reduces to a coupled system of o.d.e's (see already \eqref{eq:radial}). The radial o.d.e.\ reads
\begin{align}\label{eq:radialintro}
&	-u''(r^\ast) + \tilde V(r^\ast,\omega,\lambda_{m\ell}) u = 0 
\end{align}
for a rescaled radial variable $r^\ast \in (-\infty, \frac{\pi}{2}l)$ with $r^\ast (r=r_+) = -\infty$, $r^\ast (r=+\infty) = \frac{\pi}{2}l$. The radial o.d.e\ \eqref{eq:radialintro} couples to the angular o.d.e.\ through the potential $\tilde V$ which depends on the eigenvalues $\lambda_{m\ell}(a \omega)$ of the angular o.d.e.\
\begin{align}\label{eq:angintro}	P(a \omega) S_{m\ell}(a\omega, \cos \theta ) = \lambda_{m\ell}(a\omega) S_{m\ell}(a \omega,\cos\theta),
\end{align}
where $P(a\omega)$ is a self-adjoint Sturm--Liouville operator. The radial o.d.e.~\eqref{eq:radialintro} is equipped with suitable boundary conditions at $r^\ast = -\infty$ and $r^\ast = \frac{\pi}{2}l $  which stem from imposing regularity for $\psi$ at the event horizon and Dirichlet boundary conditions at infinity. 
This leads to the concept of a \emph{mode solution} $\psi$ of \eqref{eq:wavekerr} defined to be of the form \eqref{eq:qnms} such that $u$ solves \eqref{eq:radialintro} and $S_{m\ell}$ solves \eqref{eq:angintro} with the appropriate boundary conditions imposed. If such a solution $\psi$ were to exist for $\omega \in \mathbb R$, this would correspond to a time-periodic solution. Such solutions are however incompatible with the fact that all admissible solutions  decay. Nevertheless, there  exist ``almost solutions'' which are time-periodic. This leads us to the concept of
\paragraph{Quasimodes}
In \cite{quasimodes} it was shown that there exists a set of real   frequencies $(\omega_n,  m_n = 0, \ell_n)_{n\in \mathbb N}$ such that the corresponding functions $\psi_n$ ``almost'' solve \eqref{eq:wavekerr} in the sense that they satisfy $\Box_g \psi_n + \frac{2}{3}\Lambda \psi_n  = F_n $ with $|F_n|\lesssim \exp(- c \ell_n)$. These almost-solutions are called \emph{quasimodes} and their existence actually implies that the logarithmic decay of \cite{gustav} is sharp as shown in \cite{quasimodes}. These  quasimode frequencies are equivalently characterized by the condition that the  Wronskian $\mathfrak W[u_{\mathcal H^+}, u_{\infty}]$ of solutions $u_{\mathcal H^+}, u_{\infty}$ of \eqref{eq:radialintro} adapted to the boundary conditions satisfies 
\begin{align}\label{eq:decayofwronskian}
|\mathfrak W[u_{\mathcal H^+}, u_{\infty}] (\omega_n, m_n,\ell_n) | \lesssim e^{-c \ell_n}.
\end{align}
The reason why there exist such quasimodes is that in the high frequency limit, the potential in \eqref{eq:radialintro} admits a region of stable trapping, see already \cref{fig:potential_intro}. Alternatively and intimately related to the above, the existence of  quasimodes can be seen as a consequence of the existence of stably trapped null geodesics on the exterior of asymptotically AdS black holes.

\paragraph{Quasinormal modes}
The Wronskian $\mathfrak W[u_{\mathcal H^+}, u_{\infty}]$ has no real zeros,  $\mathfrak W[u_{\mathcal H^+}, u_{\infty}]\neq0$, however, it might very well have zeros in the lower half-plane  with \mbox{$\operatorname{Im}(\omega)<0$.} These zeros correspond to so-called \emph{quasinormal modes}  i.e.\  solutions of the form \eqref{eq:qnms} which decay in time at an exponential rate. Note that quasinormal modes do not have finite energy on $\{t =\textup{const.}\}$-slices (in particular they have infinite energy on $\Sigma_0= \{ t=0\}$). However, they have finite energy for $\{t^\ast= \textup{const.} \}$-slices, where $t^\ast$ is a suitable  time coordinate which extends regularly to  the event horizon $\mathcal H_R$, see already \eqref{eq:kerrstar}. For a more precise definition, construction and a more detailed discussion of quasinormal modes in general we refer to   \cite{gajic2019quasinormal}. Turning back to Kerr--AdS, we note that the bound \eqref{eq:decayofwronskian}  implies the existence of zeros of $\mathfrak W[u_{\mathcal H^+}, u_{\infty}]$ exponentially close to the real axis as shown in \cite{quasinormal_gannot}, see also \cite{quasi_warnick}. More precisely, it was shown that there exist axisymmetric quasinormal modes  with frequencies $m=0$ and $(\omega,\ell) = (\omega_n,\ell_n)_{n\in \mathbb N}$ satisfying
\begin{align}
 & c \ell_n \leq |\operatorname{Re}(\omega_n)| \leq C \ell_n,
\\&  0< - \operatorname{Im}(\omega_n) \leq  C  \exp(- c \ell_n).
\end{align}
While the previous results were proved in axisymmetry to simplify the analysis, in principle, they also extend to non-axisymmetric solutions as remarked in \cite{quasinormal_gannot}. 
\paragraph{Semi-classical heuristics for distribution of quasimodes and quasinormal modes}
We first turn to the heuristic distribution of the quasimode frequencies in the semi-classical (high frequency) limit. For   large $|m|$,  $m \in \mathbb Z$, $\ell \geq |m|$, we  expect a quasimode with frequencies $m,\ell, \omega$ to exist, if the potential $\tilde V(r^\ast, \omega, m, \lambda_{m\ell}(a \omega))$ appearing in the radial o.d.e.\ \eqref{eq:radialintro} satisfies (see \cref{fig:potential_intro})
\begin{itemize}
\item $\tilde V(r^\ast, \omega, m,  \lambda_{m\ell}(a \omega) ) >0 $ for $r_1^\ast < r^\ast < r_2^\ast$,
\item $\tilde V(r^\ast,\omega, m, \lambda_{m\ell}(a \omega)<0$ for $ r_2^\ast< r^\ast\leq \frac{\pi}{2}l$.
\end{itemize}
 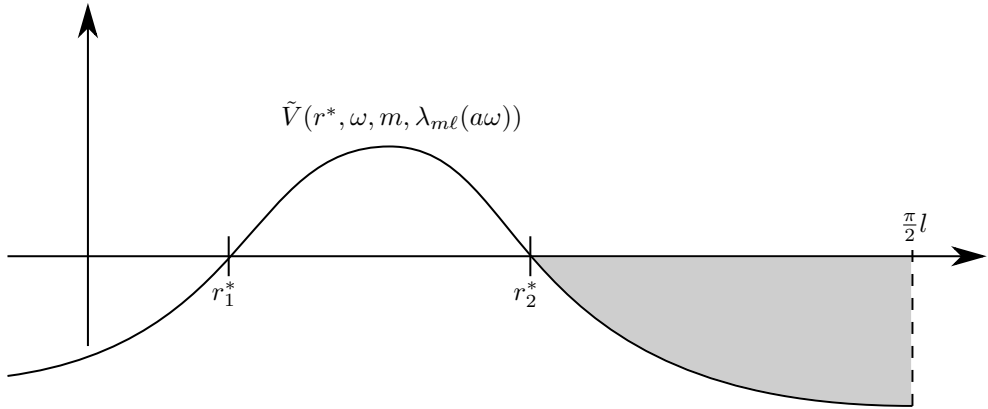
\begin{figure}[!ht]
\centering
\begingroup%
  \makeatletter%
  \providecommand\color[2][]{%
    \errmessage{(Inkscape) Color is used for the text in Inkscape, but the package 'color.sty' is not loaded}%
    \renewcommand\color[2][]{}%
  }%
  \providecommand\transparent[1]{%
    \errmessage{(Inkscape) Transparency is used (non-zero) for the text in Inkscape, but the package 'transparent.sty' is not loaded}%
    \renewcommand\transparent[1]{}%
  }%
  \providecommand\rotatebox[2]{#2}%
  \newcommand*\fsize{\dimexpr\f@size pt\relax}%
  \newcommand*\lineheight[1]{\fontsize{\fsize}{#1\fsize}\selectfont}%
  \ifx\svgwidth\undefined%
    \setlength{\unitlength}{365.25610054bp}%
    \ifx\svgscale\undefined%
      \relax%
    \else%
      \setlength{\unitlength}{\unitlength * \real{\svgscale}}%
    \fi%
  \else%
    \setlength{\unitlength}{\svgwidth}%
  \fi%
  \global\let\svgwidth\undefined%
  \global\let\svgscale\undefined%
  \makeatother%
  \begin{picture}(1,0.41569001)%
    \lineheight{1}%
    \setlength\tabcolsep{0pt}%
    \put(0,0){\includegraphics[width=\unitlength,page=1]{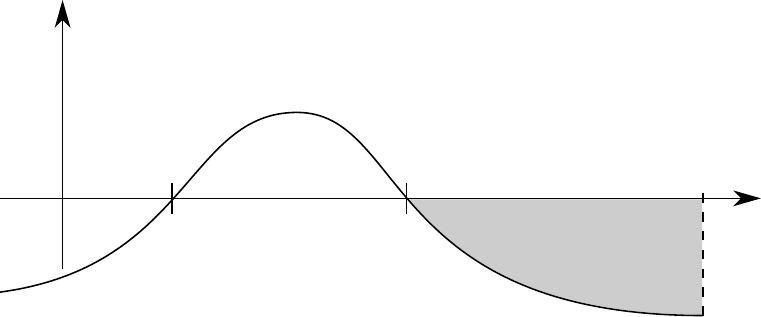}}%
    \put(0.5169779,0.10998776){\color[rgb]{0,0,0}\makebox(0,0)[lt]{\lineheight{1.25}\smash{\begin{tabular}[t]{l}$r_2^\ast$\end{tabular}}}}%
    \put(0.20895997,0.11012886){\color[rgb]{0,0,0}\makebox(0,0)[lt]{\lineheight{1.25}\smash{\begin{tabular}[t]{l}$r_1^\ast$\end{tabular}}}}%
    \put(0.27959915,0.28973286){\color[rgb]{0,0,0}\makebox(0,0)[lt]{\lineheight{1.25}\smash{\begin{tabular}[t]{l}$\tilde V(r^\ast,\omega,m,\lambda_{m\ell}(a \omega))$\end{tabular}}}}%
    \put(0.91151512,0.18038846){\color[rgb]{0,0,0}\makebox(0,0)[lt]{\lineheight{1.25}\smash{\begin{tabular}[t]{l}$\frac \pi 2l$\end{tabular}}}}%
  \end{picture}%
\endgroup%

\caption{Potential $\tilde V$ with frequency $\omega,m, \ell$ for which we expect quasimodes. The gray area is a suitable  projection of the phase space volume.}
\label{fig:potential_intro}
\end{figure}
Note that the conditions above are satisfied for a range of $\omega$ of the form $c \ell <  |\omega|< C \ell$. In addition, for a quasimode to exist, the potential has to satisfy the Bohr--Sommerfeld quantization condition (see e.g.\  \cite[Chapter VII]{MR0093319}). In our case this means that the phase space volume
\begin{align}\frac{1}{2\pi}\operatorname{vol}\left\{ (r^\ast,\xi)\colon \xi^2  +  \tilde  V(r^\ast, \omega, m, \lambda_{m\ell}(a \omega) ) < 0 , r^\ast > r_2^\ast\right\}\end{align} should be an integer multiple modulo the Maslov index up to an exponentially small error.

Thus, at least heuristically, we expect that for given but large $|m|, \ell \geq |m|$, there exist $N(m,\ell)\sim \ell $ intervals  of quasimodes with midpoint $\omega \sim \ell$ and length $e^{-c\ell}$.  While quasimode frequencies are defined through an open condition (c.f.\ \eqref{eq:decayofwronskian}), quasinormal mode frequencies will be discrete and in an exponentially small neighborhood of quasimodes. Thus, we expect the quasinormal mode frequencies to be distributed as 
\begin{align}
\begin{split}
& c \ell\leq	|\operatorname{Re}(\omega_{m\ell n})| \leq C \ell, \\
& 0 < - \operatorname{Im}(\omega_{m \ell n}) \leq C \exp(- c  \ell). 
\end{split}
\label{eq:decayinexpl}
\end{align}
Refer to \cref{fig:distributionofqnms} for a visualization of the expected distribution of quasimodes and quasinormal modes.
\begin{figure}[h!]
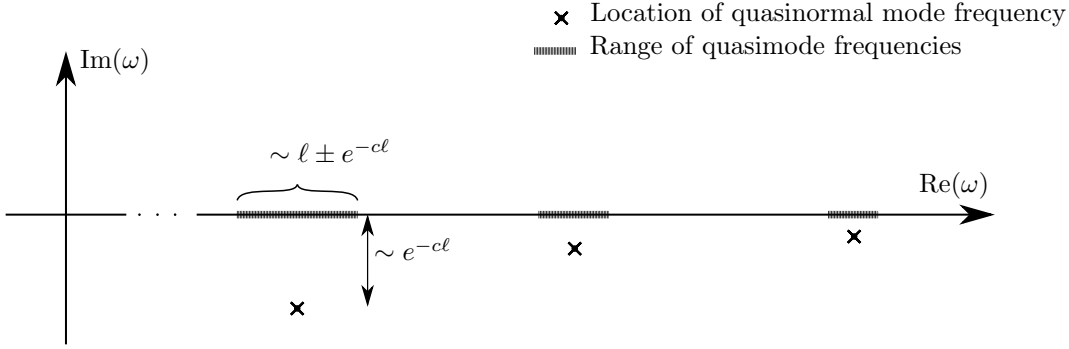
\caption{Quasimodes and quasinormal modes frequencies for large $|m|\sim \ell$.}
\label{fig:distributionofqnms}
\end{figure}

For our heuristic analysis we will now {consider} a solution $\psi$ of \eqref{eq:wavekerr} which consists of an infinite sum of weighted quasinormal modes   (Warning: A general solution cannot be written as a sum of quasinormal modes.)
\begin{align}\label{eq:intropsisum}
\psi(r,t,\theta,\phi) = \sum_{m \in \mathbb Z} \sum_{\ell \geq |m|} \sum_{n=1}^{N(m,\ell)} \tilde a (m,\ell,n) \frac{u(r,\omega_{m\ell n},m,\ell)}{\sqrt{r^2+a^2}} e^{-i\omega_{m \ell n} t} S_{m\ell}(a \omega_{m \ell n}, \cos \theta) e^{i m  \phi},
\end{align}
where we require that the weights $\tilde a(m,\ell,n)$ have   superpolynomial decay. This ensures that the initial data (posed on a $\{ t^\ast = \textup{const.}\}$-slice) are smooth where we assume that each individual quasinormal mode is suitably normalized.\footnote{By a domain of dependence argument one can then produce a solution arising from smooth data on $\Sigma_0$.} 
Restricting this solution   $\psi$  to the event horizon yields
\begin{align}\label{eq:intropsihorizon}
\psi\restriction_{\mathcal H}(v,\theta,\tilde \phi_+) = \sum_{m \in \mathbb Z} \sum_{\ell \geq |m|} \sum_{n=1}^{N(m,\ell)} a(m,\ell,n) e^{-i\omega_{m \ell n} v} S_{m\ell}(a \omega_{m \ell n}, \cos \theta) e^{i m \tilde \phi_+} 
\end{align}
for new coefficients  $a(m,\ell,n)$ which satisfy $|a(m,\ell,n)| \sim  {|\tilde a(m,\ell,n) u(r_+,\omega_{m\ell n},m,\ell)|}$. Now, note that the radial part of the quasinormal mode $|u(r,\omega_{m\ell n},m,\ell)|$ will be localized in the region of stable trapping, i.e.\ in the region $\{r^\ast \geq r_2^\ast \}$ of \cref{fig:potential_intro}. From semi-classical heuristics, we expect that only an exponentially damped proportion ``tunnels'' from the region of stable trapping through the barrier to event horizon at $r=r_+$. More precisely, the damping factor of the exponent of $|u(r_+,\omega_{m\ell n},m,\ell)|$ is expected to be proportional to
\begin{align}\int_{r_1^\ast}^{r_2^\ast} \sqrt{\tilde  V(r^\ast,  \omega_{R}, m, \lambda_{m\ell}(a  \omega_{R}) )} \d r^\ast \sim \ell. \end{align} 
Now, for any choice of superpolynomially decaying (or polynomially decaying)  weights $\tilde a (m,\ell,n)$, the new coefficients $a(m,\ell,n)$ decay exponentially  
\begin{align}\label{eq:weightsamln}
|a(m,\ell,n)| \lesssim  \exp(- C  \ell).
\end{align}
Thus, choosing coefficients $\tilde a(m,\ell,n)$ now corresponds to choosing coefficients $a(m,\ell,n)$ satisfying \eqref{eq:weightsamln} and vice versa.
 In view of this, instead of choosing $\tilde a(m,\ell,n)$, we will go forward in our heuristic discussion by choosing  coefficients $a(m,\ell,n)$ satisfying \eqref{eq:weightsamln}. The goal is to choose such coefficients such that $\psi$ blows up at the Cauchy horizon!

\subsection{Interior: Scattering from event to Cauchy horizon}
\label{sec:introin}
\begin{figure}[h]
	\centering
	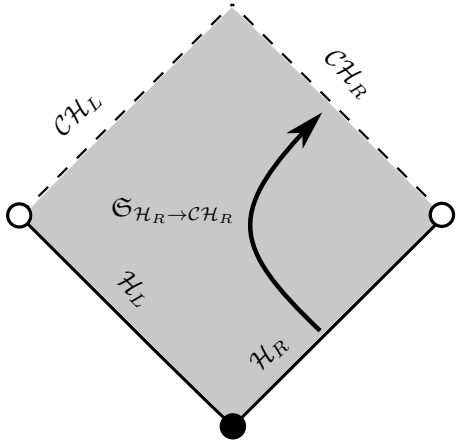
	\caption{Interior scattering $\mathfrak S_{\mathcal{H}_R \to \mathcal{CH}_R}$ from event horizon $\mathcal H_R$ to Cauchy horizon $\mathcal{CH}_R$}
	\label{fig:intscattering}
\end{figure}
We now turn to the interior problem.
We will view some aspects of the propagation of $\psi$ from the event horizon to the Cauchy horizon as a scattering problem as visualized in \cref{fig:intscattering}.  We refer to \cite{kehle2018scattering} for a detailed discussion of the scattering problem on black hole interiors. Unlike in \cite{kehle2018scattering}, we will not develop a full scattering theory for Kerr--AdS, but rather make use of  a key insight from \cite{kehle2018scattering} adapted to our context.
Recall from  \cite[Proposition 6.2]{kehle2018scattering} that on  Reissner--Nordström--AdS, the scattering operator $\mathfrak{S}_{\mathcal H_R \to \mathcal{CH}_R}$ in the interior   has a pole at the frequency $\omega=0$, which is the characteristic   frequency associated to the Killing generator of the Cauchy horizon $T$.  In the present case for Kerr--AdS, it is the vector field $K_- := T + \omega_- \Phi$ which generates the Cauchy horizon and thus the characteristic frequencies are those satisfying $\omega - \omega_- m =0$. For fixed frequency scattering, this means that the reflection coefficient $\mathfrak R$ (i.e.\ the fixed frequency scattering operator from $\mathcal H_R$ to $\mathcal{CH}_R$) has a pole at $\omega - \omega_- m =0$ such that  $\mathfrak R$ is of the form
\begin{align}\label{eq:reflectioncoefficient}
	\mathfrak R = \frac{\mathfrak r(\omega,m,\ell)}{\omega - \omega_- m },
\end{align}
where $\mathfrak r(\omega=\omega_- m,m,\ell) \neq 0$. 

There is a natural solution $\psi$ defined in the black hole interior by continuing each quasinormal mode appearing in \eqref{eq:intropsisum} into the interior. This solution is again smooth across $\mathcal H_R$ and thus can be view as a solution arising from smooth data on a spacelike hypersurface which coincides with  $\{ t^\ast =0\}$ on the exterior.
Let us assume for a moment that the fixed frequency scattering theory also carries over to complex frequencies and that we can analytically continue the reflection coefficient  $\mathfrak R$ to the complex plane. We then expect that the continued solution $\psi$ at the Cauchy horizon can be obtained by   multiplying each individual coefficient   $\psi\restriction_{\mathcal H}$ as in \eqref{eq:intropsihorizon} with the  reflection coefficient $\mathfrak R(\omega_{m\ell n}, m,\ell)$. Moreover, neglecting $\mathfrak r(\omega_{m\ell n},m,\ell)$ which is expected to be suitably bounded from below and above, and taking the $L^2(\mathbb S^2)$-norm of the $\{ u =\text{const} \}$-spheres on the Cauchy horizon $\mathcal{CH}_R$,   formally yields 
\begin{align}\label{eq:diphantineintro}
\|	\psi\restriction_{\mathcal{CH}} \|^2_{L^2(\mathbb S^2)} \sim  \sum_{m \in \mathbb Z} \sum_{\ell \geq |m|}\sum_{n=1}^{N(m,\ell)}   \frac{|a(m,\ell,n)|^2}{|\omega_{m\ell n} - \omega_- m|^2},
\end{align}
where we recall that $a(m,\ell,n)$ decay exponentially as in \eqref{eq:weightsamln}. 
\begin{figure}[!h]
	\centering
	\hspace{-2cm}\scalebox{.9}{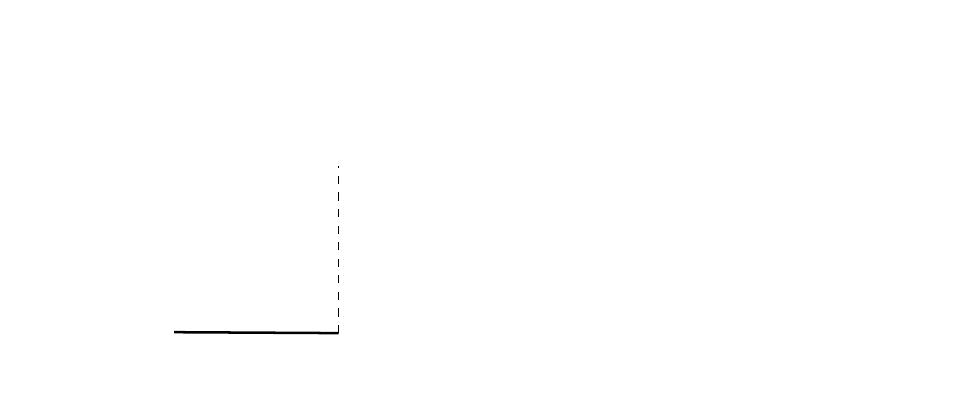}	
	
	\vspace{1cm}
	\hspace{-2cm}\scalebox{.9}{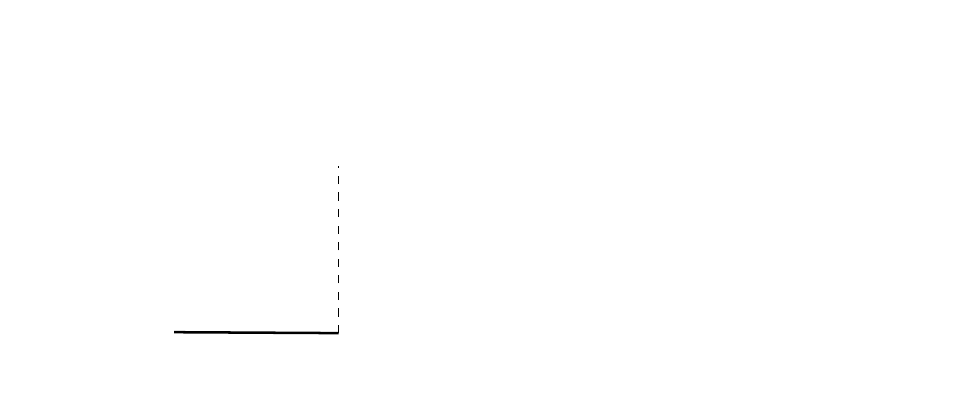}
	\caption{Reissner--Nordström--AdS (top): High frequency stably trapped perturbations are \emphb{decoupled} in frequency space from interior scattering pole at characteristic frequency   $\omega =0$ (w.r.t.\ $T$).\\ Kerr--AdS (bottom): High frequency stably trapped perturbations  \emphb{couple} in frequency space to interior scattering poles at characteristic frequency    $\omega - \omega_- m=0$ (w.r.t.\ $K_-$). }
	\label{fig:kerradscoupling}
\end{figure}
In order to resolve \cref{conj:linearana}, we have to determine whether for all coefficients $a(m,\ell,n)$ satisfying \eqref{eq:weightsamln}, the sum \eqref{eq:diphantineintro} remains uniformly bounded, or whether, for some choice of $a(m,\ell,n)$  satisfying \eqref{eq:weightsamln}, this sum is infinite. Before we address this issue in the next paragraph, we refer to \cref{fig:kerradscoupling} for an illustration of the main difference between the behavior of linear scalar perturbations on Reissner--Nordström--AdS and Kerr--AdS.

\subsection{Small divisors and relation to Diophantine approximation}\label{sec:introdioph}

The convergence of \eqref{eq:diphantineintro} is an example \emph{par excellence} of  a \emphb{small divisors problem}.
Indeed, if $|\omega_{m\ell n} - \omega_- m|$  is   exponentially small in $m,\ell,n$, the sum in \eqref{eq:diphantineintro} is infinite for suitable (in fact generic) $  a(m,\ell,n)$ satisfying \eqref{eq:weightsamln}. More precisely, for the sum in \eqref{eq:diphantineintro} to be infinite for some choice of $  a(m,\ell,n)$, in view of \eqref{eq:weightsamln}, it suffices that there exist infinitely many $(m,\ell,n)$ such that $|\omega_{m\ell n} - \omega_- m|$ decays exponentially. Thus, we conjecture blow-up if 
\begin{align}\label{eq:cond1}
		|\omega_{m\ell n} - \omega_- m| \leq  c' \exp(- C \ell) \text{ for infinitely many admissible } (m,\ell,n)  ,
\end{align}
where $(m,\ell,n)$ are admissible if  $m \in \mathbb Z , \ell \geq |m |,n = 1, \dots, N(m,\ell)$.

Conditions like \eqref{eq:cond1} lie at the heart of \emphb{Diophantine approximation}. Indeed, semi-classical heuristics as in \eqref{eq:decayinexpl} suggest that  $\operatorname{Re}(\omega_{m\ell n})$ are  uniformly distributed  and we assume for a moment that $\operatorname{Re}(\omega_{m\ell n}) = \tilde c(\ell + \frac{n}{\ell})$ for $n = 0, 1, \dots , \ell$ for a constant $\tilde c=\tilde c(M,a,\Lambda)$ and that $|\operatorname{Im}(\omega_{m\ell n})|\lesssim e^{-c\ell}$. For the sake of the purely heuristic argument assume also for a moment that the dimensionless constants $C$    as well as $c'/\tilde c$  are actually $C=c'/\tilde c=1$. Then, the ratio $r(\mathfrak m, \mathfrak a):= \frac{\omega_-}{\tilde c}$, which is dimensionless and only depends on the dimensionless black hole parameters $(\mathfrak m = M \sqrt{-\Lambda}, \mathfrak a= a \sqrt{-\Lambda})$, has to satisfy the non-Diophantine condition
\begin{align}
\label{eq:diophantine}
r( \mathfrak m, \mathfrak a)\in\mathscr{R} :=  	\left\{ x \in \mathbb R \colon \left| \frac{\ell + \frac{n}{\ell}}{m}  - x\right| < \exp(-   \ell) \text{ for } \infty\text{-many  admissible } (m,\ell,n) \right\}.
 \end{align}
Thus, from our heuristic derivation, it is natural to conjecture that linear perturbations   \emphb{blow up} at the Cauchy horizon of Kerr--AdS with mass $M = \mathfrak m/\sqrt{-\Lambda}$ and angular momentum $a=\mathfrak a/\sqrt{-\Lambda}$ if the ratio $r=r(\mathfrak m, \mathfrak a)$ satisfies the non-Diophantine condition \eqref{eq:diophantine}.  At this point it worth emphasizing that the above arguments are \emph{merely} heuristics and by no means can be turned into a proof easily. In particular, our proof does not use a quasinormal modes approach as the previous heuristics and the  non-Diophantine condition (see already \cref{sec:pblowup}) is significantly more technical (refer also to the discussion later in \cref{sec:breifdesc}).

\paragraph{The set \texorpdfstring{$\mathscr R$}{R} is Baire-generic and Lebesgue-exceptional}
The set $\mathscr{R}$ can be written as a $\limsup$ set as \begin{align}	
	\mathscr{R} = \bigcap_{m_0\in \mathbb N} \bigcup_{|m|\geq m_0} \bigcup_{\ell \geq |m|} \bigcup_{0\leq n \leq \ell} 	\left\{ x\in \mathbb R \colon \left| \frac{\ell + \frac{n}{\ell}}{m}  - x\right| < \exp(-\ell)   \right\}.
\end{align}
It is a countable intersection of open and dense sets such that $\mathscr{R}$ is of second category in view of Baire's theorem \cite{zbMATH02668286}. Thus, the set $\mathscr{R}$ is generic from a topological point of view, which we refer to as   \emphb{Baire-generic}.
On the other hand, from a measure-theoretical point of view, the set $\mathscr{R}$ is  exceptional. Indeed,  an application of the Borel--Cantelli lemma shows that the Lebesgue measure of $\mathscr{R}$ vanishes. This is the easy part of the famous theorem by Khintchine \cite{khintchine} stating that for a decreasing function $\phi$, the set
\begin{align}\label{eq:wphi}
 W[\phi]:=	\left\{ x\in \mathbb R\colon \left|x- \frac pq \right| <\frac{ \phi(q)}{q} \text{ for } \infty\text{-many rationals } \frac{p}{q} \right\}
\end{align} 
has full Lebesgue measure if and only if the sum $\sum_q \phi(q)$ diverges.  Thus, $\mathscr{R}$ is \emphb{Lebesgue-exceptional}.

\paragraph{More refined measure: The Hausdorff and packing measures}
 This naturally leads us to consider the more refined versions of measure, the so-called \emphb{Hausdorff}  and \emphb{packing measures} $H^f$, $P^f$ together with their associated \emphb{dimensions} $\dim_H$, $\dim_P$ (see \cref{sec:fractalmeas}). The Hausdorff and packing measure generalize the Lebesgue measure to non-integers. In a certain sense, they can be considered to be dual to each-other: The Hausdorff measure approximates and measures sets by a most economical covering, whereas the packing measure packs as many disjoint balls with centers inside the set. While for all sufficiently nice sets these notions agree, they indeed turn out to give different results in our context. 

We first consider the Hausdorff dimension. A version of the Borell--Cantelli lemma (more precisely the Hausdorff--Cantelli lemma) and using the natural cover for $\mathscr{R}$ shows that the set $\mathscr{R}$ is of  \emphb{Hausdorff dimension zero}. This again can be seen as a consequence of a theorem   going back to Jarn{\'i}k \cite{jarnik} and Besicovitch \cite{besi} which states the set $W[\phi]$ as in \eqref{eq:wphi} has Hausdorff measure   
\begin{align}
	H^s(W[\phi]) = \begin{cases}
	 0 & \text{ if } \sum_{q} q^{1-s} \psi^s(q) <  \infty \\
	 + \infty & \text{ if } \sum_{q} q^{1-s} \psi^s(q) = \infty 
	\end{cases}
\end{align}for  $s\in (0,1)$. However, measuring also logarithmic scales, i.e.\ considering the Hausdorff measure $H^f$ for $f = \log^t(r)$ for some $t>0$, it follows that the set $\mathscr{R}$ is of \emphb{logarithmic generalized Hausdorff dimension}.
On the other hand, using the dual notion of packing dimension, it turns out that $\mathscr{R}$ has \emphb{full packing dimension}, a consequence of the fact that it is a set of second category (Baire-generic) \cite{MR3236784}.
\paragraph{Summary of properties of \texorpdfstring{$\mathscr R$}{R}}
 To summarize, we obtain that
\begin{itemize}
	\item $\mathscr{R}$ is Baire-generic,
		\item $\mathscr{R}$ is Lebesgue-exceptional,
			\item $\mathscr{R}$ has zero Hausdorff dimension $\dim_H (\mathscr{R}) = 0$,
			\item $\mathscr{R}$ is of logarithmic generalized Hausdorff dimension,
				\item $\mathscr{R}$ has full packing dimension $\dim_P (\mathscr{R}) = 1$.
\end{itemize}
The above heuristics will enter in our revised conjectures, \cref{con:2a} and \cref{con:2b}, which transcend \cref{conj:linearana} and \cref{conj:linearanachristo} for $\Lambda <0$. Before we  turn to that in  \cref{sec:introconjectures},  we briefly discuss other aspects of PDEs and dynamical systems for which Diophantine approximation plays a crucial role.
\paragraph{Small divisors problems and Diophantine approximation in dynamical systems and PDEs} 
 Most prominently, Diophantine approximation and the small divisors problem are intimately tied to the problem of the stability of the solar system \cite{moser} and more generally, the stability of  Hamiltonian systems in classical mechanics. 
This   stability problem was partially resolved with the celebrated KAM theorem \cite{kolmogorov,arnold,moser2} which roughly states that Lebesgue-generic perturbations of integrable Hamiltonian systems  lead to quasiperiodic orbits.  
The small divisors problem and Diophantine approximation are ubiquitous in modern mathematics and  arise naturally in many other aspects of PDEs and dynamical systems. We refer to \cite{peroidicwave,peroidicwave2}   for a connection to   wave equations with periodic boundary conditions and to the more general results in \cite{MR1355946} as well as the monograph \cite{illposed}. There is also a vast recent literature on the construction of \mbox{(quasi-)}periodic orbits to nonlinear wave equations; we refer to \cite{waterwave,wangnlw,MR2967117}, the overview article \cite{berti} and the monograph \cite{berti2019quasi} and references therein for further details.  Similar results have been obtained for the Schrödinger equation on the torus in \cite{MR1995764,quantumparticletime,MR3910065,MR3486416}. Further applications of Diophantine approximation include the characterization of homeomorphisms on $\mathbb S^1$ by the Diophantine properties of their rotation numbers or analyzing the Lyapunov stability of vector fields, see the discussion in \cite{MR2109001}.  

\subsection{\texorpdfstring{\cref{con:2a}}{Conjecture 5} and \texorpdfstring{\cref{con:2b}}{Conjecture 6} replace \texorpdfstring{\cref{conj:linearana}}{Conjecture 3} and \texorpdfstring{\cref{conj:linearanachristo}}{Conjecture 4} for Kerr--AdS}
\label{sec:introconjectures}
With the above heuristics in hand, we now transcend \cref{conj:linearana} and \cref{conj:linearanachristo}  for subextremal Kerr--AdS black holes with parameters below the Hawking--Reall bound  in terms of the following two conjectures.  We denote the set of all such parameters with $\mathscr P$, see already \eqref{eq:setofallparameters}.
\begin{conjecture}\label{con:2a}
	There exists a set $\mathscr{P}_{\textup{Blow-up}} \subset \mathscr P$ of dimensionless Kerr--AdS parameters  mass    $\mathfrak m = {M}{\sqrt{-\Lambda}}$ and angular momentum $\mathfrak a =  {a}{\sqrt{-\Lambda}}$ with the following properties
	\begin{itemize}
		\item $\mathscr P_{\textup{Blow-up}}$ is \textbf{Baire-generic} (of second category),
		\item $\mathscr P_{\textup{Blow-up}}$ is \textbf{Lebesgue-exceptional} (zero Lebesgue measure),
	\end{itemize}
and such that for every   Kerr--AdS black hole with mass $M = \mathfrak m / \sqrt{-\Lambda}$ and specific angular momentum $ a = \mathfrak a /\sqrt{-\Lambda}$, where  $(\mathfrak m, \mathfrak a) \in \mathscr{P}_{\textup{Blow-up}}$, there exists a solution $\psi$ to \eqref{eq:wavekerr}, which arises from smooth and compactly supported initial data $(\psi_0,\psi_1)$ on a suitable spacelike hypersurface with Dirichlet boundary conditions at infinity, and which \textbf{blows up}   \begin{align} \label{eq:blowupl2}
	\| \psi\|_{L^2(\mathbb S^2)}(u,r)  \xrightarrow[r\to r_-]{} +\infty
	\end{align} 	at the Cauchy horizon for every $u\in \mathbb R$. 
	\end{conjecture}
\begin{rmk}\label{rmk:linearisenough}
If there exist initial data $(\psi_0,\psi_1)$ leading to a solution  $\psi$ which blows up as in \eqref{eq:blowupl2}, this then shows that initial data $(\tilde \psi_0, \tilde \psi_1)$ for which the arising solution does not blow up are exceptional in the sense that they obey the following  co-dimension 1 property: The  solution arising from the perturbed data $(\tilde \psi_0+ c  \psi_0, \tilde \psi_1 + c  \psi_1)$ blows up for each $c\in \mathbb R\setminus \{0\}$. This is analogous to the notion of genericity used by Christodoulou in his proof of weak cosmic censorship for the spherically symmetric Einstein-scalar-field system \cite{chr2,chr1}. Thus, \cref{con:2a} gives a formulation of \cref{conj:linearana}.  We note already that we will actually formulate in \cref{rmk:othernotionsofgenericity} another more refined genericity condition for the set of initial data leading to solutions which blow up as in \eqref{eq:blowupl2}.  
\end{rmk}
\begin{rmk}\label{rmk:blowupinamplitue}
Note that in \cref{con:2a}   we have replaced the statement of blow-up in amplitude from \cref{conj:linearana} with a statement about the blow-up of the $L^2(\mathbb S^2)$-norm on the sphere. Indeed, the blow-up of the $L^2(\mathbb S^2)$-norm in \cref{con:2a}, if true, implies that $\|\psi\|_{L^\infty(\mathbb S^2)}(u,r)\to +\infty$ as $r\to r_-$.   In this sense, if \cref{con:2a} is true, the amplitude  also blows up. It is however an interesting and open question whether one may actually replace the $L^\infty(\mathbb S^2)$ blow-up statement in  \eqref{eq:blowupl2} with the pointwise blow-up \begin{align}\lim_{r\to r_-}|\psi(u,r,\theta,\phi^\ast_-)|\to +\infty\end{align} for every $(\theta, \phi^\ast_-) \in \mathbb S^2$. One may even speculate about the geometry of the set of  $(\theta,  \phi_-^\ast) \in \mathbb S^2$ for which  pointwise blow-up holds.
It appears that ultimately one has to quantitatively understand the nodal domains associated to the generalized spheroidal harmonics $S_{m\ell}(a \omega_- m, \cos \theta)$ at the interior scattering poles.
\end{rmk}
\begin{rmk}\label{rmk:conjhdpacking}
Moreover, we  conjecture that the set $\mathscr P_{\textup{Blow-up}}$ has \begin{itemize}
	\item Hausdorff dimension $\dim_H (\mathscr P_{\textup{Blow-up}}) = 1 $,
\item generalized Hausdorff dimension $\dim_{gH} (\mathscr P_{\textup{Blow-up}}) = 1+\log $,
\item  full packing dimension $\dim_P (\mathscr P_{\textup{Blow-up}}) = 2$.
\end{itemize}   
\end{rmk}
Moreover, in view of our discussion we additionally conjecture
 \begin{conjecture}\label{con:2b}\hspace{-4pt}\textup{\textbf{(A)}}
There exists a set $\mathscr{P}_{\textup{Bounded}} \subset \mathscr P$ of dimensionless Kerr--AdS parameters  mass    $\mathfrak m = {M}{\sqrt{-\Lambda}}$ and angular momentum $\mathfrak a =  {a}{\sqrt{-\Lambda}}$ with the following properties
\begin{itemize}
	\item $\mathscr P_{\textup{Bounded}}$ is \textbf{Baire-exceptional} (of first category),
	\item $\mathscr P_{\textup{Bounded}}$ is \textbf{Lebesgue-generic} (full Lebesgue measure),
\end{itemize}
and such that for every  Kerr--AdS black hole with mass $M = \mathfrak m / \sqrt{-\Lambda}$ and specific angular momentum $ a = \mathfrak a /\sqrt{-\Lambda}$, where  $(\mathfrak m, \mathfrak a) \in \mathscr{P}_{\textup{Bounded}}$, all solutions $\psi$ to \eqref{eq:wavekerr}, which arise from smooth and compactly supported initial data $(\psi_0,\psi_1)$ on a spacelike hypersurfaces with Dirichlet boundary conditions at infinity,  \textbf{remain uniformly bounded}   
\begin{align}\label{eq:constantinpbounded}
 	|\psi| \leq C(\mathfrak m, \mathfrak a)  D(\psi_0,\psi_1)
 	\end{align}
 at the Cauchy horizon. Here, $D(\psi_0,\psi_1)$ is a   (higher-order) energy of the initial data and $C(\mathfrak m, \mathfrak a)$ is a constant depending on $\mathfrak m$ and $\mathfrak a$.
 
\noindent  \textup{\textbf{(B)}}
For \textbf{all}  Kerr--AdS black holes with parameters in $\mathscr P$, there exists a solution $\psi$ to \eqref{eq:wavekerr}, which arises from smooth and compactly supported initial data on a spacelike hypersurfaces with Dirichlet boundary conditions at infinity and which \textbf{blows up in energy} 
	\begin{align}
	\| \psi\|_{H^1_\textup{loc}} = + \infty
	\end{align}	at the Cauchy horizon.
\end{conjecture}
\begin{rmk}
In view of \cref{con:2a}, we expect that the constant $C(\mathfrak m, \mathfrak a)$ appearing in \eqref{eq:constantinpbounded} to be unbounded on any open set of parameters  in the sense that
\begin{align}\sup_{(\mathfrak m, \mathfrak a) \in \mathcal U\cap \mathscr{P}_{\textup{Bounded}}} C(\mathfrak m, \mathfrak a) = + \infty\end{align}
for any non-empty open $\mathcal U \subset \mathscr P$.
\end{rmk}
\paragraph{More general boundary conditions and Klein--Gordon masses}
The above conjectures are both stated for Dirichlet conditions at infinity.  Neumann conditions are also natural to consider and indeed well-posedness was proved in \cite{warnick_massive_wave,warnick_boundedness_and_growth}. For Neumann conditions we also expect the same behavior as for the case of Dirichlet boundary conditions. For other more general conditions, it may be the case that linear waves  grow exponentially (as for suitable Robin boundary conditions \cite{warnick_boundedness_and_growth}) or on the other hand  even decay superpolynomially as is the case for purely outgoing conditions \cite{MR4072237}. For even more general  boundary conditions, even well-posedness may be open.

In this paper we have focused on scalar perturbations satisfying \eqref{eq:wavekerr}. In particular, the choice of the Klein--Gordon mass parameter $\mu = \frac 23 \Lambda$  (``conformal coupling'') is the most natural as it arises from the linear scalar analog of \eqref{eq:Einstein} and also remains regular at infinity. However, in certain situations it may also be interesting to consider more general Klein--Gordon masses $\mu$ satisfying the Breitenlohner--Friedman  \cite{breitenlohner} bound  $\mu > \frac{3}{4}\Lambda$. We also expect \cref{con:2a} and \cref{con:2b} to hold for Klein--Gordon masses for which the exterior is linearly stable, i.e. for $\mu > \frac{3}{4}\Lambda$ in the case of Dirichlet boundary conditions, and for  $ \frac 34 \Lambda < \mu < \frac{5}{12} \Lambda$ together with additional assumptions in the case of Neumann boundary conditions \cite{warnick_boundedness_and_growth}.
\paragraph{Regularity of the initial data}
We stated \cref{con:2a} and \cref{con:2b} for smooth ($C^\infty$) initial data. One can also consider classes of initial data which are more regular (e.g.\ Gevrey or analytic) or less regular (e.g.\  Sobolev). 
From our heuristics, we expect that the analogs of \cref{con:2a} and  \cref{con:2b} remain valid both for rougher data in some suitably weighted Sobolev space (see  \cite{gustav}) and more regular data of Gevrey regularity with index $\sigma >1$ and analytic data ($\sigma=1$). 
Only in the exceptional and most regular case of initial data with Gevrey regularity $\sigma <1$ (note that this is more regular than analytic) in the angular direction $\partial_\phi$, we expect the analog of \cref{con:2a} to break down. In particular, for axisymmetric data (or data supported only on finitely many azimuthal modes $m$), we expect the arising solution to remain uniformly bounded at the Cauchy horizon for \emph{all} parameters in $\mathscr P$. 
\subsection{\texorpdfstring{\cref{thm:rough}}{Theorem~1}: \texorpdfstring{\cref{con:2a}}{Conjecture 5} is true}
\label{sec:intromainthm}
Our main result is the following   resolution of \cref{con:2a}. 

\addvspace{\medskipamount}
\begin{theorem}\label{thm:rough}
	\cref{con:2a} is true.
\end{theorem} 

\addvspace{\medskipamount}
The proof of  \cref{thm:rough} will be given in \cref{sec:mainthmkerr}.
\begin{rmk}\label{rmk:othernotionsofgenericity}
In the proof of \cref{thm:rough} we will not only construct a single solution which blows up leading to genericity of initial data as in \cref{rmk:linearisenough} but  we will actually obtain what is perhaps a more satisfying genericity condition on the initial data which are smooth and of compact support. We formulate this condition in  \cref{cor:genericity} in \cref{sec:mainthmkerr}. 
\end{rmk}

\begin{rmk}
	We also prove in \cref{sec:mainthmkerr} the statement about the packing dimension of $\mathscr P_{\textup{Blow-up}}$ as conjectured in  \cref{rmk:conjhdpacking}. The statements concerning the Hausdorff dimension, however, remain open.
\end{rmk}
\begin{rmk}
	In principle, our proof is expected to also apply to Neumann boundary conditions as well as to more general Klein--Gordon masses satisfying the Breitenlohner--Friedman bound \cite{breitenlohner} as discussed at the end of \cref{sec:introconjectures}.
\end{rmk}

\subsection{Outlook on \texorpdfstring{\cref{con:2b}}{Conjecture~6}}\label{sec:outlookonres}
We also expect that our methods provide a possible framework for the  resolution of \cref{con:2b}. 

 First, note that the blow-up statement of \cref{thm:rough} is strictly stronger than the $H^1_\textup{loc}$ blow-up conjectured in \cref{con:2b}(B). Thus, \cref{thm:rough}  shows that \cref{con:2b}(B) is  true for black hole parameters in the set $\mathscr P_{\textup{Blow-up}}$. For parameters not contained in $\mathscr P_{\textup{Blow-up}}$, we expect that a quasinormal mode which decays at a sufficiently slow exponential decay rate compared to the surface gravity of the Cauchy horizon will blow up in energy at the Cauchy horizon. This would show \cref{con:2b}(B).
Towards \cref{con:2b}(A),  we note that our proof, particularly   formula \eqref{eq:formulaforlimit} of \cref{prop:mainpropinmainthm}, reveals the main obstruction for boundedness. Together with the methods used  in  \cite{CKehle2019} for the Reissner--Nordström--AdS case, this can serve as a starting point for a resolution of \cref{con:2b}(A). 

\subsection{Turning the heuristics of \texorpdfstring{\cref{sec:introdioph}}{Section~1.4} into a proof of \texorpdfstring{\cref{thm:rough}}{Theorem~1}}
\label{sec:breifdesc}
We will now outline how we turn our heuristics of \cref{sec:introdioph} into a proof  of \cref{con:2a}, i.e.~\cref{thm:rough}.

 We are interested in constructing a solution of \eqref{eq:wavekerr}, arising from smooth and compactly supported initial data, which blows up as in  \eqref{eq:blowupl2}, if the dimensionless parameters $\mathfrak m, \mathfrak a$ satisfy a certain non-Diophantine condition. 

We remark that unlike in our heuristic discussion, we will not make use of quasinormal modes and the frequency analysis will be purely based on the real axis with $\omega\in \mathbb R$. Indeed, our approach can be interpreted as replacing quasinormal modes with quasimodes. This will also manifest itself in the fact that the roles of $\frac{1}{\mathfrak W[u_{\mathcal H^+},u_\infty](\omega,m,\ell)}$ and $\frac{1}{\omega - \omega_- m}$ will be changed: In the heuristic analysis, we considered the quasinormal mode frequencies $\omega_{m\ell n}$ which are (complex) roots of the Wronskian $\mathfrak W[u_{\mathcal H^+},u_\infty] $ and the small divisors came from $\frac{1}{\omega_{m\ell n} - \omega_- m}$. In the actual proof of \cref{thm:rough}, we will instead consider the real frequencies $\omega=\omega_- m$ (i.e.\ the roots of $\omega-\omega_- m =0$) and as we will see, the small divisors will then appear from the Wronskian evaluated at the characteristic frequency $\frac{1}{\mathfrak W[u_{\mathcal H^+},u_\infty](\omega=\omega_- m,m,\ell)}$. Note that the divisor $|\mathfrak W[u_{\mathcal H^+},u_\infty](\omega=\omega_- m,m,\ell)|$ is small exactly if there exists a quasimode with frequency $(\omega=\omega_- m, m,\ell)$. In view of the distribution of the quasimode frequencies discussed in \cref{sec:introex}, this will lead  to a (generalized) non-Diophantine condition which we will address in more detail further below.
\paragraph{Initial data and exterior analysis (\texorpdfstring{\cref{sec:initialdata} and \cref{sec:exterioranalysis}}{Section~6 and Section~7})}
We begin our discussion with our choice of initial data. In \cref{sec:initialdata} we will carefully impose smooth and compactly supported initial data $\Psi_0,\Psi_1 \in C_c^\infty(\Sigma_0)$ for \eqref{eq:wavekerr} on the spacelike hypersurface $\Sigma_0 = \{ t=0\}$ which---with foresight---will be chosen to satisfy 
 \begin{align}\label{eq:decaywrintro}
 \left| G(\Psi_0,\Psi_1,m_i,\ell_i)\right|  \gtrsim e^{-m_i^{\frac 13}}
 \end{align}
 for  suitable infinite sequences $m_i,\ell_i$,  where 
 \begin{align}\nonumber
 G(\Psi_0,\Psi_1,m,\ell):= 	\int_{\Sigma_0} u_\infty(r,\omega = \omega_-m, m,\ell)& S_{m\ell}(a\omega = a\omega_-m ,\cos\theta) e^{-im\phi} \\
 &\times  H(\Psi_0,\Psi_1)(\omega_- m, r,\theta,\phi) \dvol_{\Sigma_0},
\end{align}
and
\begin{align} H(\Psi_0,\Psi_1)(\omega, r,\theta,\phi) =  \frac{\Sigma}{r^2 \sqrt{r^2+a^2}}  \left( - \sqrt{-g^{tt}} \Psi_1 - i \omega g^{tt} \Psi_0 + g^{t\phi} \partial_\phi \Psi_0 \right). \end{align}
We also recall that $u_\infty(r,\omega,m,\ell)$ is the solution to the radial o.d.e.\ adapted to the Dirichlet boundary condition at $\mathcal I$. 

Complementing the data with vanishing data on $\mathcal H_L \cup \mathcal B_{\mathcal H}$, the data $\Psi_0,\Psi_1$ define a solution on the black hole interior. Different from our heuristic discussion with quasinormal modes (i.e. fixed frequency solutions) in  \cref{sec:introdioph}, in the present case, we do need to consider the analog of a ``full'' scattering operator $\mathfrak S_{\mathcal H_R \to \mathcal{CH}_R}$ from the event horizon to the Cauchy horizon which would  be of the form 
\begin{align}\label{eq:scatteringkerrads}
\mathfrak S_{\mathcal H_R \to \mathcal{CH}_R} = \mathcal{F}_{\mathcal{CH}}^{-1}\circ  \mathfrak R(\omega,m,\ell)\circ \mathcal F_{\mathcal{H}}  = \mathcal{F}_{\mathcal{CH}}^{-1}\circ   \frac{\mathfrak r (\omega,m,\ell)  }{\omega- \omega_- m}\circ \mathcal F_{\mathcal{H}} ,
\end{align}
where $\mathcal F_{\mathcal{H}}$ and $\mathcal F_{\mathcal{CH}}$ represent (generalized) Fourier transforms along the event and Cauchy horizon, respectively. 
Thus, from the exterior, we need to determine the generalized Fourier transform  $\mathcal F_{\mathcal H} [\psi\restriction_\mathcal H]$ along the event horizon. 
Such a characterization in terms of the chosen initial data from above is the content of \cref{sec:exterioranalysis}. While in the actual proof (see already \cref{prop:rep}), we will use a  suitably truncated generalized Fourier transform, we may formally think of  $\mathcal F_{\mathcal H} [\psi\restriction_\mathcal H]$ as having the form
\begin{align}\label{eq:fouriertransfromoneventhorizon}
\mathcal F_{\mathcal H} [\psi\restriction_\mathcal H] (\omega, m,\ell) \sim \frac{ \int_{\Sigma_0} u_{\infty}(r,\omega,m,\ell)  S_{m\ell}(a\omega ,\cos\theta) e^{-im\phi} H(\Psi_0,\Psi_1)(\omega, r,\theta,\phi) \dvol_{\Sigma_0}}{\mathfrak W[u_{{\mathcal H^+}}, u_{\infty} ](\omega,m,\ell) }.
\end{align}
We already remark that a consequence of the smoothness of the initial data is that $G(\Psi_0,\Psi_1,m,\ell)$ decays superpolynomially in $m$ and $\ell$, cf.\ \eqref{eq:superpolynomialdecayofinitialdata} in \cref{sec:initialdata}.

\paragraph{Interior analysis (\texorpdfstring{\cref{sec:interior},  \cref{sec:mainthmkerr} and some of \cref{sec:angularode}}{Section~8, Section~9 and some of Section~3})} 
Turning to the interior analysis, we recall from our heuristic discussion in \cref{sec:introin} that the analog of the scattering operator  \eqref{eq:scatteringkerrads} from the event to the Cauchy horizon has  poles at the characteristic frequencies $\omega- \omega_- m=0$ with respect to $K_-$. In our heuristic discussion in \cref{sec:introdioph} based on quasinormal modes and fixed frequency scattering, these poles formally lead to \eqref{eq:diphantineintro}. In the actual proof, based on frequency analysis on the real axis, the scattering poles become evident in formula \eqref{eq:formulaforlimit} stated in \cref{prop:mainpropinmainthm} which roughly translates to the statement that, as $r\to r_-  $, we have
\begin{align}\label{eq:rough}
 \| \psi(u_0, r) \|_{L^2(\mathbb S^2)}^2 \sim  \sum_{m\ell}   |\mathfrak r(\omega=\omega_-m,m,\ell)|^2 \left|\mathcal F_{\mathcal H} [\psi\restriction_\mathcal H  ]   (\omega = \omega_- m,m,\ell)\right|^2 + \textup{Err}(D),
\end{align}
where $\textup{Err}(D)$ is uniformly bounded by an (higher order)  energy of the initial data. (Note that in the actual statement of \cref{prop:mainpropinmainthm}, the Fourier transform  along the horizon $\mathcal{F}_{\mathcal{H}} [\psi\restriction_{\mathcal H}]$ appearing in \eqref{eq:rough} is  the truncated Fourier transform   $a_{\mathcal H}^{R_n} (\omega = \omega_- m,m,\ell)  = \mathcal{F}_{\mathcal{H}} [\psi\restriction_{\mathcal H}\mathbb 1_{v\leq R_n}]$ for $R_n = 2r_n^\ast - u_0 + \tilde c $.)

  Both the proof and the use of \eqref{eq:rough} lie at the heart of the proof of \cref{thm:rough}. The proof of \eqref{eq:rough} is technical and combines physical space methods with techniques from harmonic analysis.
 One of the key technical steps  (see \cref{eq:smlbounds} in \cref{sec:angularode}) is a quantitative bound (see already \eqref{eq:formulaonsml}) on the derivative of the generalized spheroidal harmonics
 \begin{align} \sup_{|a\omega-a\omega_- m| < \frac{1}{m}}\| \partial_\omega S_{m\ell}(a \omega)\|^2_{L^2} \lesssim  m\label{eq:spheroidalboundintro}\end{align}
  near the interior scattering poles $\omega = \omega_- m$. The proof of \eqref{eq:spheroidalboundintro} relies on uniform bounds (in $m,\ell$ and $\omega \approx \omega_- m$) on the  resolvent of the associated singular Sturm--Liouville operator, see the discussion in  \cref{sec:boundsonpartialxlambda}. These bounds are shown by   constructing and estimating the   integral kernel of the resolvent using suitable approximations with parabolic cylinder functions and Airy functions. For solutions of the radial o.d.e.\ in the interior, the analogous resolvent bounds   are shown in  \cref{eq:radialodeinterior}. Their proofs rely on  semi-classical approximations and estimates on Volterra integral equations. 
  
  Further ingredients to control the error term  $\textup{Err}(D)$ in \eqref{eq:rough}  are  uniform bounds on the transmission and reflection coefficients $\mathfrak T(\omega,m,\ell)$ and $\mathfrak R(\omega,m,\ell)$ for frequencies which are bounded away from the characteristic frequency $\omega=\omega_- m$.

\paragraph{Combining the exterior with the interior: Occurrence of small divisors and the proof of \texorpdfstring{\cref{thm:rough}}{Theorem~1} (\texorpdfstring{\cref{sec:mainthmkerr}}{Section~9})} 
We will now connect the exterior analysis  to the interior.
In particular, formally plugging \eqref{eq:fouriertransfromoneventhorizon} into \eqref{eq:rough} and noting that \begin{align} \mathcal F_{\mathcal H} [\psi\restriction_\mathcal H] (\omega=\omega_- m , m,\ell) \sim \frac{ G(\Psi_0,\Psi_1,m,\ell)}{\mathfrak W[u_{{\mathcal H^+}}, u_{\infty} ](\omega=\omega_- m,m,\ell) } ,\end{align} yields in the limit $r\to r_-$ that
\begin{align}\label{eq:thethingtoblowup}
 \| \psi(u_0, r) \|_{L^2(\mathbb S^2)}^2 \sim  \sum_{m\ell}   \frac{ |m|^2 |G(\Psi_0,\Psi_1,m,\ell)|^2}{|\mathfrak W[u_{{\mathcal H^+}}, u_{\infty} ](\omega=\omega_- m,m,\ell) |^2} + \textup{Err}(D),
\end{align}
where we also used that the (renormalized) reflection coefficient  satisfies $|\mathfrak r (\omega=\omega_- m, m,\ell)|\sim |m|$ which we will show in \cref{lem:lowerboundont}. Also recall from before that the error term $|\textup{Err}(D)|$ is shown to remain uniformly bounded as $r\to r_-$.
Remark that in the actual proof we will not quite show \eqref{eq:thethingtoblowup} but rather obtain  \eqref{eq:thefinalformulaforblowup} which corresponds to \eqref{eq:thethingtoblowup}   in a certain limiting sense. We also recall from the discussion of the exterior analysis that the term $|m|^2 |G(\Psi_0,\Psi_1,m,\ell)|^2$ 
which appears in the sum of \eqref{eq:thethingtoblowup} as the numerator,  decays superpolynomially in $m$ and $\ell$. 
Thus, at least formally, in order to show blow-up for \eqref{eq:thethingtoblowup}, it is necessary that small divisors in \eqref{eq:thethingtoblowup} occur infinitely often, i.e.\  that the Wronskian evaluated at the interior scattering poles $\mathfrak W[u_{{\mathcal H^+}}, u_{\infty} ](\omega=\omega_-m,m,\ell)$ (cf.\ \cref{sec:introex}) decays (at least) superpolynomially for infinitely many $(m,\ell)$. In our  proof, we will actually require from the black hole parameters $\mathfrak m, \mathfrak a$ that this Wronskian decays exponentially \begin{align}|\mathfrak W[u_{{\mathcal H^+}}, u_{\infty} ](\omega=\omega_-m,m,\ell)| \leq e^{-m^{\frac 12}}  \text{ for infinitely many } (m,\ell). \label{eq:formuladecayintro}\end{align}   

Before we address the validity of \eqref{eq:formuladecayintro}, we will assume for a moment that indeed the black hole parameters $\mathfrak m, \mathfrak a$ are such that \eqref{eq:formuladecayintro} holds true. Then, explicitly choose that the subsequences $m_i$ and $\ell_i$ in \eqref{eq:decaywrintro}    coincide with the infinite sequences which fulfill  \eqref{eq:formuladecayintro}. Then, we formally obtain the blow-up result of  \cref{thm:rough} as  
\begin{align}
\nonumber
\lim_{r\to r_-} \| \psi(u_0, r) \|_{L^2(\mathbb S^2)}^2 & \sim  \sum_{i \in \mathbb N}   |m_i|^2 \frac{\left|G(\Psi_0,\Psi_1,m_i,\ell_i)\right|^2 }{|\mathfrak W[u_{{\mathcal H^+}}, u_{\infty} ] (\omega = \omega_- m_i ,m_i,\ell_i)|^2 }  \\ &
\sim   \sum_{i \in \mathbb N}   |m_i|^2 \frac{ |e^{- m_i^{\frac 13}}  |^2}{| e^{-m_i^{\frac 12}} |^2 } =+\infty.\label{eq:thethingtoblowupactually}
\end{align}
 Similarly to the remark before,  in reality, \eqref{eq:thethingtoblowupactually} holds true only in a certain limiting sense, cf.\ \eqref{eq:thefinalformulaforblowup}--\eqref{eq:thefinalformula2} of \cref{sec:mainthmkerr}. Already from \eqref{eq:thethingtoblowupactually} and \eqref{eq:formuladecayintro} we obtain the following genericity condition 
  \begin{align} 
 \sum_{i\in \mathbb N} |m_i e^{  m_i^\frac 12}|^2 |G(  \Psi_0,  \Psi_1,m_i,\ell_i)|^2 = +\infty
 \end{align}
 on the initial data leading to blow-up. This will be formulated as \cref{cor:genericity} in \cref{sec:mainthmkerr}.
\paragraph{The non-Diophantine condition and its relation to quasimodes (\texorpdfstring{\cref{sec:radialext}, \cref{sec:pblowup}, and some   of \cref{sec:angularode}}{Section~4, Section~5, and some of Section~3})} 
Finally, this leaves us to address the question of whether the small divisors in \eqref{eq:thethingtoblowup} actually appear infinitely often, more precisely, whether \eqref{eq:formuladecayintro} holds true. The condition \eqref{eq:formuladecayintro}   constitutes a generalized non-Diophantine condition on the black hole parameters $\mathfrak m, \mathfrak a$ in view of its relation to the (discrete) Bohr--Sommerfeld quantization conditions from our heuristic discussion \cref{sec:introex}. In our actual proof, the non-Diophantine condition which we define in \cref{eq:defnpblowup} in \cref{sec:pblowup} is more  technical than \eqref{eq:formuladecayintro}, though \eqref{eq:formuladecayintro} should be considered as its key property. 
 We denote the set of  dimensionless black hole parameters $\mathfrak m,\mathfrak a$ which satisfy the condition with $\mathscr{P}_\textup{Blow-up}$. 
The statement that is $\mathscr{P}_\textup{Blow-up}$ is  Baire-generic but Lebesgue-exceptional is the content of \cref{sec:quasimodes} and \cref{sec:lebesgue}, respectively. Both proofs crucially rely on estimates developed in \cref{sec:radialext}.

Connecting to the discussion of quasimodes before, we note that the non-Diophantine condition of \eqref{eq:formuladecayintro} can be interpreted as the statement that there exist infinitely many quasimodes with frequency $\omega = \omega_- m$. This also implies that there exist infinitely many quasinormal modes with (complex) frequencies $\omega$ exponentially close to $\omega=\omega_- m$. However, note that quasimodes are more robust to perturbations in the sense that if $\omega, m,\ell$ are frequencies of a quasimode,  there exists a (exponentially small) neighborhood of $\omega$ such that for each $\tilde \omega$  in that neighborhood, the frequencies $\tilde \omega, m,\ell$ would also describe a quasimode. It is also  this robustness which is a key advantage of quasimodes over an approach based on quasinormal modes as in the heuristic discussion in \cref{sec:introdioph}.

\subsection{Outline of the paper}
\label{sec:outlineofc3}
In \cref{sec:prelimskerr} we set up the Kerr--AdS spacetime, recall the well-posedness of \eqref{eq:wavekerr} and the decay statement for solutions on the exterior. We also introduce Carter's separation of variables. \cref{sec:angularode} is devoted to the analysis of the angular o.d.e. In \cref{sec:radialext} we analyze the radial o.d.e.\ on the exterior and introduce suitable solutions of the radial o.d.e.\ associated to trapping at the interior  scattering  poles $\omega = \omega_- m$. With the estimates from \cref{sec:radialext} in hand, we define the set $\mathscr P_{\textup{Blow-up}}$ in \cref{sec:pblowup} and  show its topological and metric properties. 

Then, for arbitrary but fixed parameters $\mathfrak p \in \mathscr P_{\textup{Blow-up}}$ we define suitable compactly supported and smooth initial data in \cref{sec:initialdata}. In \cref{sec:exterioranalysis} we treat the exterior problem and conclude with a representation formula of the solution along the horizon in terms of the initial data. In \cref{sec:interior} we first show suitable estimates for solutions of the radial o.d.e.\  in the interior before we finally conclude the paper with the proof of \cref{thm:rough} in \cref{sec:mainthmkerr}.
\subsection{Acknowledgements}
The author would like to express his gratitude to his advisor Mihalis Dafermos for his support and many valuable comments on the manuscript.  The author also thanks Harvey Reall, Igor Rodnianski, Yakov Shlapentokh-Rothman and Claude Warnick for insightful discussions and helpful remarks. The author is also grateful to two anonymous referees for several helpful comments and remarks to improve the manuscript. This work was supported by the EPSRC grant EP/L016516/1, by Dr. Max Rössler, the Walter Haefner Foundation and the ETH Zürich Foundation.. The author thanks Princeton University for hosting him as a VSRC.
\section{Preliminaries}\label{sec:prelimskerr}
	\subsection{Fractal measures and dimensions}\label{sec:fractalmeas}
	\subsubsection{Hausdorff and Packing measures}
	We begin by introducing the Hausdorff and packing measure. We refer to the monograph \cite{MR3236784} for a more detailed discussion.  For an increasing \emph{dimension function} $f\colon [0,\infty) \to [0,\infty)$ we define the   Hausdorff measure $ H^f(A)$ of a set $A$ as 
	\begin{align}
	H^f(A):= \sup_{\delta >0} H^f_{\delta}(A), 
	\end{align}
	where 
	\begin{align}
	H^f_{\delta}(A) := \inf\{ \sum_{i=1}^\infty f(\textup{diam}(U_i)) \colon  \{U_i\}_{i=1}^\infty \text{ countable cover of } A, \textup{diam}(U_i) \leq \delta \}.
	\end{align}
	If $f(r) = r^s$, we write $H^s=H^{r^s}$ and for $s\in \mathbb N$, the measure $H^s$ reduces to the Lebesgue measure up to some normalization. 
	While the Hausdorff measure quantifies the size of a set by approximation it from outside via efficient coverings, we also recall the dual notation: The \emph{packing measure} quantifies the size of sets by placing as many disjoint balls with centers contained in the set. Again, for a dimension function $f$, we define the pre-measure
	\begin{align}\nonumber
	P_0^f(A) := \limsup_{\delta\to 0}  \big\{&\sum_{i=1}^\infty f(\textup{diam}(B_i))\colon \{B_i\}_{i=1}^\infty  \text{  collection of closed,  }\\&
	\text{ pairwise disjoint balls with }\textup{diam}(B_i)\leq \delta \text{ and centers in } A \big\}
	\end{align}
	and finally the packing  measure as 
	\begin{align}
	P^f(A) := \inf\left\{ \sum_{i=1}^\infty P_0^f(A_i)\colon A \subset \bigcup_{i=1}^\infty  A_i  \right\}.
	\end{align}
\subsubsection{Hausdorff and Packing dimensions}
	For $f(r)=r^s$ Hausdorff and Packing dimensions $\dim_H$ and $\dim_P$ are now characterized as the jump value, where the respective measure jumps from 0 to $\infty$, more precisely
	\begin{align}
			\dim_H (A) = \sup\{ s \colon   H^s(A) =0 \}, \;\; \dim_P (A) = \sup\{ s \colon   P^s(A) =0 \}.
	\end{align}
	We also say that a set $A$ has generalized Hausdorff dimension $\dim_{gH}(A) = s+\log $ if the jump appears for the dimension function $f(r) = r^s \log^t(r)$ for some $t>0$.
	\subsection{Kerr--AdS spacetime}
	\subsubsection{Parameter space}
	\label{sec:parameterspace}
	We let the value of the cosmological constant $\Lambda <0$ be \emph{fixed} throughout the paper. For convenience and as it is convention, we re-parametrize the cosmological constant by the AdS radius 
	\begin{align}\label{eq:defnofl}
	l:= \sqrt{\frac{-3}{\Lambda}}.
	\end{align} 

	We consider Kerr--AdS black holes which are parameterized by their mass $M>0$ and their specific angular momentum $a\neq 0$. 	Moreover, without loss of generality  we will only consider $a>0$ and require $0<a<l$ for the spacetime to be regular. For $M>0$, $0<a<l$, we consider the polynomial 
	\begin{align}
	\Delta(r) := (a^2+r^2)\left(1+\frac{r^2}{l^2}\right) - 2 Mr .\label{eq:defdelta}
	\end{align} 
	We are interested in spacetimes without naked singularities. To ensure this, we define a parameter tuple    $(M,a) \in \mathbb{R}_{>0}^2$ to be \emph{non-degenerate} if $0<a<l$ and $\Delta(r)$ defined in \eqref{eq:defdelta} has two real roots satisfying $0< r_- < r_+$. Finally, to exclude growing mode solutions (see \cite{dold}) we assume the Hawking--Reall (non-superradiant) bound 
		\begin{align}\label{eq:hawkingreallbound}
	 	r_+^2 > a l.
		\end{align}
		This leads us to the definition of the dimensionless black hole parameter space 
		\begin{align}\label{eq:setofallparameters}
		\mathscr P:= \{ (\mathfrak m, \mathfrak a) \in \mathbb R^2_{>0} \colon (M ,a):= ( \mathfrak m l/\sqrt{3} , \mathfrak a l/\sqrt{3} ) \text{ is non-degenerate and } r_+^2 > a l \}.
		\end{align}
		Note that in view of \eqref{eq:defnofl}, we have $ M = \mathfrak m/ \sqrt{-\Lambda} =\mathfrak m l/\sqrt{3} $ and $a = \mathfrak a/ \sqrt{-\Lambda} = \mathfrak a l /\sqrt{3}$. 
On the parameter space $\mathscr P$, we will also use the global coordinates $( \vartheta, \mathfrak a)$, where \begin{align}\vartheta = \vartheta(\mathfrak a, \mathfrak m) := \frac{1-a^2/l^2}{1+r_-^2/a^2}.\end{align} 
(Note that $\vartheta = a \omega_-$, where $\omega_-$ is defined in \eqref{eq:definitionofomega-} below.) Thus, for each $\mathfrak a$, there exists an interval $(\mathfrak \vartheta_1(\mathfrak a), \vartheta_2(\mathfrak a))$ and a smooth embedding  $ (\mathfrak \vartheta_1(\mathfrak a), \vartheta_2(\mathfrak a) ) \to \mathscr P, \vartheta \mapsto (\mathfrak m (\vartheta), \mathfrak a)$ which also depends smoothly on $\mathfrak a$. We define the vector field $\Gamma$ on $\mathscr P$ by
 \begin{align}
 \Gamma:= \frac{\partial}{\partial \vartheta}
\label{eq:defnofgamma}
\end{align} in coordinates $(\vartheta,\mathfrak a)$. We define  $\Phi_\tau^\Gamma$ as the flow generated by $\Gamma$.

Finally, remark that  $\mathscr{P}$ is a  Baire space as  a (non-empty) open subset of $\mathbb R^2$. In particular, this allows us to speak about the notion of  \emph{Baire-exceptional} and \emph{Baire-generic} subsets. Recall that a subset is Baire-exceptional if it is a countable union of nowhere dense sets and a subset is called Baire-generic if it is a countable intersection of open and dense sets. Note that if a subset is Baire-generic then its complement is Baire-exceptional and vice versa. Finally, in a Baire space every Baire-generic subset is dense.

\subsubsection{Kerr--AdS spacetime}
\paragraph{Fixed manifold}	We begin by constructing the Kerr--AdS spacetime. We define the exterior region $\mathcal R$ and the black hole interior $\mathcal{B}$ as smooth four dimensional manifolds diffeomorphic to $\mathbb R^2\times \mathbb S^2$. On $\mathcal R$ and on $\mathcal{B}$ we assume to have global (up to the well-known degeneracy on $\mathbb S^2$) coordinate charts
\begin{align}
&(t_{\mathcal{R}}, r_{\mathcal{R}}, \theta_{\mathcal{R}}, \phi_{\mathcal{R}} ) \in \mathbb R \times (r_+,\infty) \times \mathbb S^2,\\ 
 &(t_{\mathcal{B}}, r_{\mathcal{B}}, \theta_{\mathcal{B}}, \phi_{\mathcal{B}} ) \in \mathbb R \times (r_-,r_+) \times \mathbb S^2.
\end{align} These coordinates $(t,r,\phi,\theta)$ are called Boyer--Lindquist coordinates. If it is clear from the context which coordinates are being used, we will omit their subscripts throughout the
chapter.
\paragraph{The Kerr--AdS metric} For $(\mathfrak m, \mathfrak a ) \in \mathscr{P}$ and $ M = \mathfrak m l/ \sqrt{3} $ and $a = \mathfrak al/ \sqrt{3}$, we define the Kerr--AdS metric on $\mathcal{R}$ and $\mathcal{B}$ in terms of the Boyer--Lindquist coordinates as 
\begin{align}\nonumber
g_{\mathrm{KAdS}}:= &- \frac{\Delta - \Delta_\theta a^2 \sin^2\theta}{\Sigma} \d t\otimes \d t+  \frac{\Sigma}{\Delta} \d r \otimes \d r + \frac{\Sigma}{\Delta_\theta}\d \theta \otimes \d \theta \\
&+ \frac{\Delta_\theta(r^2 + a^2)^2 - \Delta a^2 \sin^2\theta}{\Xi^2 \Sigma} \sin^2\theta \d {\phi} \otimes \d \phi \nonumber  \\
& -  \frac{\Delta_\theta (r^2 +a^2)- \Delta}{\Xi \Sigma} a \sin^2 \theta (\d t \otimes\d \phi + \d\phi \otimes \d t), \label{eq:kerradsmetric}
\end{align}
where \begin{align}\Sigma:= r^2 + a^2 \cos^2\theta,\;\; \Delta_\theta:= 1 - \frac{a^2}{l^2} \cos^2\theta,\;\; \Xi := 1- \frac{a^2}{l^2}\end{align} and $\Delta$ is as in \eqref{eq:defdelta}. We will also write $\Delta_x := 1-\frac{a^2}{l^2}x^2$ which arises from the substitution $x=\cos\theta$ in $\Delta_\theta$. We also define
\begin{align}\label{eq:definitionofomega-}
	\omega_+:= \frac{a \Xi}{r_+^2+a^2}, \;\;	\omega_-:= \frac{a \Xi}{r_-^2+a^2}, \;\;	\omega_r:= \frac{a \Xi}{r^2+a^2}.
\end{align} Now, we time-orient the patches  $\mathcal{R}$  and $\mathcal{B}$ with $-\nabla t_{\mathcal{R}}$ and $-\nabla{r_{\mathcal B}}$, respectively. We also note that $\partial_t$ and $\partial_\phi$ are Killing fields in each of the patches.
The inverse metric reads
\begin{align}\nonumber
	g_{\textup{KAdS}}^{-1} = &\left( -\frac{(r^2+a^2)^2}{\Sigma\Delta} + \frac{a^2 \sin^2\theta}{\Sigma\Delta_{\theta}} \right) \partial_t \otimes \partial_t +  \frac{\Delta}{\Sigma} \partial_r \otimes\partial_r  + \frac{\Delta_\theta}{\Sigma}\partial_\theta \otimes \partial_\theta \\
	&+ \left( \frac{\Xi^2}{\Sigma\Delta_{\theta} \sin^2\theta} - \frac{\Xi^2 a^2}{\Sigma \Delta}\right) \partial_\phi\otimes\partial_\phi 
-\left( \frac{\Xi a (r^2+a^2)}{\Delta \Sigma} - \frac{a\Xi}{\Delta_\theta \Sigma} \right) (\partial_t \otimes \partial_\phi+ \partial_\phi \otimes \partial_t ).
\end{align}
On $\mathcal{R}$ and $\mathcal{B}$, we define the tortoise coordinate ${r^\ast}(r)$ by
\begin{align}\label{eq:defnrast} \frac{\d {r^\ast}}{\d r} (r) := \frac{r^2 +a^2}{\Delta(r)},\end{align} where $\Delta$ is as in \eqref{eq:defdelta}. For definiteness we set ${r^\ast} (r=+\infty):= \frac{\pi}{2} l$ on $\mathcal{R}$ and ${r^\ast} ( \frac{1}{2} (r_+ + r_-)) = 0$ on $\mathcal{B}$. 
\paragraph{Eddington--Finkelstein-like coordinates}
We also define Eddington--Finkel\-stein-like coordinates ($v,r,\theta,\tilde \phi_+$) in the exterior $\mathcal{R}$ as 
\begin{align}
	v(t,r):= t + r^\ast \chi_v(r), \;\;  \tilde \phi_+ (\phi,r):= \phi + \omega_+ r^\ast(r)   \chi_v(r) \; \text{mod } 2\pi,
\end{align}
where $\chi_v(r)$ is  a smooth monotone cut-off function with $\chi_v(r) = 1 $ for $r \leq  r_+ +\eta $ and $\chi_v(r) = 0$ for $r\geq r_+ +2 \eta$ for some $\eta>0$ small enough such that $   J^+( \{r> 2 r_+ \} \cap \{ t_\mathcal{R}=0 \} ) \cap \{ v=0 \}   = \emptyset $ \footnote{Note that $\nabla v$ is not timelike everywhere on $\mathcal{R}$, in particular $g(\nabla v , \nabla v) = a^2 \sin^2\theta \Sigma^{-1}\Delta_\theta^{-1}$ for $r\in[r_+,r_+ + \eta]$.} and $\eta <\frac{r_+}{4}$. In these coordinates the spacetime $(\mathcal{R},g_{\mathrm{KAdS}})$ can be extended (see \cite{gustav} for more details) to a time-oriented Lorentzian manifold   $(\mathcal D,g_{\mathrm{KAdS}})$ defined as $\mathcal D:= \{ (r,v,\theta,\tilde \phi_+) \in (r_-,\infty) \times \mathbb R \times \mathbb S^2 \}$. Moreover, the Lorentzian submanifold $(\mathcal D \cap \{ r_- < r < r_+\} , g_{\mathrm{KAdS}})$  coincides (up to time-orientation preserving isometry) with  $(\mathcal{B} ,g_{\mathrm{KAdS}})$. We identify these regions and denote the (right) event horizon as $\mathcal{H}_R:= \{r =r_+\}.$ The Killing null generator of the event horizon is \begin{align}K_+:= \partial_v + \omega_+ \partial_{\tilde \phi_+}. \end{align} 
The Killing field $K_+$ is called the Hawking vector field and is future-directed and timelike in $\mathcal{R}$, a consequence the Hawking--Reall bound $r_+ > a l$.

To attach the (left) Cauchy horizon $\mathcal{CH}_L$ we introduce in $\mathcal B$ further coordinates $(v,r,\theta,\tilde \phi_-)$, as
\begin{align}v=t+r^\ast,\;\; \tilde \phi_- (\phi,r):= \phi+\omega_- r^\ast \textup{ mod } 2\pi,\;\; r=r, \;\; \theta = \theta.\end{align} In these coordinates, the Lorentzian manifold extends smoothly to $r=r_-$ and the null hypersurface $\mathcal{CH}_L:= \{r =r_-\}$ is the left Cauchy horizon with null generator \begin{align}K_- := \partial_v + \omega_- \partial_{\tilde \phi_-}.\end{align}      Note that $\partial_v = \partial_t$ and $\partial_{\tilde \phi_-} = \partial_{\phi}$ in $\mathcal B $. 

To attach the left event horizon $\mathcal{H}_L$  we introduce new coordinates on $\mathcal{B}$ defined as $(u,r,\theta, \phi_+^\ast ) \in \mathbb R \times (r_-,r_+) \times \mathbb S^2$ by \begin{align} u(t,r) := -t + r^\ast, \;\; \phi_+^\ast := \phi - \omega_+ r^\ast \text{ mod } 2\pi, r=r, \theta=\theta \end{align} on $\mathcal{B}$ and attach $\mathcal{H}_L$ as $\mathcal{H}_L=\{r=r_+\}$.   Similarly, introduce $(u,r,\theta,\phi^\ast_{-})$ as 
\begin{align}
u(t,r) := -t + r^\ast, \;\; \phi_-^\ast := \phi - \omega_- r^\ast \text{ mod } 2\pi, r=r, \theta=\theta
\end{align} on $\mathcal{B}$ and attach the right Cauchy horizon $\mathcal{CH}_R$ as $\mathcal{CH}_R = \{ r=r_-\}$ in this coordinate system.  Indeed, $K_+$ and $K_-$ extend to  Killing vector fields expressed as $K_+:=- \partial_u  + \omega_+ \partial_{\phi_+^\ast}$ and $K_-:=- \partial_u  + \omega_- \partial_{\phi_-^\ast}$. They are past directed Killing generators of $\mathcal{H}_L$ and $\mathcal{CH}_R$, respectively. Finally, we attach the past and future bifurcation spheres $\mathcal{B}_{\mathcal H}$ and $\mathcal{B}_{\mathcal{CH}}$. Formally, they are defined as $\mathcal{B}_{\mathcal{H}}:= \{ v = - \infty\} \times \{ r=r_+\} \times \mathbb S^2 =\{ u= - \infty\} \times \{ r=r_+\} \times \mathbb S^2  $ respectively in the coordinates systems $(v,r,\theta,\tilde \phi_+)$ and $(u,r,\theta, \phi^\ast_+)$. Similarly, we have $\mathcal{B}_{\mathcal{CH}}:= \{ v = + \infty\} \times \{ r=r_-\} \times \mathbb S^2 =\{ u= + \infty\} \times \{ r=r_-\} \times \mathbb S^2  $. Finally, we define the Cauchy horizon $\mathcal{CH}:=\mathcal{CH}_L \cup \mathcal{CH}_R \cup \mathcal {B}_{\mathcal{CH}}$.  This is standard and we refer to the preliminary section of \cite{MR3607468} for more details.  The metric $g_{\mathrm{KAdS}}$ extends to a smooth Lorentzian metric on $\mathcal{B}_{\mathcal H}$, $\mathcal{B}_{\mathcal{CH}}$ and we define $(\mathcal{M}_{\mathrm{KAdS}}, g_{\mathrm{KAdS}})$ as the Lorentzian manifold  constructed above. Moreover, $T:= \partial_t$ and $\Phi:= \partial_\phi$ extend to smooth Killing vector fields on $\mathcal M_{\textup{KAdS}}$ with $K_+ = T+\omega_+ \Phi$ and $K_- = T + \omega_- \Phi$.  
\paragraph{Kerr--AdS-star coordinates}
On the exterior region $\mathcal{R}$ we define an additional system of coordinates $(t^\ast,r,\theta,\phi^\ast)$ from the Boyer--Lindquist coordinates  through
\begin{align}\label{eq:kerrstar}
	t^\ast := t + A(r), \; r=r, \; \theta = \theta,\; \phi^\ast := \phi+ B(r) - \omega_+ ( t+A(r)),
\end{align}
where $\frac{\d A}{\d r} = \frac{2 M r}{\Delta (1+\frac{r^2}{l^2})}$ and $\frac{\d B}{\d r} = \frac{a \Xi}{\Delta}$ and $A= B = 0$ at $r=+\infty$. As shown in \cite[Section~5]{warnick_boundedness_and_growth}, these coordinates extend smoothly to the event horizon $\mathcal H_R$ and we call the coordinates  $(t^\ast,r,\theta,\phi^\ast)$ covering $\mathcal R \cup \mathcal H_R$ Kerr--AdS-star coordinates. Note that the event horizon is characterized as $\mathcal H_R = \{ r= r_+\}$ and that $K_+ = \partial_{t^\ast}$ in these coordinates.

\paragraph{Foliations and Initial Hypersurface}
We foliate the region $\mathcal R \cup \mathcal H_R$ with constant $t^\ast$ hypersurfaces $\Sigma_{t^\ast}$ which are spacelike and intersect the event horizon at $r=r_+$. 
For the initial data we will consider the axisymmetric spacelike hypersurface \begin{align}
 \label{eq:defnsigma0}
\Sigma_0 :=  \Sigma_{t=0}= \mathcal{R} \cap \{ t_{\mathcal{R}} = 0 \} .  \end{align}
 Note that $\Sigma_0$ does not contain the bifurcation sphere $\mathcal{B}_{\mathcal H}$.  We will impose initial data on $\Sigma_0 \cup \mathcal{B}_{\mathcal H} \cup  \mathcal{H}_L $. We will choose the support of our initial data to lie in a compact subset $K \subset \Sigma_0 \cap \{ r\geq 2 r_+ \}$. Thus, we assume vanishing data on $\mathcal{H}_L \cup \mathcal{B}_{\mathcal H}$. This will be made precise in \cref{sec:initialdata}.
 \paragraph{Boundary conditions}
    Note that the conformal boundary $\mathcal{I}$,  expressed formally as $\{ r  = +\infty \}$, is timelike, as a consequence, $(\mathcal{M}_{\mathrm{KAdS}}, g_{\mathrm{KAdS}})$ is not globally hyperbolic. Thus, in addition to Cauchy data for \eqref{eq:wavekerr}, we will also impose Dirichlet boundary conditions at $\mathcal I = \{ r=+\infty\}$. 
    		\subsection{Conventions} \label{sec:conventions}
    If $X$ and $Y$ are two (typically non-negative) quantities, we use $X\lesssim Y$ of $Y\gtrsim X$ to denote that $X\leq C(M,a,l) Y$ for some constant $C(M,a,l) >0$  depending continuously on the black hole parameters $(M,a,l) $, unless explicitly stated  otherwise.    We also use $X=O(Y)$ for $|X| \lesssim Y$. We use $X\sim Y$ for $X \lesssim Y \lesssim X$ and if the constants appearing in $\lesssim, \gtrsim, \sim$ or $O$ depend on additional parameters $a_i$ we include those as subscripts, e.g.\ $X \lesssim_{a_1 a_2} Y$. Similarly, implicit constants in ``sufficiently small'' or ``sufficiently large'' may also depend continuously on $M,a,l$.

In \cref{sec:initialdata} we will \emph{fix} parameters $(\mathfrak m,\mathfrak a)  \in \mathscr P_\textup{Blow-up}$ and all constants appearing in $\lesssim$ and $\gtrsim$  throughout \cref{sec:initialdata} \cref{sec:exterioranalysis}, \cref{sec:interior} will only depend on this particular choice  and on $l>0$ as   in \eqref{eq:defnofl}.

Further, we denote the total variation of a function $f\colon \mathbb R \to \mathbb R$ in the interval $(a,b)$  with $\mathcal V_{a,b}(f)$ defined as $
\mathcal V_{a,b} (f) := \sup_{ P}\sum_{i=1}^{n_p -1} |f(x_{i+1}) - f(x_i)|$,
where the supremum runs over the set of all partitions  of the given interval, see \cite[Chapter~1, \S11]{olver}.
  
  	\subsection{Norms and energies}
  	To state the well-posedness result of \eqref{eq:wavekerr} and the logarithmic decay result on the Kerr--AdS exterior, we define the following norms and energies in the exterior region $\mathcal R \cup \mathcal H_R$. These are based on the works \cite{wellposed,gustav,quasimodes}, where more  details can be found. In the region $\mathcal R \cup \mathcal H_R$ we let $\slashed g$ and $\slashed \nabla$ be the induced metric and induced connection of the spheres $\mathbb S_{t^\ast,r}^2$ of constant $t^\ast $ and $r$. For a smooth function $\psi$ we denote $|\slashed\nabla \dots \slashed \nabla \psi|^2 = \slashed g^{AA'} \cdots g^{BB'} \slashed \nabla_A \dots \slashed \nabla_B \bar \psi \slashed \nabla_{A'} \dots \slashed \nabla_{B'} \psi$. Now, we define  energy densities in Kerr--AdS-star coordinates as
\begin{align}
&e_1[\psi]:= 	\frac{1}{r^2} |\partial_{t^\ast} \psi|^2 + r^2 |\partial_r \psi|^2 + |\slashed \nabla \psi|^2 + |\psi|^2,\\
&e_2[\psi] := e_1[\psi] + e_1[\partial_{t^\ast} \psi] + \sum_{i=1}^3 e_1[\Omega_i \psi] + r^4 |\partial_r \partial_r \psi|^2 + r^2 |\partial_r \slashed \nabla \psi|^2 + |\slashed \nabla \slashed \nabla \psi|^2,
\end{align}
and analogously for higher order energy densities. Here, $(\Omega_i)_{i=1,2,3}$ denote the angular momentum operators on the unit sphere in $\theta,\phi^\ast$ coordinates.  We also define the energy norms on constant $t^\ast$ hypersurfaces as 
\begin{align}
	\|\psi\|^2_{H^{0,s}_{\textup{AdS}} (\Sigma_{t^\ast})} =& \int_{\Sigma_{t^\ast}} r^s |\psi|^2 r^2 \d r \sin \theta \d \theta \d \phi^\ast,\\
 	\|\psi\|^2_{H^{1,s}_{\textup{AdS}} (\Sigma_{t^\ast})} =& \int_{\Sigma_{t^\ast}} r^s\left( r^2 |\partial_r \psi|^2 + |\slashed \nabla \psi|^2 + |\psi|^2 \right)r^2 \d r \sin \theta \d \theta \d \phi^\ast,\\
 	\|\psi\|^2_{H^{2,s}_{\textup{AdS}} (\Sigma_{t^\ast})} =& \|\psi\|^2_{H^{1,s}_{\textup{AdS}} (\Sigma_{t^\ast})} +  \int_{\Sigma_{t^\ast}}r^s\big( r^4 |\partial_r \partial_r \psi|^2 + r^2 |\slashed \nabla \partial_r \psi|^2 \nonumber \\ & \hspace{5cm} + |\slashed \nabla \slashed \nabla \psi|^2 \big)r^2 \d r \sin \theta \d \theta \d \phi^\ast.
\end{align}
We now denote the space $H_\textup{AdS}^{k,s} (\Sigma_{t^\ast})$ as the space of functions with $\nabla^i \psi \in L^2_{loc} (\Sigma_{t^\ast})$ for $i=0, \dots, k$  and such that $\| \psi\|^2_{ H^{k,s}_{\textup{AdS}} (\Sigma_{t^\ast})} < \infty$ and we denote with $CH_\textup{AdS}^{k}$ the space of functions $\psi$  on $\mathcal R \cup \mathcal H_R$ such that $\psi \in \bigcap_{q = 0, \dots, k} C^q(\mathbb R_{t^\ast}; H_{\textup{AdS}}^{k-q,s_q} (\Sigma_{t^\ast})  )$, where $s_k = -2, s_{k-1}  =0$ and $s_j =0 $ for $j =0,\dots, k-2$. 
 \subsection{Well-posedness and \texorpdfstring{$\log$}{log}-decay on the exterior region}
In the following we state well-posedness for \eqref{eq:wavekerr} and decay solutions  with Dirichlet boundary conditions. The following theorem is a summary of results by Holzegel, Smulevici and Warnick shown in \cite{hol10,wellposed,gustav,quasimodes,MR3150167}. 
\begin{theorem}[{\cite{hol10,wellposed,gustav,quasimodes,MR3150167}}] \label{thm:wellposedanddecay}
	Let the initial data $\Psi_0,\Psi_1 \in C_c^\infty(\Sigma_0)$. Assume Dirichlet boundary conditions at $\mathcal I$ and vanishing incoming data on $\mathcal{H}_L \cup \mathcal{B}_{\mathcal H}$. Then, there exists a unique solution $\psi \in C^\infty(\mathcal{M}_{\textup{KAdS}} \setminus \mathcal{CH})$ of \eqref{eq:wavekerr} such that $\psi\vert_{\Sigma_0} = \Psi_0, n_{\Sigma_0}\psi\vert_{\Sigma_0} = \Psi_1$, $\psi\restriction_{\mathcal{H}_L \cup \mathcal{B}_{\mathcal H}} =0$. The solution satisfies $\psi\restriction_{\mathcal R \cup \mathcal H_R} \in CH_\textup{AdS}^k$ for every $k \in \mathbb N$.  We also have boundedness of the energy 
\begin{align}
\label{eq:boundednesskerr}
&\int_{\Sigma_{t^\ast_2}} e_1[\psi] r^2 \sin\theta \d r \d\theta\d\phi^\ast \lesssim 	\int_{\Sigma_{t^\ast_1}} e_1[\psi] r^2 \sin\theta \d r \d\theta\d\phi^\ast
\end{align}
for $t^\ast_2 \geq t_1^\ast \geq 0 $ as well as for all higher order energies. Further, the energy along the event horizon is bounded by the initial energy as 
\begin{align} \label{eq:eventhorizonbound0}
\sum_{\substack{1\leq i_1 + i_2 \leq k \\ i_2\geq 1}}	\int_{\mathcal H_R\cap\{ t^\ast \geq t^\ast_0  \}} |\slashed \nabla^{i_1} K_+^{i_2} \psi|^2 r^2 \sin \theta  \d v\d \theta \d \tilde \phi_+ \lesssim_{k} \int_{\Sigma_{t^\ast_0}} e_k[\psi] r^2 \sin\theta \d r \d \theta \d \phi.\end{align}
for any $t_0^\ast \geq0$. 

Moreover, the energy of $\psi$ decays 
\begin{align}
&\int_{\Sigma_{t^\ast}} e_1[\psi] r^2 \sin\theta \d r \d\theta\d\phi^\ast \lesssim \frac{1}{[\log(2+t^\ast)]^2} \int_{\Sigma_{t^\ast_0}} e_2[\psi] r^2 \sin\theta \d r \d\theta\d\phi
\label{eq:decay}
\end{align}
	for all $t^\ast \geq t^\ast_0\geq 0$ and similar estimates hold for all higher order energies. Similarly, by commuting and applications of the Sobolev embeddings, $\psi$ and all its derivatives also decay pointwise
	\begin{align}
		\label{eq:logdecay}		\sum_{0 \leq i_1 + i_2 + i_3 \leq k} |\slashed \nabla^{i_1} \partial_{t^\ast}^{i_2} \partial_r^{i_3} \psi|^2 \lesssim_k \frac{1}{[\log(2+t^\ast)]^2} \int_{\Sigma_{t^\ast_0}} e_{k + 3}[\psi] r^2\sin\theta \d r \d \theta \d \phi
	\end{align}
for $t^\ast \geq t_0^\ast\geq 0 $. 
\end{theorem}
By general theory (a local in time energy estimate) all norms on the right-hand side of \eqref{eq:boundednesskerr}--\eqref{eq:logdecay}   are bounded in terms of a non-degenerate energy of $\psi$ on $\Sigma_0$, i.e.\ in terms of weighted Sobolev norms (of appropriate order) of $\Psi_0$ and $\Psi_1$ on $\Sigma_0$. In particular, since $\Psi_0, \Psi_1$ are smooth and compactly supported, all right-hand sides of \eqref{eq:boundednesskerr}--\eqref{eq:logdecay} are finite.

It should be noted that \eqref{eq:eventhorizonbound0} merely gives a bound on $\int_{\mathcal{H}_R} |K_+ \psi|^2$ which does   not control   the full $L^2$-norm of $\psi$ along in the event horizon. However,   one     obtains control of the $L^2$-norm via an ``inverse-commutation'' argument relying on \cite[Theorem 4.9]{quasi_warnick}. 
\begin{prop}\label{eq:estimateonl2norm}
	Let the initial data $\Psi_0,\Psi_1 \in C_c^\infty(\Sigma_0)$. Assume Dirichlet boundary conditions at $\mathcal I$ and vanishing incoming data on $\mathcal{H}_L \cup \mathcal{B}_{\mathcal H}$. In view of \cref{thm:wellposedanddecay}, denote by $\psi$ the unique solution with $\psi\vert_{\Sigma_0} = \Psi_0, n_{\Sigma_0}\psi\vert_{\Sigma_0} = \Psi_1$, $\psi\restriction_{\mathcal{H}_L \cup \mathcal{B}_{\mathcal H}} =0$.  Then, 
\begin{align}\label{eq:finitenessofl2norm}
D^k_{\Hp_R}[\psi]:= 	\sum_{ {0\leq i_1 + i_2 \leq k }}	\int_{\mathcal H_R } |\slashed \nabla^{i_1}K_+^{i_2} \psi|^2 r^2 \sin \theta  \d v\d \theta \d \tilde \phi_+ <  \infty
	\end{align}
for each  $k \in \mathbb N$. 
\begin{proof}
By a local in time energy estimate it suffices to show that   \begin{align}	\sum_{ {0\leq i_1 + i_2 \leq k }}	\int_{\mathcal H_R \cap \{ t^\ast \geq 0 \}} |\slashed \nabla^{i_1}K_+^{i_2} \psi|^2 r^2 \sin \theta  \d v\d \theta \d \tilde \phi_+ < \infty.\end{align}
Further, we also have that  the solution $\psi$ has  finite energy on $\Sigma_{t^\ast_0}$ for $t^\ast_0:=0$ of all orders in the sense that \begin{align} \int_{\Sigma_{t^\ast_0}} e_k[\psi] r^2 \sin\theta \d r \d \theta \d \phi < \infty \label{eq:boundonpsi1andpsi0}\end{align}
  for every $k\in \mathbb N$. We   denote $(\psi_0,\psi_1):= (\psi, K_+ \psi)\restriction_{\Sigma_{t^\ast_0}}$.

In view of the above and  \eqref{eq:eventhorizonbound0} in  \cref{thm:wellposedanddecay}, to obtain \eqref{eq:finitenessofl2norm},  it suffices to show that there exist  data $(\tilde \psi, K_+ \tilde \psi)\restriction_{\Sigma_{t^\ast_0}} = (\tilde \psi_0, \tilde \psi_1)$ for \eqref{eq:wavekerr} such that the arising solution $\tilde \psi$ satisfies  $K_+ \tilde \psi = \psi$  and moreover $(\tilde \psi_0, \tilde \psi_1)$ are sufficiently regular as to apply the (twisted) energy estimate associated to $K_+$.  We set $\tilde \psi_1 := \psi_0  $ and it remains to construct $\tilde \psi_0$ via inverting an elliptic operator to ensure $K_+ \tilde \psi = \psi$. 

To construct $\tilde \psi_0$ we will use  \cite[Theorem~4.9]{quasi_warnick}. To apply it we briefly recall the theory developed in  \cite{quasi_warnick,warnick_boundedness_and_growth}. Note from \cite[Lemma~5.2]{quasi_warnick} and the Hawking--Recall bound \eqref{eq:hawkingreallbound}  that the Kerr--AdS  exterior to the future of $\Sigma_{t^\ast_0}$ is a globally stationary asymptotically Anti-de~Sitter black hole spacetime in the sense of \cite[Definition 2.6]{quasi_warnick} with stationary vector field $K_+$. Thus, we can apply the general framework of \cite{quasi_warnick}. 
Following \cite{quasi_warnick} we write \eqref{eq:wavekerr}  as $L\psi =0$ for a strongly hyperbolic operator with $W=0$, $V=\frac{2}{l^2}$ as in \cite[Definition~2.7]{quasi_warnick}, more precisely,   $L:= A (\Box_g + \frac{2}{l^2}) $ for $ A = \frac{-1}{g(\nabla t^\ast, \nabla t^\ast )}$. As in  \cite[p.~998]{quasi_warnick} we decompose $L\tilde \psi = P_2 \tilde \psi + P_1 K_+  \tilde \psi + K_+ K_+ \tilde \psi$, where $P_1$ is a differential operator of first order on $\Sigma_{t^\ast_0}$ and $P_2$ is a (degenerate) elliptic operator on $\Sigma_{t^\ast_0}$. We further recall the following natural norms  from  \cite[p.~976]{quasi_warnick}:
\begin{align}
\| \tilde \psi \|^2_{\underline L^2 {(\Sigma_{t^\ast_0} )}} := \int_{\Sigma_{t^\ast_0}} |\tilde \psi|^2 r^{-1} \dvol_{\Sigma_{t^\ast_0}}, \;  \| \tilde \psi \|^2_{\underline H^1(\Sigma_{t^\ast_0})}  :=  \int_{\Sigma_{t^\ast_0}} \left( |\tilde \nabla\tilde \psi|^2 + |\tilde \psi|^2r^{-2} \right)  r  \dvol_{\Sigma_{t^\ast_0}}, \label{eq:norms}
\end{align}
where $\tilde \nabla_{\mu}\tilde  \psi := r^{-1} \nabla_{\mu} (r \tilde \psi), \tilde \nabla_{\mu}^\dagger \tilde \psi :=- r\nabla_{\mu} (r^{-1} \tilde \psi)$ are the twisted derivatives and the norms in \eqref{eq:norms} are with respect to the induced metric on $\Sigma_{t^\ast_0}$. As in \cite{quasi_warnick} we define respectively $\underline L^2(\Sigma_{t^\ast_0})$ and $\underline H^1(\Sigma_{t^\ast_0})$ as the completion of smooth functions on  $\Sigma_{t^\ast_0}$ which are supported away from $\mathcal I$ in the norms $\| \cdot \|_{\underline L^2 {(\Sigma_{t^\ast_0} )}}$ and $ \| \cdot \|_{\underline H^1(\Sigma_{t^\ast_0})}  $, respectively. 

In order  to construct $\tilde \psi_0$ we need to invert $P_2$, more precisely, we need to solve $P_2 \tilde \psi_0 =- P_1 \psi_0 - \psi_1$, where we note that $P_1\colon\underline H^1(\Sigma_{t^\ast_0}) \to \underline L^2(\Sigma_{t^\ast_0})$ is a bounded operator (see \cite[p.\ 1002]{quasi_warnick}). In particular, we have $P_1 \psi_0 + \psi_1 \in \underline L^2(\Sigma_{t^\ast_0})$ in view of \eqref{eq:boundonpsi1andpsi0}. 
We note that $P_2 = \hat L_0$, where $\hat L_0= \hat L_{s=0}\colon \operatorname{dom}(\hat L_{s=0})  \to \underline L^2 (\Sigma_{t^\ast_0})$ is as in \cite[eqn.~(4.1), $s=0$]{quasi_warnick}. Now, we apply \cite[Theorem~4.9]{quasi_warnick} with $k=0$ and $s=0$. Indeed, $k=0$ and $s=0$ are  valid because  $w_L=0$ (recall $W=0$ and \cite[Definition~3.7]{quasi_warnick}). Since $0\notin \Lambda_{\mathrm{QNF}}^0$ (no stationary solutions  exist \cite[Corollary~1.3]{quasi_warnick}), we have from \cite[Theorem~4.9]{quasi_warnick}  that the operator $P_2^{-1} \colon \underline L^{2}(\Sigma_{t^\ast_0} ) \to  \underline H^{1}(\Sigma_{t^\ast_0})$ exists and is bounded. Hence,  $\tilde \psi_0 = - P_2^{-1} ( P_1 \psi_0 + \psi_1) \in \underline H^{1}(\Sigma_{t^\ast_0})$. (In fact, $\tilde \psi_0$ can be shown to be more regular which is however not needed for the proof.)

Finally, let $\tilde \psi$ be the unique solution to \eqref{eq:wavekerr} arising from initial data $(\tilde \psi, K_+ \tilde \psi)\restriction_{\Sigma_{t^\ast_0}} = (\tilde \psi_0, \tilde \psi_1)$. This is well-posed by \cite[Theorem~2.3]{quasi_warnick}. Then, by construction we have that $K_+ \tilde \psi  = \psi$. Now, by the twisted energy estimate for $\tilde \psi$ associated to  $K_+$ (see  e.g.\ \cite[Proposition~3]{warnick_boundedness_and_growth} and \cite[Theorem~3.4 (i), $\gamma=0$]{quasi_warnick})  we have $ \int_{\Hp_R\cap\{t^\ast \geq t_0^\ast \}} |\psi|^2 r^2 \d v \sin\theta \d\theta \d \tilde \phi_+  \lesssim  \int_{t^\ast \geq t^\ast_0} |K_+ \tilde \psi|^2 r_+^2 \d t^\ast \sin\theta \d\theta \d \phi^\ast \lesssim  \|\tilde \psi_0 \|_{\underline H^1(\Sigma_{t^\ast_0})} + \| \tilde \psi_1\|_{\underline L^2(\Sigma_{t^\ast_0})} < \infty$.  This concludes the proof. 
\end{proof} 
\end{prop}

\subsection{Separation of variables: Radial o.d.e., angular o.d.e.\ and coupling constants \texorpdfstring{$\lambda_{m\ell}(a\omega)$}{lambdamlaw} }
The wave equation \eqref{eq:wavekerr} is formally separable \cite{carter} and we can consider pure mode solutions in the Boyer--Lindquist coordinates of the form
	\begin{align}
	\psi(t,r,\theta,\phi) = \frac{u(r)}{\sqrt{r^2 + a^2}} e^{-i\omega t} S_{m\ell}(a\omega, \cos \theta)e^{i m\phi}, \; m \in \mathbb Z, \omega \in \mathbb R
	\end{align}
	for two unknown functions $u(r)$ and $S_{m\ell}(a \omega, \cos\theta)$. 	Plugging   this ansatz into \eqref{eq:wavekerr}  leads to a coupled system of o.d.e's.  The angular o.d.e.\ is the eigenvalue equation of the operator $P (a \omega)$  which reads
	\begin{align}\label{eq:Pwithaomega}
		P(a \omega) S_{m\ell}(a\omega, \cos \theta ) = \lambda_{m\ell}(a\omega) S_{m\ell}(a \omega,\cos\theta),
	\end{align}
	where 
		\begin{align}\nonumber 
	P(\xi) f =	P_m(\xi) f = & - \frac{1}{\sin\theta} \partial_\theta (\Delta_\theta \sin\theta \partial_\theta f) + \frac{\Xi^2 m^2}{\Delta_\theta \sin^2 \theta} f - \Xi \xi^2 \Delta_\theta^{-1} \cos^2 \theta f \\ &  - 2 m\xi \frac{\Xi}{\Delta_\theta} \frac{a^2}{l^2} \cos^2 \theta f + \frac{2}{l^2}a^2 \sin^2\theta f, \; \; \; \xi \in \mathbb R. \label{eq:spheroidalharmonicoperator}
		\end{align}
 The operator \eqref{eq:spheroidalharmonicoperator} is realized as a self-adjoint operator on a suitable domain in $L^2((0,\pi); \sin\theta \d \theta)$.  As a Sturm--Liouville operator, the spectrum of $P(a\omega)$ consists of simple eigenvalues $\lambda_{m\ell}(a\omega)$, where $\ell \in \mathbb Z_{\geq |m|}$  labels the eigenvalue in ascending order.  The eigenvalues $\lambda_{m\ell}(a \omega )$ of $P(a\omega)$  couple the angular o.d.e.\ to the radial o.d.e. 
	\begin{align}\label{eq:radial}
	-u^{\prime \prime} + (V - \omega^2 )u = 0,
	\end{align}
	where ${}^\prime := \frac{\d }{\d r^\ast}$.
	We also use the notation $\tilde V:= V - \omega^2$, where the potential $V$ is given by 
	\begin{align}\label{eq:defnofvintermsofv0v1}
	V =  V_0 + V_1
	\end{align}
	with purely radial part
	\begin{align} \label{eq:defnv1}
	V_1:=  \frac{-\Delta^2 3r^2}{(r^2+a^2)^4} + \Delta \frac{5 \frac{r^4}{l^2} + 3 r^2\left(1 + \frac{a^2}{l^2} \right) - 4 M r + a^2}{(r^2+a^2)^3} - \frac{2 \Delta}{l^2} \frac{1}{r^2+a^2}
	\end{align} 
	and frequency dependent part
	\begin{align}
	V_0 := \frac{\Delta (\lambda_{m\ell}(a \omega) + \omega^2 a^2 ) - \Xi^2 a^2 m^2 - 2m\omega a \Xi (\Delta - (r^2 +a^2))}{(r^2 + a^2)^2}.
	\end{align}  
 We will be particularly interested in the case for which the frequency   $\omega$ coincides with the interior scattering poles, i.e.\ $\omega=\omega_- m$. Moreover, in order to be in the regime of stable trapping on the exterior we also want $|\omega|$ and $|m|$ to be large. Hence, we will think of $ \frac{1}{m}$ as a small semiclassical parameter. In particular, setting $\omega = \omega_- m$ in \eqref{eq:radial} and separating out the $m^2$ we obtain
 \begin{align} \label{eq:radialomega}
 	-u'' + (m^2 V_{\textup{main}} + V_{1})u =0,
 \end{align}
 where $V_1$ is as in \eqref{eq:defnv1} and
 \begin{align}\label{eq:highfrequencypart}
 	V_{\textup{main}}  := \frac{V_0(\omega = \omega_- m) - \omega_-^2 m^2}{m^2} =  \frac{\Delta}{(r^2+a^2)^2} \left(\frac{\lambda_{m \ell}(a\omega_- m)}{m^2} + \omega_-^2 a^2 - 2  a \omega_- \Xi \right) - ( \omega_- - \omega_r )^2.
 \end{align}
In \eqref{eq:highfrequencypart} we  also used  $\omega_r= \frac{a \Xi}{r^2+a^2}$ as defined in \eqref{eq:definitionofomega-}.
We begin our analysis with the angular o.d.e.~\eqref{eq:Pwithaomega} in the following \cref{sec:angularode}.
  \section{The angular o.d.e.}
  \label{sec:angularode}
For the operator $P(\xi)$ as in  \eqref{eq:spheroidalharmonicoperator} we change   variables to $x=\cos\theta$. This is a unitary transformation and thus, the eigenvalues of $P(\xi)$ are equal to the eigenvalues of $P_x$ given by \begin{align}
\label{eq:spheroidalharmonicoperatorx0}
P_x (\xi):=	 -
\frac{\d }{\d x} \left(\Delta_x (1-x^2)  \frac{\d }{\d x} \cdot \right)+ \frac{\Xi^2 m^2}{\Delta_x (1-x^2) } - \Xi \xi^2 \frac{ x^2}{\Delta_x}- 2 m\xi \frac{\Xi}{\Delta_x} \frac{a^2}{l^2} x^2 + \frac{2}{l^2}a^2 (1-x^2).
\end{align} 
The Sturm--Liouville operator $P_x$ is realized as a self-adjoint operator acting on a domain $\mathcal D \subset L^2(-1,1)$  which can be explicitly characterized as \begin{align}\mathcal D =\{ f \in L^2(-1,1) \colon f \in AC^1(-1,1), P_x f \in L^2(-1,1) , \lim_{x\to \pm 1 } (1-x^2) f'(x) = 0 \text{ if } m=0 \},\label{eq:defndomain}
\end{align}
see e.g.\ \cite[Chapter~10.3, Theorem 10.7]{teschl} which, mutatis mutandis, also applies to \eqref{eq:spheroidalharmonicoperatorx0} and \eqref{eq:defndomain}. In \eqref{eq:defndomain}, $AC^1(-1,1)$ denotes the space of differentiable functions with absolutely continuous derivative.

Having the same spectrum as $P$, the operator $P_x$  has   eigenvalues $(\lambda_{m\ell}(\xi))_{\ell \geq |m|}$ with corresponding real-analytic eigenfunctions $S_{m\ell} = S_{m\ell}(\xi,x)$ which satisfy
\begin{align}
P_x S_{m\ell} = \lambda_{m\ell} S_{m\ell} \;\; \text{ and } \;\; 	\| S_{m\ell}(\xi)\|_{L^2(-1,1)} = 1. 
\end{align}
We note that for $\xi = a=0$, the eigenvalues $(\lambda_{m\ell})_{\ell \geq|m|}$ reduce to the eigenvalues of the Laplacian on the sphere $\lambda_{m\ell}(a =\xi =0) = \ell (\ell +1)$. We also define the shifted eigenvalues
\begin{align}\Lambda_{m\ell}(\xi):= \lambda_{m\ell}(\xi) + \xi^2.\end{align}
A computation (see \cite[Proof of Lemma~5.1]{gustav}) shows that  \begin{align}P_x (\xi) + \xi^2 - \frac{2}{l^2} a^2 ( 1- x^2) \geq \Xi^2 P_x(\xi =0, a =0 )\end{align} in the sense of self-adjoint operators acting on $\mathcal D\subset L^2(-1,1)$. Hence, 
\begin{align}
	\Lambda_{m\ell}(\xi) \geq \Xi^2 \ell ( \ell + 1) \geq \Xi^2 |m| (|m|+1).
\end{align}

Having recalled basic properties of the angular problem we now focus on the interior scattering poles $\omega=\omega_- m$ for large $m$. In particular, we will only consider $m\neq 0$ for the rest of \cref{sec:angularode}. 
  \subsection{Angular potential  \texorpdfstring{$W_1$}{W1} at interior scattering poles in semi-classical limit}
  \label{sec:3.1}
 In the current \cref{sec:3.1} and in the following \cref{sec:332} we will consider the operator 
 \begin{align}
  \nonumber
  P_{\omega_-}:= P_x (\xi = a m \omega_-)  
  	= &-  \frac{\d }{\d x} \left(\Delta_x (1-x^2)  \frac{\d }{\d x} \cdot \right)+ \frac{\Xi^2 m^2}{\Delta_x (1-x^2) } - \Xi a^2m^2\omega_-^2 \frac{ x^2}{\Delta_x}\\ &- 2 m^2 a \omega_- \frac{\Xi}{\Delta_x} \frac{a^2}{l^2} x^2 + \frac{2}{l^2}a^2 (1-x^2)
 \label{eq:pomega-}
  \end{align} 
  with corresponding eigenvalues  $\lambda_{m\ell}:= \lambda_{m\ell}(a \omega_- m)$. We re-write the eigenvalue problem  \begin{align}\label{eq:evproblempomega-}
  P_{\omega_-} f= \lambda f
  \end{align} as 
  \begin{align}\label{eq:equivalentformulation}
  	\tilde P_{\omega_-} f = 0,
  \end{align}
  where 
  \begin{align}
 & \tilde P_{\omega_-}   := - \Delta_x (1-x^2) \frac{\d }{\d x} \left(\Delta_x (1-x^2)  \frac{\d }{\d x} \cdot \right) + m^2 W_1(x) + P_\textup{error},\\
& P_\textup{error}:= \Delta_x (1-x^2) \frac{2}{l^2}a^2 (1-x^2)
\end{align}
 and
  \begin{align} \label{eq:defnW1}
  	W_1 :=  \Xi^2 - \left[ \Xi a^2 \omega_-^2 + 2 a \omega_- \Xi \frac{a^2}{l^2} \right] x^2 (1-x^2) - \tilde \lambda \Delta_x (1-x^2),
  \end{align}
  with 
  \begin{align}\label{eq:tildelambda}
  	\tilde \lambda:= \frac{\lambda}{m^2}.
  \end{align}
    In the semi-classical limit $m^2 \to\infty$ we   consider $P_\textup{error}$ as a perturbation and $W_1$ determines the leading order terms of the eigenvalues and eigenfunctions. Consequently, our analysis focuses on $W_1$ which we analyze in the following lemma.
  \begin{lemma}\label{lem:w1} 
	\begin{enumerate}Let $W_1$ be the angular potential defined in \eqref{eq:defnW1}.
\item 	For $\tilde \lambda <  \Xi^2$, we have $W_1 > 0$ for $x\in [0,1]$. 	
\item	For $\tilde \lambda = \Xi^2$, we have $W_1 >0$ on $(0,1]$ and $W_1(x=0) =0$.
\item	For $\tilde \lambda > \Xi^2 $, the potential $W_1$ has exactly one root in $x\in [0,1]$ and satisfies 
  	\begin{align}\label{eq:caselambdatildegeqxi^2}
  	\frac{\d W_1}{\d x} \gtrsim \tilde \lambda x
  	\end{align}
  	for $x\in [0,1]$. We call this root $x_0$ which also satisfies  $x_0 \in (0,1)$.
  	\end{enumerate}
  	\begin{proof}
  		  We start by expanding $W_1$ and obtain
  		\begin{align}
  		W_1(x) = \Xi^2-\tilde \lambda + a_1 x^2 + a_2 x^4
  		\end{align}
  		with 
  		\begin{align} \nonumber 
  		a_1 & = \tilde \lambda \left(1+\frac{a^2}{l^2}\right) - \frac{a^4(a^2-l^2)^2 (a^2+l^2+2 r_-^2)}{l^6 (a^2+r_-^2)^2}
  		\\ \nonumber	&= \Xi^2 \left(1+\frac{a^2}{l^2}\right) - \frac{a^4(a^2-l^2)^2 (a^2+l^2+2 r_-^2)}{l^6 (a^2+r_-^2)^2} + (\tilde \lambda-\Xi^2) \left(1+\frac{a^2}{l^2}\right) 
  		\\	&=  \frac{(a-l)^2 (a+l)^2 (2a^2 l^2 r_-^2 +(a^2+l^2)r_-^4)}{l^6 (a^2+r_-^2)^2} + (\tilde \lambda-\Xi^2) \left(1+\frac{a^2}{l^2}\right)
  		\end{align}
  		and
  		\begin{align}
  		a_2 & = \frac{a^4(a^2-l^2)^2 (a^2+l^2+2r_-^2)}{l^6(a^2+r_-^2)^2} - \frac{a^2}{l^2}\tilde \lambda.
  		\end{align}
  		We also note that \begin{align}\label{eq:w1of0}  		W_1(x=0) = \Xi^2 - \tilde{\lambda}.\end{align}
  	  
  	  	We now consider the case $\tilde \lambda \geq \Xi^2$ and remark that \begin{align}\frac{\d W_1}{\d x } = 2 a_1 x + 4 a_2 x^3 .  \end{align}
  		We look at two cases now, $a_2\geq 0$ and $a_2 <0$. If $a_2 \geq 0$, then we directly infer that $\frac{\d W_1}{\d x} \geq 2 a_1 x$. If $a_2 <0$, then we use that $x^3< x$ and estimate
  		\begin{align}\frac{\d W_1}{\d x } = 2 a_1 x + 4 a_2 x^3 \geq ( 2a_1 + 4 a_2 )x.  \end{align}
  		 Now, a direct computation yields 
  		 \begin{align}
  		 	2a_1 + 4 a_2 = 2 \Xi \left( \Xi \frac{a^2}{l^2} \frac{a^2+l^2+2 r_-^2}{(a^2+r_-^2)^2} + \tilde \lambda \right) \gtrsim \tilde \lambda  	.	 \end{align}
  		Note that this shows \eqref{eq:caselambdatildegeqxi^2} for $x\in[0,1]$ and we conclude \textit{3.} Together with \eqref{eq:w1of0}, this also shows that  $W_1(x)>0$ for $x\in (0,1]$ and $\tilde \lambda = \Xi^2$ such that we have \textit{2.}
  		
  		Finally for $\tilde \lambda < \Xi^2$, we have $W_1>0$ everywhere because for each fixed $x$, the function $\tilde \lambda \mapsto W_1(x)$ is strictly  decreasing.
  	  	\end{proof}
  \end{lemma}

 \subsection{Angular eigenvalues at interior scattering poles with \texorpdfstring{$\lambda_{m_i\ell_i} = \tilde \lambda_0 m_i^2 +O(1)$}{lambda = tilde lambda m2 +O(1) }}
 \label{sec:332}
For the proof of the Baire-genericity of the set $\mathscr{P}_{\textup{Blow-up}}$ in \cref{sec:quasimodes} we will use that there exists a sequence of angular eigenvalues of the form $\lambda_{m_i\ell_i} =\lambda_{m_i\ell_i}(\omega=\omega_- m)= \tilde \lambda_0 m_i^2 +O(1)$ at the interior scattering poles. To show this, we will use the following well-known result (\cref{prop:highfreq}) on  the semi-classical distribution of eigenvalues. The proof of \cref{prop:highfreq} relies on suitable connection formulas of Airy functions and can be found in \cite[Chapter~13, \S8--\S9.1]{olver}. We also recall the definition $\mathcal V$ of the total variation as in \cref{sec:conventions}.
 
 \begin{prop}[{\cite[Chapter~13, \S8.2--\S9.1]{olver}}]\label{prop:highfreq}
 	Consider a parameter $\epsilon \in I_0:= [\epsilon_0,\epsilon_1]$. 
Let $f_\epsilon(x),g_\epsilon(x) \in C^2(\mathbb R_x)$   for all $\epsilon\in I_0$ where we assume that  $f_\epsilon$ and $g_\epsilon$ depend continuously on $\epsilon$. Assume that $f_\epsilon(x)/[(x-\hat x_0(\epsilon)) (x-x_0(\epsilon))]$  is positive and  bounded away from zero uniformly for   $\epsilon\in I_0$. In particular, assume that $f_\epsilon$ has two simple roots at $\hat x_0(\epsilon)< x_0(\epsilon)$ with $-\infty < \inf_{\epsilon} \hat x_0(\epsilon)$ and  $ \sup_{\epsilon}  x_0(\epsilon)<+\infty$.  Assume further that the roots do not coalesce, i.e.\ that there exists a $x_1\in \mathbb R$ with $ \sup_{\epsilon}\hat x_0(\epsilon) < x_1 < \inf_{\epsilon}  x_0(\epsilon) $.
  Assume that for all $\epsilon\in I_0$ and for some $c>0$ sufficiently large,
$\int_{c}^x \sqrt{f_\epsilon}$ diverges as $x\to +\infty$ and $\int_{x}^{-c} \sqrt{f_\epsilon}$  diverges as $x\to -\infty$.

To make the above statements quantitative, define the error-control functions \begin{align}\label{eq:hepsilon1} H_\epsilon(x):= \int_{x_0(\epsilon)}^x \frac{1}{|f_\epsilon|^{\frac 14}} \frac{\d^2}{\d y^2} \left( \frac{1}{|f_\epsilon|^{\frac 14}}\right) - \frac{g_\epsilon}{|f_\epsilon|^{\frac 12}} - \frac{5|f_\epsilon|^2}{16|\zeta_\epsilon|^3}\d y,\\
	\hat H_\epsilon(x):= \int_{\hat x_0(\epsilon)}^x \frac{1}{|f_\epsilon|^{\frac 14}} \frac{\d^2}{\d y^2} \left( \frac{1}{|f_\epsilon|^{\frac 14}}\right) - \frac{g_\epsilon}{|f_\epsilon|^{\frac 12}} - \frac{5|f_\epsilon|^2}{16|\hat \zeta_\epsilon|^3}\d y\label{eq:hepsilon2} 
\end{align}
for $|\zeta_\epsilon|^3 := \left| \frac 32\int_{x_0(\epsilon)}^x\sqrt{f_\epsilon} \d y\right|^{2}  $, $|\hat \zeta_\epsilon|^3 := \left| \frac 32\int_{\hat x_0(\epsilon)}^x\sqrt{f_\epsilon} \d y\right|^{2}  $
and set
\begin{align}\nonumber 
	B(f_\epsilon,g_\epsilon)=\left| \frac{f_\epsilon'(x_1)}{f^{\frac 32}_\epsilon(x_1)}\right| &+\left|  \int_{x_1}^{x_0(\epsilon)} \sqrt{f_\epsilon}(x) \d x\right|^{-1} + \left|\int_{\hat x_0(\epsilon)}^{x_1} \sqrt{f_\epsilon}(x) \d x  \right|^{-1} \\ & + \mathcal{V}_{ x_1, +\infty}   (H_\epsilon) + \mathcal{V}_{ -\infty, x_1}   (\hat H_\epsilon). \label{eq:defnbfge}
\end{align}
Assume $B_0 := \sup_{\epsilon\in I_0} B(f_\epsilon,g_\epsilon)<+\infty$.

Then, for all $u$ sufficiently large and all $\epsilon\in I_0$, if the differential equation 
\begin{align}\label{eq:odebohrsommerfeld}
w''=(u^2 f_\epsilon + g_\epsilon )w,
\end{align}
admits a bound state  $w$ (i.e.\ a solution which is recessive at both ends $x\to \pm \infty$), then
\begin{align}\label{eq:bohrsommerfeld}
\frac 2\pi 	\int_{\hat x_0(\epsilon)}^{x_0(\epsilon)} \sqrt{-f_\epsilon} \d x + \vartheta_{f_\epsilon,g_\epsilon,u,n} = \frac{2n+ 1 }{u}
\end{align}
for some positive integer $n\in \mathbb N$ and for an  error  function $\vartheta_{f_\epsilon,g_\epsilon,u,n}$ which obeys $|\vartheta_{f_\epsilon,g_\epsilon,u,n}| \lesssim_{B_0} u^{-2}$. In particular, the implicit constant is independent of $\epsilon, u,n$. 

Conversely,  for all $u$ sufficiently large,  there exists an error  function $\vartheta_{f_\epsilon,g_\epsilon,u,n}$   which depends continuously on $\epsilon \in I_0$ and  satisfies $ |\vartheta_{f_\epsilon,g_\epsilon,u,n}| \lesssim_{B_0} u^{-2}$ such that \eqref{eq:odebohrsommerfeld} admits a bound state $w$ if there exist an $\epsilon \in I_0$ and a   $n \in \mathbb N$  satisfying \eqref{eq:bohrsommerfeld}.  
\begin{proof}
The above result follows from \cite[Chapter~13, \S8.2]{olver} but for convenience of the reader we briefly outline the steps in  \cite[Chapter~13, \S8.2]{olver} in our context.
We begin by noting that our assumption $B_0<\infty$ shows that the error terms in   (8.03) and (8.05) of  \cite[Chapter~13, \S8.2]{olver} are controlled by $O_{B_0}(u^{-1})$ uniformly for $u$ sufficiently large and $\epsilon\in I_0$.  Following the argument of \cite[Chapter~13, \S8.2]{olver}  we then conclude that for $u$ sufficiently large, \eqref{eq:odebohrsommerfeld} admits a bound state $w$  if and only if $\sin\left(u \int_{\hat x_0(\epsilon)}^{x_0 (\epsilon)} \sqrt{-f_\epsilon} \d x - \frac \pi 2\right)  = \Theta(f_\epsilon,g_\epsilon,u)$, for some error function  $\Theta$ which satisfies $\sup_{\epsilon\in I_0} |\Theta(f_\epsilon,g_\epsilon,u)| \lesssim_{B_0} u^{-1}$. By virtue of $f_\epsilon$ and $g_\epsilon$ depending continuously on $\epsilon$, we also obtain that $\Theta(f_\epsilon,g_\epsilon,u)$ depends continuously on $\epsilon$. Inverting $\sin$ around its zeros yields the claim. 
\end{proof}
 \end{prop}
 With the above proposition in hand we proceed to the main proposition of this subsection, where we recall that we still consider the case $\omega=\omega_- m$.
 \begin{prop}\label{prop:angular} Let $\mathfrak p_0 \in \mathscr{P}$ be arbitrary but fixed. 
 Then, for almost every $\tilde \lambda_0 \in  (\Xi^2,\infty)$ (more precisely, for every $\tilde \lambda_0 \in (\Xi^2,\infty) \setminus \mathcal{N}_{\mathfrak p_0}$ for some Lebesgue null set  $\mathcal{N}_{\mathfrak p_0}$), there exists a strictly increasing sequence of natural numbers $(m_i)_{i \in \mathbb N}$  such that for every $i\in \mathbb N$, the operator
 	$	P_{\omega_- } $ admits an eigenvalue $\lambda_i := \lambda_{m_i \ell_i}  = \lambda_{m_i \ell_i}(\omega = \omega_- m)$ satisfying \begin{align}\tilde \lambda_i := { \lambda_i } m_i^{-2} = \tilde \lambda_0  + \lambda_\textup{error}^{(i)}m_i^{-2}, \end{align} where $|\lambda_{\textup{error}}^{(i)}| \leq C(\tilde \lambda_0,\mathfrak p_0)$ as $m_i\to\infty$ for some constant $C(\tilde \lambda_0,\mathfrak p_0)>0$. Moreover,  $m_i \leq \ell_i \leq m_i^2$.
 \end{prop}
 \begin{proof}
We consider the   formulation of the angular o.d.e.\  in \eqref{eq:equivalentformulation} and moreover change coordinates 
 	\begin{align}
 	y (x) = \int_{0}^x \frac{1}{\Delta_{\tilde x} (1-\tilde x^2)} \d \tilde x
 	\end{align}
 	such that \begin{align}\frac{\d x}{\d y } = \Delta_x (1-x^2).\end{align}
 This yields the equivalent eigenvalue problem
 	\begin{align}
 	-  \frac{\d^2}{\d y^2}g +   (m^2 W_1  +   P_\textup{error}) g =0 \label{eq:eigenvaluepb}
 	\end{align}
 	for $g$ in a dense domain of $L^2(\mathbb R,  \d y)$.
 From \cref{lem:w1} we have that $W_1$ has a unique positive root   for $\tilde \lambda>\Xi^2$ which we denote with $y_0(\tilde \lambda) := y(x_0(\tilde \lambda))$. We also define
 	\begin{align}
\xi (\tilde \lambda):= \int_{-y_0(\tilde \lambda)}^{y_0(\tilde \lambda)} \sqrt{-W_1}  \d y,
 	\end{align}
 	where we recall that $W_1$ is symmetric around the origin.  	For the potential $W_1$, we have (e.g.\ \cite[p.~118]{bsbook}) that $ 	\xi \colon (\Xi^2, + \infty) \to \mathbb R, \tilde \lambda\mapsto \xi (\tilde \lambda)$ is a strictly increasing smooth (even real-analytic) function. Further note that \begin{align}\frac{\d \xi}{\d \tilde \lambda} = \int_{-y_0(\tilde \lambda)}^{y_0(\tilde \lambda)}\frac{\Delta_{x(y)} (1-x(y)^2)}{2\sqrt{- W_1}}\d y >0 \label{eq:monotone}\end{align}so by the inverse function theorem, $\xi$ has a smooth inverse.
 
By a standard result on Diophantine approximation (see e.g.\ \cite[Corollary (ii) after Theorem~6.2]{harman1998metric}), we have that for each $x \in \mathbb R_{>0} \setminus \mathcal N$, where $\mathcal N$ is a Lebesgue null set, there exist sequences of natural numbers $(n_i )_{i\in \mathbb N}$ and $(m_i )_{i\in \mathbb N}$ with $n_{i+1} > n_i$ and $m_{i+1} > m_i$  such that \begin{align} 
 	\Big |\frac{2}{\pi} x - \frac{2n_i +1}{m_i} \Big| \leq \frac{1}{m_i^2}
 	 	\end{align}
 	 	 for all $i \in \mathbb N$. Indeed, the assumptions of  \cite[Corollary (ii) after Theorem~6.2]{harman1998metric} are satisfied as  $\sum_{m\in \mathbb N} \frac{1}{m}$ is divergent and $2n+1 = 	1\, (\textup{mod } 2) $, $m= 0\,  (\textup{mod } 1)$, where $(1,2,0,1)$ is pairwise coprime, i.e. $(1,2,1,1)=1$ in the notation of \cite[Corollary (ii) after Theorem~6.2]{harman1998metric}.  Alternatively, the result also follows from  \cite[Theorem~6.6]{harman1998metric} after noting that the sum of the lower asymptotic densities $\underline d$ of the odd and the natural numbers exceeds 1, i.e.\ $\underline d(\mathbb N) + \underline d (2\mathbb N +1) = 1 +\frac 12 >1$. 
 	 	 
 	 	 	Now, since $\xi$ has a smooth inverse, there exists a  Lebesgue null set $\mathcal N_{\mathfrak p_0}:= \xi^{-1}(\mathcal N) \subset (\Xi^2, \infty)$ such that for each $\tilde \lambda_0 \in  (\Xi^2, \infty)\setminus \mathcal N_{\mathfrak p_0}$ we have 
 	 	 	 \begin{align}\label{eq:approxnimi}
 	 	 	\Big |\frac{2}{\pi} \xi( \tilde \lambda_0 )- \frac{2n_i +1}{m_i} \Big| \leq \frac{1}{m_i^2}
 	 	 	\end{align}
for a sequence   of natural numbers $(n_i )_{i\in \mathbb N}$ and $(m_i )_{i\in \mathbb N}$ with $n_{i+1} > n_i$ and $m_{i+1} > m_i$.
 	 	 
 	 	  Now, we will apply \cref{prop:highfreq}. For  $\tilde \lambda_0 \in (\Xi^2,\infty)\setminus \mathcal N_{\mathfrak p_0}$, choose a small neighborhood $\mathcal U_{\tilde \lambda_0}$ such that for all $\tilde \lambda$ in the closure of $U_{\tilde \lambda_0}$, we have $\tilde \lambda > \Xi^2$. We will now consider $\tilde \lambda \in \mathcal U_{\tilde \lambda_0}$ which will take the role of $\epsilon$ appearing in \cref{prop:highfreq}.
We will now show that indeed the assumptions of  \cref{prop:highfreq} are satisfied. First, note that  $W_1$ and $P_\textup{error}$ are smooth for all $\tilde \lambda  \in \mathcal U_{\tilde \lambda_0}$. Further, uniformly in $\mathcal U_{\tilde \lambda_0}$,  $\int^y_0 \sqrt{| W_{1}|} \d \tilde y $ diverges as $y\to \pm \infty$. Moreover, the potential $W_1$ has two simple roots which do not coalesce uniformly in $\mathcal U_{\tilde \lambda_0}$ in view of \cref{lem:w1}. In particular,  this also shows that in a region $[0,c]$ for any fixed $c>0$ (in particular containing the right turning point), the total variation of $H_{\tilde \lambda}$ is bounded uniformly in $\mathcal U_{\tilde \lambda_0}$; analogously for  $\hat H_{\tilde \lambda}$ in $[-c,0]$. Remark from \cite[Chapter~11, \S3]{olver} that indeed \eqref{eq:hepsilon1} and \eqref{eq:hepsilon2} are by construction the quantitative versions of the qualitative statement of two non-coalescing roots.  
To show that $\mathcal{V}_{ 0, +\infty}   (H_{\tilde \lambda})$ and $ \mathcal{V}_{ -\infty, 0}   (\hat H_{\tilde \lambda}) $ remain bounded at $\pm\infty$, respectively, we note that for $|y|\to\infty$, we have \begin{align}  \frac{\Xi^2}{2}\leq W_1 \leq \Xi^2, \text{ and } \left|\frac{\d W_1}{\d y}\right|, \left| \frac{\d^2 W_1}{\d y^2} \right| \lesssim \frac{\d x}{\d y}\end{align}  
 	 	 as well as
 	 	  \begin{align}|P_\textup{error}|\lesssim \frac{\d x}{\d y}.\end{align} 
 	 Inserting these bounds in  \eqref{eq:hepsilon1} and \eqref{eq:hepsilon2}  we obtain
 	 	 \begin{align}
 	 	  \mathcal{V}_{ 0, +\infty}   (H_{\tilde \lambda}) + \mathcal{V}_{ -\infty, 0}   (\hat H_{\tilde \lambda})  \lesssim 1
 	 	 \end{align}
  	 uniformly in $\mathcal U_{\tilde \lambda_0}$. The others bounds of \eqref{eq:defnbfge} (uniformly in $\mathcal U_{\tilde \lambda_0}$)  also follow directly from the previous estimates  and we obtain $B_0 = \sup_{\tilde \lambda \in \mathcal U_{\tilde \lambda_0}} \in B(W_1,P_\textup{error})\lesssim 1$.
 	 	  	 	 Thus, from \cref{prop:highfreq} we now conclude that the eigenvalues $\lambda =\tilde \lambda m^2$ for  $\tilde \lambda $ in a neighborhood of $\tilde{\lambda}_0$ are characterized by
 	 	 \begin{align}\label{eq:highfreq}
 	 	 \frac 2 \pi	\xi(\tilde \lambda) + \vartheta_{\tilde \lambda_0,m,n,B(W_1,P_{\textup{error}})} = \frac{2n +1}{m}
 	 	 \end{align}
 	 	 for   $n\in \mathbb N$, where $|\vartheta_{\tilde \lambda_0,m,n,B(W_1,P_{\textup{error}})}| \lesssim_{\tilde \lambda_0} m^{-2}$.
 	 	 
 	 	 Now, for fixed $\tilde \lambda_0 \in (\Xi^2,\infty)\setminus \mathcal N_{\mathfrak p_0}$, let the sequence $(m_i,n_i)_{i\in \mathbb N}$ as above be such that \eqref{eq:approxnimi} holds. Then, we obtain associated eigenvalues from \eqref{eq:highfreq} which satisfy
 	 	 \begin{align}
 	 	 	\tilde \lambda_i = \xi^{-1}\left(\frac{\pi}{2}\frac{2n_i+1}{m_i} - \frac \pi 2  \vartheta_{\tilde \lambda_0,m_i,n_i,B(W_1,P_{\textup{error}})}  \right) =  \xi^{-1}\left( \xi(\tilde \lambda_0) + O_{\tilde \lambda_0}(m_i^{-2})\right) = \tilde \lambda_0 + O_{\tilde \lambda_0}(m_i^{-2}).
 	 	 \end{align}
 	 The last equality holds due Taylor's theorem and \eqref{eq:monotone}.
 \end{proof}
   
   \subsection{Bounds on \texorpdfstring{$\partial_\xi \lambda_{m\ell}$}{dxilambdaml} and \texorpdfstring{$\partial_\xi S_{m\ell}$}{dxisml} near interior scattering poles}
   \label{sec:boundsonpartialxlambda}
In the proof of \cref{thm:rough} in \cref{sec:mainthmkerr} we will need to control the quantities $\partial_\omega \lambda_{m\ell}(a \omega)$ and $\partial_\omega S_{m\ell}(a \omega)$ near the interior scattering poles, i.e. for $\omega \approx \omega_- m$. We will choose our initial data in \cref{sec:initialdata} to be supported on angular modes $m>0$ which are large and positive.   Thus, for the rest of this subsection, we assume that $m>0$ and think of $1/m$ as a semiclassical parameter.
We first note that $\xi \mapsto S_{m\ell}(\xi,x)$ is smooth as $\xi$ is a smooth parameter of the  angular o.d.e.\ \eqref{eq:spheroidalharmonicoperator} solved by $S_{m\ell}$.    Now, a direct computation shows that 
   \begin{align}\partial_\xi S_{m\ell} = \frac{\partial S_{m\ell}(\xi,x)}{\partial \xi}\end{align}
   solves the inhomogeneous o.d.e.\
   \begin{align}\label{eq:inhom}
   	(P_x - \lambda_{m\ell} ) {\partial_\xi S_{m\ell}} = (\partial_\xi P_x - \partial_\xi \lambda_{m\ell}) S_{m\ell}
   \end{align}
   with Dirichlet boundary conditions at $x = \pm 1$, where
     \begin{align} \label{eq:partialxofp}
     \partial_\xi P_x = \frac{\partial P_x (\xi)}{\partial \xi } = - 2 \Xi \xi \frac{x^2}{\Delta_x} - 2 m \frac{\Xi}{\Delta_x} \frac{a^2}{l^2} x^2.
     \end{align}
We will first consider $\partial_\xi \lambda_{m\ell}$.
   \begin{lemma}\label{lem:partialxlambda}
   	The eigenvalues $\lambda_{m\ell}(\xi)$ of $P_x(\xi)$ as in \eqref{eq:spheroidalharmonicoperatorx0} satisfy
   	   \begin{align}|\partial_\xi \lambda_{m\ell}(\xi)|\leq |\langle S_{m\ell}, \partial_\xi P_x S_{m\ell}\rangle_{L^2(-1,1)} |\end{align} 
   	   and thus, 
   	   \begin{align}
   	   \sup_{\xi \in (a m \omega_- - \frac 1m, a m \omega_- + \frac 1m)} |\partial_\xi \lambda_{m\ell} (\xi) | \lesssim |m|.
   	   \label{eq:boundonpartialxilambda}
   	   \end{align}
   \end{lemma}
  \begin{proof}
Taking the $L^2$-inner product of \eqref{eq:inhom} with $S_{m\ell}$ and using that $P_x$ is self-adjoint, shows that 
\begin{align}\label{eq:hisorthgonal}
	\langle S_{m\ell}, (\partial_\xi P_x - \partial_\xi \lambda_{m\ell} )S_{m\ell}\rangle_{L^2(-1,1)} =0
\end{align}from which we obtain 
\begin{align}
|\partial_\xi \lambda | \leq| \langle S_{m\ell},  \partial_\xi P_x S_{m\ell}\rangle_{L^2(-1,1)}| \leq \| \partial_\xi P_x\|
\end{align} 
in view of $\langle S_{m\ell} , S_{m\ell} \rangle_{L^2(-1,1)} =1 $. Here $ \| \partial_\xi P_x\|$ denotes the operator norm which is equal to the $L^\infty$ norm as $\partial_\xi P_x$ is a multiplication operator (see \eqref{eq:partialxofp}). 
Now, the claim follows from the fact that  $\| \partial_\xi P_x \|_{L^\infty}\lesssim |\xi| + |m|$. 
\end{proof} 

It is more difficult to obtain estimates for $\partial_\xi S_{m\ell}$ which we express as
 \begin{align}\partial_\xi S_{m\ell} =  \operatorname{Res}(\lambda_{m\ell}; P_x ) \Pi_{S_{m\ell}}^\perp H, \label{eq:orthogonal}\end{align} 
where \begin{align}H=  (\partial_\xi P_x - \partial_\xi \lambda_{m\ell}) S_{m\ell}\end{align} is the inhomogeneous term of \eqref{eq:inhom},  $\operatorname{Res}(\lambda ; P_x ) $ is the resolvent  and $\Pi_{S_{m\ell}}^\perp $ is the orthogonal projection on the orthogonal complement of $S_{m\ell}$. At this point we also remark that both $\partial_\xi S_{m\ell}$ and  $H$ are orthogonal to $S_{m\ell}$ which follows from $\xi \mapsto \langle S_{m\ell}, S_{m\ell} \rangle_{L^2(-1,1)}=1$ and \eqref{eq:hisorthgonal}, respectively. 

 A possible way to control the  resolvent operator $\operatorname{Res}(\lambda_{m\ell}; P_x ) \Pi_{S_{m\ell}}^\perp$ is to show  lower bounds on the spectral gaps $|\lambda_{m,\ell}(a\omega) - \lambda_{m,\ell+1}(a\omega)|$  \emph{uniformly} in $m,\ell \to \infty$ and $\omega\approx \omega_- m$.   Our   approach is based on an explicit construction of the resolvent kernel via suitable approximations with parabolic cylinder functions and Airy functions. 

We begin by noting that from standard results on solutions to Sturm--Liouville problems,  each eigenfunction  $S_{m\ell}$ is either symmetric or anti-symmetric around $x=0$. If  $S_{m\ell}$ is antisymmetric around $x=0$ we have $S_{m\ell} (x=0) =0$, i.e. Dirichlet boundary conditions at $x=0$. Analogously, if $S_{m\ell}$ is symmetric, we have Neumann boundary conditions at $x=0$, i.e.\ $\frac{\d }{\d x} S_{m\ell} (x=0) = 0$. Also note that $\partial_\xi S_{m\ell}$ inherits the symmetry properties of $S_{m\ell}$. Hence, the problem reduces to studying the  interval $x \in [0,1)$ with Dirichlet/Neumann boundary conditions at $x=0$ and Dirichlet boundary conditions at $x=1$. In view of the above, $\partial_\xi S_{m\ell}$ will satisfy
\begin{align}
\label{eq:boundaryconditions}
&\partial_\xi S_{m\ell}(x=0,\xi) = 0 \text { or } \frac{\d}{\d x}\partial_\xi S_{m\ell}(x=0,\xi)  = 0
\end{align}
depending on $S_{m\ell}(x=0) =0$ or $\frac{\d}{\d x} S_{m\ell}(x=0) =0$, respectively, as well as
\begin{align}
 \partial_\xi S_{m\ell}(x=1,\xi) = 0 .
\end{align}
In addition to satisfying the above boundary conditions, $\partial_\xi S_{m\ell}$ is also a solution of the inhomogenous o.d.e.~\eqref{eq:inhom} which we explicitly write out as 
\begin{align}
\Big[ - & \frac{\d }{\d x} \left(\Delta_x (1-x^2)  \frac{\d }{\d x} \cdot \right)+ \frac{\Xi^2 m^2}{\Delta_x (1-x^2) } - \Xi m^2 a^2 \omega_-^2 \frac{ x^2}{\Delta_x}- 2 m^2 a \omega_- \frac{\Xi}{\Delta_x} \frac{a^2}{l^2} x^2 \nonumber \\  + &\frac{2}{l^2}a^2 (1-x^2) - \Xi (2\epsilon a\omega_-  + \epsilon^2 m^{-2}) \frac{ x^2}{\Delta_x} -2\epsilon \frac{\Xi}{\Delta_x} \frac{a^2}{l^2}x^2 - \lambda_{m\ell} \Big] \partial_\xi S_{m\ell} \nonumber \\ &= \left[\partial_\xi \lambda_{m\ell} + 2 \Xi (a m \omega_- + \frac{\epsilon}{m})\frac{x^2}{\Delta_x} + 2 m \frac{\Xi}{\Delta_x} \frac{a^2}{l^2} x^2\right] S_{m\ell},\label{eq:inhom2}  
\end{align}
 where $|\epsilon|<1$ is such that $\xi = a m \omega_- + \frac{\epsilon}{m}$. 
   Moreover, $\partial_\xi S_{m\ell}$ and $H$ admit the same symmetries as $S_{m\ell}$ such that, both $H$ and $\partial_\xi S_{m\ell}$ are orthogonal to $S_{m\ell}$ in $L^2([0,1))$.
    Also recall that 
   \begin{align}
   	\langle S_{m\ell} , S_{m\ell} \rangle_{L^2(-1,1)} = \int_{-1}^1 S_{m\ell}^2  \d x = 1
   \end{align}
   such that 
   \begin{align}
\int_{0}^1 S_{m\ell}^2 \d x = \frac 12.
  \end{align}

 As in the proof of \cref{prop:angular}, we introduce the variable $y=y(x)$ through the conditions 
 \begin{align}
 y(0) =0, \; \frac{\d y}{\d x} = \frac{1}{\Delta_x (1-x^2)}
 \end{align}
  as well as the associated Hilbert space $L^2([0,\infty), w(y) \d y)$, where \begin{align}w(y) = \Delta_{x(y)} (1-x(y)^2).\end{align} This can be computed explicitly as 
\begin{align}
y(x) = \frac{1}{ 2\Xi} \left( \log(1+x) - \log(1-x) + \frac{a}{l}\log(1-\frac{a}{l} x) - \frac{a}{l} \log (1+\frac{a}{l}x ) \right).
\end{align}
Note that
\begin{align}\label{eq:defnofxiy}
e^{2 \Xi y } = \frac{1+x}{1-x} \left(\frac{1-\frac a l x}{1+\frac a l x}\right)^{\frac al}.
\end{align}
In this new variable, we define \begin{align} & s_1(y) := S_{m\ell}(x(y)) \text{ and } s_p(y) := \partial_\xi S_{m\ell}(x(y))
\label{eq:defns1}
\end{align}
such that 
\begin{align}
\int_0^\infty s_1^2(y) \Delta_x (1-x^2(y)) \d y = \int_{0}^1 S_{m\ell}^2 \d x = \frac 12. \label{eq:normalizationofs1}
\end{align}
Then,   we re-write \eqref{eq:inhom2} as 
\begin{align} \nonumber
-&\frac{\d^2}{\d y^2} s_p + m^2\left( \Xi^2 - \left[ \Xi a^2 \omega_-^2 + 2 a \omega_- \Xi \frac{a^2}{l^2} \right] x^2 (1-x^2) - \tilde \lambda \Delta_x (1-x^2) \right) s_p  \\ +& \Delta_x (1-x^2) \left( \frac{2}{l^2} a^2 (1-x^2)- \Xi (2\epsilon\omega_-  + \epsilon^2 m^{-2}) \frac{ x^2}{\Delta_x} -2 \epsilon \frac{\Xi}{\Delta_x} \frac{a^2}{l^2}x^2 \right) s_p  \nonumber \\ & =
\Delta_x (1-x^2)  \left[\partial_\xi \lambda  + 2 \Xi (a m \omega_- + \frac{\epsilon}{m})\frac{x^2}{\Delta_x} + 2 m \frac{\Xi}{\Delta_x} \frac{a^2}{l^2} x^2\right] s_1. \label{eq:inhomo}
\end{align}
We recall the definition of $W_1$ in \eqref{eq:defnW1} as 
  \begin{align} 
W_1 (x(y)) =  \Xi^2 - \left[ \Xi a^2 \omega_-^2 + 2 a \omega_- \Xi \frac{a^2}{l^2} \right] x(y)^2 (1-x(y)^2) - \tilde \lambda \Delta_{x(y)} (1-x(y)^2),
\end{align}
and  define \begin{align}\label{eq:defnofw222}&
W_2 (x(y)) :=  w(y) \left( \frac{2}{l^2} a^2 (1-x(y)^2)- \Xi (2\epsilon a\omega_-  + \epsilon^2  m^{-2}) \frac{ x(y)^2}{\Delta_{x(y)}} -2\epsilon \frac{\Xi}{\Delta_{x(y)}} \frac{a^2}{l^2}x(y)^2 \right)\end{align}
as well as 
\begin{align}& F (x(y)) := w(y)\left(  \partial_\xi \lambda  + 2 \Xi (a m \omega_- + \frac{\epsilon}{m})\frac{x(y)^2}{\Delta_{(y)}} + 2 m \frac{\Xi}{\Delta_{x(y)}} \frac{a^2}{l^2} x(y)^2 \right).
\end{align}
Thus, \eqref{eq:inhomo}  reads
\begin{align}\label{eq:inhomogen}
-\frac{\d^2 }{\d y^2}s_p + (m^2 W_1 + W_2)s_p = F s_1 ,
\end{align}
where we recall that $s_p$ satisfies Dirichlet/Neumann boundary conditions at $y=0$ and vanishes at $y=+\infty$. We also note that the previous orthogonality properties remain, i.e. both $s_p$ and $w^{-1} F s_1$ are orthogonal to $s_1$ in the Hilbert space $L^2([0,\infty), w(y) \d y)$.

In order to construct the resolvent operator, we  will first state the existence of a further suitable solution $s_2$  to the homogeneous equation \begin{align}
-\frac{\d^2 }{\d y^2} g + (m^2 W_1 + W_2)g = 0 \label{eq:homo}
\end{align} which is linearly independent from $s_1$. This is the content of the following lemma which  will be proved in \cref{sec:proofofs2}. 
\begin{lemma}
	\label{prop:existenceofs2}
	Let $m\in \mathbb N$ sufficiently large as in \cref{sec:proofofs2}. For $\xi \in ( a \omega_- m -\frac{1}{m} , a \omega_- m + \frac{1}{m})$, there exists a solution $s_2$ to \eqref{eq:homo} with $\mathfrak W (s_1,s_2) =1$.
	 Moreover,  $g_p$  defined as
	\begin{align}\label{eq:defgp}
	g_p(y):= s_2(y) \int_{y}^{\infty} s_1^2(\tilde y) F(\tilde y) \d \tilde y +  s_1(y)  \int_{0}^y s_2(\tilde y) s_1(\tilde y) F(\tilde y) \d\tilde y  
	\end{align} satisfies 
	\begin{align}
\| g_p\|_{L^2([0,\infty), w(y) \d y) }^2 = 	\int_{0}^{\infty} g_p(y)^2 (1-x(y)^2) \Delta_x \d y \lesssim m.
	\end{align}
	\begin{proof}
		This is proved in \cref{sec:proofofs2}, more specifically the claim follows  from \cref{prop:estimateongp} and \cref{prop:estimateongp2}.
	\end{proof}
\end{lemma}

With $s_2$ in hand we will now construct the integral kernel of the resolvent $\operatorname{Res}(\lambda_{m\ell}; P_x ) \Pi_{S_{m\ell}}^\perp$ in   $y$-coordinates. More specifically, we show
\begin{lemma} 
	The solution $s_p (y) = \partial_\xi S_{m\ell}(x(y)) = \operatorname{Res}(\lambda_{m\ell}; P_x ) \Pi_{S_{m\ell}}^\perp (H) (x(y))$ of \eqref{eq:inhomogen} satisfies
\begin{align}\label{eq:defsp}
	s_p(y) =  g_p(y) + c_{p1} s_1(y)  =  s_2(y) \int_{y}^{\infty} s_1^2(\tilde y) F(\tilde y) \d \tilde y +  s_1(y) \left( \int_{0}^y s_2(\tilde y) s_1(\tilde y) F(\tilde y) \d\tilde y +c_{p1}\right),
\end{align}
for some constant $c_{p1} \in \mathbb R$.
\begin{proof}
Since $s_p$ is a solution of the inhomogeneous o.d.e.\ \eqref{eq:inhomogen}, it can be written (using  $\mathfrak W (s_1,s_2)=1$) as
\begin{align}\label{eq:generalsolutionofangular}
	s_p (y)=   s_2(y) \left( \int_{y}^{\infty} s_1^2(\tilde y) F(\tilde y) \d \tilde y + c_{p2}\right) +  s_1(y) \left( \int_{0}^y s_2(\tilde y) s_1(\tilde y) F(\tilde y) \d\tilde y +c_{p1}\right)
\end{align}
for some constants $c_{p1} , c_{p2} \in \mathbb R$.  It remains to show that $c_{p2} =0$ and we consider the cases of Dirichlet/Neumann conditions of $s_1$ at $y=0$ independently.

First, assume that $s_1(y=0)=0$, then we also have that $s_p (y=0) =0$ (see \eqref{eq:boundaryconditions}). Evaluating the right hand side of \eqref{eq:generalsolutionofangular} at $y=0$ we obtain 
\begin{align}\nonumber
s_2(0)& \left( \int_{0}^{\infty} s_1^2(\tilde y) F(\tilde y) \d \tilde y + c_{p2}\right) +  s_1(0) \left( \int_{0}^0 s_2(\tilde y) s_1(\tilde y) F(\tilde y) \d\tilde y +c_{p1}\right) \\ 
& = s_2(0) \left( \int_{0}^{\infty} s_1^2(\tilde y) F(\tilde y) \d \tilde y + c_{p2}\right) = s_2(0)   c_{p2},
\end{align}
where we have used that $s_1(y=0)=0$ and that $s_1$ is $L^2([0,\infty), w(y) \d y)$-orthogonal to $w^{-1} s_1 F$. Moreover, from the Wronskian condition $\mathfrak W(s_1,s_2) = 1$ we have that  $s_2(y=0) \neq 0$. Thus, $c_{p2}=0$ follows from $s_p(y=0) =0$.

Now, if $s_1$ satisfies the Neumann condition $\frac{\d}{\d y} s_1(y=0)=0$, then so does $s_p$, i.e.~$\frac{\d}{\d y} s_p (y=0) =0$. Differentiating the right hand side of \eqref{eq:generalsolutionofangular} and evaluating this at $y=0$ yields 
\begin{align} 
\frac{\d}{\d y}s_2(0) c_{p2} - s_2(0) s_1^2(0) F(0) + s_1(0)^2 s_2(0) F(0) = \frac{\d}{\d y}s_2(0) c_{p2}.
\end{align}
From the Wronskian identity we again have that $\frac{\d}{\d y} s_2 (0) \neq 0$ such that $c_{p2} =0$ follows from $\frac{\d}{\d y} s_p (y=0) =0$.
\end{proof}
\end{lemma}

Up to the completion of the proof of \cref{prop:existenceofs2}, which is the content of \cref{sec:proofofs2}, we will now show the   main proposition of this subsection.
\begin{prop}\label{eq:smlbounds}
For all $m\in \mathbb N$ sufficiently large, the eigenfunctions $S_{m\ell}(\xi, \cos \theta)$  of the operator $P$ defined in \eqref{eq:spheroidalharmonicoperator} satisfy
\begin{align}
\label{eq:formulaonsml}
\sup_{\omega \in (\omega_- m - \frac{1}{am}, \omega_- m + \frac{1}{a m } )} \| \partial_\omega S_{m\ell}(a \omega, \cdot )\|_{L^2 ([0,\pi]; \sin\theta \d \theta)} \lesssim m^\frac{1}{2}.
\end{align}
\begin{proof}
	First note that $\xi = a \omega $ such that $\partial_\omega = a \partial_\xi $. Then, we have
	\begin{align}\nonumber
	\|\partial_\xi S_{m\ell} \|^2_{L^2(-1,1)}& = \|  \operatorname{Res}(\lambda_{m\ell}; P_x ) \Pi_{S_{m\ell}}^\perp H \|_{L^2(-1,1)}^2
	= 2\|s_p\|_{L^2([0,\infty), w(y) \d y)}^2 \\
	&= 2 \|\Pi_{s_1}^\perp g_p \|_{L^2([0,\infty), w(y) \d y)}^2 \leq 2  \| g_p \|_{L^2([0,\infty), w(y) \d y)}^2,\label{eq:estimateonparxisml} \end{align}
	where we have used that \begin{align}
	s_p = \Pi_{s_1}^\perp g_p.
	\end{align}
	 Here, $ \Pi_{s_1}^\perp$ is the projection on the orthogonal complement of $s_1$ in $L^2([0,\infty), w(y) \d y)$.	
	The estimate \eqref{eq:formulaonsml} follows now from \eqref{eq:estimateonparxisml} and \cref{prop:existenceofs2}.
\end{proof}
\end{prop}

\subsection{Semi-classical resolvent estimates near interior scattering poles} \label{sec:proofofs2}
Throughout this subsection (\cref{sec:proofofs2}) we assume that 
\begin{align}\xi \in \left( a \omega_- m -\frac{1}{m} , a \omega_- m + \frac{1}{m}\right)
\label{eq:restrictionofxi}\end{align}
and $m>0$. The goal of this subsection is to show {\cref{prop:existenceofs2}. We first argue that for sufficiently large $m$, we only need to consider the case $\tilde{\lambda} > \Xi^2$ as all eigenvalues $\lambda_{m\ell}(a\omega_- m)$ at the interior scattering poles are larger than $\Xi^2 m^2$.
\begin{lemma}\label{cor:max} For  sufficiently large $m$, we have $\inf_{y\in \mathbb R}\left( m^2 W_1(y) + W_2(y)\right) >0 $ for any  $\tilde \lambda \leq \Xi^2 $.
\begin{proof}  By monotonicity of $W_1$ with respect to $\tilde \lambda$, it suffices to show the result for $\tilde \lambda = \Xi^2$. We recall   from the definition of $W_2$ in \eqref{eq:defnofw222} that $W_2$ is uniformly bounded  and satisfies \begin{align}
	W_2 (x=0) = \frac{2a^2}{l^2} >0.\end{align}
Since $W_1\geq 0$ in view of  \cref{lem:w1}, we have positivity  in a neighborhood $U$ around $y=0$, i.e.\ $\inf_{y\in U}\left( m^2 W_1(y) + W_2(y)\right) >0 $. Outside that neighborhood, in view of \cref{lem:w1}, we have that $\inf_{y\in \mathbb R \setminus U}  W_1(y)  >0$. To conclude we use that $W_2$  is uniformly bounded and the claim follows for all $m$ sufficiently large. 
\end{proof}
\end{lemma}

\begin{lemma}
For $\xi$ as in \eqref{eq:restrictionofxi} and for sufficiently large $m$  as in \cref{cor:max}, any eigenvalue  $\lambda_{m\ell}(\xi) = m^2 \tilde \lambda$ of $P_x$ satisfies $\tilde \lambda > \Xi^2 $.
\begin{proof} 
This is immediate as for $\tilde \lambda \leq \Xi^2$ and sufficiently large $m$, the operator $ - \frac{\d^2}{\d y^2} + m^2 W_1 + W_2$  is strictly positive in view of \cref{cor:max}. 
\end{proof}
\end{lemma}

Thus, it suffices to show \cref{prop:existenceofs2} for $\tilde \lambda >\Xi^2$ and we consider the case $\tilde \lambda \in (\Xi^2,\Xi^2+1]$ in \cref{sec:tildelambda<xi+1} and the case $\tilde \lambda \in (\Xi^2+1,\infty)$ in \cref{sec:tildelambda>xi+1}.

\subsubsection{The case \texorpdfstring{$\Xi^2 < \tilde \lambda \leq \Xi^2 +1 $}{lambda < xi}}\label{sec:tildelambda<xi+1}
Let $\tilde \lambda \in (\Xi^2 , \Xi^2 +1]$. In this range, $\tilde \lambda $ can be arbitrarily close to $\Xi^2$. As $\tilde \lambda \to \Xi^2$, the root $y_0>0$ of the potential $W_1(y)$ coalesces with $y=0$. Thus, our estimates need to be uniform in this limit and the appropriate approximation will be given by parabolic cylinder functions. To do so we will introduce the following Liouville transform which is motivated by \cite{second-order}. We define a new variable\footnote{Here and in the following, $\xi$ is not to be mixed up with $\xi$ appearing in \eqref{eq:spheroidalharmonicoperatorx0}.} \begin{align}\xi=\xi(y)\end{align}
uniquely through the conditions
\begin{align}\label{eq:defndxidy}
 \left(\frac{\d \xi }{\d y}\right)^2 = \frac{W_1(y)}{\xi^2 - \alpha^2} \text{ for } y \neq y_0,  
\end{align}
 $\xi(y_0) = \alpha >0 $ and $\xi(y=0)=0$.  By construction, this defines $\xi = \xi(y)$ as a smooth (even real-analytic) increasing function with values in $[0,\infty)$, see also \cite[Section~2.2]{second-order}. Note that this holds true as the right hand side satisfies
\begin{align}\frac{W_1(y)}{\xi^2 - \alpha^2}>0\end{align} for $y>0$. Equivalently, the function $\xi(y)$ can be expressed as
\begin{align}
&	\int_{y}^{y_0} (-W_1)^{\frac 12} \d \tilde{y} = \int_{\xi(y)}^{\alpha} ( \alpha^2 - \tau^2 )^{\frac 12} \d \tau \text{ for }  y \leq y_0, \\
&	\int_{y_0}^{y} W_1^{\frac 12} \d \tilde{y} = \int_{\alpha}^{\xi(y)} (\tau^2 - \alpha^2)^{\frac 12} \d \tau \text{ for } y_0 \leq y<\infty. \label{eq:defnxilarge}
\end{align}
We also consider $y=y(\xi)$ as a function $\xi$ and define 	
\begin{align}\label{eq:defnsigma1and2}
\sigma_1 := \left(\frac{\d y}{\d \xi}\right)^{-\frac 12} s_1,
\end{align}
where we recall that $s_1$ was defined in \eqref{eq:defns1}.
In this new variable $\xi$, the function  $\sigma_1=\sigma_1(\xi)$ satisfies
\begin{align}\label{eq:newode}
-\frac{\d^2 \sigma}{\d \xi^2} +\left[ m^2 (\xi^2 - \alpha^2) + \Psi \right] \sigma=0,
\end{align}
where the error function $\Psi$ is given by
\begin{align} \label{eq:defnofpsi}
\Psi = \left( \frac{\d y}{\d \xi } \right)^2 W_2 + \left( \frac{\d y}{\d \xi} \right)^{\frac 12} \frac{\d^2}{\d \xi^2} \left(\frac{\d y}{\d \xi} \right)^{-\frac 12}.
\end{align}
Since $W_1$ is analytic and non-increasing in $\tilde \lambda$, we apply \cite[Lemma~1]{second-order} to conclude that $\Psi$ is continuous for $(\xi,\alpha) \in [0,\infty)\times [0, A]$, where $A = \xi(y_0 ( \tilde \lambda = \Xi^2+1))$. Now, we define the error-control function (see (6.3) of \cite{second-order}) 
\begin{align}\label{eq:defnf1}
	F_1 := \int_0^{\xi} \frac{|\Psi|}{\Omega(\xi \sqrt{2m})} \d \xi 
\end{align}
with $\Omega(x) = |x|^{\frac 13}$. We will now bound the total variation of the error-control function $F_1$ in \eqref{eq:defnf1}. To do so we first show 
\begin{lemma}
	The smooth and monotonic functions $\xi= \xi(y)$ and $y=y(\xi)$ as defined in \eqref{eq:defndxidy} satisfy
	\begin{align}	
	\label{eq:xi^2isy}
	&\xi^2(y) \sim y\\
	&	\frac{\d y }{\d \xi} \sim \xi \label{eq:estimateonder}\\
	&\left|	\frac{\d^2 y }{\d \xi^2}\right| \lesssim  1 \label{eq:estimatesecondder}\\
	&	\left|	\frac{\d^3 y }{\d \xi^3}\right| \lesssim  \xi^{-1} \label{eq:estimatesthirdorder}
\end{align}
for all $\xi$ sufficiently large.
	\begin{proof}\label{lem:xiofy}
We estimate \begin{align}
			\frac{\d \xi }{\d y } \lesssim \sqrt{ \frac{\Xi^2}{\xi^2 - \alpha^2}  } \lesssim \frac{1}{\xi}
		\end{align}
		for all $\xi$ large enough, where we have used that $W_1 \sim \Xi^2$ for large $\xi$. Similarly, \begin{align}\frac{\d \xi }{\d y } \gtrsim \frac{1}{\xi}\end{align}
		for $\xi$ large which shows \eqref{eq:estimateonder}. Upon integrating the inequalities, we obtain \eqref{eq:xi^2isy}. 
		
		For \eqref{eq:estimatesecondder}, we differentiate \eqref{eq:defndxidy} to obtain 
		\begin{align}
	\left|	\frac{\d^2 y }{\d \xi^2}\right| = \left| \frac{\d }{\d \xi } \sqrt{\frac{\xi^2 - \alpha^2}{W_1(y(\xi))}} \right| \lesssim \sqrt{\frac{W_1}{\xi^2 - \alpha^2} }  \left| \frac{\xi}{W_1} + \frac{\xi^2}{W_1^2} \frac{\d W_1}{\d x} \frac{\d x}{\d y} \frac{\d y }{\d \xi}\right| \lesssim 1,
		\end{align}
		where we have used that \begin{align}W_1 \sim 1, \frac{\d W_1}{\d x } \lesssim 1, \frac{\d x }{\d y } \lesssim e^{-2 \Xi y} ,  \text{ and }\frac{\d y }{\d \xi} \lesssim \xi\end{align} for $\xi$ large enough. In particular, it follows that $|\frac{\d W_1}{\d \xi}|\lesssim e^{-\xi}$ and similarly that $\left| \frac{\d^2 W_1}{\d \xi^2}\right| \lesssim e^{-\xi}$. 
	
		Finally, we proceed to \eqref{eq:estimatesthirdorder} by estimating for large $\xi$ 
		\begin{align}\nonumber
	\left|	\frac{\d^3 y }{\d \xi^3}\right|  & = \left| \frac{\d^2 }{\d \xi^2 } \sqrt{\frac{\xi^2 - \alpha^2}{W_1(y(\xi))}} \right| = \left| \frac{\d}{\d \xi } \left( \frac{\xi}{\sqrt{W_1 (\xi^2 - \alpha^2) }} - \frac{\sqrt{\xi^2-\alpha^2} }{2W_1^{\frac 32}} \frac{\d W_1}{\d \xi } \right)  \right|\\
	& \lesssim \left| \frac{1}{\sqrt{W_1(\xi^2-\alpha^2)} }  \right|  +    \frac{\left|\frac{\d W_1}{\d \xi} \xi^2 \right| + | W_1 \xi | } { \left| W_1 (\xi^2-\alpha^2)\right|^{\frac 32} }  + \left| \frac{\xi \frac{\d^2 W_1}{\d \xi^2}}{W_1^{\frac 32}} \right| + \left| \frac{\frac{\d W_1}{\d \xi} }{W_1^{\frac 32}} \right| + \left|\frac{\xi }{W_1^{\frac 52}} \left( \frac{\d W_1}{\d \xi} \right)^{2} \right|  \lesssim \xi^{-1} 
		\end{align}
	in view of the above estimates. 
	\end{proof}
\end{lemma}
This allows us now to estimate the total variation of the error control function $F_1$.
\begin{lemma}\label{lem:boundoftotalvariation}
The error control function $F_1$ as defined in \eqref{eq:defnf1} satisfies
\begin{align}
	\mathcal{V}_{0,\infty} (F_1) \lesssim \frac{1}{ m^{\frac 16}}.
\end{align}
\begin{proof}
	As $\Psi$ is continuous on $[0,\infty)$, it suffices to control the integral for large $\xi$. We control both terms of 
	\begin{align}\nonumber
		\Psi& = \left( \frac{\d y}{\d \xi } \right)^2 W_2 + \left( \frac{\d y}{\d \xi} \right)^{\frac 12} \frac{\d^2}{\d \xi^2} \left(\frac{\d y}{\d \xi} \right)^{-\frac 12} \\
		& = \left( \frac{\d y}{\d \xi } \right)^2 W_2 + \left( \frac{\d y}{\d \xi} \right)^{-2} \left(  - \frac 12  \frac{\d y}{\d \xi} \frac{\d^3 y}{\d \xi^3}  + \frac 34\left(\frac{\d^2 y}{\d \xi^2} \right)^{2}  \right) \label{eq:estimateonpsi}
	\end{align}
independently. For large $\xi$, we estimate the first term as 
	\begin{align}
	 \left|	\left(\frac{\d y }{\d \xi } \right)^2 W_2\right|\leq |W_2|{\frac{\xi^2 - \alpha^2}{|W_1|} }\lesssim \xi^2 |W_2|
	\end{align}
	in view of \begin{align}W_1 \geq \frac{\Xi^2}{2}\end{align} for $\xi$ sufficiently large. 
	Further, for $\xi$ sufficiently large we have $|W_2|\lesssim e^{-2\Xi y}$  and thus,
	\begin{align}\xi^2 |W_2|\lesssim \xi^2 e^{- 2 \Xi y(\xi)}\lesssim e^{-\xi}\end{align}in view of \cref{lem:xiofy}.
	
	 For the second term of \eqref{eq:estimateonpsi}, we use \cref{lem:xiofy} to estimate  
	\begin{align}
\left| \left( \frac{\d y}{\d \xi} \right)^{-2} \left(  - \frac 12  \frac{\d y}{\d \xi} \frac{\d^3 y}{\d \xi^3}  + \frac 34\left(\frac{\d^2 y}{\d \xi^2} \right)^{2}  \right) \right| \lesssim  \xi^{-2} 
	\end{align} 
	for $\xi$ sufficiently large. Hence, 
	\begin{align}
		|\Psi|\lesssim (1+\xi)^{-2}
	\end{align}
	for $\xi$ sufficiently large. Recall that $\Psi$ is continuous everywhere and $\Omega = |x|^{\frac 13}$ such that \begin{align} \label{eq:estimateonvariation}
		\mathcal{V}_{0,\infty} (F_1) \lesssim \int_0^\infty \frac{|\Psi|}{\xi^{\frac 13} m^{\frac 16}} \d \xi \lesssim m^{-\frac 16}.
	\end{align}
\end{proof}
\end{lemma}
Having controlled the error terms we now proceed to the definition of our fundamental solutions based on appropriate parabolic cylinder functions.
We will apply \cite[Theorem~1]{second-order} which we recall for convenience of the reader in the following.
\begin{prop}[{\cite[Theorem~1]{second-order}}] \label{prop:olverthm1}
Assume that for each value of $m$, the function $\Psi(m,\alpha,\xi)$ as defined in \eqref{eq:defnofpsi} is continuous in the region $\alpha \in [0,A], \xi \in [0,\infty)$ and $\mathcal V_{0,\infty}(F_1)$ converges uniformly with respect to $\alpha$, where $F_1$ is as in \eqref{eq:defnf1}.
  Then,  the o.d.e.\ \eqref{eq:newode} has solutions $w_1(m,\alpha,\xi)$ and $w_2(m,\alpha,\xi)$ which are continuous, have continuous first and second partial $\xi$-derivatives and are given by
\begin{align}
& w_1(m,\alpha,\xi) = U\left(-\frac 12 m \alpha^2, \xi \sqrt{2m} \right) + \epsilon_1(m,\alpha,\xi),\\
& w_2 (m,\alpha,\xi) = \bar U\left(-\frac 12 m \alpha^2, \xi \sqrt{2m} \right) + \epsilon_2 (m,\alpha,\xi),
\end{align}
where   $U$ and $\bar U$ are parabolic cylinder functions defined in \cref{defn:uubar} in the appendix. The error terms satisfy
\begin{align} \nonumber 
\frac{|\epsilon_1(m,\alpha,\xi) | }{M_U( - \frac 12 m \alpha^2, \xi \sqrt{2m})}, & \frac{\partial_\xi \epsilon_1(m,\alpha, \xi ) }{\sqrt{2m} N_U(-\frac 12 m \alpha^2,\xi\sqrt{2m})} \\ & \leq E_U^{-1} (-\frac 12 m \alpha^2, \xi \sqrt{2m}) \left[ e^{\frac 12 \pi^{\frac 12} m^{-\frac 12} l_1 (-\frac 12 m\alpha^2) \mathcal V_{\xi,\infty }(F_1)  }-1 \right]\\ \nonumber 
\frac{|\epsilon_2(m,\alpha,\xi) | }{M_U( - \frac 12 m \alpha^2, \xi \sqrt{2m})}, & \frac{\partial_\xi \epsilon_2(m,\alpha, \xi ) }{\sqrt{2m} N_U(-\frac 12 m \alpha^2,\xi\sqrt{2m})} \\ & \leq E_U (-\frac 12 m \alpha^2, \xi \sqrt{2m}) \left[ e^{\frac 12 \pi^{\frac 12} m^{-\frac 12} l_1 (-\frac 12 m\alpha^2) \mathcal V_{0,\xi}(F_1)  }-1 \right]
\end{align}
and 
\begin{align}\label{eq:defnofl1}
l_1(b) = \sup_{x\in (0,\infty)} \left( \Omega(x) \frac{M_U^2(b,x)}{\Gamma(\frac 12 - b)} \right), b \leq 0.
\end{align}
\end{prop}
For the definitions of $M_{U},N_U,E_U$ refer to \cref{sec:appweight}.
We will now apply the previous proposition with our estimate at hand.
\begin{prop}\label{prop:existenceofw1w2}
There exist solutions $w_1$ and $w_2$ of \eqref{eq:newode}  satisfying
\begin{align}
&	w_1 = U\left(-\frac 12 m \alpha^2, \xi \sqrt{2m} \right) + \tilde \eta_1, \\
&	w_2  = \bar U\left(-\frac 12 m \alpha^2, \xi \sqrt{2m} \right) + \tilde \eta_2.
\end{align}

The error terms satisfy
\begin{align}
&	\tilde \eta_1 = E_U^{-1}\left(-\frac 12 m \alpha^2, \xi \sqrt{2m}\right) M_U\left(- \frac 12 m \alpha^2, \xi \sqrt{2m}\right) O (m^{-\frac 23})  \label{eq:esteta2} \\
	&\tilde \eta_2 = E_U\left(-\frac 12 m \alpha^2, \xi \sqrt{2m}\right) M_U\left(- \frac 12 m \alpha^2, \xi \sqrt{2m}\right) O (m^{-\frac 23}) \\
	&	\partial_\xi \tilde \eta_1 = E_U^{-1}\left(-\frac 12 m \alpha^2, \xi \sqrt{2m}\right) N_U\left(- \frac 12 m \alpha^2, \xi \sqrt{2m}\right) O (m^{-\frac 16}) \label{eq:esteta'2}  \\
	&	\partial_\xi \tilde \eta_2 = E_U\left(-\frac 12 m \alpha^2, \xi \sqrt{2m}\right) N_U\left(- \frac 12 m \alpha^2, \xi \sqrt{2m}\right) O (m^{-\frac 16})
\end{align} 
uniformly in $\tilde \lambda \in [\Xi^2, \Xi^2 + 1]$ and $\xi \in [0,\infty)$. Moreover, $\tilde \eta_2 (\xi=0) = \partial_\xi \tilde \eta_2(\xi=0)=0$ and $\lim_{\xi \to\infty} \tilde \eta_1(\xi) = \lim_{\xi \to \infty} \partial_\xi \tilde \eta_1 (\xi ) =0$.  
\begin{proof}
	We have chosen $\Omega(x) = |x|^{\frac 13}$ in \eqref{eq:defnf1}. For this choice of $\Omega$, the quantity $l_1$ as defined in \eqref{eq:defnofl1} satisfies $l_1(b) \lesssim 1$ uniformly in $b\leq 0$ which follows from \cref{prop:boundsonmu} and \cref{defn:zetaU}, see also  \cite[equation (6.15)]{second-order}. 
	Now, we recall that $\Psi$ is continuous for $(\xi,\alpha) \in [0,\infty) \times [0,A]$ and from \cref{lem:boundoftotalvariation} we have $ 
			\mathcal{V}_{\xi,\infty} (F_1) , 	\mathcal{V}_{0,\xi} (F_1) \leq \mathcal{V}_{0,\infty} (F_1) \lesssim \frac{1}{ m^{\frac 16}}$. Hence, we apply \cref{prop:olverthm1} and moreover estimate the error terms as
		\begin{align}
\left| e^{\frac 12 \pi^{\frac 12} m^{-\frac 12} l_1 (-\frac 12 m\alpha^2) \mathcal V_{\xi,\infty }(F_1)  }-1 \right|, \left| e^{\frac 12 \pi^{\frac 12} m^{-\frac 12} l_1 (-\frac 12 m\alpha^2) \mathcal V_{0,\xi}(F_1)  }-1 \right| \lesssim m^{-\frac{2}{3}}
		\end{align}
	from which the error bounds follow. 
 
\end{proof}
\end{prop}
\begin{rmk} As $x\to\infty$, the function $U$ is recessive (decaying), whereas $\bar U$ is dominant (growing). Hence, $w_1$ is recessive and $w_2$ is dominant. We refer to  \cite[Chapter~5, \S7.2]{olver2014asymptotics} for further details.
\end{rmk}
 \begin{lemma}\label{lem:boundsonwrosnkianw1w2}
 	The Wronskian $\mathfrak W(w_1,w_2)$ satisfies
 	\begin{align}
 		|\mathfrak W (w_1,w_2)| \sim \sqrt m \Gamma\left(\frac 12 +\frac 12 m \alpha^2\right)
 	\end{align}
 	for $m$ sufficiently large.
 	\begin{proof}
 		Since the Wronskian is independent of $\xi$, we compute it at $\xi=0$ such that $\tilde \eta_2 = \partial_\xi \tilde \eta_2 =0$. We obtain
 		\begin{align} \nonumber 
\mathfrak W (w_1,w_2) &=   \mathfrak W\left( U\left(-\frac 12 m \alpha^2, \xi \sqrt{2m} \right) + \tilde \eta_1, \bar U\left(-\frac 12 m \alpha^2, \xi \sqrt{2m} \right)  \right) \\
&=\label{eq:principaltermofwrosnkian} \mathfrak W\left( U\left(-\frac 12 m \alpha^2, \xi \sqrt{2m} \right)  , \bar U\left(-\frac 12 m \alpha^2, \xi \sqrt{2m} \right)  \right)  
 	\\&+  \mathfrak W\left(   \tilde \eta_1, \bar U\left(-\frac 12 m \alpha^2, \xi \sqrt{2m} \right)  \right),
  	\end{align} 
  where each Wronskian is evaluated at $\xi=0$. 
 		We begin by computing \eqref{eq:principaltermofwrosnkian}. For   $U(b,x)$ and $\bar U (b,x)$ we have the Wronskian identity $\mathfrak W(U,\bar U) = \sqrt{\frac{2}{\pi}}\Gamma(\frac 12 - b)$, see \cite[Equation~(5.8)]{second-order}. Thus, the chain rule yields
 		\begin{align}
 \mathfrak W\left( U\left(-\frac 12 m \alpha^2, \xi \sqrt{2m} \right)  , \bar U\left(-\frac 12 m \alpha^2, \xi \sqrt{2m} \right)  \right) = \sqrt{2m}  \sqrt{\frac{2}{\pi}}\Gamma\left(\frac 12 + \frac 12 m \alpha^2\right).
 		\end{align}
 	Now, we use \eqref{eq:esteta2}, \eqref{eq:esteta'2}, \cref{defn:munu} and \eqref{eq:valuesatzero1}--\eqref{eq:valuesatzero4} to estimate
 	\begin{align}\nonumber 
&  \left|\mathfrak W\left(   \tilde \eta_1, \bar U\left(-\frac 12 m \alpha^2, \xi \sqrt{2m} \right)  \right)(\xi=0)\right|\\ \nonumber &   \lesssim \sqrt{m} \left| \tilde \eta_1(\xi=0) \bar U' \left(-\frac 12 m \alpha^2, 0\right) \right|  +  \left|  \eta_1'(\xi=0) \bar U \left(-\frac 12 m \alpha^2, 0\right) \right| \\\nonumber
&\lesssim  m^{-\frac 16} \left| M_U\left(- \frac 12 m \alpha^2, 0 \right)  \bar U' \left(-\frac 12 m \alpha^2, 0\right) \right| + m^{-\frac 16} \left| N_U\left(- \frac 12 m \alpha^2, 0 \right)  \bar U \left(-\frac 12 m \alpha^2, 0\right) \right|
\\ \nonumber & \lesssim m^{- \frac 16} \left( \left| \sqrt{ U^2 + \bar U^2} \left(-\frac 12 m \alpha^2, 0\right)  \bar U' \left(-\frac 12 m \alpha^2, 0\right) \right|  + \left| \sqrt{ {U'}^2 + ({\bar{U}'})^2}\left(-\frac 12 m \alpha^2, 0\right) \bar U \left(-\frac 12 m \alpha^2, 0\right) \right|   \right) \\\nonumber
& \lesssim m^{-\frac 16} 2^{\frac 12 m \alpha^2} \Gamma\left(\frac 14+\frac 14 m \alpha^ 2\right)\Gamma\left( \frac 34 + \frac 14 m\alpha^2  \right) =  m^{-\frac 16} 2^{\frac 12 m \alpha^2} \Gamma\left(\frac 12 +\frac 12 m \alpha^ 2\right) 2^{1-2(\frac 14 + \frac 14 m \alpha^2)} \sqrt{\pi}  \\
& \lesssim m^{-\frac 16} \Gamma\left(\frac 12 + \frac 12 m \alpha^2\right),
 	\end{align}
 where we also used the Legendre duplication formula $\Gamma(x) \Gamma(x+ \frac 12) = 2^{1-2x} \sqrt \pi \Gamma(2x)$. This concludes the proof.
 		\end{proof}
 \end{lemma}
 	\begin{lemma}
 		The function $\sigma_1$ defined in \eqref{eq:defnsigma1and2} has the form  \begin{align}\label{eq:s1approx}
 		&	\sigma_1 = A_1 w_1,
 		\end{align}
 		where $w_1$ is as in \cref{prop:existenceofw1w2} and $A_1\neq 0$ is a real constant. 
 	\end{lemma}
 	\begin{proof} Both functions $\sigma_1$ and $w_1$ are non-trivial solutions to \eqref{eq:newode} which are recessive as $\xi\to \infty$ ($y\to \infty$). The claim follows now as the space of solutions of \eqref{eq:newode} which are recessive as $\xi\to \infty$ is one-dimensional (see e.g.\ \cite[Chapter~5, \S7.2]{olver2014asymptotics}). 
 	\end{proof}
 	Using the parabolic cylinder functions, we now define a  solution $\sigma_2$ which is linearly independent of $\sigma_1$.  
 \begin{definition}
 	We define the solution $\sigma_2$  of \eqref{eq:newode} as
 	\begin{align} \label{eq:defnsigma2}
\sigma_2:= \frac{1}{A_1 \mathfrak{W}(w_1,w_2)}  w_2
 	\end{align}
 	and the solution $s_2$ to \eqref{eq:homo} as
 	\begin{align}\label{eq:defs2}
 		s_2(y):= \left(\frac{\d y}{\d \xi}\right)^{\frac 12} \sigma_2(\xi(y)).
 	\end{align}
 \end{definition}
  	A direct computation shows   
 \begin{lemma}
 	We have \begin{align}
 		\mathfrak W_y(s_1,s_2) = \mathfrak W_\xi (\sigma_1,\sigma_2) =1.
 	\end{align}
 	Here, $\mathfrak W_y$ and  $\mathfrak W_\xi$   denote the Wronskians with respect to the $y$ and $\xi$ variable. 
 \end{lemma}

\begin{lemma}\label{eq:estimatesonsigma1sigma2}
	With $\sigma_1$ and $\sigma_2$ as defined  in \eqref{eq:defnsigma1and2} and \eqref{eq:defnsigma2} we have 
	\begin{align}&|\sigma_1| \lesssim |A_1| E_U^{-1}\left(-\frac 12 m \alpha^2, \xi \sqrt{2m} \right) M_U\left(-\frac 12 m \alpha^2, \xi \sqrt{2m} \right), \\&
	|\sigma_2|\lesssim  \left| \frac{1}{A_1 \mathfrak W(w_1,w_2)}\right|  E_U\left(-\frac 12 m \alpha^2, \xi \sqrt{2m} \right) M_U\left(-\frac 12 m \alpha^2, \xi \sqrt{2m} \right),  \\
	&	|\sigma_1(\xi) \sigma_2(\xi)| \lesssim \frac{1}{|\mathfrak W(w_1,w_2)|} M^2_U\left(-\frac 12 m \alpha^2, \xi \sqrt{2m} \right).
	\end{align}
	\begin{proof}
We estimate using \eqref{eq:estimatesonU} that
\begin{align} \nonumber
			|\sigma_1| &  =| A_1 w_1| =|A_1| \left|  U\left(-\frac 12 m \alpha^2, \xi \sqrt{2m} \right) + \tilde \eta_1 \right| \\
			&\lesssim |A_1| E_U^{-1}\left(-\frac 12 m \alpha^2, \xi \sqrt{2m} \right) M_U\left(-\frac 12 m \alpha^2, \xi \sqrt{2m} \right)
		\end{align}
		and
 	\begin{align} \nonumber
				|\sigma_2| & = \left| \frac{1}{A_1 \mathfrak W(w_1,w_2)}\right|  |w_2|\leq   \left| \frac{1}{A_1 \mathfrak W(w_1,w_2)}\right| \left| \bar U\left(-\frac 12 m \alpha^2, \xi \sqrt{2m} \right) + \tilde \eta_2 \right| \\
				&\lesssim \left| \frac{1}{A_1 \mathfrak W(w_1,w_2)}\right|  E_U\left(-\frac 12 m \alpha^2, \xi \sqrt{2m} \right) M_U\left(-\frac 12 m \alpha^2, \xi \sqrt{2m} \right).
	\end{align}
	\end{proof}
\end{lemma}
Now, we recall the definition of $g_p$ in \eqref{eq:defgp} as 
\begin{align}\label{eq:defngp2}
g_p (y):=   s_2(y) \int_{y}^{\infty} s^2_1(\tilde y) F(\tilde y) \d \tilde y +  s_1(y) \int_{0}^y s_1(\tilde y) s_2(\tilde y) F(\tilde y) \d\tilde y.
\end{align}
for $s_1$ as in \eqref{eq:defns1} and where we take $s_2$ as in \eqref{eq:defs2}. Now, we are in the position to show the main lemma of \cref{sec:tildelambda<xi+1}.
\begin{lemma}\label{prop:estimateongp}
	Let $\tilde \lambda \in (\Xi^2 , \Xi^2 + 1]$ and let $s_2$ as in \eqref{eq:defs2}. Then, $g_p$ satisfies
	\begin{align}\label{eq:lefthandside}
	\int_{0}^{\infty} g_p(y)^2 (1-x(y)^2) \Delta_x \d y \lesssim m .
	\end{align}
	\begin{proof}
		We plug \eqref{eq:defngp2} into the left hand side of \eqref{eq:lefthandside} and we will estimate both terms independently. 
		
For the first term, we change variables from $y$ to $\xi$, use that $x \mapsto E_U(b,x)$ is non-decreasing, as well as \cref{eq:estimatesonsigma1sigma2} to estimate
		\begin{align} \nonumber
		\int_0^\infty  s_2^2(y)  	& \left(\int_{y}^\infty s^2_1(\tilde y) F(\tilde y) \d\tilde y\right)^2 (1-x(y)^2) \Delta_x \d y \\ \nonumber & = \int_0^\infty \sigma_2^2(\xi)  \left(\int_{\xi}^\infty \sigma_1^2 (\tilde \xi) F(\tilde \xi) \d\tilde \xi\right)^2 (1-x(\xi)^2) \Delta_x(\xi) \d \xi \\
			&\lesssim \int_{0}^{\infty} \Bigg[ \frac{\left|M_U \left(- \frac 12 m \alpha^2, \xi \sqrt{2m} \right)\right|^2}{|\mathfrak{W}(w_1,w_2)|^2}  \Big( \int_{\xi}^{\infty} |\sigma_1(\tilde \xi) F(\tilde \xi)| \nonumber \\ &   \hspace{60pt}|M_U(-\frac 12 m \alpha^2, \tilde \xi \sqrt{2m})| \d \tilde \xi \Big)^2    (1-x(\xi)^2) \Delta_x(\xi) \Bigg]\d \xi.
		\end{align}
		Now, we use the bounds on $M_U$ and $\mathfrak W(w_1,w_2)$ from \cref{prop:boundsonmu} and \cref{lem:boundsonwrosnkianw1w2} to deduce
		\begin{align}\nonumber
	\int_0^\infty s_2^2(y) & \left(\int_{y}^\infty s^2_1(\tilde y) F(\tilde y) \d\tilde y\right)^2 (1-x(y)^2) \Delta_x \d y \\ \nonumber &
	\lesssim \frac 1m \int_0^\infty (1-x(\xi)^2) \Delta_x(\xi) \d \xi\left( \int_{0}^{\infty} |\sigma_1(\tilde \xi)| |F (\tilde \xi)| \d \tilde \xi\right)^2
	\\ \nonumber &\lesssim \frac{1}{m} \int_0^\infty |s_1|^2 (1-x(y)^2) \Delta_x(y) \d y \int_0^\infty\frac{ |F|^2 }{\Delta_x (1-x(y)^2)} \frac{\d \xi}{\d y} \d y \\ \nonumber & \lesssim \frac{1}{m} \int_0^\infty \Delta_x (1-x^2)  \left[\partial_\xi \lambda  + 2 \Xi (a m \omega_- + \frac{\epsilon}{m})\frac{x^2}{\Delta_x} + 2 m \frac{\Xi}{\Delta_x} \frac{a^2}{l^2} x^2\right]^2  \frac{\d \xi}{\d y} \d y \\ & \lesssim m , 
		\end{align}
		where we used the Cauchy--Schwarz inequality and the fact that $s_1$ satisfies \eqref{eq:normalizationofs1} as well as \eqref{eq:boundonpartialxilambda}.

		For the second term we argue similarly and obtain
		\begin{align}\nonumber
		&	\left|\int_{0}^\infty s_1^2(y)\left(\int_{0}^y s_2 (\tilde y) s_1(\tilde y) F(\tilde y) \d \tilde y\right)^2 (1-x(y)^2) \Delta_x \d y \right| \\ \nonumber & \lesssim \int_0^\infty s_1^2(y) (1-x(y)^2) \Delta_x \d y \left(\int_0^\infty s_2 (\tilde y) s_1(\tilde y ) F(\tilde y) \d \tilde y\right)^2\\ \nonumber
			& \lesssim \left( \int_0^\infty \frac{M^2_U\left(-\frac 12 m \alpha^2, \xi \sqrt{2m} \right)}{|\mathfrak W(w_1,w_2)|}   |F(\xi) |\d \xi  \right)^2
			\\  \nonumber & \lesssim \frac{1}{ m} \left(\int_0^\infty   \Delta_x (1-x(y)^2)  \left|\partial_\xi \lambda  + 2 \Xi (a m \omega_- + \frac{\epsilon}{m})\frac{x^2}{\Delta_x} + 2 m \frac{\Xi}{\Delta_x} \frac{a^2}{l^2} x^2\right|   \frac{\d \xi}{\d y}   \d y\right)^2
			\\ & \lesssim m.
		\end{align}
	\end{proof}
\end{lemma}
\subsubsection{The case \texorpdfstring{$ \tilde \lambda \in (\Xi^2 + 1 , \infty) $}{tilde lambda near xi2}}
\label{sec:tildelambda>xi+1}
For the parameter range $\tilde \lambda \in (\Xi^2 + 1, \infty)$ we consider $\lambda = m^2 \tilde \lambda  $ as a large parameter and re-write the o.d.e.\ \eqref{eq:inhomogen} as 
\begin{align}
- \frac{\d^2}{\d y^2} s_p + m^2 \tilde{\lambda} \tilde W_1 s_p + W_2 s_p = F s_1 , 
\end{align}
where 
\begin{align}\label{eq:tildew1toest}
	\tilde W_1 = \tilde W_1 (y) := \frac{W_1}{\tilde \lambda} = \frac{ \Xi^2}{\tilde \lambda} - \left[ \Xi a^2 \omega_-^2 + 2 a \omega_- \Xi \frac{a^2}{l^2} \right] \frac{ x^2 (1-x^2) }{\tilde \lambda} -  \Delta_x (1-x^2).
\end{align}
We also recall the homogeneous o.d.e.\ \eqref{eq:homo} 
\begin{align}\label{eq:homo2}
- \frac{\d^2}{\d y^2} g + m^2 \tilde{\lambda} \tilde W_1 g + W_2 g = 0 .
\end{align}
Recall also that $s_1$ as defined in \eqref{eq:defns1} is a solution of \eqref{eq:homo2}.
As before, we define $y_0$ as the unique non-negative root of $\tilde W_1(y)$. It satisfies \begin{align} \label{eq:estimateony_0}
		y_0 \sim   \log (\tilde \lambda)
	\end{align}
for sufficiently large $\tilde \lambda$, where we note that $y_0$ becomes arbitrarily large for $\tilde \lambda \to \infty$. Indeed, to show \eqref{eq:estimateony_0}, we note that for large $\tilde \lambda$,   from \eqref{eq:tildew1toest} we see that $x(y_0)$ satisfies \begin{align}\Xi^2 = \left[ \Xi a^2 \omega_-^2 + 2 a \omega_- \Xi \frac{a^2}{l^2} \right]  { x(y_0)^2 (1-x(y_0)^2) } +  \tilde \lambda \Delta_{x(y_0)} (1-x(y_0)^2) \sim \tilde \lambda (1-x(y_0)).\end{align}
Then, \eqref{eq:estimateony_0} follows from $1-x \sim e^{- 2 \Xi y}$ for $y$ sufficiently large (recall \eqref{eq:defnofxiy}). 
Our estimates will be uniform in the limit $\tilde \lambda \to \infty$. 
	\begin{lemma}\label{prop:estimatesontildew1}
In the region $ 0 \leq   y \leq y_0 -1$ we have
\begin{align} \label{eq:estiminregionI}
&\frac{1}{\tilde \lambda}\lesssim	- \tilde W_1 \lesssim 1,\;\;
	 \frac{\d\tilde W_1}{\d y}\lesssim |\tilde W_1|, \;\;	\left| \frac{\d^2 \tilde W_1}{\d y^2}\right|\lesssim \frac{\d \tilde W_1}{\d y } + (1-x(y)^2 ) |\tilde W_1|.
\end{align} For $ y_0 - 1 \leq   y \leq y_0 +1$, we have
\begin{align}\label{eq:estimatesonw1inintermedaite}
\frac{|y-y_0|}{\tilde \lambda} \lesssim  	| \tilde W_1| \lesssim \frac{1}{\tilde \lambda}, \;\; \frac{\d \tilde{W}_1}{\d y }  \sim \frac{1}{\tilde \lambda}, \;\;
\left|\frac{\d^2 \tilde W_1}{\d y^2} \right|\lesssim \frac{1}{ \tilde \lambda}, \;\;
\left|\frac{\d^3 \tilde W_1}{\d y^3} \right|\lesssim \frac{1}{ \tilde \lambda}, \;\;
\left|\frac{\d^4 \tilde W_1}{\d y^4} \right|\lesssim \frac{1}{ \tilde \lambda}.
\end{align}
 For $  y_0 +1 \leq  y < \infty$, we have
  \begin{align}\tilde W_1 \sim \frac{1}{\tilde \lambda} \text{ and }
 	\frac{\d \tilde W_1}{\d y}, \left| 	\frac{\d^2 \tilde W_1}{\d y^2} \right| \lesssim \frac{\d x }{\d y} \lesssim \frac{1}{\tilde \lambda}.
 	\end{align}
\begin{proof}
From \cref{lem:w1} we have that $\tilde W_1$ is increasing on $y \in [0,\infty)$ and moreover, for $y \in [y_0-1,y_0+1]$ we have that \begin{align} \label{eq:lowerboundondwdx}
	\frac{\d \tilde W_1}{\d y} = \frac{\d \tilde W_1}{\d x } \frac{\d x}{\d y} \gtrsim x(y) \frac{\d x }{\d y} \gtrsim \frac{1}{\tilde \lambda}.
\end{align}
Thus, for $0 \leq y\leq y_0 -1$, 
\begin{align} \label{eq:lowerboundonw1}
	-\tilde W_1(y) \geq -\tilde W_1(y_0 - 1) \geq \int_{y_0-1}^{y_0} 	\frac{\d \tilde W_1}{\d y}  \d \tilde y \gtrsim \frac{1}{\tilde \lambda}.
\end{align} 
Moreover, for $0 \leq y \leq y_0 - 1$,
\begin{align}
	\frac{\d \tilde W_1}{\d y } \lesssim  \frac{\d x}{\d y } = \Delta_x (1-x^2) \lesssim |\tilde W_1| + \frac{1}{\tilde \lambda} \lesssim |\tilde W_1|
\end{align}
using the definition of $\tilde W_1$ and $|\tilde W_1 | \gtrsim \frac{1}{\tilde \lambda}$. Similarly, we obtain 

 \begin{align}\left|\frac{\d^2 \tilde W_1}{\d y^2} \right|  \lesssim (1-x(y)^2) \frac{\d\tilde W_1}{\d x } + (1-x(y)^2)^2 \left| \frac{\d^2 \tilde W_1}{\d x^2} \right|\lesssim \frac{\d \tilde W_1}{\d y } + (1-x(y)^2 ) |\tilde W_1| .\label{eq:estimateondwdy^2}
\end{align}

In the region $y \in [y_0 -1 , y_0 + 1 ]$, recall from \eqref{eq:lowerboundondwdx} that $\frac{\d \tilde W_1}{\d y } \gtrsim \frac{1}{\tilde \lambda}$. Moreover, just as in \eqref{eq:estimateondwdy^2}, we obtain
\begin{align}
\left|	\frac{\d \tilde W_1}{\d y} \right|, \left|	\frac{\d^2 \tilde W_1}{\d y^2} \right|, \left|	\frac{\d^3 \tilde W_1}{\d y^3} \right|, \left|	\frac{\d^4 \tilde W_1}{\d y^4} \right|  \lesssim \frac{\d x}{ \d y } \lesssim \frac{1}{\tilde \lambda}.
\end{align}

In the region $y\in (y_0 + 1 , +\infty)$, analogous to \eqref{eq:lowerboundonw1}, we have
\begin{align}
 \frac{1}{\tilde \lambda}\lesssim	\tilde W_1 \lesssim  \frac{1}{\tilde \lambda}
\end{align}
and moreover, \begin{align}
	\frac{\d \tilde W_1}{\d y}, \left| 	\frac{\d^2 \tilde W_1}{\d y^2} \right| \lesssim \frac{\d x }{\d y} \lesssim \frac{1}{\tilde \lambda}.
\end{align}

\end{proof}
	\end{lemma}

With the estimates of \cref{prop:estimatesontildew1} in hand we will define the  variable $\varsigma$ as 
\begin{align}
	\frac 23 \varsigma^{\frac 32} = \int_{y_0}^y \sqrt{\tilde W_1 (y)} \d y 
\end{align}
for  $y\geq y_0$ and
\begin{align}
\frac 23 \left( -\varsigma\right)^{\frac 32} = \int_{y}^{y_0} \sqrt{-\tilde W_1 (y)} \d y 
\end{align}
for $y \leq y_0$. We denote 
\begin{align}
\varsigma_0 := \varsigma(y=0) = - \left( \frac{3}{2}\int_{0}^{y_0} \sqrt{- \tilde W_1(y)} \d y \right)^{\frac 23}.
\end{align}
We further introduce the error control function 
\begin{align}\label{defn:herrorcontrol}
	H(y):= \int_{y_0}^{y}  \frac{1}{| \tilde  W_1|^{\frac 14}}\frac{\d^2}{\d y^2} \left(|\tilde W_1|^{-\frac 14}\right) - \frac{W_2}{|\tilde W_1|^{\frac 12}} - \frac{5 |\tilde W_1|^{\frac 12}}{16 \varsigma(y)^3} \d y .
\end{align}
The fact that $H$ is absolutely continuous is a standard result and follows from \cite[Lemma, Section~4]{turning}, see also \cite[Lemma~3.1, Chapter~11]{olver}. In the following we establish a quantitative version of this.
\begin{lemma}\label{lem:errorboundonH}
	The error control function $H$ defined in \eqref{defn:herrorcontrol} satisfies
\begin{align}
\mathcal V_{0, \infty}(H) \lesssim \tilde \lambda^{\frac 12} .
\end{align}
\begin{proof}Since $H$ is absolutely continuous we compute
\begin{align}\nonumber
	\mathcal{V}_{0,\infty}(H) =& \int_{0}^{y_0 -1 } \left| \frac{1}{| \tilde  W_1|^{\frac 14}}\frac{\d^2}{\d y^2} \left(|\tilde W_1|^{-\frac 14}\right) - \frac{W_2}{|\tilde W_1|^{\frac 12}} - \frac{5 |\tilde W_1|^{\frac 12}}{16 \varsigma(y)^3} \right|\d y  \\ \nonumber & + \int_{y_0 - 1}^{y_0 +1 } \left| \frac{1}{|\tilde  W_1|^{\frac 14}}\frac{\d^2}{\d y^2} \left(|\tilde W_1|^{-\frac 14}\right) - \frac{W_2}{|\tilde W_1|^{\frac 12}} - \frac{5 |\tilde W_1|^{\frac 12}}{16 \varsigma(y)^3} \right|\d y
	\\ \nonumber & +\int_{y_0 + 1}^{+\infty } \left| \frac{1}{|\tilde  W_1|^{\frac 14}}\frac{\d^2}{\d y^2} \left(|\tilde W_1|^{-\frac 14}\right) - \frac{W_2}{|\tilde W_1|^{\frac 12}} - \frac{5 |\tilde W_1|^{\frac 12}}{16 \varsigma(y)^3} \right|\d y \\ &=: I + II + III
\end{align}
and estimate each term independently relying on \cref{prop:estimatesontildew1}.

\paragraph{Term \texorpdfstring{$I$}{I}} We estimate   term $I$ as 
\begin{align}\label{eq:termI}
I \lesssim	\int_0^{y_0 - 1} \frac{1}{|\tilde W_1|^{\frac 52} }\left(\frac{\d \tilde W_1}{\d y}\right)^2
 + \frac{1}{|\tilde W_1|^{\frac 32} }\left|\frac{\d^2 \tilde W_1}{\d y^2}\right| +\frac{|W_2|}{|\tilde W_1|^{\frac 12}}+ \frac{|\tilde W_1|^{\frac 12}}{\varsigma^3} \d y.\end{align} 
  We consider the first term appearing in \eqref{eq:termI} and in view of  \cref{prop:estimatesontildew1} we obtain
 \begin{align}
 \int_0^{y_0 - 1} &\frac{1}{|\tilde W_1|^{\frac 52} }\left(\frac{\d \tilde W_1}{\d y}\right)^2
 \d y \lesssim  \int_0^{y_0 - 1} \frac{\frac{\d \tilde W_1}{\d y}}{|\tilde W_1|^{\frac 32} } 
 \d y  \lesssim \frac{1}{|\tilde W_1|^{\frac 12}} \Big\vert_{0}^{y_0-1} \lesssim \tilde \lambda^{\frac 12} .
 \end{align} 
For the second term involving the second derivative, we use \eqref{eq:estiminregionI} to conclude that 
  \begin{align}
  	\int_0^{y_0 - 1} \frac{1}{|\tilde W_1|^{\frac 32} }\left|\frac{\d^2 \tilde W_1}{\d y^2}\right| \d y \lesssim \int_0^{y_0 - 1} \frac{ \frac{\d \tilde W_1}{\d y}}{|\tilde W_1|^{\frac 32}} +  \frac{1-x(y)^2}{|\tilde W_1|^{\frac 12}} \d y \lesssim \tilde \lambda^{\frac 12}.
  \end{align}  
 For the third term we use that $|W_2|\lesssim 1-x(y)^2$ such that
 \begin{align}
 	\int_0^{y_0-1} \frac{|W_2|}{|\tilde W_1|^{\frac 12}} \lesssim \tilde \lambda^{\frac 12} \int_0^\infty (1-x(y)^2) \d y \lesssim \tilde \lambda^{\frac 12}.
 \end{align} 
 For the last term in \eqref{eq:termI}, we  have
 \begin{align} \nonumber
 \int_{0}^{y_0 -1}	\frac{|\tilde W_1|^{\frac 12}}{\varsigma^3}  \d y &\lesssim  \int_{0}^{y_0 -1} \frac{\sqrt{-\tilde W_1}}{\left( \int_{y}^{y_0} \sqrt{-\tilde W_1} \d \tilde y  \right)^2} \d y \lesssim \frac{1}{ \int_{y_0 - 1}^{y_0} \sqrt{-\tilde W_1} \d \tilde y} \\ & \lesssim \frac{1}{ \int_{y_0 - 1}^{y_0} \sqrt{|\tilde y - y_0| \tilde \lambda^{-1}} \d \tilde y}\lesssim \tilde \lambda^{\frac 12}.
 \end{align}

 \paragraph{Term \texorpdfstring{$II$}{II}}
 For this term, we use Taylor's theorem around the point $y=y_0$. We will now only consider the region $y\in [y_0-1,y_0+1]$ and use the notiation $' = \frac{\d}{\d y} $ throughout the following paragraph. 
We write $\tilde W_1 (y) = (y-y_0)  \tilde W_1' (y_0)+\frac 12   (y-y_0)^2    \tilde W_1''  (y_0) +\frac 16 (y-y_0)^3   \tilde W_1'''  (y_0)  + \frac{1}{24} (y-y_0)^{4} R_{W_1}(y)$, where the remainder $R_{W_1}(y)$ is smooth and satisfies  $R_{W_1}(y) = \tilde W_1'''' (\xi)$ for some $\xi$ between $y_0$ and $y$. Thus,
\begin{align}\tilde W_1 (y) =  (y-y_0)  \tilde W_1' (y_0)+\frac 12   (y-y_0)^2    \tilde W_1''  (y_0) +\frac 16 (y-y_0)^3   \tilde W_1'''  (y_0)  +   O(\tilde \lambda^{-1}(y-y_0)^{4}  )
\end{align} in view of \eqref{eq:estimatesonw1inintermedaite}.
We note that $\zeta^3 = \frac{9}{4} (\int^{y}_{y_0} \sqrt{|\tilde W_1(\tilde y )|} \d \tilde y)^2  $ which we expand as 
\begin{align}\nonumber 
\zeta^3(y) = & \tilde W'(y_0) (y-y_0)^3 + \frac{3}{10} \tilde W_1''(y_0) (y-y_0)^4 \\& + \frac{1}{700} \left( - \frac{3\tilde W_1''(y_0)^2 }{W'(y_0)} + 50 \tilde W_1'''(y_0) \right) (y-y_0)^5 +\tilde \lambda^{-1} O(|y-y_0|^6).
\end{align}
Then, we further expand
  \begin{align}\label{eq:estiamteonw1divxi}
- \frac{5}{16} \frac{\tilde W_1}{\xi^3} = \frac{-5}{16|y-y_0|^2} - \frac{\tilde W_1''(y_0)}{16 W'(y_0)(y-y_0) } + \frac{117\tilde W_1''(y_0)^2 - 200 \tilde W_1'(y_0) \tilde W_1'''(y_0) }{6720 \tilde W_1'(y_0)^2} + O(|y-y_0| ).
  \end{align}
Here, the error in  $O(|y-y_0| )$ is estimated using  \eqref{eq:estimatesonw1inintermedaite} and we note that the homogeneity of \eqref{eq:estiamteonw1divxi} is such that no powers of $\tilde \lambda$ occur. 

We also expand $ {| \tilde  W_1|^{\frac 14}}\frac{\d^2}{\d y^2} \left(|\tilde W_1|^{-\frac 14}\right) $ around $y=y_0$ and obtain 
\begin{align}\nonumber 
 {| \tilde  W_1|^{\frac 14}}  \frac{\d^2}{\d y^2} \left(|\tilde W_1|^{-\frac 14}\right) = & \frac{5}{16 |y-y_0|^2} + \frac{\tilde W_1''(y_0)}{16 \tilde W_1'(y_0) (y-y_0)} \\ & + \frac{9 \tilde W_1''(y_0)^2 - 8 \tilde W_1'(y_0) \tilde W_1'''(y_0)}{192 \tilde W_1'(y_0)^2} + O(|y-y_0|).
\end{align}
Thus,
\begin{align}
 {| \tilde  W_1|^{\frac 14}}  \frac{\d^2}{\d y^2} \left(|\tilde W_1|^{-\frac 14}\right)  - \frac{5}{16} \frac{\tilde W_1}{\xi^3} = \frac{9 \tilde W_1''(y_0) - 10 \tilde W_1'''(y_0) \tilde W_1'(y_0) }{140 \tilde W_1'(y_0)^2} + O(|y-y_0|) 
\end{align}
from we  estimate  (using  \cref{prop:estimatesontildew1})
 \begin{align} \nonumber
 \Bigg| & \frac{1}{| \tilde  W_1|^{\frac 14}}\frac{\d^2}{\d y^2} \left(|\tilde W_1|^{-\frac 14}\right) - \frac{W_2}{|\tilde W_1|^{\frac 12}} - \frac{5 |\tilde W_1|^{\frac 12}}{16 \varsigma(y)^3} \Bigg| \\ & \lesssim \frac{1}{|W_1|^{\frac 12}}\left|\frac{|9 \tilde W_1''(y_0) - 10 \tilde W_1'''(y_0) \tilde W_1'(y_0)| }{140 \tilde W_1'(y_0)^2}  + |W_2| + O(|y-y_0|) \right| \lesssim\frac{1}{|W_1|^{\frac 12}} \lesssim  \frac{\tilde \lambda^{\frac 12 }}{|y-y_0|^{\frac 12}}\label{eq:estimateintheregiony0pm1}
 \end{align}
 uniformly in $y \in [y_0-1,y_0 + 1]$ from which we conclude that $|II|\lesssim \tilde \lambda^{\frac 12}$.  
 \paragraph{Term \texorpdfstring{$III$}{III}}
In the region $y\in [y_0+1, \infty)$ we first have
\begin{align}
	\varsigma^3(y) \gtrsim \left(\int_{y_0}^{y_0+1} \sqrt{\tilde W_1} \d \tilde y \right)^2+ \left(\int_{y_0+1}^{y} \sqrt{\tilde W_1} \d \tilde y\right)^2\gtrsim \frac{1}{\tilde \lambda} + (y-y_0 - 1)^2 \frac{1}{\tilde \lambda}
\end{align}
such that 
\begin{align}\label{eq:w1111}
	\frac{|\tilde W_1|^{\frac 12}}{\varsigma^3} \lesssim\frac{ \tilde \lambda^{\frac 12} }{1+(y-y_0-1)^2}
\end{align}
which is integrable at $y=\infty$.
Moreover, we have
\begin{align}\label{eq:w1112}
\Bigg| \frac{1}{| \tilde  W_1|^{\frac 14}}\frac{\d^2}{\d y^2} \left(|\tilde W_1|^{-\frac 14}\right) & - \frac{W_2}{|\tilde W_1|^{\frac 12}}   \Bigg| \lesssim \frac{\d x}{\d y}\left( \frac{\frac{\d \tilde W_1}{\d y}}{\tilde W_1^{\frac 32}} + \frac{1}{\tilde W_1^{\frac 12}}\right)\lesssim  \frac{\d x}{\d y} \tilde \lambda^{\frac 12}.
\end{align}
Combining the estimates \eqref{eq:w1111} and \eqref{eq:w1112} we obtain that $|III|\lesssim \tilde \lambda^{\frac 12}$ which concludes the proof.
\end{proof}
\end{lemma}
Finally, we also introduce 
\begin{align}
	\hat W_1 := \frac{\tilde W_1}{\varsigma} \text{ or equivalently } \hat W_1 = \left( \frac{\d \varsigma}{\d y} \right)^2
\end{align}
which we will bound from below in the following. To do that we also use the error-control function  $M_{\Ai}$ which is defined in \eqref{eq:definsofeaimai} in the appendix.
\begin{lemma}\label{lem:boundsonmoverw}
We  have
\begin{align}
	\frac{M_{\Ai}(\lambda^{\frac 13} \varsigma(y))}{\hat W_1^{\frac 14}} \lesssim \tilde \lambda^{\frac 16}.
\end{align}
\begin{proof}
	First, for $y_0 -1 \leq y \leq y_0$ we have
\begin{align}
\frac 23	(-\varsigma)^{\frac 32} = \int_{y}^{y_0} \sqrt{ |\tilde W_1| } \d \tilde y \leq (y_0 - y) \sqrt{ - \tilde W_1(y)}
\end{align}
	and for $y_0 \leq y \leq y_0 + 1 $ we have  \begin{align}
	\frac 23	\varsigma^{\frac 32} = \int_{y_0}^{y} \sqrt{ \tilde W_1 } \d \tilde y \leq (y-y_0)  \sqrt{\tilde W_1(y)},\end{align}
	where we have used the monotonicity of $\tilde W_1$. Hence, 
	\begin{align}\hat W_1 = \frac{\tilde W_1}{\varsigma} \gtrsim \left(\frac{\tilde W_1}{y-y_0}\right)^{\frac 23}\gtrsim \tilde \lambda^{-\frac 23} \end{align}
for $|y-y_0|\leq 1$. Now, using $ M_\Ai(x) \lesssim 1$ we conclude that for $|y-y_0|\leq 1$ we have
\begin{align}
	\frac{M_{\Ai}(\lambda^{\frac 13} \varsigma(y))}{\hat W_1^{\frac 14}} \lesssim 	\frac{1}{\hat W_1^{\frac 14}} \lesssim \tilde \lambda^{\frac 16}.
\end{align}
For $|y-y_0 | \geq 1$, we use that $M_{\Ai}(x) \lesssim |x|^{- \frac 14}$ to obtain 

\begin{align}
	\frac{M_{\Ai}(\lambda^{\frac 13} \varsigma(y))}{\hat W_1^{\frac 14}} \lesssim \frac{|\varsigma|^{\frac 14}}{|\varsigma|^{\frac 14} \lambda^{\frac{1}{12}} |\tilde W_1|^{\frac 14}}\lesssim \frac{\tilde{\lambda}^{\frac 14}}{\lambda^{\frac{1}{12}} } \lesssim \tilde \lambda^{\frac{1}{6}} m^{-\frac 16} \lesssim \tilde \lambda^{\frac{1}{6}}.
\end{align}
\end{proof}
\end{lemma}
Now, we are in the position to define the following fundamental solutions.
\begin{prop}\label{prop:defnofw1w2}
	There exist solutions $w_1$ and $w_2$ of \eqref{eq:homo2} satisfying
	\begin{align} \label{eq:defnw1}
		& w_1 = \frac{1}{\hat W_1^{\frac 14}} \left( \Ai( \lambda^{\frac 13} \varsigma(y)) + \eta_{\Ai}(\lambda,y) \right)\\
		& w_2 = \frac{1}{\hat W_1^{\frac 14}} \left( \Bi( \lambda^{\frac 13} \varsigma(y)) + \eta_{\Bi}(\lambda,y) \right), \label{eq:defnw2}
	\end{align}
	where 
	\begin{align}
	&	\frac{|\eta_{\Ai}(\lambda, y)|}{M_{\Ai}(\lambda^{\frac 13} \varsigma)}, 	\frac{|\partial_y \eta_{\Ai}(\lambda, y)|}{\lambda^{\frac 13} N_{\Ai}(\lambda^{\frac 13} \varsigma) \hat W_1^{\frac 12}}	 \lesssim {E^{-1}_{\mathrm{Ai}}(\lambda^{\frac 13} \varsigma)} m^{-1},
		\\
		&	\frac{|\eta_{\Bi}(\lambda, y)|}{M_{\Ai}(\lambda^{\frac 13} \varsigma)}, 	\frac{|\partial_y \eta_{\Bi}(\lambda, y)|}{\lambda^{\frac 13} N_{\Ai}(\lambda^{\frac 13} \varsigma) \hat W_1^{\frac 12}} \lesssim E_{\mathrm{Ai}}(\lambda^{\frac 13} \varsigma) m^{-1}.
	\end{align}
	Moreover, the Wronskian of $w_1$ and $w_2$ satisfies \begin{align} \label{eq:wronskianw1w2}
		|\mathfrak W(w_1,w_2) |\sim \lambda^{\frac 13}.
	\end{align}
	\begin{proof}
		This follows from \cite[Chapter~11, Theorem~3.1]{olver} and the error bounds follow from the bounds on $\mathcal{V}_{0,\infty}(H)$ in \cref{lem:errorboundonH}. The Wronskian identity is a direct consequence of the chain rule.
	\end{proof}
\end{prop}
\begin{lemma}
There exists a constant $A_2\neq 0$ such that $s_1 = A_2 w_1$, where $w_1$ is defined in \eqref{eq:defnw1} and $s_1$ is defined in \eqref{eq:defns1}.
	\begin{proof}
		Note that both, $w_1$ and $s_1$ are recessive as $y\to \infty$. Since the space of solutions which are recessive at $y\to\infty$ is one-dimensional, we conclude that $s_1$ and $w_1$ are linearly dependent.
	\end{proof}
	\label{prop:a2}
\end{lemma}
In view of \cref{prop:a2} we define 
\begin{align}\label{eq:defns2new}
	s_2 := \frac{1}{A_2 \mathfrak{W}(w_1,w_2)} w_2,
\end{align} 
where $w_2$ is as in \eqref{eq:defnw2}. Note that this implies that 
\begin{align}
	\mathfrak W (s_1,s_2) = 1.
\end{align}
\begin{lemma}\label{prop:estimateongp2}
	Let $\tilde \lambda \in [\Xi^2+1, \infty)$ and let $s_2$ as in \eqref{eq:defns2new}. Then, $g_p$ defined in \eqref{eq:defngp2} 
	satisfies
	\begin{align}\label{eq:lefthandside2}
	\int_{0}^{\infty} g_p(y)^2 (1-x(y)^2) \Delta_x \d y \lesssim m.
	\end{align}
	\begin{proof}
		Analogous to the proof of \cref{prop:estimateongp} we first estimate \begin{align} \nonumber
		\int_0^\infty s_2^2(y) & \left(\int_{y}^\infty s^2_1(\tilde y) F(\tilde y) \d\tilde y\right)^2 (1-x(y)^2) \Delta_x \d y \\ & = \int_0^\infty \frac{w_2^2(y)}{\mathfrak W(w_1,w_2)^2}  \left(\int_{y}^\infty w_1(\tilde y) s_1 (\tilde y) F(\tilde y) \d\tilde y \right)^2 (1-x(y)^2) \Delta_x(y) \d y . \label{eq:boundsinhighfreq}
		\end{align}
		Now, we use \cref{prop:defnofw1w2} and standard bounds on Airy functions from \cref{sec:airyfunctions}, as well as  \cref{lem:boundsonmoverw} to obtain \begin{align}&|w_1(y)|\lesssim \left| E^{-1}_{\Ai}(\lambda^{\frac 13} \varsigma(y)) \frac{M_{\Ai}(\lambda^{\frac 13} \varsigma(y))}{\hat W_1^{\frac 14} (y) } \right| \lesssim E^{-1}_{\Ai}(\lambda^{\frac 13} \varsigma(y)) \tilde \lambda^{\frac 16},\\ & |w_2(y)|\lesssim \left| E_{\Ai}(\lambda^{\frac 13} \varsigma(y)) \frac{M_{\Ai}(\lambda^{\frac 13} \varsigma(y))}{\hat W_1^{\frac 14} (y) } \right| \lesssim E_{\Ai}(\lambda^{\frac 13} \varsigma(y)) \tilde \lambda^{\frac 16}. \end{align}
		Now, plugging these estimates into \eqref{eq:boundsinhighfreq} and using that $E_{\textup{Ai}}^{-1}(\lambda^{\frac 13} \varsigma(y))$ is a decreasing function, we conclude
		\begin{align} \nonumber
			\int_0^\infty s_2^2(y) & \left(\int_{y}^\infty s^2_1(\tilde y) F(\tilde y) \d\tilde y\right)^2 (1-x(y)^2) \Delta_x \d y \\ \nonumber
&	\lesssim \int_0^\infty   \frac{\tilde \lambda^{\frac 23}}{\mathfrak W(w_1,w_2)^2} \left( \int_{y}^\infty s_1(\tilde y) F(\tilde y) \d \tilde y\right)^2 (1-x(y)^2) \Delta_x \d y.
		\\ \nonumber & \lesssim \frac{\tilde \lambda^{\frac 23}}{\mathfrak W(w_1,w_2)^2} \int_0^\infty (1-x(y)^2) \Delta_x \d y 		\\ \nonumber &\;\;\;\; \times \int_0^\infty s_1^2( y) (1-x(y)^2) \Delta_x \d y \int_0^\infty F^2(y) \frac{1}{(1-x(y)^2) \Delta_x} \d y 
			\\   & \lesssim \frac{\tilde \lambda^{\frac 23} m^2}{\mathfrak W(w_1,w_2)^2} . 
		\end{align}
		
		For the second term, we argue similarly and estimate
		\begin{align}\nonumber
		\int_0^\infty s_1^2(y) & \left(\int_{0}^y s_2(\tilde y)  s_1(\tilde y) F(\tilde y) \d\tilde y\right)^2 (1-x(y)^2) \Delta_x \d y \\ \nonumber& \leq  \left(\int_0^\infty | s_2(\tilde y)  s_1(\tilde y) F(\tilde y)| \d\tilde y\right)^2 \int_0^\infty  s_1^2(y)(1-x(y)^2) \Delta_x \d y \\ \nonumber & = \frac 12  \left(\int_0^\infty | s_2(\tilde y)  s_1(\tilde y) F(\tilde y)| \d\tilde y\right)^2 \\ & \lesssim \frac{\tilde \lambda^{\frac 23}}{\mathfrak W (w_1, w_2)^2}  \left(\int_{0}^{\infty} |F| \d \tilde y \right)^2 \lesssim 	 \frac{\tilde \lambda^{\frac 23} m^2}{\mathfrak W (w_1, w_2)^2}.
		\end{align}
		
		Thus, we conclude 	\begin{align}
		\int_{0}^{\infty} g_p(y)^2 (1-x(y)^2) \Delta_x \d y \lesssim 	 \frac{\tilde \lambda^{\frac 23} m^2}{\mathfrak W (w_1, w_2)^2} \lesssim	 \frac{\tilde \lambda^{\frac 23} m^2}{\lambda^{\frac 23}}\lesssim m^{\frac 23} \lesssim m.
		\end{align}
	\end{proof}
	\end{lemma}

\section{The radial o.d.e.\ on the exterior}
\label{sec:radialext}
We will now derive for which frequency parameters  $(\omega,m,\ell)$ the poles of the interior scattering operator at $\omega = \omega_- m$  coincide with frequency parameters which are exposed to stable trapping on the black hole exterior. This allows us then to define the set $\mathscr{P}_{\textup{Blow-up}}$ in \cref{sec:defnblow}. Thus, we will analyze the radial o.d.e.\ at frequency $\omega=\omega_- m$.
\subsection{Resonance: Radial o.d.e.~at interior scattering poles  allows for stable trapping}
We will first determine the range of angular eigenvalues $\lambda_{m\ell}(a m \omega_-)$ at the interior scattering poles $\omega=\omega_- m$  for which the radial o.d.e.\ admits stable trapping.
Recall from \eqref{eq:highfrequencypart} that the normalized high frequency part of the potential with $\omega=\omega_- m$ is given by 
\begin{align}\label{eq:Vomega-}
&V_{\textup{main}} =  \frac{\Delta}{(r^2+a^2)^2} \left( \frac{\lambda_{m\ell}(a\omega_- m)}{m^2} + \omega_-^2 a^2 - 2  a \omega_- \Xi \right) - ( \omega_- - \omega_r )^2.
\end{align}
Note that \begin{align}&V_{\textup{main}} \to - (\omega_- - \omega_+)^2 < 0 \text{ as } r \to r_+,\\
&V_{\textup{main}} \to l^{-2}\left( \frac{\lambda_{m\ell}(a\omega_- m)}{m^2} + \omega_-^2 a^2 - 2  a \omega_- \Xi \right) -   \omega_-^2    \text{ as } r \to \infty,\label{eq:vmainatinf}
\\
& V_{\textup{main}}(r_\mathrm{cut}) \geq \Big(l^{-2}  + 3 \frac{M^2}{\Xi} \Big(\frac{9M^2}{\Xi^2} + a^2\Big)^{-2} \Big) \left(\frac{\lambda_{m\ell}(a\omega_- m)}{m^2} + \omega_-^2 a^2 - 2  a \omega_- \Xi \right) -  \omega_-^2 ,\label{eq:vmainatrcut}
\end{align}
where $r_{\mathrm{cut}} := \frac{3M}{\Xi}$. Remark that \eqref{eq:vmainatrcut} follows from (using $l^2 = a^2 + l^2 \Xi$)
\begin{align}
\frac{\Delta( r)}{(r^2+a^2)^2} = l^{-2}  \frac{ (r^2+a^2)(r^2 + l^2) }{(r^2+a^2)^2}   - \frac{2Mr}{(a^2+r^2)^2 } = l^{-2} + \frac{\Xi(r^2+a^2) - 2Mr}{(a^2+r^2)^2},
\end{align}
 together with $\Xi (r_{\mathrm{cut}}^2 + a^2)  - 2 M r_{\mathrm{cut}} = \frac{3M^2}{\Xi} + a^2 \Xi$.

For the potential $V_{\textup{main}}$ to admit a region of stable trapping, we require that   $V_{\textup{main}}$ has two roots $r_{1} < r_{2}$, see already \cref{fig:turningpoints}. A sufficient condition for that is that the angular eigenvalues $\lambda_{m\ell}(a \omega_- m) m^{-2}$ are such that \eqref{eq:vmainatinf} is negative and \eqref{eq:vmainatrcut} is positive. In this case, we denote with $r_1 = r_1( {\lambda_{m\ell}(a\omega_- m)}{m^{-2}}, \mathfrak p) < r_2 = r_2( {\lambda_{m\ell}(a\omega_- m)}{m^{-2}}, \mathfrak p)$ the two largest roots. (We will show later that indeed, these are the only roots for $r\geq r_+$.)
This leads us to define the following range of angular eigenvalues
\begin{align}
E:=\bigsqcup_{\mathfrak p \in \mathscr P} \{\mathfrak p\}\times E_\mathfrak p = \bigsqcup_{\mathfrak p \in \mathscr P} \{\mathfrak p \}\times (\mu_0(\mathfrak p), \mu_1 (\mathfrak p) ),\label{eq:fibrebundle}
\end{align}
where 
\begin{align}\nonumber
E_\mathfrak p := \bigg\{\tilde \mu \in&  ( \Xi^2, \omega_-^2 (l^2-a^2) +2a\omega_- \Xi)  \colon    \text{ Every }   \mu \in [\tilde \mu, \omega_-^2 (l^2-a^2) +2a\omega_- \Xi) \text{ satisfies} \\ \nonumber
&
\Big(l^{-2}  + 3 \frac{M^2}{\Xi} \Big(\frac{9M^2}{\Xi^2} + a^2\Big)^{-2} \Big) \left(\mu + \omega_-^2 a^2 - 2  a \omega_- \Xi \right) -    \omega_-^2 >0 ,
\\ &  \Gamma \Big(\Delta^{-1}\Big)(\mu + a^2 \omega_-^2 - 2 a \omega_- \Xi) - \Gamma\Big(  \frac{(r^2+a^2)^2}{\Delta^2}(\omega_- - \omega_r)^2\Big)  < -\frac{\omega_-}{a}l^4 r^{-4} \text{ for } r \geq r_2(\mu, \mathfrak p) \bigg\} . \label{eq:defnsp}
\end{align}
Remark that the last condition in \eqref{eq:defnsp} will be used in \cref{lem:vomega-estimates} while the other conditions in \eqref{eq:defnsp} guarantee that $V_\textup{main}$ has two roots. Here, we also recall the definition of $\Gamma = \frac{\partial}{\partial \vartheta}$ in \eqref{eq:defnofgamma}, where $\vartheta= a \omega_-$. 

By construction, $E_\mathfrak p$ is connected. Indeed, note that if $\tilde \mu_1,\tilde \mu_2 \in E_{\mathfrak p}$, then every $\tilde \mu_3$ with $\tilde \mu_1\leq  \tilde  \mu_3 \leq \tilde \mu_2 $ satisfies $\tilde \mu_3\in E_{\mathfrak p}$. Thus, $E_\mathfrak p$ is a (a priori possibly empty) bounded interval.  We define $\mu_0 (\mathfrak p):=\inf E_\mathfrak p$, $\mu_1( \mathfrak p):=\sup E_\mathfrak p$.   We will show that $E$ is a fiber bundle. To do so we first show that $E_\mathfrak p$ is open and non-empty. 
\begin{lemma}\label{eq:lemspnonempty}
	For any $\mathfrak p\in \mathscr{P}$, the set $E_\mathfrak p $ defined in \eqref{eq:defnsp} is  a non-empty bounded open interval. 
	\begin{proof}
	First, we will show that  \begin{align}
\Xi^2 < \omega_-^2 (l^2-a^2) +2a\omega_- \Xi \end{align} which in turn follows from \begin{align}r_-^2 < a l .\end{align}
		To see that $r_-^2 < a l$ holds true, we write $\Delta (r)$ in terms of $r_-$. More precisely, from $\Delta(r_-)=0$ we have
		\begin{align}
			2M = r_-^{-1} (a^2+r_-^2) \left( 1+ \frac{r_-^2}{l^2}\right)
			\end{align} 
		from which we obtain
		\begin{align}
\Delta(r) = (r^2+a^2)(1+ \frac{r^2}{l^2}) - \frac{r}{r_-} (a^2+r_-^2) \left( 1+ \frac{r_-^2}{l^2}\right).
		\end{align}
After a polynomial long division, this reduces to 
		\begin{align}
		\Delta(r) = l^{-2} (r - r_-) (r^3 + r^2 r_- + r(r_-^2 + a^2 + l^2) - a^2 l^2 r_-^{-1}).
		\end{align}
		Hence, \begin{align} 0> \partial_r \Delta (r_-) = l^{-2} r_-^{-1} ( 3 r_-^4 + r_-^2 a^2 + r_-^2 l^2 - a^2 l^2)\end{align}
		implies \begin{align}3 r_-^4 < a^2 l^2\end{align} from which  \begin{align}
		\label{eq:onr-}
	r_-^2 < a l
		\end{align}
		follows.  
		
		Note that for $\mu = \omega_-^2 (l^2-a^2) +2a\omega_- \Xi$ we have 
		\begin{align} \nonumber 
			\Big(l^{-2}  + &  3 \frac{M^2}{\Xi} \Big(\frac{9M^2}{\Xi^2} + a^2\Big)^{-2} \Big)  \left(\mu + \omega_-^2 a^2 - 2  a \omega_- \Xi \right) -    \omega_-^2 \\
			& = 	\Big(l^{-2}  + 3 \frac{M^2}{\Xi} \Big(\frac{9M^2}{\Xi^2} + a^2\Big)^{-2} \Big)\omega_-^2 l^2  -    \omega_-^2 \nonumber  \\
			& = 3 \frac{M^2}{\Xi} \Big(\frac{9M^2}{\Xi^2} + a^2\Big)^{-2}  \omega_-^2 l^2  > 0.
		\end{align}
	Thus, for $\mu < \omega_-^2 (l^2-a^2) +2a\omega_- \Xi$ and $\omega_-^2 (l^2-a^2) +2a\omega_- \Xi-\mu$ sufficiently small, we have \begin{align}\Big(l^{-2}  + 3 \frac{M^2}{\Xi} \Big(\frac{9M^2}{\Xi^2} + a^2\Big)^{-2} \Big) \left(\mu + \omega_-^2 a^2 - 2  a \omega_- \Xi \right) -    \omega_-^2 >0. \end{align}
		Now, note that for $\mu = \lambda_{m\ell}(a\omega_- m) m^{-2} = \omega_-^2 (l^2-a^2) +2a\omega_- \Xi$, we have that $V_{\textup{main}}\to 0$ as $r\to\infty$. Thus, the largest root $r_2$ satisfies $r_2(\mu,\mathfrak p) \to +\infty$ as $\mu \to  \omega_-^2 (l^2-a^2) +2a\omega_- \Xi$ from below.
		Hence, the claim follows from the fact that to highest order in $r$, the last condition in \eqref{eq:defnsp} is fulfilled. More precisely, \begin{align} \nonumber 
	\lim_{r\to\infty} r^4 \left(	\Gamma \Big(\Delta^{-1}\Big)(\mu + a^2 \omega_-^2 - 2 a \omega_- \Xi) - \Gamma\Big(  \frac{(r^2+a^2)^2}{\Delta^2}(\omega_- - \omega_r)^2\Big) \right) = -  l^4  \Gamma(\omega_-^2) = \frac{-2   \omega_-}{a} l^4 <0 \end{align}
where we used $\Gamma(a \omega_-)=1$ and $\Gamma(a)=0$. 	
\end{proof}
\end{lemma}
From \eqref{eq:defnsp} it also follows that $\mu_0$ and $\mu_1$ are manifestly continuous functions on $\mathscr{P}$. Hence, $E$ is a  (topological) fiber bundle. Now, also note that $E$ is trivial with global trivialization $\varphi_E\colon E\to \mathscr P \times (0,1),  (\mathfrak p,\mu) \mapsto (\mathfrak p,\frac{\mu}{\mu_1 - \mu_0} - \frac{\mu_0}{\mu_1 - \mu_0})$ and we find (using this trivialization) two global sections
 \begin{align}\label{eq:defnsigma12} \sigma_1 \in \Gamma(E)  \text{ and } \sigma_2 \in \Gamma(E)  \text { with }  \sigma_1(\mathfrak p) <\sigma_2(\mathfrak p) 
 \end{align}  for all $\mathfrak p \in \mathscr{P}$ (in mild abuse of notation). 
For definiteness, we take 
\begin{align}\sigma_1(\mathfrak p) := \varphi_E^{-1}\big(\mathfrak p,\frac 12\big)\text{ and }\sigma_2(p):= \varphi_E^{-1}\big(\mathfrak p, \frac 34\big).
\end{align}
Having constructed $\sigma_1$ and $\sigma_2$, we will now show the existence of exactly two turning points $r_1 < r_2$ of $V_{\textup{main}}$.
\begin{lemma} \label{prop:trappingexist} Let $m^{-2}\lambda_{m\ell}(a \omega_- m) \in (\sigma_1, \sigma_2)$ as chosen in \eqref{eq:defnsigma12}. Then, $V_{\textup{main}}$ 
	has a maximum $r_\textup{max}\in (r_+,\infty)$, and two roots $r_1, r_2$ with $r_+ < r_1 < r_\textup{max} < r_2 <\infty$ such that
	\begin{align}
	\label{eq:vw11}	&V_{\textup{main}} > 0 \text{ for } r \in ( r_1 , r_2),\\ \label{eq:vw21}
	&V_{\textup{main}} < 0 \text{ for } r\in [r_+,\infty) \setminus [r_1,r_2]
	\end{align}
	and $r_2 - r_1 \gtrsim  1$.
	\begin{proof}
		By construction of $\sigma_1$ and $\sigma_2$, for any $m^{-2}\lambda_{m\ell}\in (\sigma_1, \sigma_2 )$, the potential $V_{\textup{main}}$ has a maximum and satisfies \begin{align}\lim_{r\to\infty} V_{\textup{main}} <0, V_{\textup{main}}(r=r_+) <0 \text{ and }V_{\textup{main}} (r = r_{\mathrm{cut}}) >0,\end{align}
		 where $r_{\mathrm{cut}} = \frac{3M}{\Xi}$. See also \cite[Lemma~3.1]{quasimodes}.
		
		We will show now that $V_{\textup{main}}$ has exactly two roots in $[r_+,\infty)$ from which \eqref{eq:vw11} and \eqref{eq:vw21} follow. Indeed, in view of the above, $V_{\textup{main}}$ either has two or four roots in $[r_+,\infty)$. To exclude the case of four roots, it suffices to exclude the case of three critical points in $[r_+,\infty)$. To see this, note that
		\begin{align} \nonumber
		\frac{\d V_{\textup{main}}}{\d r } = &\frac{( - 2 \Xi r^3+6 M r^2 - 2 \Xi a^2 r - 2 M a^2 ) m^{-2} (\lambda_{m\ell}(a\omega_- m)  + a^2\omega_-^2 - 2 a \omega_- \Xi )  }{(r^2+a^2)^3}\\
		& + \frac{ 4 a r\Xi (  \omega_r - \omega_- ) (r^2 + a^2) }{(r^2+a^2)^3}
		\end{align}
		has at most three real roots, one of which is in $[r_+,\infty)$ in view of the construction above. Indeed, one other root has to lie in $(-\infty,r_-]$ as \begin{align}\lim_{r\to-\infty} \frac{\d V_{\textup{main}}}{\d r }  > 0\text{ and } \frac{\d V_{\textup{main}}}{\d r }(r=r_-) = \frac{\partial_r \Delta (r_-)}{(r_-^2 + a^2)^2}  < 0.\end{align} Thus, $V_{\textup{main}}$ has at least one and at most two critical points in $[r_+,\infty)$ from which we deduce \eqref{eq:vw11} and \eqref{eq:vw21}.
	\end{proof}
\end{lemma}
\subsection{Fundamental pairs of solutions associated to trapping}\label{sec:fundamentalsolutions}
We will now define various solutions to the radial o.d.e.\ associated to the boundary and to the turning points. Note that the turning points define the transition from the trapping region to the semiclassically forbidden region. 
\subsubsection{Solutions associated to the boundary}
\label{sec:soltoboundary}
We first define the associated solution to the radial o.d.e.\ \eqref{eq:radial} which satisfies the Dirichlet boundary conditions at ${r^\ast} = \frac{\pi}{2} l $.
\begin{definition}\label{def:uinfty} For all frequencies $\omega\in \mathbb R,m \in \mathbb Z,\ell \in \mathbb Z_{\geq |m|}$ we define the solution $u_\infty$ as the unique solution to   \eqref{eq:radial} satisfying
\begin{align}
&u_\infty \left(\frac{\pi}{2} l \right) = 0\\
& u_{\infty}'\left( \frac{\pi}{2}l \right) =  1,
\end{align}
where we recall that $~^\prime = \frac{\d}{\d r^\ast}$. 
\end{definition}
We now define solutions associated to the event horizon $\mathcal H^+$. The limit $r^\ast\to-\infty$ of the radial o.d.e.\ \eqref{eq:radial} constitutes a regular singular point. 
In particular, in view of \cite[Theorem 2.2, Chapter~6]{olver2014asymptotics} (set $f = (\omega-\omega_+ m)^2$ and $g= V-\omega^2 - (\omega-\omega_+ m)^2 = O_{\omega,m,\ell}(\Delta)$  as $r^\ast\to-\infty$) we can define the following unique solutions to \eqref{eq:radial}.
 
\begin{definition}\label{defn:uhruhlexterior}
For all frequencies $\omega\in \mathbb R, m \in \mathbb Z,\ell \in \mathbb Z_{\geq |m|}$ we also define $\uhplus $ and $\uhminus$ as the unique solutions to  \eqref{eq:radial} satisfying
\begin{align}
&	u_{\mathcal{H}^+} = e^{-i (\omega  -\omega_+m )  {r^\ast}} + O_{\omega,m,\ell}(\Delta) \text{ as } {r^\ast} \to -\infty,\\
&	u_{\mathcal{H}^-} = e^{i (\omega - \omega_+ m ) {r^\ast}} + O_{\omega,m,\ell}(\Delta) \text{ as } {r^\ast} \to -\infty.
\end{align}
\end{definition} 
\begin{rmk}\label{rmk:volterra1}
Equivalently, we can define $u_{\mathcal{H}^+} $ (and similarly $u_{\mathcal{H}^+}$) as the unique solution to the Volterra  integral equation
\begin{align}\label{eq:defnofuhrthroughvolterra0}
	u_{\mathcal{H}^+} (r^\ast) = e^{-i (\omega  -\omega_+m )  {r^\ast}} + \int_{-\infty}^{r^\ast} \frac{\sin((\omega-\omega_+m)(r^\ast-y))}{\omega-\omega_+m} (V(y) - \omega^2 +   ( \omega-\omega_+ m)^2)	u_{\mathcal{H}^+} (y) \d y,
\end{align}
where $V$ is as in \eqref{eq:defnofvintermsofv0v1} and we note that $V(y) - \omega^2 +   ( \omega-\omega_+ m)^2= O_{\omega,m,\ell} (\Delta)$ as $r^\ast \to -\infty$. Existence and uniqueness is standard, see e.g.\
 \cite[Theorem 10.1, Chapter 6]{olver2014asymptotics}. Note also that the fact that 	$u_{\mathcal{H}^+}$ defined by \eqref{eq:defnofuhrthroughvolterra0} indeed solves the o.d.e.\ can be checked explicitly. 
\end{rmk}

\subsubsection{Solutions associated to turning points at interior scattering poles}\label{sec:solassocittoturningpoints}
For the solutions associated to the turning points we only consider the radial o.d.e.\ for $\omega= \omega_- m$ which we recall from \eqref{eq:radialomega} as 
 \begin{align}  
	-u'' + (m^2 V_{\textup{main}} + V_{1})u =0.
\end{align} 
Throughout \cref{sec:solassocittoturningpoints} we assume that $\sigma_1 \leq \lambda_{m\ell}m^{-2} \leq \sigma_2$ and in view of  \cref{prop:trappingexist} we denote the  turning points of $V_{\textup{main}}$ with $r_1^\ast := r^\ast (r_1) <  r^\ast(r_2)=: r_2^\ast $, where we recall that they do not coalesce for $\sigma_1 \leq \lambda_{m\ell}m^{-2} \leq \sigma_2$.
We will now define solutions associated to the turning points $r_1^\ast$ and $r_2^\ast$ as illustrated in \cref{fig:turningpoints}. We will closely follow \cite[Chapter~11, \S3]{olver2014asymptotics}, see also \cite{turning} for the original publication. 
 \begin{figure}[!h]
	\centering
\begingroup%
  \makeatletter%
  \providecommand\color[2][]{%
    \errmessage{(Inkscape) Color is used for the text in Inkscape, but the package 'color.sty' is not loaded}%
    \renewcommand\color[2][]{}%
  }%
  \providecommand\transparent[1]{%
    \errmessage{(Inkscape) Transparency is used (non-zero) for the text in Inkscape, but the package 'transparent.sty' is not loaded}%
    \renewcommand\transparent[1]{}%
  }%
  \providecommand\rotatebox[2]{#2}%
  \newcommand*\fsize{\dimexpr\f@size pt\relax}%
  \newcommand*\lineheight[1]{\fontsize{\fsize}{#1\fsize}\selectfont}%
  \ifx\svgwidth\undefined%
    \setlength{\unitlength}{365.25610054bp}%
    \ifx\svgscale\undefined%
      \relax%
    \else%
      \setlength{\unitlength}{\unitlength * \real{\svgscale}}%
    \fi%
  \else%
    \setlength{\unitlength}{\svgwidth}%
  \fi%
  \global\let\svgwidth\undefined%
  \global\let\svgscale\undefined%
  \makeatother%
  \begin{picture}(1,0.41569001)%
    \lineheight{1}%
    \setlength\tabcolsep{0pt}%
    \put(0,0){\includegraphics[width=\unitlength,page=1]{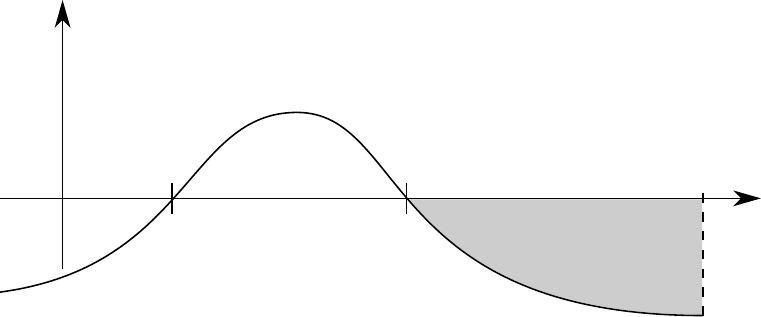}}%
    \put(0.5169779,0.10998776){\color[rgb]{0,0,0}\makebox(0,0)[lt]{\lineheight{1.25}\smash{\begin{tabular}[t]{l}$r_2^\ast$\end{tabular}}}}%
    \put(0.20895997,0.11012886){\color[rgb]{0,0,0}\makebox(0,0)[lt]{\lineheight{1.25}\smash{\begin{tabular}[t]{l}$r_1^\ast$\end{tabular}}}}%
    \put(0.27959915,0.28973286){\color[rgb]{0,0,0}\makebox(0,0)[lt]{\lineheight{1.25}\smash{\begin{tabular}[t]{l}$V_{\textup{main}}(r^\ast,\omega,m,\lambda_{m\ell}(a \omega))$\end{tabular}}}}%
    \put(0.91151512,0.18038846){\color[rgb]{0,0,0}\makebox(0,0)[lt]{\lineheight{1.25}\smash{\begin{tabular}[t]{l}$\frac \pi 2l$\end{tabular}}}}%
  \end{picture}%
\endgroup%

	\caption{Rough shape of the potential $V_{\textup{main}}$ and the turning points $r_1^\ast$ and $r_2^\ast$.}
	\label{fig:turningpoints}
\end{figure}
 We begin by defining new variables $\xi_1$ and $\xi_2$  in the following which will by construction solve 
\begin{align}
\xi_1 \left( \frac{\d \xi_1 }{\d r^\ast }\right)^2 = V_{\textup{main}}, \;\;\; \xi_2 \left( \frac{\d \xi_2 }{\d r^\ast }\right)^2 = V_{\textup{main}}.
\end{align}
 
\begin{definition}\label{eq:defnofxi1andxi2}
For some fixed  $\epsilon=\epsilon(\mathfrak p)>0$ sufficiently small depending smoothly on $\mathfrak p$, we define 
\begin{align} \label{eq:defnxi1}
	&\xi_1 ({r^\ast}, m):= \begin{cases}
	- \left( \frac 32 \int^{{r^\ast_1}}_{{r^\ast}} (- V_{\textup{main}})^{\frac 12} \d y \right)^{\frac 23} & {r^\ast} \in (-\infty ,{r^\ast_1}) \\
	\left( \frac 32 \int_{{r^\ast_1}}^{{r^\ast}}  V_{\textup{main}}^{\frac 12 } \d y \right)^{\frac 23} & {r^\ast} \in  [{r^\ast_1}, {r^\ast_2} - \epsilon ]
	\end{cases}\\
	&\xi_2 ({r^\ast}, m):= \begin{cases}
	\left( \frac 32 \int^{{r^\ast_2}}_{{r^\ast}}  V_{\textup{main}}^{\frac 12} \d y \right)^{\frac 23} & {r^\ast} \in [{r^\ast_1} + \epsilon, {r^\ast_2}) \\
	-\left( \frac 32 \int_{{r^\ast_2}}^{{r^\ast}} (- V_{\textup{main}})^{\frac 12 } \d y \right)^{\frac 23} & {r^\ast}  \in [{r^\ast_2} , \frac{\pi}{2} l]
	\end{cases}\\
	&	\hat f_1:= \frac{V_{\textup{main}}}{\xi_1} = \left(\frac{\d \xi_1 }{\d r^\ast}\right)^2\text{ for } {r^\ast} \in (-\infty, r_{2}^{\ast} - \epsilon],\\
	& \hat f_2:= \frac{V_{\textup{main}}}{\xi_2} = \left(\frac{\d \xi_1 }{\d r^\ast}\right)^2 \text{ for }{r^\ast} \in [r_{1}^\ast + \epsilon, \frac{\pi}{2}l], \label{eq:defnofxi2withf2}
\end{align}
where we note that $r^\ast \mapsto \xi_1$ is monotonically increasing and $r^\ast \mapsto \xi_2$ is monotonically decreasing. In particular, we  choose $\epsilon(\mathfrak p)>0$ sufficiently small  to ensure that 
\begin{align} \frac{\xi_2 (r_1^\ast+2\epsilon)}{\xi_1(r_1^\ast+2 \epsilon)}  \geq 3. \label{eq:smallnessofepsilon}
	\end{align} 
\end{definition}
As in \cite[Chapter~11, \S3]{olver2014asymptotics} we note that with the above definitions the new unknowns $W_1:= \left( \frac{\d \xi_1}{\d r^\ast} \right)^{\frac 12} u $ and  $W_2:= \left( - \frac{\d \xi_2}{\d r^\ast} \right)^{\frac 12} u $ respectively solve 
\begin{align}\label{eq:transformedformofode}
\frac{\d^2 W_1}{\d \xi_1^2} = (m^2 \xi_1 + \Psi_1 )W_1
\end{align} 
\begin{align}\label{eq:transformedformofode2}
	\frac{\d^2 W_2}{\d \xi_2^2} = (m^2 \xi_2 + \Psi_2)W_2
\end{align} 
for the error functions
\begin{align}\label{eq:defnofpsi1intermsoff1}
\Psi_1 = \frac{1}{\hat f_1^{\frac 14}} \frac{\d^2 \hat f_1^{\frac 14}}{\d \xi_1^2} + \frac{V_1}{\hat f_1} = -  \frac{1}{ {\hat f_1}^{\frac 34}} \frac{\d^2 }{\d {r^\ast}^2} \left( \frac{1}{ {\hat f_1}^{\frac 14}} \right) + \frac{V_1}{\hat f_1}
\end{align}
and
\begin{align}
	\Psi_2 = \frac{1}{\hat f_2^{\frac 14}} \frac{\d^2 \hat f_2^{\frac 14}}{\d \xi_2^2} + \frac{V_1}{\hat f_2} = -  \frac{1}{ {\hat f_2}^{\frac 34}} \frac{\d^2 }{\d {r^\ast}^2} \left( \frac{1}{ {\hat f_2}^{\frac 14}} \right) + \frac{V_1}{\hat f_2}.
\end{align}
As in \cite[Chapter~11, \S3]{olver2014asymptotics} this can be equivalently written as 
\begin{align} \label{eq:psi1alternative}
\Psi_1   = \frac{5}{16 \xi_1^2} + \left[ 4 V_{\textup{main}} \frac{\d^2 }{\d {r^\ast}^2} V_{\textup{main}}  - 5\left(   \frac{\d  }{\d {r^\ast} } V_{\textup{main}}  \right)^2   \right]  \frac{\xi_1}{16 V_{\textup{main}}^3} + \frac{\xi_1 V_1 }{V_{\textup{main}} }
\end{align}
and 
 \begin{align} \label{eq:psi2alternative}
	\Psi_2   = \frac{5}{16 \xi_2^2} + \left[ 4 V_{\textup{main}} \frac{\d^2 }{\d {r^\ast}^2} V_{\textup{main}}  - 5\left(   \frac{\d  }{\d {r^\ast} } V_{\textup{main}}  \right)^2   \right]  \frac{\xi_2}{16 V_{\textup{main}}^3} + \frac{\xi_2 V_1 }{V_{\textup{main}} }.
\end{align}

\begin{lemma}\label{lem:smoothnessandsmoothdependceofpsi1psi2}
The functions $(-\infty, r_2^\ast-\epsilon]\ni r^\ast \mapsto \xi_1/(r^\ast-r_1^\ast), [r_1^\ast+\epsilon, \frac{\pi l }{2}] \ni r^\ast \mapsto \xi_2/(r_2^\ast - r^\ast)$ are smooth positive functions in their respective domains. Moreover, $(-\infty, r_2^\ast-\epsilon]\ni r^\ast \mapsto \xi_1(r^\ast)$ (resp.\ $ [r_1^\ast+\epsilon, \frac{\pi l }{2}] \ni r^\ast \mapsto  \xi_2(r^\ast)$) is a strictly increasing (resp.\ decreasing) smooth function. 

The error control functions $(-\infty, r_2^\ast-\epsilon]\ni r^\ast \mapsto  \Psi_1(r^\ast)$ and $ [r_1^\ast+\epsilon, \frac{\pi l }{2}] \ni r^\ast \mapsto \Psi_2(r^\ast)$ are also smooth functions. Moreover, $\xi_1,\xi_2,\Psi_1,\Psi_2$ depend smoothly on $(r^\ast, \mathfrak p)$. 
\end{lemma}
\begin{proof}
Mutatis mutandis, this follows from \cite[Chapter 11, \S3, Lemma 3.1]{olver2014asymptotics}. For convenience of the reader we will   give a proof for the case of $\xi_1$ and $\Psi_1$ following the proof of \cite[Chapter 11, \S3, Lemma 3.1]{olver2014asymptotics}. 
We first write $p(r^\ast):= \frac{V_{\textup{main}}}{r^\ast - r^\ast_1}$ which is smooth as  $V_{\textup{main}}$ has a simple root at $r^\ast_1$. We also set $q(r^\ast):= (r^\ast - r^\ast_1)^{-\frac 32} \int_{r^\ast_1}^{r^\ast} (y - r^\ast_1)^{\frac 12	} p(y) \d y$ for $r^\ast > r_1^\ast$ such that by construction, $\frac{\xi_1 (r^\ast )}{r^\ast-r^\ast_1} = (\frac 32 q(r^\ast))^{\frac 23}$.  Moreover, $q$ is also smooth with $n$-th derivative $q^{(n)} (r^\ast) = (r^\ast-r_1^\ast)^{\frac{ 2n +1}{2}} \int_{r^\ast_1}^{r^\ast} (y-r^\ast_1)^{\frac{2n-1}{2}} p^{(n)} (y)\d y $ which follows from integrating by parts. By the mean value theorem we have that for every $n$, the derivative $q^{(n)}(r^\ast)$ has a right limit $r^\ast \to r^\ast_1$ and in particular, $\lim_{r^\ast\to r_1^\ast} q(r^\ast) = \frac 23 p(r^\ast_1) \neq 0$, where the last property follows from the fact that $V_{\textup{main}}$ has a simple root at $r_1^\ast$. Thus, $\frac{\xi_1 (r^\ast )}{r^\ast-r^\ast_1} = (\frac 32 q(r^\ast))^{\frac 23}$ is positive and extends smoothly across  $r^\ast= r_1^\ast$. Arguing completely analogously for the region $r^\ast < r^\ast_1$ and noting that by construction the left and right derivatives agree at $r^\ast=r^\ast_1$, we obtain that $\frac{\xi_1 (r^\ast )}{r^\ast-r^\ast_1}$ is a smooth positive function on $(-\infty, r_2^\ast-\epsilon]$. 

Moreover, as $r_1^\ast$ and $V_{\textup{main}}$ depend smoothly on $(r^\ast,\mathfrak p)$, we also have that $\xi_1$ depends smoothly on $(r^\ast,\mathfrak p)$.
Now, we note that $\hat f_1 (r^\ast)= p(r^\ast)^2 ( \frac{3}{2}q(r^\ast))^{- \frac 23}$ such that $\hat f_1$ is a smooth positive function of $r^\ast$ and $\hat f_1$ also depends smoothly on $\mathfrak p$ in view of the above properties  shown about $p$ and $q$.  In view of \eqref{eq:defnofpsi1intermsoff1}, this finally shows that $\Psi_1$ is a smooth function of $r^\ast$ (and of $\xi_1$) and also depends smoothly on $\mathfrak p$.  The claim about $\xi_2$ and $\Psi_2$ follows completely analogously. 
\end{proof}

\begin{definition}
	We now define the error control functions 
	\begin{align}
	H_1({r^\ast}) &:= \int_{{r^\ast_1}}^{{r^\ast}} \left\{ \frac{1}{|V_{\textup{main}}|^{\frac 14}} \frac{\d^2}{\d {r^\ast}^2} \left(\frac{1}{|V_{\textup{main}} |^{\frac 14} }\right) - \frac{V_{1}}{|V_{\textup{main}}|^{\frac 12}} - \frac{5 |V_{\textup{main}}|^{\frac 12}}{16 |\xi_{1}|^3} \right\} \d y ,\\
	H_2({r^\ast}) &:= \int^{{r^\ast_2}}_{{r^\ast}} \left\{ \frac{1}{|V_{\textup{main}}|^{\frac 14}} \frac{\d^2}{\d {r^\ast}^2} \left(\frac{1}{|V_{\textup{main}} |^{\frac 14} }\right) - \frac{V_{1}}{|V_{\textup{main}}|^{\frac 12}} - \frac{5 |V_{\textup{main}}|^{\frac 12}}{16 |\xi_{2}|^3} \right\}\d y
	\end{align}
which are equivalently characterized (using \eqref{eq:psi1alternative}, see also \cite[Chapter~11, \S3.3, equ.\ (3.07), (3.08)]{olver2014asymptotics}) as 
\begin{align}
	&H_1(r^\ast) = - \int_0^{\xi_1(r^\ast)} |v|^{-\frac 12} \Psi_1(v) \d v\\
	& H_2(r^\ast) = - \int_0^{\xi_2(r^\ast)} |v|^{-\frac 12} \Psi_2(v) \d v.
\end{align}  
\end{definition}

\begin{lemma}\label{prop:properrorairy}
The error control functions $H_1$ and $H_2$ satisfy
	\begin{align}
	&\mathcal{V}_{{-\infty} , r_2^\ast - \epsilon }\left( H_1 \right) \lesssim  1\label{eq:hbound1}\\
	&\mathcal{V}_{r_1^\ast + \epsilon, l\frac \pi 2} \left( H_2 \right) \lesssim    1.\label{eq:hbound2}
	\end{align}
		\begin{proof}
			We begin with $H_2$. We have that $\xi_2( [r_1^\ast+\epsilon, \frac{\pi}{2}l] )$ is compact as $\xi_2$ is continuous. Moreover, since $\xi_2 \mapsto \Psi_2 $ is continuous we have $\sup_{\xi_2( [r_1^\ast+\epsilon, \frac{\pi}{2}l] )} |\Psi_2| \lesssim 1 $. As  $|v|^{-\frac 12} \in L^1_{loc}$ we obtain \eqref{eq:hbound2} as 
			\begin{align}\label{eq:boundonh2}
\mathcal{V}_{r_1^\ast + \epsilon, l\frac \pi 2} \left( H_2 \right)  \leq  \int_{\xi_2(\frac{\pi l}{2})}^{\xi_2(r_1^\ast+\epsilon)} |v|^{-\frac 12} |\Psi_2(v) | \d v \lesssim 1. 
			\end{align}

			For $H_1$, we have to deal with the unbounded region $r^\ast \in (-\infty, r_2^\ast - \epsilon]$. We 	decompose \begin{align}\mathcal{V}_{{-\infty} , r_2^\ast- \epsilon }\left( H_1 \right) = \mathcal{V}_{{-\infty} , r_1^\ast - 1 }\left( H_1 \right) + \mathcal{V}_{ r_1^\ast - 1 , r_2^\ast - \epsilon }\left( H_1 \right).  \end{align}
			Completely analogous to the   proof of the bound on $H_2$ we have 
			\begin{align} \mathcal{V}_{ r_1^\ast - 1 , r_2^\ast - \epsilon }\left( H_1 \right) \lesssim 1.
			\end{align} 
			For the term $\mathcal{V}_{{-\infty} , r_1^\ast - 1 }\left( H_1 \right)$  we remark  that
			 \begin{align}  -V_{\textup{main}}\sim 1 ,\;\;   |V_{\textup{main}}'| , |V_{\textup{main}}''| \lesssim  e^{2\kappa_+ {r^\ast}} \text{ and }|V_{1}|\lesssim  e^{2 \kappa_+ {r^\ast}} 
			 \end{align}
			for $r^\ast \in (-\infty, r_1^\ast - 1).$
			 Using the lower bound $-V_{\textup{main}}$, we infer from \eqref{eq:defnxi1} that  \begin{align}- \xi_1(r^\ast) \gtrsim (-r^\ast)^{\frac 23}\end{align} 
			 for $r^\ast \in (-\infty, r_1^\ast - 1)$. Hence, 
			\begin{align} \nonumber
			\mathcal{V}_{{-\infty} , r_1^\ast - 1 }\left( H_1 \right) & \lesssim  \int_{-\infty}^{r_1^\ast - 1}\left| \frac{V''_{\textup{main}}}{|V_{\textup{main}}|^{\frac 32}}\right| + \left| \frac{V'^2_{\textup{main}}}{|V_{\textup{main}}|^{\frac 52}} \right| + \left| \frac{V_1}{|V_{\textup{main}}|^{\frac 12}}\right| + \frac{|V_{\textup{main}}|^{\frac 12}}{|\xi_1|^3} \d r^\ast \\
			& \lesssim  \int_{-\infty}^{r_1^\ast - 1} e^{2 \kappa_+ r^\ast} + \frac{1}{|r^\ast|^2} \d r^\ast \lesssim  1.
			\end{align} 
		\end{proof}
\end{lemma}
	
With the bounds in \cref{prop:properrorairy} in hand, we apply  \cite[Chapter~11, Theorem~3.1]{olver} which allow us to define the following. 
\begin{prop} \label{prop:defnsofuaianduabi}
	We define solutions to the radial o.d.e.\ \eqref{eq:radial} for $\omega = \omega_- m$ as
	\begin{align}\label{eq:defai1}
	&u_{\mathrm{Ai}1} ({r^\ast},m) := \hat f_1^{\frac 14}({r_1^\ast}) \hat f_1^{- \frac 14}({r^\ast}) \left\{ \mathrm{Ai}(m^{\frac 23} \xi_1) + \epsilon_{\mathrm{Ai}1}(m,{r^\ast} )\right\} \text{ for } {r^\ast} \in (-\infty, r_2^\ast - \epsilon],\\  \label{eq:defbi1}
	&	u_{\mathrm{Bi}1} ({r^\ast},m) :=  \hat f_1^{\frac 14}({r_1^\ast}) \hat f_1^{- \frac 14}({r^\ast}) \left\{ \mathrm{Bi}(m^{\frac 23} \xi_1) + \epsilon_{\mathrm{Bi}1}(m,{r^\ast} )\right\} \text{ for } {r^\ast} \in (-\infty, r_2^\ast - \epsilon],\\
	&u_{\mathrm{Ai}2} ({r^\ast},m) := \hat f_2^{\frac 14}({r_2^\ast}) \hat f_2^{- \frac 14}({r^\ast}) \left\{ \mathrm{Ai}(m^{\frac 23} \xi_2) + \epsilon_{\mathrm{Ai}2}(m,{r^\ast} )\right\} \text{ for }{r^\ast} \in [r_1^\ast + \epsilon, \frac{\pi}{2}l], \label{eq:uai2} \\ 
	&	u_{\mathrm{Bi}2} ({r^\ast},m) := \hat   f_2^{\frac 14}({r_2^\ast}) \hat f_2^{- \frac 14}({r^\ast}) \left\{ \mathrm{Bi}(m^{\frac 23} \xi_2) + \epsilon_{\mathrm{Bi}2}(m,{r^\ast} )\right\} \text{ for }{r^\ast} \in [r_1^\ast + \epsilon, \frac{\pi}{2}l]. 
	\end{align}
	Moreover,
 \begin{align}
	&|\epsilon_{\mathrm{Ai}1}|\lesssim M_\Ai(m^{\frac 23} \xi_1) E_\Ai^{-1}(m^{\frac 23} \xi_1) m^{-1}, \label{eq:eps1ai}\\
		&|\epsilon'_{\mathrm{Ai}1}| \lesssim  \hat f_1^{\frac 12} N_\Ai(m^{\frac 23} \xi_1) E_\Ai^{-1}(m^{\frac 23} \xi_1) m^{-\frac 13}, \label{eq:eps1ai'}\\
	&|\epsilon_{\mathrm{Bi}1}|\lesssim_\epsilon
	 M_\Ai(m^{\frac 23} \xi_1) E_\Ai(m^{\frac 23} \xi_1) m^{-1},\label{eq:eps1bi}\\
	 &|\epsilon'_{\mathrm{Bi}1}| \lesssim  \hat f_1^{\frac 12} N_\Ai(m^{\frac 23} \xi_1) E_\Ai(m^{\frac 23} \xi_1) m^{-\frac 13}, \label{eq:eps1bi'}\\
	&|\epsilon_{\mathrm{Ai}2}|\lesssim   M_\Ai(m^{\frac 23} \xi_2) E_\Ai^{-1}(m^{\frac 23} \xi_2) m^{-1},\label{eq:eps2ai}\\
		&|\epsilon'_{\mathrm{Ai}2}| \lesssim  \hat f_2^{\frac 12} N_\Ai(m^{\frac 23} \xi_2) E^{-1}_\Ai(m^{\frac 23} \xi_2) m^{-\frac 13}, \label{eq:eps2ai'}\\
	&|\epsilon_{\mathrm{Bi}2}|\lesssim  M_\Ai(m^{\frac 23} \xi_2) E_\Ai(m^{\frac 23} \xi_2) m^{-1},\label{eq:eps2bi} \\
	&|\epsilon'_{\mathrm{Bi}2}| \lesssim \hat f_2^{\frac 12} N_\Ai(m^{\frac 23} \xi_2) E_\Ai(m^{\frac 23} \xi_2) m^{-\frac 13}, \label{eq:eps2bi'}
	\end{align}
	where \eqref{eq:eps1ai}--\eqref{eq:eps1bi'} hold uniformly in $r^\ast \in (-\infty, {r^\ast_2} - \epsilon]$ and \eqref{eq:eps2ai}--\eqref{eq:eps2bi'} hold uniformly in $r^\ast \in [{r^\ast_1}+\epsilon, \frac{\pi}{2}l]$.
Further, we choose the initialization such that
	\begin{align}
	&	|\epsilon_{\mathrm{Ai}2}(r^\ast)|\lesssim
	M_\Ai(m^{\frac 23} \xi_2) E_\Ai^{-1}(m^{\frac 23} \xi_2) \left( \exp\left[ 2\mathcal{V}_{r_1^\ast + \epsilon, r^\ast}(H_2)  m^{-1}\right]  -1 \right) 	,	\\&
	|\epsilon_{\mathrm{Bi}2}(r^\ast)|\lesssim M_\Ai(m^{\frac 23} \xi_2) E_\Ai(m^{\frac 23} \xi_2) \left( \exp\left[ 2\mathcal{V}_{r^\ast,l \frac \pi 2}(H_2)  m^{-1}\right]  -1 \right) 		
	\end{align}
	and in particular,
	$|\epsilon_{\mathrm{Ai}2}(\frac \pi 2 l)|\lesssim m^{- \frac{7}{6}}, \epsilon_{\mathrm{Bi}2}(\frac \pi 2 l)=0$.
	\begin{proof}
The proof follows from  \cite[Chapter~11, Theorem 3.1]{olver1961error}. For convenience of the reader we will outline the proof for $u_{\mathrm{Ai}2}$. Indeed, with the ansatz of \eqref{eq:uai2} we note that from \eqref{eq:transformedformofode2}, we have $W_2 = \hat f_2^{\frac 14}(r_2^\ast) \{ \Ai (m^{\frac 23} \xi_2) + \epsilon_{\mathrm{Ai}2}(\xi_2)\}$. Thus, from \eqref{eq:transformedformofode2} and variation of parameters (see \cite[Chapter~11, equ.\ (3.12)]{olver2014asymptotics} or \cite{turning}), the error  $\epsilon_{\mathrm{Ai}2}(\xi_2)$ solves
\begin{align}\label{eq:volterraintegralequation}
\epsilon_{\mathrm{Ai}2} (\xi_2) =  - \pi m^{-\frac 23} \int^{\xi_2(r^\ast_1+\epsilon)}_{\xi_2} K(\xi_2, v) \Psi_2(v) \left[ \epsilon_{\mathrm{Ai}2}(v) + \Ai ( m^{\frac 23} v) \right] \d  v,
\end{align}
where $K(\xi_2,v) = \Bi(m^{\frac 23} \xi_2) \Ai(m^{\frac 23} v) - \Bi(m^{\frac 23} v) \Ai(m^{\frac 23} \xi_2)$. We note that as in \cite[Chapter~11, equ.\ (3.12)]{olver2014asymptotics}, the kernel  $K$ satisfies 
\begin{align} \label{eq:estimatesonkernel1}
|K(\xi_2, v) |\leq E_{\Ai}^{-1}( m^{\frac 23}   \xi_2 ) E_{\Ai} (m^{\frac 23} v ) M_{\Ai}(m^{\frac 23} \xi_2) M_{\Ai}(m^{\frac 23  } v )
\end{align}
for $v\geq \xi_2$ and similarly,
\begin{align}\label{eq:estimatesonkernel2}
|\partial_{\xi_2 } K(\xi_2, v) | \leq m^{\frac 23} E_{\Ai}^{-1}(m^{\frac 23} \xi_2) E_{\Ai}(m^{\frac 23} v ) N_{\Ai} (m^{\frac 23} \xi_2) M_{\Ai}(m^{\frac 23} v)
\end{align}
for $v\geq \xi_2$. Now, as in  \cite[Chapter~11, equ.\ (3.12)]{olver2014asymptotics}, the  integral equation \eqref{eq:volterraintegralequation} is solved using \cite[Chapter~6, Theorem~10.2]{olver2014asymptotics} which shows the bounds together with  \cref{prop:properrorairy}. 
	\end{proof}
\end{prop}
\section{The non-Diophantine condition}
\label{sec:pblowup}
With the fundamental solutions from \cref{sec:fundamentalsolutions}, we are now in the position to define the set of black hole parameters $\mathscr P_{\textup{Blow-up}}$. The set $\mathscr P_{\textup{Blow-up}}$   will be defined  in \cref{eq:defnpblowup} as a suitable $\limsup$ set which constitutes a  (generalized) non-Diophantine condition.
\subsection{Definition of the non-Diophantine condition as the set \texorpdfstring{$\mathscr{P}_{\textup{Blow-up}}$  }{Pblowup}} We first define Wronskians of solutions of the radial o.d.e.\ which will play a fundamental role in the estimates. 
\label{sec:defnblow}
\begin{definition}\label{defn:w1w2}
	We define $\mathfrak{W}_1 \colon\mathscr P \times \mathbb Z_m \times \mathbb Z_{\ell \geq |m|} \to \mathbb C$ and $\mathfrak{W}_2 \colon\mathscr P \times \mathbb Z_m \times \mathbb Z_{\ell \geq |m|} \cap \{\sigma_1 \leq \lambda_{m\ell}m^{-2} \leq \sigma_2\} \to \mathbb C$ as\begin{align}
	&	\mathfrak{W}_1 (\mathfrak p,m,\ell):= \mathfrak W [\uhplus,u_\infty](m,\ell,\omega = \omega_- m,\mathfrak  p),\\
		&	\mathfrak{W}_2 (\mathfrak p,m,\ell):= \mathfrak W [u_{\textup{Ai}2},u_\infty](m,\ell,\omega = \omega_- m,\mathfrak p).
	\end{align}
	Note that this is well-defined as the Wronskians $\mathfrak W_1$ and $\mathfrak W_2$    only depend on $\mathscr P$ (by construction). Moreover, by continuous  dependence on parameters of solutions to o.d.e.'s,   the Wronskians $\mathfrak W_1$ and $\mathfrak W_2$   are continuous functions on $\mathscr P$ for fixed $m,\ell$. 
\end{definition}
\begin{rmk}
Note that, as discussed in the introduction, the Wronskian  $\mathfrak W_1$ does not vanish.  Nevertheless, $\mathfrak W_1$ can be very small (as $m,\ell\to \infty$)  which corresponds to  frequency parameters associated to stable trapping. On the other hand,  $\mathfrak W_2$ may vanish and this indeed corresponds to stable trapping. In particular, if $\mathfrak W_2$ vanishes, then the solution $u_\infty$ is a multiple of the $u_{\Ai2}$ which is exponentially damped in the semi-classical forbidden region. In this case, we infer that $\mathfrak W_1$ is exponentially small and indeed, we are in the situation of stable trapping. This would then show that there exists a quasimode with frequency $\omega = \omega_- m$. This intuition leads to the following non-Diophantine condition for the set $\mathscr P_{\textup{Blow-up}}$.
\end{rmk}

\begin{definition}\label{defn:um0}
For $m_0 \in \mathbb{N}$ we define
	\begin{align}
	&	U_{m_0} := \bigcup_{\substack{ m \geq m_0 \\ m \in \mathbb N } } \; \bigcup_{\substack{m \leq \ell \leq m^2\\ \ell \in \mathbb{N} }} U(m,\ell),
	\end{align}
	where 
	\begin{align} \nonumber
		U(m,\ell):= \Big\{  \mathfrak p \in \mathscr{P}\colon & |\mathfrak W_{1}(\mathfrak p,m,\ell)|< e^{-\sqrt{m }}, \sigma_1(\mathfrak p) < \lambda_{m\ell}(a\omega_-(\mathfrak p) m) < \sigma_2(\mathfrak p) ,
		 \\ \nonumber &|\mathfrak{W}_2( \mathfrak p,m,\ell)| < e^{-\ell} e^{-m}, 
		 |\Gamma \mathfrak W_{2} (\mathfrak p,m,\ell)| > 1,\\ &
		|\mathfrak W_2 ( \Phi^\Gamma_\tau (\mathfrak p) ,m,\ell)| > e^{-\ell} e^{-m} \textup{ for all } |\tau| \in [e^{-\ell} e^{-m}, \frac{1}{m^2}] \Big\}.
	\end{align}
 \end{definition}

\begin{definition}\label{eq:defnpblowup}
	We define 
	\begin{align}
	\mathscr{P}_{\textup{Blow-up}}:= \bigcap_{m_0\in  \mathbb N} U_{m_0}.
	\end{align}
\end{definition}
While a priori the set $\mathscr{P}_{\textup{Blow-up}}$ could be empty, we will show in the following that it is  {dense} in $\mathscr{P}$ and  {Baire-generic}, i.e.\ a countable intersection of open and dense sets. 

\subsection{Topological genericity: \texorpdfstring{$\mathscr{P}_{\textup{Blow-up}}$}{Pblowup} is Baire-generic}\label{sec:quasimodes}
We will first show that each $U_{m_0}$ is dense. To do so, we let $m_0$ and $\mathfrak p_0 = (\mathfrak m_0, \mathfrak a_0) \in \mathscr{P}$ be arbitrary and fixed throughout \cref{sec:quasimodes}. Also, let $\mathcal U \subset \mathscr P$  be an open neighborhood of $\mathfrak p_0$.  We will show that there exists an element of $U_{m_0}$ which is contained in $\mathcal U$.  
We now define a curve of parameters through $\mathfrak p_0$ as follows.
	\begin{definition}
	For $\delta=\delta(\mathfrak p_0, \mathcal U)>0$ sufficiently small, we  define the smooth embedded curve $\gamma_{\delta}(\mathfrak p_0) \subset \mathcal U $ through $\mathfrak p_0$ as  \begin{align}\gamma_{\delta}(\mathfrak p_0):= \{ \mathfrak p = (\mathfrak m, \mathfrak a)   \in  \mathscr P \colon \mathfrak a = \mathfrak a_0 , |\vartheta(\mathfrak p) -\vartheta(\mathfrak p_0)|\leq \delta \}.\end{align}
	\end{definition}
Throughout \cref{sec:quasimodes} we will only consider 
\begin{align}
\mathfrak p \in \gamma_{\delta}(\mathfrak p_0).
\end{align}

	We parameterize $\gamma_{\delta}(\mathfrak p_0)$ with $\vartheta  \in (\vartheta_0 - {\delta}, \vartheta_0 + {\delta})$, where $\vartheta_0 = \vartheta(\mathfrak p_0)  $.  
\begin{rmk} Note that the expression $\Xi $ is (by construction) constant on $\gamma_{\delta}(\mathfrak p_0)$.  We also note that  $\gamma_{\delta}(\mathfrak p_0)$ can be seen as a variation of $\mathfrak m$ keeping $\mathfrak a$ fixed. It is however more convention to use the coordinates $(\vartheta, \mathfrak a)$ as they are adapted to the interior scattering pole   in view of $\vartheta = a\omega_-$. 
\end{rmk}
\begin{lemma}\label{lem:gammaoneigv}
	 The angular eigenvalues at the interior scattering poles $\omega = \omega_- m$ satisfy
	\begin{align}
	\Gamma({\lambda_{m\ell}(a\omega_-m)} + a^2\omega_-^2 m^2 - 2 a \omega_- m^2 \Xi ) \in [- 4 m^2, 0].
	\end{align}
\begin{proof}
Note that  \begin{align}	{\lambda_{m\ell}(a\omega_-m)} + a^2\omega_-^2 m^2 - 2 a \omega_- m^2 \Xi \end{align} is an eigenvalue of the operator 
\begin{align}  \nonumber
P_{\omega_-} &	+ a^2\omega_-^2 m^2 - 2 a \omega_- m^2 \Xi 
 \\ & =   - \frac{\d}{\d x} \left(\Delta_x (1-x^2) \frac{\d}{ \d x} \cdot \right)  + \frac{m^2}{\Delta_x}\left( \frac{\Xi}{\sqrt{1-x^2}}-a \omega_- \sqrt{1-x^2} \right)^2 + 2\frac{a^2}{l^2} (1-x^2).
\end{align}
Now, by construction of $\Gamma$ in \eqref{eq:defnofgamma} we have $\Gamma(a\omega_-) = \Gamma(\vartheta )=1$ and $\Gamma(a) =0$ to compute
\begin{align} \nonumber
	\Gamma(P_{\omega_-} + & a^2 \omega_-^2 m^2 - 2 a \omega_- m^2 \Xi ) = -2 \frac{m^2}{\Delta_x} ( \Xi - a \omega_- (1-x^2)) + 2 a \omega_- m^2 - 2 m^2 \Xi \\ \nonumber
&	= - 2m^2 \frac{\Xi}{\Delta_x} \frac{  r_-^2 + a^2 -a^2 + a^2 x^2 - a^2 \Delta_x + (r_-^2 +  a^2) \Delta_x }{r_-^2 + a^2}\\
& =  - 2m^2 \frac{\Xi}{\Delta_x} \left( 1 - \frac{  r_-^2 -a^2 }{r_-^2 + a^2}\Delta_x \right) \in [-4m^2,0]
\end{align}
as $\left|\frac{  r_-^2 -a^2 }{r_-^2 + a^2}\Delta_x \right|\leq 1$ and $0\leq \frac{\Xi}{\Delta_x}\leq 1$. 
\end{proof}
\end{lemma}
\begin{lemma}\label{lem:boundongammapesilon2}
We have $|\Gamma (\epsilon_{\mathrm{Ai}2} (\frac \pi 2 l ) ) | \lesssim m^{- \frac 16}$ where $\epsilon_{\mathrm{Ai}2}$ is defined in \cref{prop:defnsofuaianduabi}. 
\begin{proof}
 First, we consider $\epsilon_{\mathrm{Ai}2}$ as a function of $\xi_2$ defined in \cref{eq:defnofxi1andxi2}.  Then, applying $\Gamma$ to \eqref{eq:volterraintegralequation} shows that  $\Gamma \epsilon_{\mathrm{Ai}2}$ solves 
\begin{align}\nonumber
	\Gamma	\epsilon_{\mathrm{Ai}2} (\xi_2)   = & - \pi m^{-\frac 23} \int^{\xi_2(r^\ast_1+\epsilon)}_{\xi_2} K(\xi_2, v) \Psi_2(v) \Gamma \epsilon_{\mathrm{Ai}2}(v)     \d  v\\
	&-  \pi m^{-\frac 23} \int^{\xi_2(r^\ast_1+\epsilon)}_{\xi_2} K(\xi_2, v) \Gamma (\Psi_2(v) ) \left[  \epsilon_{\mathrm{Ai}2}(v) + \Ai ( m^{\frac 23} v) \right] \d  v \nonumber \\
	& - \Gamma(\xi_2(r_1^\ast+\epsilon)) \pi m^{-\frac 23}  K(\xi_2, \xi_2(r_1^\ast+\epsilon) ) \Psi_2(\xi_2(r_1^\ast+\epsilon) ) \Ai ( m^{\frac 23} \xi_2(r_1^\ast+\epsilon) ).
 \end{align} 
We denote $k (\xi_2) := -   \Gamma(\xi_2(r_1^\ast+\epsilon)) \pi m^{-\frac 23}  K(\xi_2, \xi_2(r_1^\ast+\epsilon) ) \Psi_2(\xi_2(r_1^\ast+\epsilon) ) \Ai ( m^{\frac 23} \xi_2(r_1^\ast+\epsilon) )$ which satisfies 
\begin{align}\label{eq:estimatesonk}
|k (\xi_2)|\lesssim m^{-\frac 23}  M_{\Ai}(m^{\frac 23} \xi_2 ) M^2_{\Ai}(m^{\frac 23} \xi_2(r^\ast_1 + \epsilon) )   E_{\Ai}^{-1}(m^{\frac 23} \xi_2),
\end{align}
where we used that $|\Gamma(\xi_2(r^\ast_1+\epsilon))|\lesssim 1$ and $|\Psi_2(\xi_2(r_1^\ast+\epsilon))|\lesssim 1$. 
Now, we re-write the equation for the unknown $ (\Gamma 	\epsilon_{\mathrm{Ai}2}   -k )(\xi_2) $ as 
\begin{align}\nonumber
&(	\Gamma \epsilon_{\mathrm{Ai}2}  - k) (\xi_2) =  -\pi m^{-\frac 23} \int^{\xi_2(r^\ast_1+\epsilon)}_{\xi_2} K(\xi_2, v)  \Psi_2(v)( \Gamma \epsilon_{\mathrm{Ai}2}   - k)(v)  \d  v\\
& - \pi m^{-\frac 23} \int^{\xi_2(r^\ast_1+\epsilon)}_{\xi_2} K(\xi_2, v)\left[  \Gamma (\Psi_2(v) )   (\epsilon_{\mathrm{Ai}2}(v) + \Ai ( m^{\frac 23} v) ) + \Psi_2(v) k(v)\right] \d  v.
\end{align}

First, from \eqref{lem:smoothnessandsmoothdependceofpsi1psi2} we know that $\Psi_2$ depends smoothly on $\mathfrak p$ and in particular, that  $\xi_2 \mapsto \Gamma \Psi_2(\xi_2)$ is smooth.  Since $|l^2 \Gamma V_{\textup{main}}|\lesssim 1 $ (using \cref{lem:gammaoneigv}) and $|l^2\Gamma V_1|\lesssim 1$, this is made quantitative to obtain $|\Gamma \Psi_2|\lesssim 1$ uniformly in the compact set $\xi_2[r_1^\ast + \epsilon, l \frac \pi 2]$. 

Now, we apply  \cite[Theorem 10.1, Chapter 6]{olver2014asymptotics}  and in the notation of  \cite[Theorem 10.1, Chapter 6]{olver2014asymptotics}  we set
\begin{align}
&\mathrm{K}(\xi_2,v) = \pi m^{-\frac 23} K(\xi_2,v) |v|^{\frac 12},\\	
&P_0(\xi_2) = E^{-1}_{\Ai} (m^{\frac 23}\xi_2)	M_{\Ai}(m^{\frac 23 }\xi_2),\\
&P_1(\xi_2) = m^{\frac 23} E^{-1}_{\Ai} (m^{\frac 23}\xi_2)	N_{\Ai}(m^{\frac 23 }\xi_2),\\
& Q(v) = \pi m^{-\frac 23}  |v|^{\frac 12 } E_{\Ai}(m^{\frac 23} v) M_{\Ai} (m^{\frac 23} v),\\
&\psi_0(v)  = |v|^{-\frac 12} \Psi_2(v),\\
&\phi(v) = |v|^{-\frac 12}, \\
& \psi_1(v) = 0, \\
& J(v)=  \Gamma (\Psi_2(v) )   (\epsilon_{\mathrm{Ai}2}(v) + \Ai ( m^{\frac 23} v) ) + \Psi_2(v) k(v) .
\end{align}

We further have that 
\begin{align}&\Phi(\xi_2) := \int_{\xi_2}^{\xi_2(r^\ast_1 + \epsilon)} |\phi(v)|\d v= \int_{\xi_2}^{\xi_2(r^\ast_1 + \epsilon)} |v|^{-\frac 12} \d v,\\
&	\Psi_0(\xi_2) := \int_{\xi_2}^{\xi_2(r^\ast_1 + \epsilon)} |\psi_0(v)|\d v =  \int_{\xi_2}^{\xi_2(r^\ast_1 + \epsilon)}  |v|^{-\frac 12} |\Psi_2(v) |\d v
	\end{align}
exist and satisfy $\sup_{v \in \xi_2( r_1^\ast+\epsilon, l\frac \pi 2)}\Phi(v)\lesssim 1 $, as well as  $\sup_{v \in \xi_2( r_1^\ast+\epsilon, l\frac \pi 2)} \Psi_0(v) \lesssim 1$ as in \eqref{eq:boundonh2} of  \cref{prop:properrorairy}.  With the above choices, \cref{prop:properrorairy} and \eqref{eq:estimatesonkernel1},  \eqref{eq:estimatesonkernel2}, we have that assumptions (i)--(vi) of \cite[Theorem 10.1, Chapter 6]{olver2014asymptotics}  are satisfied. 

Now, we  compute $\kappa := \sup_{v\in \xi_2[r_1^\ast + \epsilon, l \frac \pi 2]} Q(v) |J(v)  |$, for which we note that the largest term in $J$ is coming from $\Gamma(\Psi_2) \Ai(m^{\frac 23} v )$ in view of the shown properties for $\epsilon_{\Ai 2}$ from \eqref{eq:eps2ai} and $k(v)$ from \eqref{eq:estimatesonk}.  Since $|\Gamma(\Psi_2)(v) \Ai(m^{\frac 23} v ) |\lesssim M_{\Ai}(m^{\frac 23 }v) E^{-1}_{\Ai}(m^{\frac 23} v)$, we obtain $\kappa \lesssim m^{-\frac 23}$. 
Further, for $\kappa_0 := \sup_{v\in \xi_2[r_1^\ast + \epsilon, l \frac \pi 2]} Q(v) |P_0(v)  |$ we directly obtain the estimate $\kappa_0 \lesssim m^{-\frac 23}$. 

Thus, the assumptions of \cite[Theorem 10.1, Chapter 6]{olver2014asymptotics}  are satisfied and we obtain that 
$\Gamma\epsilon_{\Ai2} - k$ satisfies 
\begin{align}
\sup_{\xi_2 \in \xi_2( r_1^\ast+\epsilon, l \frac \pi 2 )} |\Gamma\epsilon_{\Ai2} - k| (\xi_2)\leq \sup_{\xi_2 \in \xi_2( r_1^\ast+\epsilon, l \frac \pi 2 )} P_0(\xi_2) \kappa \Phi(\xi_2) \exp(\kappa_0 \Psi_0(\xi_2)) \lesssim m^{- \frac 23} . 
\end{align}
Thus,
\begin{align}
| (\Gamma\epsilon_{\Ai2} )( \xi_2 ({\pi l }/{2}))| \lesssim m^{-\frac 23} + |k(  {\pi l }/{2})| \lesssim m^{-\frac 23} 
\end{align}
in view of \eqref{eq:estimatesonk}.
Finally, \begin{align}\nonumber 
	\Gamma ( \epsilon_{\Ai2} ( \xi_2 ({\pi l }/{2}) )) & \leq | (\Gamma\epsilon_{\Ai2} )( \xi_2 ({\pi l }/{2}))    | + \left|\frac{ \d \epsilon_{\Ai2}}{\d \xi_2} (\xi_2 ( l  \pi /2) ) \Gamma (\xi_2 ( l  \pi/ 2)) \right| \\
	& \lesssim m^{-\frac 23} +  \left|\frac{ \d \epsilon_{\Ai2}}{\d r^\ast} ( l  \pi /2 ) \left( \frac{\d \xi_2 }{\d r^\ast}\right)^{-1}( l  \pi / 2) \Gamma (\xi_2 ( l  \pi / 2)) \right|  
	\nonumber \\& \lesssim m^{-\frac 23} + m^{-\frac 13 } N_{\Ai} (m^{\frac 23} \xi_2 ( l   \pi /2) ) \lesssim m^{-\frac 23}  +  m^{-\frac 13 } m^{\frac 23 \cdot  \frac 14 }  \xi_2 ( l   \pi /2)^{\frac 14}\lesssim m^{- \frac 16}\end{align} which concludes the proof.
\end{proof}
\end{lemma}
\begin{rmk}\label{rmk:lambdailambdamli}
From \cref{prop:angular} and \cref{lem:gammaoneigv} we have that for almost every $\tilde \lambda_0 > \Xi^2$, there exist sequences $(m_i)_{i\in \mathbb N}$ and $(\ell_i)_{i\in \mathbb N}$ ($m_i \leq \ell_i \leq m_i^2$) with $m_i\to \infty ,\ell_i\to\infty$ as $i\to\infty$ such that the angular eigenvalues satisfy \begin{align}\label{eq:asymptoticremark}
	\lambda_i= \lambda_{m_i \ell_i }(\omega=\omega_-m_i) = \tilde \lambda_i m_i^2 = \tilde \lambda_0 m_i^2 + \lambda_\textup{error}^{(i)},  \end{align}
	where $|\lambda_\textup{error}^{(i)}(\mathfrak p_0)|\leq C(\tilde \lambda_0,\mathfrak p_0) $ and $|\lambda_\textup{error}^{(i)}(\mathfrak p)|\leq  C(\tilde \lambda_0,\mathfrak p_0) (1+ |\vartheta_0 - \vartheta| m_i^2) \lesssim 1+  \delta m_i^2$   uniformly for $\mathfrak p \in \gamma_{\delta}(\mathfrak p_0)$ as $m_i \to \infty$.   Moreover, we   assume without loss of generality that $m_{i+1} > m_i$ and note that the choice of the subsequences $m_i,\ell_i$ depends  on $\mathfrak p_0.$
\end{rmk}
\begin{lemma}\label{lem:vomega-estimates} Let $\lambda_1 := \sup_{\mathfrak p \in \gamma_{\delta}(\mathfrak p_0)} \sigma_1(\mathfrak p)$ and $ \lambda_2:= \inf_{\mathfrak p \in \gamma_{\delta}(\mathfrak p_0)} \sigma_2(\mathfrak p)$ and choose $\delta >0$ potentially smaller such that $\lambda_1 < \lambda_2$. Let $\tilde \lambda_0 \in (\lambda_1,\lambda_2) \setminus \mathcal{N}_{\mathfrak p_0}$ (see \cref{rmk:lambdailambdamli}) be arbitrary. 
Let   $\tilde \lambda_i = \tilde \lambda_0 + \lambda_\textup{error}^{(i)} m_i^{-2}$ be the associated angular eigenvalues from \cref{prop:angular} such that from \eqref{eq:Vomega-} we have the associated potential $V_{\textup{main}}  = \frac{\Delta}{(r^2+a^2)^2} \left( \tilde \lambda_i + \omega_-^2 a^2 - 2  a \omega_- \Xi \right) - ( \omega_- - \omega_r )^2  $. 

Then, for all $\mathfrak p \in  \gamma_{\delta}(\mathfrak p_0)$, 
	and for all $i\in \mathbb N$ sufficiently large, we have 
	\begin{align}
	&c(\delta, \mathfrak p_0)\leq  | \Gamma  \xi_{\infty}^{(i)} | \leq C(\delta, \mathfrak p_0) , \label{eq:dxida}
	\end{align}
  where $\xi_{\infty}^{(i)} : \gamma_{\delta}(\mathfrak p_0) \to \mathbb R$ is   defined as
\begin{align}	 \xi_\infty^{(i)}:=  \int_{r^\ast_2}^{\frac \pi 2l} \sqrt{|V_{\textup{main}}|} \d r^\ast =  \int_{r_2}^\infty   \sqrt{- \frac{(r^2+a^2)^2}{\Delta^2}   V_{\textup{main}}  }  \d r      \label{eq:dxidrminus}
	\end{align}
	and $c(\delta, \mathfrak p_0),C(\delta, \mathfrak p_0)>0$ only depend  on $\delta>0 $ and $ \mathfrak p_0 $.   
		\begin{proof} 
			We use the product rule to compute for $r\geq r_2$
			\begin{align}\nonumber \Gamma \left( \frac{(r^2+a^2)^2}{   \Delta^2}	V_{\textup{main}} \right) = & \Delta^{-1}  \Gamma( \tilde \lambda_i + \omega_-^2 a^2 - 2 a \omega_- \Xi)  +  \Gamma(\Delta^{-1} ) ( \tilde \lambda_i + \omega_-^2 a^2 - 2 a \omega_- \Xi) \\
				& - \Gamma \left(  \frac{(r^2+a^2)^2}{\Delta^2}(\omega_- - \omega_r)^2 \right). \end{align} 
From \cref{lem:gammaoneigv} we know that $ \Delta^{-1}  \Gamma( \tilde \lambda_i + \omega_-^2 a^2 - 2 a \omega_- \Xi) \leq 0$. Moreover, by choice of $\tilde \lambda_0$ in the assumptions of   \cref{lem:vomega-estimates}, we have that $\tilde \lambda_i \in E_{\mathfrak p}$ for all $\mathfrak p \in \gamma_{\delta}(\mathfrak p_0)$ and for all $i$ sufficiently large.
Thus, using the definition of $E_{\mathfrak p}$ in \eqref{eq:defnsp}  we have
\begin{align}\Gamma(\Delta^{-1} ) ( \tilde \lambda_i + \omega_-^2 a^2 - 2 a \omega_- \Xi) - \Gamma \left(  \frac{(r^2+a^2)^2}{\Delta^2}(\omega_- - \omega_r)^2 \right) <- \frac{\omega_- l^2}{a r^4} \end{align}  
for all $r\geq r_2$, all $i$ sufficiently large and all $\mathfrak p \in \gamma_{\delta}(\mathfrak p_0)$. 

Hence,
\begin{align}
 \Gamma \left( \frac{(r^2+a^2)^2}{   \Delta^2}	V_{\textup{main}} \right) \leq - \frac{\omega_-l^4 }{ar^4}
\end{align}
which shows 
\begin{align}|\Gamma \xi_\infty^{(i)}| > c(\delta,\mathfrak p_0)\end{align} for all $i$ sufficiently large and for all parameters in $\gamma_\delta(\mathfrak p_0)$ by choosing $\delta>0$ sufficiently small. For the upper bound we also  use \cref{lem:gammaoneigv}  to obtain that $ |\Gamma  \sqrt{|V_{\textup{main}}}| = \frac{|\Gamma V_{\textup{main}}| }{2| V_\textup{main}^{\frac 12} |} \lesssim \frac{1}{\sqrt{|r^\ast - r_2^\ast|}} $ uniformly in $i$. The last implicit constant can also be chosen to be uniform on $\gamma_{\delta}(\mathfrak p_0)$ such that $ |\Gamma \xi_\infty^{(i)}|\leq   \int_{r^\ast_2}^{\frac \pi 2l} |\Gamma \sqrt{|V_{\textup{main}}|}| \d r^\ast \leq C(\delta,\mathfrak p_0)$ follows. 
\end{proof}
\end{lemma}

  Now, recall the definition of $\mathfrak W_1$ and $\mathfrak W_2$ from \cref{defn:w1w2}.
\begin{prop}\label{prop:density}
Let $m_0 \in \mathbb N$. Then, there exist a parameter
 \begin{align*}\mathfrak p_\textup{Blow-up}  \in \gamma_{\delta}(\mathfrak p_0)\subset \mathcal U\end{align*} and an $i \in \mathbb N$ such that $ m_0 \leq m_i \leq \ell_i \leq m_i^2 $ with
	 \begin{align}
	 \sigma_1(\mathfrak p_\textup{Blow-up})< \lambda_{m_i \ell_i} (\omega=\omega_- m_i)< \sigma_2(\mathfrak p_\textup{Blow-up}), \end{align}  
	\begin{align}
	|	\mathfrak W_1 ( \vartheta_\textup{Blow-up})  | < e^{- \sqrt m_i}
	\end{align}
	as well as 
	\begin{align}\nonumber
	&|	\mathfrak W_2 ( \vartheta_\textup{Blow-up} ) |=0 \text{ and }|\Gamma\mathfrak W_2( \vartheta_\textup{Blow-up} )|  > 1, \\ &|\mathfrak W_2 (  \vartheta)|> e^{- \ell_i } e^{-m_i} \text{ for all } e^{-\ell_i} e^{-m_i} < |\vartheta_{\textup{Blow-up}}-  \vartheta| < \frac{1}{m_i^2}.
	\end{align} 
\end{prop}
The proof of \cref{prop:density} relies on the following two lemmata and will be given thereafter.
First, we will start by showing that for every $m_i \geq m_0$ sufficiently large, there exists a $\mathfrak p_\textup{Blow-up} \in \gamma_{\delta}(\mathfrak p_0)$ such that $\mathfrak W_2 = 0$ and $|\Gamma \mathfrak W_2 | >1$. We will state this as the following lemma. 
\begin{lemma}\label{lem:keylemmainquasimodes}
For   every $  \tilde m_0 >0$ there exists an $i\in \mathbb{N}$ with $m_i > \tilde  m_0$ and a parameter $\vartheta_{\textup{Blow-up}}$ with $|\vartheta_{\textup{Blow-up}} - \vartheta(\mathfrak p_0)|\leq \delta $ such that 
\begin{enumerate}
	\item $\mathfrak{W}_2(\vartheta_{\textup{Blow-up}}, m_i) =  0$,
	\item  $u_{\mathrm{Ai}2}  = \alpha_\infty \hat f_2^{\frac 12}(\frac \pi 2l) u_\infty$ for $\vartheta=\vartheta_{\textup{Blow-up}}$ with $|\alpha_\infty(\vartheta_\textup{Blow-up}) | \sim  m_i^{\frac 56}$,
	\item  $|\Gamma \mathfrak W_2(\vartheta_{\textup{Blow-up}}, m_i) |>1$,
\item For all $\vartheta$ with $  e^{-\ell_i} e^{-m_i} < |\vartheta - \vartheta_{\textup{Blow-up}}|<\frac{1}{m_i^2}$, we have $|\mathfrak W_2(\vartheta)| >  e^{-\ell_i} e^{-m_i}$.
\end{enumerate} 
\begin{proof}
Throughout the proof of \cref{lem:keylemmainquasimodes} we will use the convention that all constants appearing in $\lesssim$, $\gtrsim$, $\sim$ and $O$ only depend on $\mathfrak p_0,l$ and $\delta>0$.  
	
Let   $ \tilde m_0>0$.  
We begin by showing \textit{1}. From \cref{prop:defnsofuaianduabi} and \eqref{eq:dxidrminus} we have 
 \begin{align}\nonumber \mathfrak{W}[u_{\mathrm{Ai}2}, u_\infty](m_i,\vartheta) = u_{\mathrm{Ai}2}({r^\ast} = l \pi/2 , m_i) &=   \hat f_2^{ \frac 14}(r_2^\ast) \hat f_2^{- \frac 14}(\pi l /2) \left\{ \mathrm{Ai}( m_i^{\frac 23} \xi_2 (l \pi/2)) + \epsilon_{\mathrm{Ai}2}(m_i,l \pi/2 )\right\}\\ 
 &	= \hat f_2^{ \frac 14}(r_2^\ast) \hat f_2^{- \frac 14}(\pi l /2)  \left\{ \mathrm{Ai}\left( - \left(\frac{3}{2}m_i \xi_\infty^{(i)}\right)^{\frac 23} \right) +O(m_i^{- \frac 76}) \right\}\label{eq:expansionofair} \end{align}
uniformly on $\gamma_{\delta}(\mathfrak p_0)$.
Now, for all $m_i>\tilde m_0$ sufficiently large, we  use the asymptotics for the Airy  functions as shown in  \cref{lem:asymptoticsforairy} to conclude that 
\begin{align}
	\mathrm{Ai}\left( - \left(\frac{3}{2}m_i \xi_\infty^{(i)}\right)^{\frac 23} \right) +O(m_i^{- \frac 76}) = \frac{1}{\sqrt{\pi}} \left(\frac{3}{2}m_i \xi_\infty^{(i)}\right)^{- \frac 16} \left( \cos\left( m_i \xi_\infty^{(i)} - \frac \pi 4\right)  + O(m_i^{-1}) \right). \label{eq:estimateonairy123}
\end{align}	
Thus, in order to conclude that $\mathfrak{W}[u_{\mathrm{Ai}2}, u_\infty](m_i,\vartheta) =0 $ for some value on $\gamma_{\delta}(\mathfrak p_0)$, we have  to vary $\mathfrak p(\vartheta) \in \gamma_\delta(\mathfrak p_0)$ such that $  m_i \xi_\infty^{(i)}$ goes through a period of $\pi$. Here, we also use that all terms in $O(m_i^{-1})$ depend continuously on $\mathfrak p$. Thus, it suffices to let $\xi_\infty^{(i)}$ go through a period of $\pi m_i^{-1}$. From \eqref{eq:dxida} we have
 \begin{align}\label{eq:estiamteongammaxii}
| \Gamma \xi_\infty^{(i)} | \sim 1 
 \end{align}
uniformly on $\gamma_{\delta}(\mathfrak p_0)$. Thus, by potentially choosing $m_i>\tilde m_0$ even larger, there exists a parameter $\vartheta_{\textup{Blow-up}}$ with 
  \begin{align}
  |\vartheta_{\textup{Blow-up}}-\vartheta(\mathfrak p_0)|\lesssim \frac{1}{m_i} 
  \end{align} 
  such that 
  \begin{align}
  \mathfrak{W}_2({\vartheta_\textup{Blow-up}},m_i) =0  \text{ and } \mathfrak p(\vartheta_{\textup{Blow-up}}) \in \gamma_\delta(\mathfrak p_0). 
  \end{align}
We finally note that from \eqref{eq:expansionofair} and \eqref{eq:estimateonairy123} we have that $\mathrm{Ai}\left( - \left(\frac{3}{2}m_i \xi_\infty^{(i)}\right)^{\frac 23} \right) + \epsilon_{\mathrm{Ai}2}(m_i,l \pi/2 )  = O(m_i^{-\frac 76})$ for all $\vartheta$ with $|\vartheta-\vartheta_{\textup{Blow-up}}|\leq m_{i}^{-2}$ in view of \eqref{eq:estiamteongammaxii}.

Having found $m_i$ and $\vartheta_{\textup{Blow-up}}$, we will now prove \textit{2.} 
For $\vartheta=\vartheta_\textup{Blow-up}$  we have from \cref{prop:defnsofuaianduabi} and  $\textit{1.}$ that 
\begin{align} \nonumber
  |u_{\mathrm{Ai}2}^\prime ( \pi l/2) | &= \hat f_2^{ \frac 14}(r_2^\ast) \hat f_2^{- \frac 14}( \pi l/2)   \left\{ \mathrm{Ai}'( m_i^{\frac 23} \xi_2 (l \pi/2)) m_i^{\frac 23} \partial_{r^\ast} \xi_2 (l \pi/2) + \epsilon_{\mathrm{Ai}2}'(m_i,l \pi/2 )\right\}\\  \nonumber
  &  = \hat f_2^{ \frac 14}(r_2^\ast) \hat f_2^{ \frac 14}(\pi l/2) \left\{ \frac{1}{\sqrt \pi}  \sin\left(m_i \xi_\infty^{(i)} - \frac \pi 4 \right) m_i^{\frac 56} + O(m_i^{-\frac 16})\right\}  \\ & \sim \hat f_2^{\frac 12} ( l \pi/2 ) m_i^{\frac 56}.
\end{align}
Here, to estimate $\epsilon_{\mathrm{Ai}2}'(m_i,l \pi/2 )$ we used \eqref{eq:eps2ai'}  and \eqref{eq:errorestimateonnai}. We further used  \cref{lem:asymptoticsforairy}, as well as \eqref{eq:defnofxi2withf2} together with $\hat f_2(r_2^\ast) \sim \hat f_2 (l \pi/2)$. Finally, we used that $|\sin\left(m_i \xi_\infty^{(i)} - \frac \pi 4 \right)| \sim  1 $ as $|\cos\left(m_i \xi_\infty^{(i)} - \frac \pi 4 \right)|= O(m_i^{-1}) $.  
   Thus, for $\vartheta =\vartheta_\textup{Blow-up}$ we have \begin{align}u_{\Ai2} = \hat f_2^{\frac 12} (l \pi/ 2)\alpha_\infty u_\infty \text{ with } |\alpha_\infty |\sim m_i^{\frac 56}.\end{align}
   
For \textit{3}, we will in fact show the stronger statement that $|\Gamma \mathfrak W_2(\vartheta , m_i) |>1$  for all $\vartheta$ with $|\vartheta-\vartheta_{\textup{Blow-up}}|\leq m_i^{-2} $. We first recall that
	\begin{align}\label{eq:lowerboundongammaxiinfty}
	\left|\Gamma  \xi_\infty^{(i)}  \right| \sim 1
	\end{align}
	on $\gamma_{\delta}(\mathfrak p_0)$
	in view of \eqref{eq:dxida}.
	 Now, we take the derivative of \eqref{eq:expansionofair} with respect to $\Gamma$. 
First we consider the term   when $\Gamma$ hits $ \hat f_2^{ \frac 14}(r_2^\ast) \hat f_2^{- \frac 14}(\pi l /2) $. We recall from the proof of \textit{1} that for $\vartheta$ with $|\vartheta-\vartheta_{\textup{Blow-up}}|\leq m_i^{-2}$ we have $ \mathrm{Ai}\left( - \left(\frac{3}{2}m_i \xi_\infty^{(i)}\right)^{\frac 23} \right)   + \epsilon_{\mathrm{Ai}2}(m_i,l \pi/2 )= O(m_i^{-\frac 76})$ such that 
\begin{align*}
 \Gamma \left( \hat f_2^{ \frac 14}(r_2^\ast) \hat f_2^{- \frac 14}(\pi l /2) \right) \left\{ \mathrm{Ai}( m_i^{\frac 23} \xi_2 (l \pi/2)) + \epsilon_{\mathrm{Ai}2}(m_i,l \pi/2 )\right\} = O(m_i^{-\frac 76}).
\end{align*} 
Now, we consider  the    term when $\Gamma$ hits  $  \epsilon_{\mathrm{Ai}2}(m_i,l \pi/2 )$. Using \cref{lem:boundongammapesilon2}   we have that \begin{align*} \hat f_2^{ \frac 14}(r_2^\ast) \hat f_2^{- \frac 14}(\pi l /2) \Gamma   \epsilon_{\mathrm{Ai}2}(m_i,l \pi/2 ) = O(m_i^{-\frac 16}). \end{align*} 
Finally, we consider the term when $\Gamma$ hits $\mathrm{Ai}\left( - \left(\frac{3}{2}m_i \xi_\infty^{(i)}\right)^{\frac 23} \right)$. By the  chain rule, we directly compute \begin{align}\left|\Gamma \mathrm{Ai}\left( - \left(\frac{3}{2}m_i \xi_\infty^{(i)}\right)^{\frac 23} \right) \right|  \sim \left| \mathrm{Ai}'\left( - \left(\frac{3}{2}m_i \xi_\infty^{(i)}\right)^{\frac 23} \right)  m_i^{\frac 23}\right|\end{align}
	 in view of \eqref{eq:lowerboundongammaxiinfty}. Similar to  the proof of \textit{2.},   from $\mathrm{Ai}\left( - \left(\frac{3}{2}m_i \xi_\infty^{(i)}\right)^{\frac 23} \right) = O(m_i^{-\frac 76}) $, we have  that  $\left|\mathrm{Ai}'\left( - \left(\frac{3}{2}m_i \xi_\infty^{(i)}\right)^{\frac 23} \right) \right|\sim m_i^{\frac 14 \frac 23} \sim m_i^{\frac 16}$.   Putting everything together, this  shows that 
	\begin{align}\label{eq:455}
	\left|	\Gamma \mathfrak W_2 \right| \sim  m_i^{\frac 16}  m_i^{\frac 23}+ O(m_i^{-\frac 16}) + O(m_i^{-\frac 76})  \sim m_i^{\frac 56}
	\end{align}
for all  $\vartheta$ with $|\vartheta-\vartheta_{\textup{Blow-up}}| < m_i^{-2}$. In particular, this shows \textit{3}.  Upon integration this also shows~\textit{4}.
\end{proof}
\end{lemma}
\begin{lemma}\label{lem:rad2} 	There exists a constant $c>0$ (only depending on $\mathfrak p_0$ and $\delta>0$) such that for $\mathfrak p \in \gamma_\delta(\mathfrak p_0)$ we have 
\begin{align}
	\label{eq:wronskianestimate}
		|\mathfrak{W}(u_{\mathrm{Ai}2}, u_{\mathrm{Bi}1}) | \lesssim    \hat f_1^{\frac 12}(r_1^\ast) e^{- c m_i } \text{ and } |\mathfrak{W}(u_{\mathrm{Ai}2}, u_{\mathrm{Ai}1}) | \lesssim \hat  f_1^{\frac 12}(r_1^\ast) e^{- c m_i }
\end{align}
	for all $m_i$ sufficiently large. Moreover, there exist  constants $\alpha_1 = \alpha_1(m_i) \in \mathbb R$ and $\beta_1= \beta_1(m_i) \in \mathbb R$ satisfying $ |\alpha_1| \lesssim e^{- c m_i }  $ and $|\beta_1|\lesssim e^{- c m_i } $ such that $u_{\mathrm{Ai}2} = \alpha_1 u_{\mathrm{Ai}1} + \beta_1 u_{\mathrm{Bi}1}$.
\begin{proof}
	We start by proving	$|\mathfrak{W}(u_{\mathrm{Ai}2}, u_{\mathrm{Bi}1}) | \lesssim  \hat f_1^{\frac 12}(r_1^\ast)e^{- c m_i }$. We evaluate the Wronskian $\mathfrak{W}(u_{\mathrm{Ai}2}, u_{\mathrm{Bi}1})$ at ${r^\ast} := {r^\ast_1} + 2 \epsilon (\mathfrak p_0)$, where $\epsilon(\mathfrak p_0 )$ is as in \eqref{eq:defnxi1}. By potentially choosing $\delta>0$ smaller, we have that $ {r^\ast_1} + 2 \epsilon (\mathfrak p_0) \geq  {r^\ast_1} +  \epsilon (\mathfrak p)$ for all $\mathfrak p \in \gamma_\delta(\mathfrak p_0)$. Then, using standard bounds on Airy functions from \cref{lem:asymptoticsforairy}  we obtain
\begin{align}
	&|u_{\mathrm{Ai}2}( {r^\ast_1} + 2 \epsilon (\mathfrak p_0) )| \lesssim  \frac{1}{m_i^{\frac 16} \xi_2^{\frac 14}(r_1^\ast +  2 \epsilon (\mathfrak p_0))} e^{- \frac 23 m_i \xi_2^{\frac 32}({r^\ast_1} +2 \epsilon (\mathfrak p_0) )},\\
	& |u_{\mathrm{Ai}2}^\prime( {r^\ast_1} + 2 \epsilon (\mathfrak p_0) )| \lesssim  m_i^{\frac 16} \xi_2^{\frac 14}(r_1^\ast + 2 \epsilon (\mathfrak p_0)) \hat f_2^{\frac 12}(r_2^\ast) e^{- \frac{2}{3} m_i \xi_2^{\frac 32}({r^\ast_1} +2 \epsilon (\mathfrak p_0) )},\\
		&|u_{\mathrm{Bi}1}( {r^\ast_1} + 2 \epsilon (\mathfrak p_0) )| \lesssim  \frac{1}{m_i^{\frac 16} \xi_1^{\frac 14}(r_1^\ast + 2 \epsilon (\mathfrak p_0))} e^{ \frac 23 m_i \xi_1^{\frac 32}({r^\ast_1} +2 \epsilon (\mathfrak p_0) )}    ,\\
		& |u_{\mathrm{Bi}1}^\prime( {r^\ast_1} + 2 \epsilon (\mathfrak p_0) )| \lesssim  m_i^{\frac 16} \xi_1^{\frac 14}(r_1^\ast + 2 \epsilon (\mathfrak p_0)) \hat f_1^{\frac 12}(r_1^\ast) e^{ \frac{2}{3} m_i \xi_1^{\frac 32}({r^\ast_1} +2 \epsilon (\mathfrak p_0) )}.
	\end{align}
	Now,  by choosing $\delta>0$ potentially smaller,  in view of \eqref{eq:smallnessofepsilon}, we have $ \frac{\xi_2 (r_1^\ast+2\epsilon(\mathfrak p_0 ))}{\xi_1(r_1^\ast+2 \epsilon(\mathfrak p_0))} \geq 2 $ for all $\mathfrak p \in \gamma_\delta(\mathfrak p_0)$. Thus, there exists  a constant $c=c(\mathfrak p_0, \delta)>0$ such that 
	\begin{align} \xi_2^{\frac 32}({r^\ast_1} + 2 \epsilon (\mathfrak p_0)) -  \xi_1^{\frac 32}({r^\ast_1} + 2 \epsilon (\mathfrak p_0)) \geq c
	\end{align}
for all $\mathfrak p \in \gamma_\delta(\mathfrak p_0)$. Now, the first estimate follows by evaluating the Wronskian $\mathfrak{W}(u_{\mathrm{Ai}2}, u_{\mathrm{Bi}1})$ at ${r^\ast} = {r^\ast_1} +2 \epsilon (\mathfrak p_0) $ and the fact that ${\hat f_1(r_1^\ast)}/{\hat f_2(r_2^\ast)} \sim 1$. The second estimate of \eqref{eq:wronskianestimate} follows in the same manner but it is easier as $u_{\mathrm{Ai}2}$ is already exponentially small in the region between the turning points ${r^\ast_1}$ and ${r^\ast_2}$ since
	\begin{align}
&|u_{\mathrm{Ai}1}( {r^\ast_1} +2 \epsilon (\mathfrak p_0))| \lesssim  \frac{1}{m_i^{\frac 16} \xi_1^{\frac 14}(r_1^\ast + 2 \epsilon (\mathfrak p_0))} e^{- \frac 23 m_i \xi_1^{\frac 32}({r^\ast_1} +2 \epsilon (\mathfrak p_0))},  \\ & |u_{\mathrm{Ai}1}^\prime( {r^\ast_1} + 2 \epsilon (\mathfrak p_0))| \lesssim m_i^{\frac 16} \xi_1^{\frac 14}(r_1^\ast + 2 \epsilon (\mathfrak p_0)) \hat f_1^{\frac 12}( r_1^\ast) e^{- \frac{2}{3} m_i \xi_1^{\frac 32}({r^\ast_1} +2 \epsilon (\mathfrak p_0))}.
	\end{align}
	
	For the second part of the lemma we first note that \begin{align}\alpha_1 = \frac{\mathfrak W (u_{\mathrm{Ai}2}, u_{\mathrm{Bi}1})}{\mathfrak W (u_{\mathrm{Ai}1}, u_{\mathrm{Bi}1})} , \beta_1 = \frac{\mathfrak W (u_{\mathrm{Ai}2}, u_{\mathrm{Ai}1})}{\mathfrak W (u_{\mathrm{Bi}1}, u_{\mathrm{Ai}1})} .\end{align} To conclude it suffices to show that \begin{align}\mathfrak W (u_{\mathrm{Ai}1}, u_{\mathrm{Bi}1}) \sim \hat f_1^{\frac 12} (r_1^\ast)  m_i^{\frac 23}.\end{align} In view of the error bounds from \eqref{eq:eps1ai}--\eqref{eq:eps1bi'} and the chain rule, we conclude that 
	\begin{align}
|\mathfrak W_{r^\ast}(u_{\Ai1}, u_{\Bi1})| \sim \hat f_1^{\frac 12} (r_1^\ast)  m_i^{\frac 23 }\mathfrak{W}_{x}(\Ai(x), \Bi(x)) \sim \hat f_1^{\frac 12} (r_1^\ast) m_i^{\frac 23}
	\end{align}
	for all $m_i$ sufficiently large.
	\end{proof}
\end{lemma}
Now, we are in the position to prove \cref{prop:density}.
\begin{proof}[Proof of \cref{prop:density}]
Let $m_0 \in \mathbb N$ be arbitrary. Using  \cref{lem:keylemmainquasimodes}, we let  $m_i >m_0$ and fix $\mathfrak p_\textup{Blow-up}  \in \gamma_\delta(p_0)\subset \mathcal U$  such that $\mathfrak W_2 =0$ and $| \Gamma \mathfrak W_2 |>1$ as well as $|\mathfrak W_2(\vartheta)| > e^{-\ell_i} e^{-m_i}$ for  $e^{-\ell_i} e^{-m_i} < |\vartheta - \vartheta_{\textup{Blow-up}}|<\frac{1}{m_i^2}$. We moreover have \begin{align}u_\infty = \alpha_\infty^{-1} \hat f_2^{- \frac 12}(l \pi/ 2)   u_{\mathrm{Ai}2} = \alpha_\infty^{-1}  \hat f_2^{- \frac 12}(l \pi/ 2)  \left( \alpha_1 u_{\mathrm{Ai}1} + \beta_1 u_{\mathrm{Bi}1} \right),\end{align} where $|\alpha_\infty| \sim m_i^{\frac 56}$.
	Thus, in view of \cref{lem:rad2} we have
\begin{align} 
\nonumber
|	\mathfrak{W}[u_\infty,\uhplus] | & =|\alpha_\infty^{-1} \hat f_2^{- \frac 12}(l \pi/ 2) \left( 	\alpha_1 \mathfrak{W}[u_{\Ai1},\uhplus]   + \beta_1 \mathfrak{W}[u_{\Bi1},\uhplus] \right) |  \\ &   \lesssim \hat f_2^{- \frac 12}(l \pi/ 2) m_i^{-\frac 56} e^{-c m_i} \left(| \mathfrak{W}[u_{\Ai1},\uhplus]| + |\mathfrak{W}[u_{\Bi1},\uhplus]| \right)
\end{align}
for some constant $c= c(\mathfrak p_0)>0$.
To estimate $\mathfrak{W}[u_{\Ai1},\uhplus] $ and $\mathfrak{W}[u_{\Bi1},\uhplus] $ we infer from \cref{lem:asymptoticsforairy} and \eqref{eq:defai1}, \eqref{eq:defbi1} together with the associated error bounds,  that
\begin{align}
&|	u_{\Ai1} | \lesssim   m_i^{- \frac  16}, \;\;\; |	u'_{\Ai1} | \lesssim \hat f_1^{\frac 12}(r_1^\ast)  m_i^{\frac 56}, \;\; |	u_{\Bi1} | \lesssim m_i^{- \frac  16}, \;\;\; |	u'_{\Bi1} | \lesssim \hat f_1^{\frac 12}(r_1^\ast)m_i^{\frac 56}
\end{align} 
for all $r^\ast $ sufficiently small and particularly as $r^\ast \to -\infty$. Moreover, as $r^\ast \to -\infty$, we have that \begin{align}\uhplus = e^{-i(\omega_- - \omega_+) m r^\ast}\end{align} such that 
\begin{align}
|	\mathfrak W[u_{\Ai1}, \uhplus] | , | 	\mathfrak W[u_{\Bi1}, \uhplus ] |\lesssim \hat f_1^{\frac 12 }(r_1^\ast) m_i^{\frac 56}.
\end{align}
Thus, by potentially choosing $m_i$ even larger (i.e.\ choose $ \tilde m_0$ larger in \cref{lem:keylemmainquasimodes}) and noting that $\hat f_2^{ \frac 12}(\frac \pi 2 l) \sim (\omega_- - \omega_+) \sim \hat f_1^{\frac 12} (r_1^\ast)$, we have
\begin{align}
	|	\mathfrak{W}[u_\infty,\uhplus] | \lesssim m_i^{-\frac 56} m_i^{\frac 56} e^{-c m_i}  =   e^{-c m_i}
\end{align}
and thus, 
\begin{align}
	|	\mathfrak{W}[u_\infty,\uhplus] |< e^{- \sqrt{m_i}}
\end{align}
for all $m_i$ sufficiently large.
\end{proof}

Now, we can conclude 
\begin{prop}\label{thm:Baire}
	The set $\mathscr{P}_{\textup{Blow-up}}$ is a Baire-generic subset of $\mathscr{P}$.
	\begin{proof}
		Since $\mathfrak p_0 \in \mathscr{P}$ and $\mathcal U \subset \mathscr P$, $\mathcal U \ni\mathfrak p_0$ were arbitrary, \cref{prop:density} shows that for any $m_0 \in \mathbb N$ sufficiently large, the set $U_{m_0}$ as defined in \cref{defn:um0} is dense in $\mathscr{P}$.  Since $  \mathfrak W_1$, $\mathfrak W_2$, $\sigma_1$ and $\sigma_2$ are continuous, $U_{m_0}$ is manifestly open. Thus, in view of Baire's theorem \cite{zbMATH02668286}, $\mathscr{P}_{\textup{Blow-up}}$ is Baire-generic and in particular dense.		
	\end{proof}
\end{prop}
\subsection{Metric genericity: \texorpdfstring{$\mathscr{P}_{\textup{Blow-up}}$}{Pblowup} is Lebesgue-exceptional and  2-packing dimensional}
\label{sec:lebesgue}
	\begin{prop}\label{prop:lebesgue}
The set  $\mathscr{P}_{\textup{Blow-up}}$ is a Lebesgue null set.
		\begin{proof}
			It suffices to show that $\mathscr P_{\textup{Blow-up}} \cap C$ has vanishing Lebesgue measure (denoted by $|\cdot |$) for any closed square $C$ contained in $\mathscr P$ with side length less than unity. Let $C$ be such a square contained in $\mathscr P$. Throughout the proof, all constants appearing in $\lesssim, \gtrsim, \sim$ and $O$ will only depend on the square $C$.
	We start by estimating $U(m,\ell)\cap C$ with the co-area formula: We have
			\begin{align}
				| U(m,\ell)\cap C | = \int_{\tilde{ \mathfrak  a }\in (\mathfrak a_1, \mathfrak a_2)}    H^1 ( U(m,\ell)\cap C  \cap \{\mathfrak a= \tilde{\mathfrak a} \}) \d{\tilde {\mathfrak  a}}.
			\end{align}
			 We recall that $H^1$ denotes the one dimensional Hausdorff measure. As $|\mathfrak a_2 - \mathfrak a_1|\leq 1$, it suffices to estimate $H^1 ( U(m,\ell)\cap C  \cap \{\mathfrak a= \tilde{\mathfrak a} \}) $ uniformly for $\tilde {\mathfrak a} \in (\mathfrak a_1, \mathfrak a_2)$.
			
			 For each $\tilde {\mathfrak a} \in (\mathfrak a_1, \mathfrak a_2)$  we claim that $U(m,\ell) \cap C \cap \{\mathfrak a= \tilde{\mathfrak a} \}$ can be decomposed into at most $O(m^2)$ many subsets, each of which with diameter at most $ O(e^{-\ell} e^{-m})$. More precisely, for $\vartheta_1< \vartheta_2$ let  $(\vartheta_1, \tilde{ \mathfrak a})$ and  $(\vartheta_2, \tilde {\mathfrak a})$ be elements of $U(m,\ell) \cap C \cap \{\mathfrak a= \tilde{\mathfrak a} \}$ in coordinates $(\vartheta, \mathfrak a)$. Then, we claim that either, $|\vartheta_2 - \vartheta_1|\leq 2 e^{-\ell} e^{-m}$ or $|\vartheta_2 - \vartheta_1| > \frac{1}{m^2}$. 
			 
			 Indeed, note that $(\vartheta_2, \tilde {\mathfrak a}) = \Phi^\Gamma_{|\vartheta_2-\vartheta_1|} ((\vartheta_1, \tilde {\mathfrak a}))$. Thus, from the definition of $U(m,\ell)$ and since both, $(\vartheta_1, \tilde {\mathfrak a}), (\vartheta_2, \tilde {\mathfrak a}) \in U(m,\ell)$, we conclude that
			 \begin{align} |\vartheta_2 - \vartheta_1|< 2 e^{-\ell} e^{-m} \text{ or } \frac{1}{m^2}< |\vartheta_2 - \vartheta_1|.\end{align}
			Hence, we decompose $U(m,\ell) \cap C \cap \{\mathfrak a= \tilde{\mathfrak a} \}$ into $O(m^2)$ many subsets, each of which has diameter at most  $O(e^{-m} e^{-\ell})$ which is uniform in $\tilde{\mathfrak a}$.  Thus, \begin{align} {H}^1(U(m,\ell) \cap C \cap \{\mathfrak a= \tilde{\mathfrak a} \}) \lesssim m^2 e^{-\ell} e^{-m}\end{align} which implies \begin{align}| U(m,\ell)\cap C | \lesssim  m^2 e^{-\ell} e^{-m}.\end{align} Now,
			 \begin{align}
			 U_{m,C} := \bigcup_{m \leq \ell \leq m^2} U(m,\ell) \cap C
			\end{align}
			satisfies \begin{align}\label{eq:estimateonumc}|U_{m,C}|\lesssim  e^{-m}.\end{align}
			Using the definition of $\mathscr P_{\textup{Blow-up}}$ from \cref{eq:defnpblowup} we compute 
			\begin{align}\nonumber 
\mathscr P_{\textup{Blow-up}} \cap C & = \big( \bigcap_{m_0 \in \mathbb N} U_{m_0}\big)  \cap C = \bigcap_{m_0 \in \mathbb N} (U_{m_0}\cap C) = \bigcap_{m_0 \in \mathbb N} \big( \big(\bigcup_{m\geq m_0} \bigcup_{m\leq \ell \leq m^2} U(m,\ell) \big) \cap C\big) \\
& = \bigcap_{m_0 \in \mathbb N}  \big(\bigcup_{m\geq m_0} \bigcup_{m\leq \ell \leq m^2} (U(m,\ell) \cap C ) \big)  = \bigcap_{m_0 \in \mathbb N}  \bigcup_{m\geq m_0} U_{m,C} = \limsup_{m\to\infty} U_{m,C}. 
			\end{align}
		With \eqref{eq:estimateonumc} we  conclude
			\begin{align}
				|	\mathscr P_{\textup{Blow-up}}\cap C| = 0
			\end{align}
		in view of the Borel--Cantelli lemma.		
	\end{proof}
	\end{prop} 
			\begin{prop}\label{prop:packingdimension}
				The set $\mathscr P_{\textup{Blow-up}}$ has full packing dimension, i.e. $\dim_P (\mathscr P_{\textup{Blow-up}}) = 2$.
				\begin{proof}
					This follows from \cref{thm:Baire} and \cite[Corollary 3.10]{MR3236784}.
				\end{proof}
			\end{prop}
\section{Construction of the initial data}
\label{sec:initialdata}
Having constructed the set $\mathscr{P}_{\textup{Blow-up}}$, we will turn to the problem of showing blow-up.  We begin by fixing an arbitrary parameter \begin{align}\label{eq:fixedp}
\mathfrak p  = (\mathfrak m, \mathfrak a) \in \mathscr{P}_{\textup{Blow-up}}.
\end{align}
which we keep fixed through the rest of the paper, i.e.\ throughout  \cref{sec:initialdata}, \cref{sec:exterioranalysis}, \cref{sec:interior} and \cref{sec:mainthmkerr}. This also fixes the mass $M = \mathfrak m l/\sqrt{3}$ and angular momentum $a = \mathfrak a l/\sqrt{3}$.
As stated in the conventions in \cref{sec:conventions}, all constants appearing in $\lesssim$, $\gtrsim$, $\sim$ and $O$ will now depend on $\mathfrak p$ as fixed in \eqref{eq:fixedp} (and on  $l>0$ as fixed in \eqref{eq:defnofl}) throughout  \cref{sec:initialdata}, \cref{sec:exterioranalysis}, \cref{sec:interior} and \cref{sec:mainthmkerr}.

By construction of $\mathscr{P}_{\textup{Blow-up}}$ and since $\mathfrak p \in \mathscr{P}_{\textup{Blow-up}}$, there exists an infinite sequence
 \begin{align}\label{eq:mili} m_i\to\infty, \ell_i\to\infty \end{align} 
 with \begin{align}
 \label{eq:smallnessofwronskian}
 &|\mathfrak W [\uhplus,u_\infty](\omega = \omega_- m_i, m_i,\ell_i)| < e^{-\sqrt m_i}
 \end{align}
 and
 \begin{align}  \frac{\lambda_{m_i \ell_i}(a \omega_- m_i)}{m_i^2} \in ( \sigma_1(\mathfrak p), \sigma_2(\mathfrak p) ).\end{align}  
 Without loss of generality we also assume that all $m_i$ are taken sufficiently large, i.e.\ \begin{align}\label{eq:defnofm0}
 m_i \geq m_0 \end{align} for all $i \in \mathbb N$ and for a sufficiently large $m_0 = m_0(\mathfrak p) \in \mathbb N$ only depending on the choice of $\mathfrak p$. 
 
We will now carefully choose initial data for \eqref{eq:wavekerr} with compact support in  $K$, which we define in the following.
\begin{lemma}\label{lem:K}
There exists a compact interval $K  \subset (-\infty, \frac{\pi}{2}l)$, an $\epsilon >0$ and a constant $c >0$ such that for every $i\in \mathbb N$, there exists a subinterval $K_i = [r_i^\ast-\frac{c}{m_i}, r_i^\ast + \frac{c}{m_i}] \subset K$  with   \begin{align}
u_\infty^{\omega_-}  := u_\infty(\omega=\omega_- m_i, m_i,\ell_i, r^\ast) \geq \frac{\epsilon}{m_i}
\end{align}
for all $r^\ast \in K_i$.
Moreover,  we choose $K $ such that $\inf K > 3 r_+$. 
\begin{proof}
By \cref{def:uinfty}, $u_\infty^{\omega_-} = u_\infty(\omega=\omega_- m_i, m_i,\ell_i, r^\ast)$ is a solution to \eqref{eq:radialomega}, i.e.\ a solution to
\begin{align}
	-u'' + (m_i^2 V_{\textup{main}} + V_{1})u =0,
\end{align}
where 
\begin{align}
	&V_{\textup{main}} =  \frac{\Delta}{(r^2+a^2)^2} \left( \frac{\lambda_{m\ell}(a \omega_- m_i)}{m_i^2} + \omega_-^2 a^2 - 2  a \omega_- \Xi \right) -  ( \omega_- - \omega_r )^2,\\
	&V_{1} =   \frac{-\Delta^2 3r^2}{(r^2+a^2)^4} - \Delta \frac{5 \frac{r^4}{l^2} + 3 r^2\left(1 + \frac{a^2}{l^2} \right) - 4 M r + a^2}{(r^2+a^2)^3} - \frac{2 \Delta}{l^2} \frac{1}{r^2+a^2}
\end{align}
	with $u_\infty^{\omega_-} (l\pi/2) =0$ and ${u_\infty^{\omega_-} }^\prime(l\pi /2) =1 $. Since \begin{align}\frac{\lambda_{m_i \ell_i}(a \omega_- m_i)}{m_i^2} \in (\sigma_1(\mathfrak p) , \sigma_2(\mathfrak p)),\end{align}  there exists a $\delta >0$ and a $\tilde r_2^\ast$ such that  \begin{align}V_{\textup{main}} < - \delta \text{ for } r^\ast \in \left[ \tilde r_2^\ast, l\frac{\pi}{2} \right),\end{align} see \cref{prop:trappingexist}. Without loss of generality we can assume that $\tilde r_2^\ast > r^\ast (r=3r_+)$. In particular, for $m_i$ sufficiently large (take $m_0(\mathfrak p)$ possibly larger in \eqref{eq:defnofm0}), we have that \begin{align}  V_{\textup{main}} + V_1m_i^{-2} <- \delta 
	\end{align}
	 for $r^\ast \in [\tilde r_2^\ast, \frac{\pi}{2}l)$. Now, let $K  := [\tilde r_2^\ast, r_3^\ast] \subset (\tilde r_2^\ast, \frac{\pi}{2}l)$ be a compact subinterval for $r_3^\ast > r_2^\ast$ fixed, e.g.\ $r_3^\ast = \frac{1}{2} (r_2^\ast + l \frac{\pi}{2})$. In the region $[\tilde r_2^\ast , \frac{\pi}{2}l)$, the smooth potential $V_{\textup{main}}$ satisfies \begin{align}V_{\textup{main}} < - \delta,  |V'_{\textup{main}} |\lesssim 1 \text{ and }|V''_{\textup{main}}|\lesssim 1. 
	\end{align}
	Moreover, $|V_1|\lesssim 1$ uniformly in $[\tilde r_2^\ast , \frac{\pi}{2}l)$. This allows us to approximate $u_\infty^{\omega_-} $ via a WKB approximation. First, we introduce the error-control function \begin{align}F_\infty(r^\ast) := \int_{r^\ast}^{\frac \pi 2l} {|V_{\textup{main}}|}^{-\frac 14} \frac{\d^2 }{\d y^2}\left( {|V_{\textup{main}}|}^{-\frac 14}  \right) - \frac{V_{1}}{{|V_{\textup{main}}|}^{\frac 12}} \d y\end{align} 
	and note that $F_\infty(\frac{\pi}{2}l) = 0$. In view of the above bounds on $V_{\textup{main}}$ and $V_1$ we obtain 
	\begin{align}
		\mathcal{V}_{\tilde r_2^\ast, \frac{\pi}{2}l }(F_\infty) \lesssim 1.
	\end{align}
Hence,  using the boundary conditions $u_\infty^{\omega_-}  (\frac \pi 2 l) =0$, $\frac{\d }{\d r^\ast} u_\infty^{\omega_-} (\frac \pi 2 l) =1$ we obtain from \cite[Chapter~6, Theorem~2.2]{olver} that the solution $u_\infty^{\omega_-} $ is given as 
\begin{align}\label{eq:uinftyw-}
u_\infty^{\omega_-}  = \frac{A}{m_i |V_{\textup{main}}|^{\frac 14}} \sin\left(-m_i \int_{r^\ast}^{\frac \pi 2 l} \sqrt{|V_{\textup{main}}|} \d y \right) ( 1+ \epsilon_{u_\infty}),
\end{align}
where $ A =  |V^{-\frac 14}_{\textup{main}}(r^\ast =l \frac{\pi}{2})|$ satisfies $|A|\sim 1$ and $\epsilon_{u_\infty}$ satisfies $\epsilon_{u_\infty}\left( \frac{\pi}{2}l \right) = \epsilon_{u_\infty}'\left( \frac{\pi}{2}l \right) = 0$ as well as 
\begin{align}
	|\epsilon_{u_\infty}|,\frac{|\epsilon^\prime_{u_\infty}|}{2 m_i {|V_{\textup{main}}|}^{\frac 12}} \lesssim \frac{\mathcal{V}_{\tilde r_2^\ast,\frac{\pi}{2}l} (F_\infty)}{m_i}  \lesssim \frac{1}{m_i}.
\end{align}
Indeed, note that the condition $u_\infty^{\omega_-}  (\frac \pi 2 l) =0$, $\frac{\d }{\d r^\ast} u_\infty^{\omega_-} (\frac \pi 2 l) =1$  and $\epsilon_{u_\infty}\left( \frac{\pi}{2}l \right) =  \epsilon_{u_\infty}'\left( \frac{\pi}{2}l \right) = 0$ force $u_\infty^{\omega_-}$ to be of the form \eqref{eq:uinftyw-}.
Now, since $u_\infty^{\omega_-} $ oscillates with period proportional to $m_i$ in view of \eqref{eq:uinftyw-}, there exists a  compact subinterval $K_i \subset K$ of the form $K_i=[r_i^\ast - \frac{c}{m_i}, r_i^\ast+ \frac{c}{m_i}]$ for some $c>0$ such that for all $r^\ast \in K_i$, we have \begin{align}u_\infty^{\omega_-} (r^\ast, m_i,\ell_i,) \geq \frac{\epsilon}{m_i}.\end{align}
\end{proof}\label{lem:data}
\end{lemma}
We are now in the position to define our initial data which will be supported in the compact set $K$ as defined in \cref{lem:data}. We assume without loss of generality that all $m_i$ are sufficiently large such that we can apply \cref{lem:data}. First, let $\chi\colon \mathbb R \to [0,1]$ be a smooth bump function satisfying $\chi =0$ for $|x| \geq 1$ and $\chi =1$ for $|x|\leq \frac{1}{2}$. Then, for $i\in \mathbb N$ we set \begin{align} \label{defn:chii}\chi_i\colon (-\infty, l \pi/2) \to [0,1], r^\ast \mapsto \chi(c^{-1}{m_i} (r^\ast - r_i^\ast)).\end{align}

\begin{definition}\label{defn:initialdata}
Let $m_i$, $\ell_i$ be as in \eqref{eq:mili}. For each $i\in \mathbb N$, let $K_i\subset K$ be the associated subinterval as specified in \cref{lem:K} and let $\chi_i$ defined as in \eqref{defn:chii}. Then, we define initial data on $\Sigma_0$ as
	\begin{align}
&	\psi\restriction_{\Sigma_0} = \Psi_0 := 0 ,\\
	& n_{\Sigma_0}\psi\restriction_{\Sigma_0}(r,\theta,\phi)  = \Psi_1 (r,\theta,\phi) := \sum_{i \geq i_0} e^{-m_i^{\frac 13}} \psi_i(r,\theta,\phi),
	\end{align}
	where
	\begin{align}
		 \psi_i(r,\theta, \phi) =   \frac{\sqrt{r^2+a^2} \chi_i(r^\ast(r))}{-2\Sigma \sqrt{-g^{tt}}(r,\theta) u_\infty^{\omega_-} (r^\ast(r)) } S_{m_i\ell_i}(a\omega_- m_i, \cos\theta) e^{i m_i \phi}.
	\end{align}
\end{definition}

Having set up the initial data we proceed to 
\begin{definition}\label{defninitialdata} Throughout the rest of   \cref{sec:exterioranalysis}, \cref{sec:interior} and \cref{sec:mainthmkerr} we define  $\psi\in C^\infty(\mathcal{M}_{\textup{Kerr--AdS}} \setminus \mathcal{CH})$ to be the unique smooth solution to  \eqref{eq:wavekerr} of the mixed Cauchy-boundary value problem with vanishing data on $\mathcal{H}_L \cup \mathcal{B}_{\mathcal{H}}$, Dirichlet boundary conditions at infinity and the  initial data $(\Psi_0,\Psi_1) \in C_c^\infty(\Sigma_0)$ posed on $\Sigma_0$ specified in \cref{defn:initialdata}. This is well-posed in view of \cref{thm:wellposedanddecay}.\end{definition}
\begin{rmk}
By a domain of dependence argument one can also view $\psi$ as arising from smooth and compactly supported initial data posed on a spacelike hypersurface connecting both components of $\mathcal I$ as depicted in \cref{fig:adsintro}.
\end{rmk}
\begin{rmk}We note that our initial data are only supported on the positive azimuthal frequencies $m = m_i$. The same will apply to the arising solution $\psi$.\end{rmk}

In the following we define the quantity $a_\mathcal H$ from our initial data. This $a_\mathcal H$ will turn out (at least in a limiting sense) to be the (generalized) Fourier transform of the solution $\psi\restriction_{\mathcal H}$  along the event horizon. 
\begin{definition}\label{defn:amathcalhr}  For the initial data $ \Psi_0,\Psi_1 $ as in \cref{defn:initialdata} we define
	\begin{align} \nonumber
a_{\mathcal H}(\omega , m,\ell) := \frac{1}{\sqrt{2\pi}\mathfrak W[\uhplus,u_\infty]}  &\int_{r_+}^{\infty}\int_{\mathbb S^2}	\bigg\{ \frac{\Sigma}{\sqrt{r^2+a^2}} u_\infty    e^{-im\phi} S_{m\ell}(a \omega,\cos\theta) \\
& \times \left( -2 \sqrt{-g^{tt}} \Psi_1 - i \omega g^{tt} \Psi_0   \right)\bigg\} \d \sigma_{\mathbb S^2} \d r .
	\end{align}   
\end{definition}
Now, we will show that $a_{\mathcal H}$ has ``peaks'' at the interior scattering poles $\omega = \omega_- m$ for infinitely many $m$. This is a consequence of our careful  choice of initial data. We formulate this in
 \begin{lemma}\label{eq:defnah}
For $a_{\mathcal H }$ as in \cref{defn:initialdata} we have 
  \begin{align}
&a_{\mathcal{H} } (\omega = \omega_- m, m ,\ell) = a_{\mathcal H }(\omega= \omega_- m_i, m_i,\ell_i) \delta_{mm_i } \delta_{\ell\ell_i} ,\end{align}
where \begin{align}\label{eq:exponentiallyhigh}
& |a_{\mathcal H }(\omega=\omega_- m_i, m_i,\ell_i)|    \gtrsim e^{\frac 12 \sqrt{m_i}}
 	\end{align}
 	for $(m_i,\ell_i)$  as in \eqref{eq:mili}.
 	\begin{proof}
 	As $\Psi_0 = 0$, we compute \begin{align}\nonumber
 			&\int_{r_+}^\infty \int_{\mathbb S^2}\frac{\Sigma}{ \sqrt{r^2+a^2}} u_\infty^{\omega_-}  e^{-i m \phi} S_{m\ell}(a\omega_- m,\cos\theta) (-2 \sqrt{-g^{tt}}) \Psi_1 \d \sigma_{\mathbb S^2} \d r \\ 
 			& =  e^{- m_i^{\frac 13}}\delta_{m m_i}\delta_{\ell \ell_i} \int_{r_+}^\infty \chi_i(r) \d r \sim  e^{- m_i^{\frac 13}}m_i^{-1} \delta_{m m_i}\delta_{\ell \ell_i} .
 			\label{eq:superpolynomialdecayofinitialdata}
 		\end{align}
 		To conclude we use that from \eqref{eq:smallnessofwronskian} we have \begin{align}|\mathfrak W[\uhplus, u_\infty] ( \omega = \omega_- m_i, m_i, \ell_i)| < e^{- \sqrt m_i} .\end{align}
 	\end{proof}
 	 \end{lemma}

\section{Exterior analysis: From the initial data to the event horizon}
\label{sec:exterioranalysis}
\subsection{Cut-off in time and inhomogeneous equation}
\label{sec:cutoff}
We will now consider the $\psi$ as defined in \cref{sec:initialdata}. The goal of this section is to determine the Fourier transform of $\psi$ along the event horizon. To do so we will first take   a time cut-off of $\psi$. To do so, we let  \begin{align}\label{eq:defnofchi} \chi\colon \mathbb R \to [0,1]\end{align} be a smooth and monotone cut-off function with $\chi(x) =0$ for $x \leq 0$, $\chi(x) = 1$ for $x \geq 1$. Now, define $\chi_\epsilon^R(v):= \chi(v/\epsilon) \chi(R - v) $ such that $\chi^R_\epsilon \to \mathbb{1}_{(0,\infty)}$ pointwise as $\epsilon \to 0$ and $R\to\infty$. Moreover, \begin{align}\label{eq:distributions} \partial_v (\chi (v/\epsilon)  )\to \delta_0(v) \text{ and }\partial^2_v ( \chi(v/\epsilon) ) \to  \delta_0'(v)\end{align} as $\epsilon \to 0 $ in the sense of distributions.
	On $\mathcal{R}\cup \mathcal{H}_{R}$ we set 
	\begin{align}\psi^R_{\epsilon}(v,r,\theta,\tilde \phi_+) := \psi (v ,r, \theta,\tilde \phi_+)  \chi^R_\epsilon(v) \text{ and }\psi^R := \psi(v,r,\theta,\tilde \phi_+) \chi(R-v)\label{defn:psiRepsilon}
	\end{align}
	and note that $\psi^R_\epsilon$ is smooth and compactly supported in $v$ and satisfies the inhomogeneous equation
	\begin{align} \label{eq:psiepsilonrinhom}
		\Box_{g_{\textup{Kerr--AdS}}} \psi^R_\epsilon + \frac{2}{l^2} \psi^R_\epsilon = F_\epsilon^R := 2 (\partial_{v} \chi^R_\epsilon )(\nabla v)\psi + \psi \Box_{g_{\textup{Kerr--AdS}}} \chi^R_\epsilon.
	\end{align}
	Analogously, $\psi^R$ satisfies the inhomogeneous equation with \begin{align} 	\Box_{g_{\textup{Kerr--AdS}}} \psi^R + \frac{2}{l^2} \psi^R = F^R := 2 (\partial_{v} \chi^R )(\nabla v)\psi + \psi \Box_{g_{\textup{Kerr--AdS}}} \chi^R.\end{align}
	As in \cite[Section~5.1]{gustav} we have \begin{align}&|F^R|^2 r^2 \lesssim \frac{1}{r^2} |\partial_v \psi|^2 + r^2 |\partial_r \psi|^2 + |\slashed \nabla \psi|^2 + |\psi|^2,\\
	& |F_\epsilon^R|^2 r^2 \lesssim  \frac{1}{\epsilon^2 r^2} |\partial_v \psi|^2 + r^2 |\partial_r \psi|^2 + |\slashed \nabla \psi|^2 + \frac{1}{\epsilon^2} |\psi|^2.\end{align}
In view of our coordinates, we also have that for each $r>r_+$, the function $\psi_\epsilon^R(t,r,\theta,\phi)$ is compactly supported in $\mathbb R_t$ with values in $C^\infty( \mathbb S^2)$. 
  This allows us to apply Carter's separation of variables to express $\psi_\epsilon^R$ as 
\begin{align}
&	\psi_\epsilon^R (t,r,\theta,\phi) = \frac{1}{\sqrt {2\pi}} \int_{\mathbb R} \mathfrak{F}[\psi^R_\epsilon](\omega,r,\theta,\phi)  e^{-i \omega t}\d \omega,
\end{align}
where
for each $r > r_+$, \begin{align}
	\mathfrak{F}[\psi^R_\epsilon](\omega,r,\theta,\phi) := \frac{1}{\sqrt{2\pi}}  \int_{\mathbb R} \psi^R_\epsilon(t,r,\theta,\phi)  e^{i \omega t}\d t
\end{align}
is  a Schwartz function on $\mathbb R_\omega$ with values in $C^\infty(\mathbb S^2)$. For the definition of Fr\'echet~space-valued Schwartz functions refer to  \cite[p.~533]{treves} or to \cite[Definition~3]{seminar_schwartz} by L.\ Schwartz himself. Note however that we will only use Fr\'echet space-valued Schwartz functions in a qualitative way and in view of this we will not go into more details.
  We further decompose $\mathfrak{F}[\psi^R_\epsilon](\omega,r,\theta,\phi)$ in (generalized) spheroidal harmonics
\begin{align}
	\mathfrak{F}[\psi^R_\epsilon](\omega,r,\theta,\phi)= \sum_{m\ell} \hat{\psi^R_\epsilon}(\omega,m,\ell,r) S_{m\ell}(a \omega, \cos\theta) e^{i m \phi}, \label{eq:hatpsi}
\end{align}
where 
\begin{align}
\hat{\psi^R_\epsilon}(\omega,m,\ell,r) := \int_{\mathbb S^2} 	\mathfrak{F}[\psi^R_\epsilon](\omega,r,\theta,\phi) e^{-im\phi} S_{m\ell}(a\omega,\cos\theta) \d \sigma_{\mathbb S^2} \label{eq:defnpsiepsilonR}
\end{align} is smooth in $\omega$ and $r> r_+$ for fixed $m$ and $\ell$ and moreover  \begin{align}\hat{\psi^R_\epsilon}  \in L^2(\mathbb R_\omega \times \mathbb Z_m \times \mathbb Z_{\ell \geq |m|}; C^\infty(\tilde r ,\infty) )\end{align}
in view of Plancherel's theorem for every $\tilde r > r_+$. Equivalently, we have 
\begin{align}\label{eq:524}
	\hat{\psi^R_\epsilon}(\omega,m,\ell,r) = \frac{1}{\sqrt {2\pi}} \int_{\mathbb R}  \int_{\mathbb S^2} \psi^R_\epsilon(t,r,\theta,\phi) e^{i \omega t} e^{-i m \phi}S_{m\ell}(a \omega, \cos \theta) \d \sigma_{\mathbb S^2} \d t
\end{align}
for each $r>r_+$. Now, note that $\psi^R_\epsilon(v,r,\theta,\tilde \phi_+)$ is smooth and compactly supported on $\mathbb R_v$ all the way to $r=r_+$ and thus, takes values in the space $C^\infty( [r_+,\infty)_r \times \mathbb S^2_{\theta,\tilde \phi_+})$. After a change of coordinates in \eqref{eq:524} we obtain that
\begin{align}\hat \psi^R_\epsilon(\omega,r,m,\ell) e^{i(\omega-\omega_+m) r^\ast} \label{eq:psihatmal}\end{align} extends smoothly to $r=r_+$ ($r^\ast \to -\infty$). 
Similarly to the above, we have
\begin{align}
	\widehat{\Sigma F^R_\epsilon}(\omega,m ,\ell, r) =\frac{1}{\sqrt {2\pi}} \int_{\mathbb R}  \int_{\mathbb S^2} \Sigma F^R_\epsilon(t,r,\theta,\phi) e^{i \omega t} e^{-i m \phi}S_{m\ell}(a \omega,\cos \theta)\d \sigma_{\mathbb S^2} \d t.
\end{align}
Now, we define 
\begin{align}\label{eq:defnuepsilon}
	&u_\epsilon^R = u_\epsilon^R(\omega, m,\ell,r):= (r^2 + a^2)^{\frac 12} \hat{\psi_\epsilon^R}(\omega, m,\ell,r) \end{align}and
	\begin{align}\label{eq:defnH}
	&H_\epsilon^R (\omega, m,\ell,r):= \frac{\Delta}{(r^2+a^2)^{\frac 32 }} \widehat {\Sigma F_\epsilon^R} (\omega, m,\ell,r).
\end{align}
Since $\hat{\psi_\epsilon^R}$ defined in \eqref{eq:defnpsiepsilonR} is smooth, we have that $u_\epsilon^R$ as defined in \eqref{eq:defnuepsilon} is a smooth function of $\omega$ and $r >r_+$. Moreover, we can also differentiate under the integral sign in \eqref{eq:524} and since $\psi_\epsilon^R$ satisfies \eqref{eq:psiepsilonrinhom}, we have that $u_\epsilon^R$ satisfies the inhomogeneous radial o.d.e.\
\begin{align}\label{eq:inhomogeneousterm}
	-{u_\epsilon^R}'' + (V -\omega^2)u_\epsilon^R = H_\epsilon^R
\end{align}
pointwise for each $\omega,m,\ell$ on $r^\ast \in (-\infty, \frac{\pi}{2} l]$, where we recall that $' = \frac{\d }{\d r^\ast}$.

\subsection{Estimates for the inhomogeneous radial o.d.e.}
\begin{lemma}\label{lem:boundaryconditionsonhandu}
The solution $u_\epsilon^R$ as defined in \eqref{eq:defnuepsilon} satisfies the boundary conditions 
\begin{align}
\label{eq:dirichlet}
&u_\epsilon^R = 0 \text{ for } r^\ast = \frac{\pi}{2}l,\\
&{u_\epsilon^R}^\prime + i (\omega - \omega_+m) u_\epsilon^R = O(\Delta) \text{ as } r^\ast \to -\infty\label{eq:outgoing}
\end{align}
and   the inhomogeneity $H_\epsilon^R$ defined in \eqref{eq:inhomogeneousterm} also satisfies
\begin{align}
\label{eq:boundarycondH1}
& H_\epsilon^R =0 \text{ for } r^\ast = \frac{\pi}{2}l, \\
&{H_\epsilon^R}^\prime + i (\omega - \omega_+m) H_\epsilon^R = O(\Delta) \text{ as } r^\ast \to -\infty.
\label{eq:boundarycondH2}
\end{align}
\begin{proof}
To see \eqref{eq:dirichlet} note that 
\begin{align}
|u_\epsilon^R| \leq (r^2 +a^2)^{\frac 12}| \hat{\psi^R_\epsilon} | = \left|\int_{\mathbb R} \int_{\mathbb S^2}  (r^2+a^2)^{\frac 12}  {\psi_\epsilon^R}(r,t,\theta,\phi) e^{i\omega t} S_{m\ell}(a \omega, \cos\theta) e^{- i m \phi}  \d \sigma \d t\right|.
\end{align}
In view of the compact support of $\psi_\epsilon^R$ in $t$, it suffices to show that the pointwise limit \begin{align}\lim_{r\to\infty} r\psi^R_\epsilon(t,r,\theta,\phi)=0\end{align} holds true. But this follows from the fact that $\psi^R_\epsilon \in CH^1_{\textup{AdS}}$---a consequence of the well-posedness in \cref{thm:wellposedanddecay}. 

For \eqref{eq:outgoing}, we use \eqref{eq:psihatmal} to see that $\partial_r (\hat \psi^R_\epsilon(\omega,r,m,\ell) e^{i(\omega-\omega_+m) r^\ast})$ extends smoothly to $r=r_+$. Thus, using $\partial_r=  \frac{r^2+a^2}{\Delta} \partial_{r^\ast} $, we  infer that
\begin{align}
{u_\epsilon^R}^\prime + i (\omega - \omega_+m) u_\epsilon^R = O(\Delta) \text{ as } r^\ast \to -\infty.
\end{align}
Analogously, we obtain \eqref{eq:boundarycondH1} and \eqref{eq:boundarycondH2}.
\end{proof}
\end{lemma}

\begin{lemma}\label{prop:microlocalhorizontrace}
	We represent $u_\epsilon^R$ as \begin{align}
		u_\epsilon^R (r^\ast) = \frac{1}{\mathfrak{W}[\uhplus,u_\infty]} \left\{\uhplus \int_{r^\ast}^{\frac \pi 2} u_\infty H^R_\epsilon \d y + u_\infty \int_{-\infty}^{r^\ast} \uhplus H^R_\epsilon \d y \right\}\label{eq:rep}.
	\end{align}
	Moreover, \begin{align} \lim_{r_\ast \to -\infty} u^R_\epsilon (r^\ast)e^{i (\omega - \omega_+ m) r^\ast}  = 	a^R_{\epsilon,\mathcal H}, \label{eq:microlocalhorizontrace1} \end{align}  
	where $a_{\epsilon, \mathcal{H}}^R$  is defined as
	\begin{align}
		a^R_{\epsilon,\mathcal H} :=  \frac{1}{\mathfrak{W}[\uhplus,u_\infty]}  \int_{-\infty }^{\frac \pi 2} u_\infty H_\epsilon^R \d y  . \label{eq:microlocalhorizontrace}
	\end{align}
\end{lemma}
\begin{proof}
	First, since there do not exist pure mode solutions as shown in \cite[Theorem~1.3]{gustav}, the Wronskian $\mathfrak W[\uhplus,u_\infty] $ never vanishes. Thus, \eqref{eq:rep} is well-defined and in view of the boundary conditions of $u_\epsilon^R$ and $H_\epsilon^R$ as shown in \cref{lem:boundaryconditionsonhandu}, a direct computation shows \eqref{eq:rep}.  
	To show \eqref{eq:microlocalhorizontrace1} we first note that 
  that 
  \begin{align}
  	|u_{\infty}|\leq C_{m\ell \omega} |r^\ast|
  \label{eq:boundonuinfty}
\end{align} for all $r^\ast$. Indeed, \eqref{eq:boundonuinfty} holds true as for each $\omega,m,\ell$ there exist constants $a_s , a_c$ only depending on the $\omega,m,\ell$ such that 
	\begin{align} \label{eq:decompositionofuinfty} u_\infty = a_s u_s + a_c u_c,\end{align}
	where $u_s$ and $u_c$ are solutions to the radial o.d.e.\ satisfying $u_s \sim \frac{\sin( ( \omega - \omega_+ m)r^\ast)}{\omega - \omega_+ m}$ and $u_c \sim \cos( ( \omega - \omega_+ m)r^\ast)$ as $r^\ast \to -\infty$. For $\omega\neq \omega_+ m$,  $u_s$ and $u_c$ are defined as  $u_s= \frac{1}{2i} \frac{u_{\mathcal H^-} - u_{\mathcal H^+}}{\omega - \omega_+ m}$ and $u_c =  \frac{1}{2i} (u_{\mathcal H^-} + u_{\mathcal H^+})$, where $u_{\mathcal H^+}$ and $u_{\mathcal H^-}$  are defined in \cref{defn:uhruhlexterior}. This definition uniquely extends to  $\omega = \omega_+ m$ with the asymptotics   $u_s \sim r^\ast$ and $u_c \sim 1$ as $r^\ast \to -\infty$.  In particular, $\mathfrak W (u_s, u_c) = -1 $ for all $\omega,m,\ell$ which justifies \eqref{eq:decompositionofuinfty}. 
	
	We now obtain \eqref{eq:microlocalhorizontrace1} since
	\begin{align} \nonumber
\limsup_{r^\ast \to -\infty} & \left| u_\infty \int_{-\infty}^{r^\ast} \uhplus H^R_\epsilon \d y\right|^2 \leq C_{m\ell \omega}^2 \limsup_{r^\ast \to -\infty} \left( |r^\ast| \int_{r_+}^{r(r^\ast)} \frac{| \widehat{\Sigma F^R_{\epsilon}}|^2}{r^2+a^2} \d r \int_{-\infty}^{r_\ast} |\uhplus|^2 \frac{\Delta}{{r^2+a^2}} \d y \right)\\
& \leq C_{m\ell \omega}^2 \limsup_{r^\ast \to -\infty} \left(  \int_{r_+}^{r(r^\ast)} \frac{| \widehat{\Sigma F^R_{\epsilon}}|^2}{r^2+a^2} \d r \int_{-\infty}^{r_\ast} |\uhplus|^2 \frac{|y|^2 \Delta}{{r^2+a^2}} \d y \right) = 0
\label{eq:estim}	\end{align}
 	because \begin{align}\int_{r_+}^{r(r^\ast)}\frac{| \widehat{\Sigma F^R_{\epsilon}}|^2}{ r^2+a^2  } \d r <\infty, \; \sup_{r^\ast \in (-\infty,r_1)}|\uhplus|<\infty,\end{align} 
 	and $|r^\ast|^2 \Delta $   decays exponentially as $r^\ast \to -\infty$. In \eqref{eq:estim} we also used that $\d r = \frac{\Delta}{r^2 + a^2} \d r^\ast$. 
\end{proof}
\begin{lemma}\label{lem:HrepstoHr}
	The inhomogeneous term $H^R_\epsilon$ has the pointwise limit
	\begin{align} \nonumber 
	H^R &:= \lim_{\epsilon \to 0} H^R_\epsilon\\ &=\frac{\Delta}{(r^2+a^2)^{\frac 32}}  \frac{1}{\sqrt{2\pi}} \int_{\mathbb S^2} \Sigma   e^{-im\phi} S_{m\ell}(a \omega) \left( -2 \sqrt{-g^{tt}} \Psi_1 - i \omega g^{tt} \Psi_0  \right) \d \sigma_{\mathbb S^2} \nonumber \\  &+\frac{ \Delta}{(r^2+a^2)^{\frac 32}}\frac{e^{-i(\omega -\omega_+m)r^\ast} }{\sqrt{2\pi}} \int_{R-1}^{R}\int_{\mathbb S^2}\Sigma F^R(v,r,\theta,\tilde \phi_+) e^{i \omega v} e^{-i m\tilde{\phi}_+} S_{m\ell}(a\omega) \d \sigma_{\mathbb S^2} \d v.
	\end{align}
	In addition, 
	\begin{align}
		a^R_{\epsilon, \mathcal H} \to a_{\mathcal H}^R:= 	 \frac{1}{\mathfrak{W}[\uhplus,u_\infty]}  \int_{-\infty }^{\frac \pi 2} u_\infty H^R \d r^\ast \label{eq:defnar}
	\end{align}
	pointwise  as $\epsilon \to 0$. 

	Moreover,   we have
	\begin{align}\nonumber
	H^R \to H := \frac{\Delta}{({r^2+a^2})^{\frac 32}}  \frac{1}{\sqrt{2\pi}} \int_{\mathbb S^2} & \Sigma   e^{-im\phi} S_{m\ell}(a \omega)\\ & \left( -2 \sqrt{-g^{tt}} \Psi_1 - i \omega g^{tt} \Psi_0  \right) \d \sigma_{\mathbb S^2} 
	\end{align}
	and 
	\begin{align}
	a_{\mathcal H}^R \to a_{\mathcal H}
	\end{align}
  pointwise    as $R\to\infty$. Recall that $a_\mathcal H$ is defined in \cref{defn:amathcalhr}.
	\begin{proof}We start with the decomposition $F_\epsilon^R =F_\epsilon + F_R$, where the support of $F_\epsilon$ is in $\{0\leq t \leq \epsilon \} \cap \{ r \geq 2r_+\}$ for $\epsilon>0$ small enough and the support of $F_R$ is in the set $\{R-1 \leq v \leq R\} $. 
		
		We first consider   $F_\epsilon $ and write its (generalized) Fourier transform as
		\begin{align}
		\widehat{\Sigma  F_\epsilon} (\omega,m,\ell,r) &=\frac{1}{\sqrt {2\pi}} \int_{\mathbb S^2} \int_{\mathbb R}    F_\epsilon(t,r,\theta,\phi) e^{i \omega t}  \d t \, \, \Sigma e^{-i m \phi}S_{m\ell}(a \omega,\cos\theta) \d \sigma_{\mathbb S^2}.
		\end{align}
	Recall that for all $\epsilon >0$ sufficiently small, we have that $v (t,r,\theta,\phi) = t$ on the support of $F_{\epsilon}$. Thus,
	\begin{align}\nonumber 
	F_\epsilon	&= 2   \partial_t (\chi(t/\epsilon) )(\nabla t)\psi + \psi \Box_{g_{\textup{Kerr--AdS}}} \chi(t/\epsilon)\\ 
& =  - 2  \partial_t (\chi(t/\epsilon) ) \sqrt{-g^{tt}} n_{\Sigma_t}  \psi +  \psi g^{tt} \partial_t^2 ( \chi(t/\epsilon)), 
	\end{align} 
where $\chi$ is as in \eqref{eq:defnofchi}. We also used that for any $f(t,r,\theta,\phi)=f(t)$ only depending on $t$ in Boyer--Lindquist coordinates, we have $\Box_{g_{\textup{Kerr--AdS}}} f = g^{tt} \partial_t^2 f$.   Now, in view of \eqref{eq:distributions}  we obtain
 \begin{align}F_\epsilon \to F_0 := -2 \sqrt{-g^{tt} } \delta_{t=0} \Psi_1  + g^{tt} \delta'_{t=0} \Psi_0  \end{align}  as $\epsilon \to 0$ in the sense of distributions (compactly supported distributions) on $\mathbb R_t$ with values in $C^\infty ( (r_+, \infty) \times \mathbb S^2) $. Hence,
		\begin{align} \nonumber
	&	\widehat{\Sigma   F_{\epsilon}}  (\omega, m, \ell, r) = \frac{1}{\sqrt{2\pi}}   \int_{\mathbb S^2} \int_{\mathbb R} e^{-im\phi} S_{m\ell}(a \omega,\cos\theta) F_\epsilon e^{i \omega t}  \d t\d \sigma_{\mathbb S^2}\\ & \;\;\;\;\; \to\frac{1}{\sqrt{2\pi}}   \int_{\mathbb S^2}  e^{-im\phi} S_{m\ell}(a \omega,\cos\theta) \left( - 2  \sqrt{-g^{tt}} \Psi_1 - i \omega g^{tt} \Psi_0  \right) \d \sigma_{\mathbb S^2}
		\end{align}
		pointwise. (Note that the above pointwise limit can also be shown via integration by parts in $t$  without appealing to distribution theory, cf.\ \eqref{eq:estimateonsigmafe} below.)
		Thus,   \begin{align}\nonumber
		H^R &= \lim_{\epsilon \to 0} H^R_\epsilon = \\ \nonumber & = \frac{\Delta}{(r^2+a^2)^{\frac 32}}  \frac{1}{\sqrt{2\pi}} \int_{\mathbb S^2} \Sigma  e^{-im\phi} S_{m\ell}(a \omega) \left( - 2\sqrt{-g^{tt}} \Psi_1 - i \omega g^{tt} \Psi_0    \right) \d \sigma_{\mathbb S^2} \\ &+\frac{\Delta}{ (r^2+a^2)^{\frac 32}}\frac{e^{-i(\omega -\omega_+m)r^\ast} }{\sqrt{2\pi}} \int_{R-1}^{R}\int_{\mathbb S^2} \Sigma F_R(v,r,\theta,\tilde \phi_+) e^{i \omega v} e^{-i m\tilde{\phi}_+} S_{m\ell}(a \omega) \d \sigma_{\mathbb S^2} \d v
		\end{align}
		pointwise.
		
		Now, to show that $a_{\epsilon, \mathcal{H}}^R \to a_{\mathcal H}^R$ it suffices to show 
		\begin{align}\label{eq:limittointerchange}
		  \int_{-\infty }^{\frac \pi 2} u_\infty \frac{\Delta}{(r^2+a^2)^{\frac 32}} \widehat{ \Sigma F_\epsilon}\d y  \to 	  \int_{-\infty }^{\frac \pi 2} u_\infty \frac{\Delta}{(r^2+a^2)^{\frac 32}} \widehat{ \Sigma F_0}\d y  
		\end{align}
		pointwise  as $\epsilon \to 0$. Again, recall from \cref{defn:initialdata} that our initial data are compactly supported in $K$. Thus, by finite speed of propagation we have that    $r^\ast \mapsto F_\epsilon(r^\ast)$ is compactly supported (uniformly in the other coordinates) in an open neighborhood $K^o \supset K$ of $K$ for all $0 < \epsilon< \epsilon_0$ sufficiently small.  Note that $K^o \setminus K$   can be made arbitrarily small by choosing $\epsilon_0>0$ sufficiently small. Further, we show below that $\sup_{0 < \epsilon< \epsilon_0} \sup_{r^\ast} |\widehat{\Sigma F_\epsilon} | < \infty$ so we can interchange the integral with the limit $\epsilon \to 0$ in \eqref{eq:limittointerchange}.
		
		We will now justify $\sup_{0 < \epsilon< \epsilon_0} \sup_{r^\ast} |\widehat{\Sigma F_\epsilon} | < \infty$.
		Indeed, using that $\partial_t(\chi(t/\epsilon))$ is only supported in $[0,\epsilon]$ and integrating by parts  we have
		\begin{align} \nonumber
|\widehat{\Sigma F_\epsilon} | & \lesssim \left| \int_{\mathbb S^2} \int_0^\epsilon e^{-im\phi} S_{m\ell}(a\omega,\cos\theta)  e^{i\omega t} \left( - 2 \partial_t(\chi(t/\epsilon)) \sqrt{-g^{tt}} n_{\Sigma_t} \psi + \psi g^{tt} \partial_t^2(\chi(t/\epsilon)) \right)  \d \sigma_{\mathbb S^2} \d t\right|\\ \nonumber
& \lesssim   \left| \int_{\mathbb S^2} \int_0^\epsilon e^{-im\phi} S_{m\ell}(a\omega,\cos\theta)  \left(  2 \chi(t/\epsilon) \sqrt{-g^{tt}}  \partial_t (e^{i\omega t} n_{\Sigma_t} \psi )+   g^{tt} \partial_t^2 (e^{i\omega t} \psi) \chi(t/\epsilon) \right)  \d \sigma_{\mathbb S^2} \d t\right|
\\ & +   \left| \int_{\mathbb S^2} e^{-im\phi} S_{m\ell}(a\omega,\cos\theta)  \left(  2    e^{i\omega \epsilon}\sqrt{-g^{tt}} n_{\Sigma_{t=\epsilon}} \psi (t=\epsilon)+ ( \partial_t (e^{i\omega t} \psi g^{tt})) ( t= \epsilon) \chi(t/\epsilon) \right)  \d \sigma_{\mathbb S^2} \right|. \label{eq:estimateonsigmafe}
		\end{align}
	Since $\psi$ and all its derivatives are uniformly bounded and moreover are supported in $K^o$ in the $r^\ast$ coordinate we obtain $\sup_{0 < \epsilon< \epsilon_0} \sup_{r^\ast} |\widehat{\Sigma F_\epsilon} | < \infty$ for $\epsilon_0>0$ sufficiently small.
	
Next, we will show that $H^R \to H$ as $R\to\infty$.
		As $\psi$ and its derivatives decay pointwise at a logarithmic rate (see \cref{thm:wellposedanddecay}), we obtain \begin{align} \sup_{r \in (r_+,\infty),\theta,\tilde \phi_+\in \mathbb S^2} |F_r|(v,r,\theta,\tilde \phi_+) \to 0 \end{align} as $R\to \infty$. Thus, we have
		\begin{align}\nonumber
		\left| \int_{R-1}^{R}\int_{\mathbb S^2} \Sigma F_R(v,r,\theta,\tilde \phi_+) e^{i \omega v} e^{-i m\tilde{\phi}_+} S_{m\ell}(a \omega) \d \sigma_{\mathbb S^2} \d v \right|\\  \lesssim r^2 \sup_{\substack{v\in (R-1,R) \\ r \in (r_+,\infty),\theta,\tilde \phi_+\in \mathbb S^2} }|F_r|(v,r,\theta,\tilde \phi_+)  \to 0 
		\end{align}
		pointwise as $R\to\infty$. This shows $H^R \to H$ pointwise.
		
		Finally, to show that   $a_{\mathcal H}^R \to a_{\mathcal H}$ as $R\to\infty$, we estimate
		\begin{align} \nonumber
	&	\left|	\int_{-\infty}^{\frac \pi 2 } u_\infty \frac{\Delta}{ (r^2+a^2)^{\frac 32}}\frac{e^{-i(\omega -\omega_+m)r^\ast} }{\sqrt{2\pi}} \int_{R}^{R+1}\int_{\mathbb S^2} \Sigma F_R(v,r,\theta,\tilde \phi_+) e^{i \omega v} e^{-i m\tilde{\phi}_+} S_{m\ell}(a \omega) \d \sigma_{\mathbb S^2} \d v \d r^\ast \right|^2
	 \\ \nonumber
	& \lesssim \int_{-\infty}^{\frac \pi 2} \int_{\mathbb S^2} |S_{m\ell}(a \omega)|^2 |u_\infty|^2\frac{1}{r^2} \frac{\Delta}{r^2+a^2}  \d \sigma_{\mathbb S^2} \d r^\ast  \sup_{v\in (R,R+1)}\int_{r_+}^{\infty} \int_{\mathbb S^2} \Sigma^2 |F_R|^2  \frac{1}{r^2+a^2} r^2 \frac{\Delta}{r^2+a^2}  \d \sigma_{\mathbb S^2} \d r^\ast
	 \\ \nonumber
		& \lesssim \int_{-\infty}^{\frac \pi 2} |u_\infty|^2\frac{1}{r^2} \frac{\Delta}{r^2+a^2}  \d r^\ast  \sup_{v\in (R,R+1)}\int_{r_+}^{\infty} \int_{\mathbb S^2} \Sigma^2 |F_R|^2  \frac{1}{r^2+a^2} r^2  \d \sigma_{\mathbb S^2} \d r \\ \nonumber
		& \lesssim   \int_{-\infty}^{\frac \pi 2} C^2_{m\ell \omega} |r^\ast|^2 \frac{1}{r^2} \frac{\Delta}{r^2+a^2}  \d r^\ast  \sup_{v\in (R,R+1)}\int_{r_+}^{\infty} \int_{\mathbb S^2} r^2 |F_R|^2   r^2  \d \sigma_{\mathbb S^2} \d r \\
		&\lesssim C^2_{m\ell\omega} \sup_{v\in (R,R+1)}\int_{r_+}^{\infty} \int_{\mathbb S^2} e_1[\psi]  r^2  \d \sigma_{\mathbb S^2} \d r  \to 0 
		\end{align}
		as $R\to \infty$. Here, we have used \eqref{eq:boundonuinfty} again.
		Hence,	 $a_{\mathcal H}^R \to a_{\mathcal H}$ as $R\to\infty$ pointwise for each $\omega, m,\ell$. 
	\end{proof}
\end{lemma}
\subsection{Representation formula for \texorpdfstring{$\psi$}{Psi} at the event horizon}
In what follows we will prove a representation formula of the truncated solution $\psi^R$ along the event horizon in terms of the initial data and an error term which vanishes in the limit $R\to \infty$. More precisely, we will represent $\psi^R$ through  $a_{\mathcal{H}}^R$ which is defined in \eqref{eq:defnar}. Note that in the limit $R\to \infty$, we have $a_{\mathcal H}^R \to a_{\mathcal H}$ which in turn  only depends  on the initial data (see \cref{defn:amathcalhr}).
\label{sec:rep}
\begin{prop}\label{prop:rep}
	Let $a_{\mathcal H}^R$ be as defined in \eqref{eq:defnar}. Then,  on the event horizon $\mathcal{H}_R$ we have
	\begin{align}
	\psi^R (v,r_+,\theta, \tilde \phi_+) = \frac{1}{\sqrt{2\pi (r_+^2+a^2)}} \sum_{m\ell}\int_{\mathbb R}a_{\mathcal{H}}^R S_{m\ell}(a\omega, \cos \theta) e^{i m \tilde \phi_+} e^{-i \omega v} \d \omega,\end{align}
	in $L^2(\mathbb R_v \times \mathbb S^2)$.
	Moreover, 
	\begin{align}\label{eq:representationformulainproof}
a_{\mathcal{H}}^R = \sqrt{\frac{r_+^2+a^2}{2\pi}}\int_{\mathbb R \times \mathbb S^2} \psi^R (v,r_+,\theta,\tilde{\phi}_+ ) S_{m\ell}(a\omega, \cos \theta) e^{-i m \tilde \phi_+} e^{i \omega v} \d \sigma_{\mathbb S^2} \d v
\end{align} 
pointwise and in $L^2(\mathbb R_\omega \times \mathbb Z_{m} \times \mathbb Z_{\ell \geq|m|})$.
	\begin{proof}
	We have  \begin{align}\psi_\epsilon^R (v,r,\theta,\tilde{\phi}_+ )= \frac{1}{\sqrt{2\pi (r^2+a^2)}} \sum_{m\ell}\int_{\mathbb R}e^{i (\omega - \omega_+ m)r^\ast} u_\epsilon^R S_{m\ell}(a\omega, \cos \theta) e^{i m \tilde \phi_+} e^{-i \omega v} \d \omega\end{align}
	and 
	\begin{align}\label{eq:limitofeandu}
	e^{i (\omega - \omega_+ m)r^\ast} u_\epsilon^R 
	 = \sqrt{\frac{r^2+a^2}{2\pi}}\int_{\mathbb R \times \mathbb S^2} \psi_\epsilon^R (v,r,\theta,\tilde{\phi}_+ ) S_{m\ell}(a\omega, \cos \theta) e^{-i m \tilde \phi_+} e^{i \omega v} \d \sigma_{\mathbb S^2} \d v
	\end{align}
	for $r_+ < r < r_+ + \eta $.
Now, since $\psi_\epsilon^R$ is compactly supported in $v$ uniformly as $r_\ast \to-\infty$, we can interchange the limit $r^\ast \to -\infty$ with the integral over $v$. Thus, sending $r\to r_+$ ($r^\ast \to -\infty$) in \eqref{eq:limitofeandu} yields in view of \cref{prop:microlocalhorizontrace} that 
\begin{align} \label{eq:558}
		a_{\epsilon, \mathcal H}^R
			= \sqrt{\frac{r_+^2+a^2}{2\pi}}\int_{\mathbb R \times \mathbb S^2} \psi_\epsilon^R (v,r_+,\theta,\tilde{\phi}_+ ) S_{m\ell}(a\omega, \cos \theta) e^{-i m \tilde \phi_+} e^{i \omega v} \d \sigma_{\mathbb S^2} \d v,
\end{align} 
where $a_{\epsilon, \mathcal H}^R$ is defined in \eqref{eq:microlocalhorizontrace}. 
Now we will perform the limit $\epsilon\to 0$ on both sides of \eqref{eq:558} independently. First, from \cref{lem:HrepstoHr} we have that  \begin{align}a_{\epsilon,\mathcal{H}}^R \to a_{\mathcal{H}}^R  = \frac{1}{\mathfrak{W}[u_{{\mathcal H^+}},u_\infty]}  \int_{-\infty }^{\frac \pi 2} u_\infty H^R \d y\end{align} as $\epsilon \to 0$ pointwise.
Moreover,  $\psi_\epsilon^R$ has compact support on $\mathbb R_v$ uniformly as $\epsilon \to 0$ and $\psi_\epsilon^R \to \psi^R$ pointwise and in $L^2(\mathbb R_v \times \mathbb S^2)$  as $\epsilon\to 0$. Thus, the right hand side of \eqref{eq:558} converges pointwise and due to Plancherel also  in $L^2(\mathbb R_\omega \times \mathbb Z_{m} \times \mathbb Z_{\ell\geq |m|})$ as $\epsilon \to 0$. Hence, $a_{\epsilon,\mathcal{H}}^R \to a_{\mathcal{H}}^R  $ also holds in $L^2(\mathbb R_\omega \times \mathbb Z_{m} \times \mathbb Z_{\ell\geq |m|})$ and we conclude
			\begin{align}
				a_{\mathcal{H}}^R = \sqrt{\frac{r_+^2+a^2}{2\pi}}\int_{\mathbb R \times \mathbb S^2} \psi^R (v,r_+,\theta,\tilde{\phi}_+ ) S_{m\ell}(a\omega, \cos \theta) e^{-i m \tilde \phi_+} e^{i \omega v} \d \sigma_{\mathbb S^2} \d v
			\end{align}
			which holds pointwise and in $L^2(\mathbb R_\omega \times \mathbb Z_{m} \times \mathbb Z_{\ell \geq|m|})$.
	 And by Plancherel we also have
					\begin{align}
				\psi^R (v,r_+,\theta, \tilde \phi_+) = \frac{1}{\sqrt{2\pi (r_+^2+a^2)}} \sum_{m\ell}\int_{\mathbb R}a_{\mathcal{H}}^R S_{m\ell}(a\omega, \cos \theta) e^{i m \tilde \phi_+} e^{-i \omega v} \d \omega\end{align}
				in $L^2(\mathbb R_v \times \mathbb S^2)$.
\end{proof}
\end{prop}

\section{Interior analysis: Estimates on radial o.d.e.\ and interior scattering poles}\label{sec:interior}
Having established the behavior of our solution $\psi$ on the exterior $\mathcal R$, we will now consider the interior region $\mathcal B$ characterized by $r\in (r_-, r_+)$. We will first consider the interior radial o.d.e.\ and prove a suitable representation formula on the interior. We also recall that in the interior region the tortoise coordinate is defined in \eqref{eq:defnrast} as
\begin{align}\label{defnofrastinsec8}
\frac{	\d r^\ast}{\d r} = \frac{r^2 + a^2}{\Delta},
\end{align}
where $r^\ast(\frac{r_+ + r_-}{2})  = 0$ and that $\Delta <0$ in the whole interior region.
\begin{rmk} As our initial data are only supported on azimuthal modes  $m$ which are large and positive, we only need to consider $m$ sufficiently large. 
\end{rmk}
\subsection{Radial o.d.e.\ on the interior: fixed frequency scattering}
\label{eq:radialodeinterior}
We recall the radial o.d.e. \eqref{eq:radial} and write it in the interior $ r_- < r <r_+$ as
\begin{align}\label{eq:interiorradial}
	-u'' + \left(\frac{\Delta L}{(r^2+a^2)^2} - (m \omega_r - \omega)^2 + V_1 \right) u =0,
\end{align}
where \begin{align}L := \lambda_{m\ell} +a^2 \omega^2 - 2m \omega a \Xi\end{align} and $V_1$ is defined in \eqref{eq:defnv1}. Note that $L\geq 0$ follows from \cite[Lemma~5.4]{gustav}. Also note that $V_1 = O(|\Delta|)$ uniformly for $r^\ast \in (-\infty,\infty)$. We will   treat $V_1$ as a perturbation and recall that the   high-frequency part of the potential is given by \begin{align}\label{eq:principlepart} V_\sharp := \frac{\Delta L}{(r^2+a^2)^2} - (m\omega_r - \omega)^2.\end{align} 
(Note that $V_\textup{main} = \frac{V_\sharp (\omega = \omega_- m)}{m^2}$.)
Analogously to \cref{defn:uhruhlexterior}, we define fundamental pairs of solutions to the radial o.d.e.\ corresponding to the event and Cauchy horizon, respectively.
\begin{definition}\label{eq:defnu1}
We  define  solutions $\uhr,\uhl$ to \eqref{eq:interiorradial} in the interior through the condition
\begin{align}
&\uhr = e^{-i (\omega -\omega_+m)r^\ast} + O_{\omega,m,\ell}(\Delta)\\
&\uhl = e^{i (\omega -\omega_+m)r^\ast}+ O_{\omega,m,\ell}(\Delta)
\end{align}
as $r^\ast \to -\infty$. For $\omega\neq \omega_+ m$, they form a fundamental pair. For $\omega=\omega_+ m$ the solutions $\uhl$ and $\uhr$ are linearly dependent. 

	Analogously, we define  
		\begin{align}
		&\uchl = e^{-i (\omega -\omega_-m)r^\ast}+O_{\omega,m,\ell}(\Delta)\\
	&	\uchr = e^{i (\omega -\omega_-m)r^\ast}+O_{\omega,m,\ell}(\Delta)
		\end{align}as $r^\ast \to +\infty$. For $\omega\neq \omega_- m$, they form a fundamental pair. For $\omega=\omega_- m$ the solutions $\uchl$ and $\uchr$ are linearly dependent.
\end{definition}
\begin{rmk}\label{rmk:volterra2}
As in \cref{rmk:volterra1} we can equivalently define $\uhr$ (and analogously $\uhl, \uchl, \uchr$) as the unique solution to the Volterra integral equation 
\begin{align}\label{eq:defnofuhrthroughvolterra}
\uhr (r^\ast) = e^{-i(\omega-\omega_+m)r^\ast } +\int_{-\infty}^{r^\ast} \frac{\sin((\omega-\omega_+m)(r^\ast-y))}{\omega-\omega_+m} (V_\sharp (y) +V_1 (y)+   ( \omega-\omega_+ m)^2)	\uhr  (y) \d y.
\end{align}
\end{rmk}
We moreover define reflection and transmission coefficients. 
\begin{definition}\label{defn:scatteringcoef}
	For $\omega\neq \omega_- m$ define the transmission coefficient $\mathfrak T = \mathfrak T (\omega,m,\ell)$ and the reflection coefficient $\mathfrak R = \mathfrak R (\omega,m,\ell)$ as the unique coefficients such that\begin{align}
		\uhr (r^\ast, \omega,m,\ell)= \mathfrak T(\omega,m,\ell) \uchl(r^\ast, \omega,m,\ell) + \mathfrak R(\omega,m,\ell) \uchr (r^\ast, \omega,m,\ell).
	\end{align}  
	Equivalently, we have
	\begin{align}
		&\mathfrak T (\omega,m,\ell)= \frac{\mathfrak W[\uhr ,\uchr] (\omega,m,\ell) }{\mathfrak{W}[\uchl,\uchr] (\omega,m,\ell)} = \frac{\mathfrak W[\uhr, \uchr](\omega,m,\ell)}{2i(\omega-\omega_- m)}, \\
		&\mathfrak R (\omega,m,\ell)= \frac{\mathfrak W[\uhr ,\uchl] (\omega,m,\ell)}{\mathfrak{W}[\uchr,\uchl] (\omega,m,\ell) } = \frac{\mathfrak W[\uhr ,\uchl] (\omega,m,\ell)}{-2i(\omega-\omega_- m)}.
	\end{align}
	Further, we define the renormalized transmission and reflection coefficient \begin{align}&\mathfrak t(\omega,m,\ell) := (\omega - \omega_- m) \mathfrak T (\omega,m,\ell)=  \frac{1}{2i} \mathfrak W[\uhr, \uchr](\omega,m,\ell),\\
&	\mathfrak r (\omega,m,\ell):= (\omega - \omega_- m) \mathfrak R(\omega,m,\ell) = -\frac{1}{2i} \mathfrak W[\uhr, \uchl](\omega,m,\ell)\end{align}
which satisfy
	\begin{align}
	\mathfrak t^{\omega_-}(m,\ell) = - \mathfrak r^{\omega_-}(m,\ell),
	\end{align}
	where 
	\begin{align}
		\mathfrak t^{\omega_-}(m,\ell) 	:=	\mathfrak t(\omega = \omega_- m, m, \ell ) \text{ and }  \mathfrak r^{\omega_-} (m,\ell):= \mathfrak r(\omega = \omega_- m,m,\ell).
	\end{align}
\end{definition}

\begin{lemma}
	The transmission and reflection coefficients satisfy the Wronskian identity 
\begin{align}
	|\mathfrak T(\omega,m,\ell)|^2 = |\mathfrak R(\omega,m,\ell)|^2 + \frac{\omega-\omega_+m}{\omega-\omega_- m}
\end{align}
for $\omega \in \mathbb R \setminus \{  \omega_- m\} $. 
\begin{proof}
We decompose \begin{align}\uhr & = \mathfrak T \uchl + \mathfrak R \uchr. \label{eq:uhrforwronskian1}\end{align}
Since the potential of the o.d.e.\ \eqref{eq:interiorradial} is real-valued we have that $\uhl = \bar \uhr$. Thus,
\begin{align}\label{eq:uhrforwronskian2} \uhl = \bar{\uhr} = \bar{ \mathfrak T} \bar{\uchl }+\bar{ \mathfrak R} \bar {\uchr }  =\bar{ \mathfrak T}\uchr +\bar{ \mathfrak R} \uchl.  \end{align}
Now, using $\mathfrak W (\uhr, \uhl) = 2 i (\omega-\omega_+ m)$, \eqref{eq:uhrforwronskian1} and \eqref{eq:uhrforwronskian2} yields the result. 
\end{proof}
\end{lemma}
We begin by showing $L^\infty$ estimates for the solutions defined in \eqref{eq:defnu1}. To do so we will consider the cases $|\omega - \omega_r  m| \geq \ecut m$ for all $r\in [r_-,r_+]$. Note that $\ecut >0 $ will be fixed in \cref{eq:eigenvalueestimate} only depending on the black hole parameters. 

\begin{lemma} \label{prop:uniformboundsonu1} Assume that $|\omega - \omega_r m|\geq \ecut m$ for all $r\in [r_-,r_+]$ and for some $\ecut>0$. 
	Then,
	\begin{align}\label{eq:estuhr}
	&	\|\uhr \|_{L^\infty(\mathbb R)} \lesssim 1, 
				\|\uhr'\|_{L^\infty(\mathbb R)} \lesssim |\omega| + | m| + L^\frac 12,\\ \label{eq:estcuhl}
					&	\|\uchl\|_{L^\infty(\mathbb R)} \lesssim 1, 
					\|\uchl '\|_{L^\infty(\mathbb R)} \lesssim |\omega| + | m|  + L^\frac 12 ,\\\label{eq:estcuhr}
						&	\|\uchr\|_{L^\infty(\mathbb R)} \lesssim 1, 
						\|\uchr'\|_{L^\infty(\mathbb R)} \lesssim |\omega| + | m| + L^\frac 12.
	\end{align}
	\begin{proof} We first consider the case that $\omega - \omega_r m \geq \ecut m$ for all $r\in [r_-, r_+]$.

First, we also assume $L^{\frac 12} \leq  |\omega|+|m|$. Then, in view of the assumptions, the principal part of the potential  $V_\sharp$ satisfies \begin{align}-  V_\sharp \gtrsim m^2 + \omega^2 \text{ and }  \left|\frac{ V_\sharp'}{ V_\sharp}\right|, \left|\frac{ V_\sharp''}{ V_\sharp} \right| \lesssim  |\Delta|
\end{align}
and the error term satisfies $|V_1|\lesssim |\Delta|$. Thus, the error control function 
\begin{align}
F_{\uhr_1} (r^\ast):= \int_{-\infty}^{r^\ast} \frac{1}{| V_\sharp|^{\frac 14}} \frac{\d^2}{\d y^2}\left(|  V_\sharp|^{-\frac 14}\right) - \frac{V_1}{|  V_\sharp|^{\frac 12}} \d y
\end{align}
satisfies $ \mathcal{V}_{-\infty,\infty}(F_{\uhr_1} ) \lesssim \frac 1m.$ 

In the case $L^\frac 12 \geq |\omega|+|m|$ we have  
\begin{align}-  V_\sharp \gtrsim |\Delta| L +  m^2 + \omega^2 \text{ and }  \left|\frac{ V_\sharp'}{ V_\sharp}\right|, \left|\frac{ V_\sharp''}{ V_\sharp} \right| \lesssim  \frac{|\Delta|L }{|\Delta| L + m^2 + \omega^2} \label{eq:estimatesonvsharp}
\end{align}
uniformly for $r^\ast \in \mathbb R$. Making use of \eqref{eq:estimatesonvsharp}, we estimate the total variation of $F_{\uhr_1}$ in this frequency range as 
\begin{align}\nonumber 
 \mathcal{V}_{-\infty,\infty}(F_{\uhr_1} )  & \lesssim \int_{\mathbb R} \frac{|V_{\sharp}'|^2}{|V_\sharp|^{\frac 52}} + \frac{|V_{\sharp}''|}{|V_\sharp|^{\frac 32}} +  \frac{|V_1|}{|V_\sharp|^{\frac 12}} \d y\lesssim  \int_{\mathbb R}  \frac{|\Delta| L }{(|\Delta| L + m^2 + |\omega|^2)^{\frac 32}}\d y + \frac{1}{m}
  \\ \nonumber & \lesssim \int_{r_-}^{r_+} \frac{L}{( |r-r_-| |r- r_+| L + m^2 +  |\omega|^2 )^{\frac 32}} \d r + \frac{1}{m}
\\ &  \nonumber  \lesssim L^{-\frac 12}  \int_{r_-}^{r_+} \frac{1}{( |r-r_-| |r- r_+|  + m^2L^{-1} +  |\omega|^2 L^{-1}  )^{\frac 32}} \d r + \frac{1}{m}
\\ &   \lesssim L^{-\frac 12} (m^2 L^{-1} + \omega^2 L^{-1})^{- \frac 12}  + \frac{1}{m}  \lesssim \frac 1m. 
\end{align}
Here, we used the definition of $r^\ast$ in \eqref{defnofrastinsec8} as well as $|\Delta| \sim |r -r_-||r - r_+|$ uniformly for $r\in (r_-, r_+)$. 

Thus, in both of the above cases, $L^{\frac 12 }\leq |\omega| + |m|$ and $L^{\frac 12} \geq |\omega| + |m|$, the above allows us to apply standard estimates on WKB approximation such as \cite[Chapter~6, Theorem~2.2]{olver} and deduce that 
\begin{align}
	\uhr = A_\uhr \frac{|V_\sharp (-\infty)|^{\frac 14}}{| V_\sharp(r^\ast)|^{\frac 14} } \exp \left(-i  \int_{0}^{r^\ast} |  V_\sharp(y)|^{\frac 12} \d y  \right)  \left( 1 + \epsilon_{\uhr} (r^\ast) \right),
\end{align}
for some $A_\uhr$ with $|A_\uhr| =1$. Moreover, 
\begin{align}\sup_{r^\ast \in \mathbb R}|\epsilon_{\uhr}(r^\ast) |\lesssim \frac{1}{m}, \;\; \sup_{r^\ast \in \mathbb R}\left|\frac{\epsilon'_{\uhr}(r^\ast)}{|V_\sharp|^{\frac 12}} \right|\lesssim \frac 1m \; \text{ and }\; \epsilon_\uhr (-\infty) =\epsilon_\uhr'(-\infty) = 0.
\end{align} 
This shows that 
\begin{align}
	\| \uhr \|_{L^\infty(\mathbb R)} \lesssim 1 \text{ and } 	\| \uhr' \|_{L^\infty(\mathbb R)} \lesssim \| |V_\sharp|^{\frac 12} \|_{L^\infty(\mathbb R)}\lesssim |\omega| + |m| + L^\frac 12.
\end{align}

Similarly, we show that the above holds for $  \omega_r m -\omega\geq \ecut m$. This shows \eqref{eq:estuhr}. The bounds \eqref{eq:estcuhl} and \eqref{eq:estcuhr} are shown completely analogously and their proofs are omitted.
	\end{proof}
\end{lemma}
\begin{lemma}\label{eq:eigenvalueestimate} There exists a constant $\epsilon_{\textup{cut}}>0$   such that the following holds true.
	Assume that $|\omega - \omega_r m|\leq \ecut m$ for some $r \in [r_-,r_+]$, then $	L\gtrsim m^2$.
\end{lemma}
\begin{proof}
	Note that $L $ is larger than the lowest eigenvalue of the operator  $P(a \omega) + a^2 \omega^2 - 2a\Xi \omega m$, where $P$ is as in \eqref{eq:Pwithaomega}. Since
	\begin{align}\nonumber
	P (a \omega) + a^2 \omega^2 - 2a\Xi \omega m =  & - \frac{1}{\sin\theta}\partial_\theta  (\Delta_\theta \sin \theta \partial_\theta \cdot )  \\ & + \frac{1}{\Delta_\theta}\left(m \frac{\Xi}{\sin\theta}-a \omega \sin\theta \right)^2 + 2\frac{a^2}{l^2} \sin^2\theta,
	\end{align}
	it suffices to show that the second term is bounded from below by $O(m^2)$. To do so, let $r\in [r_-,r_+]$ such that $|\omega - \omega_r m | \leq \ecut m$. Then, in view of \begin{align} a\omega_r =\Xi \frac{a^2}{r^2+a^2},\end{align} 
	we conclude
	\begin{align} \nonumber
	\left(m \frac{\Xi}{\sin\theta}-a \omega \sin\theta\right)^2 &= \left(m \frac{\Xi}{\sin\theta}-a \omega_r m \sin\theta +  a (\omega_r m - \omega)\sin\theta\right)^2 \\
	&\geq \frac{m^2 \Xi^2}{\sin^2\theta}\left( 1-    \frac{a^2}{a^2+r^2} \sin^2\theta -  \left|\frac{a}{ \Xi}\frac{\omega_r m - \omega}{m}\sin^2\theta\right| \right)^2 \gtrsim
	m^2
	\end{align}
	for sufficiently small $\ecut>0$. 
\end{proof}
In the rest of the section we will make use of
\begin{definition}  For all frequencies $\omega,m,\ell$ we define \begin{align}&\tilde \uhr := e^{i(\omega - \omega_- m)r^\ast} \uhr,\\&\tilde \uhl := e^{-i(\omega - \omega_- m)r^\ast} \uhl,\\ &\tilde \uchr := e^{-i(\omega - \omega_- m)r^\ast} \uchr,\\ & \tilde \uchl := e^{i(\omega - \omega_- m)r^\ast} \uchl.\end{align}
\end{definition}
\begin{lemma}  \label{prop:u1linftybounds}
Assume that	 $|\omega - \omega_r m|\leq \ecut  m$ for some $r\in [r_-,r_+]$ and assume that $m \in \mathbb N$ is sufficiently large. 
Define \begin{align}&R_1:= -\frac{1}{2\kappa_+}\log(L), \\ & R_2:= \frac{1}{2|\kappa_-|}\log(L).\end{align} Then,
 \begin{align}\label{eq:537}
 &	\|\uhr \|_{L^\infty(-\infty,R_2]} \lesssim 1, 	\|\uhr' \|_{L^\infty(-\infty,R_2]} \lesssim |\omega| + |m|  + L^{\frac12} ,\\
  &	\|\uchl \|_{L^\infty[R_1,\infty)} \lesssim 1, 	\|\uchl' \|_{L^\infty[R_1,\infty)} \lesssim |\omega| + |m| + L^{\frac12},\\ 
  &	\|\uchr \|_{L^\infty[R_1,\infty)} \lesssim 1, 	\|\uchr' \|_{L^\infty[R_1,\infty)} \lesssim |\omega| + |m|  +  L^{\frac12},
	\end{align}
	and
\begin{align}\label{eq:partwuhratw-1}
&	|\partial_\omega \uhr|(R_1) \lesssim \log(L), 	|\partial_\omega \uhr' |(R_1) \lesssim \log(L) (|\omega| + |m| ),\\
&	|\partial_\omega \uchr|(R_2) \lesssim \log(L), 	|\partial_\omega \uchr' |(R_2) \lesssim \log(L) (|\omega| + |m| ), \label{eq:partwuhratw-2}\\ 
&	|\partial_\omega \uchl|(R_2) \lesssim \log(L), 	|\partial_\omega \uchl' |(R_2) \lesssim \log(L) (|\omega| + |m| ).\label{eq:partwuhratw-3}
\end{align}
Moreover, 
\begin{align}\label{eq:partwuhr1}
&\| \partial_{\omega} \tilde \uhr \|_{L^\infty(-\infty, R_1)} \lesssim 1, \;
\| \partial_{\omega} \tilde \uhr ' \|_{L^\infty(-\infty, R_1)} \lesssim 1,\\ \label{eq:partwuhr2}
&\| \partial_{\omega} \tilde \uchr \|_{L^\infty(R_2, \infty)} \lesssim 1, \;
\| \partial_{\omega} \tilde \uchr ' \|_{L^\infty(R_2, \infty)} \lesssim 1,
\\\label{eq:partwuhr3}
&\| \partial_{\omega} \tilde \uchl \|_{L^\infty(R_2, \infty)} \lesssim 1, \;
\| \partial_{\omega} \tilde \uchl ' \|_{L^\infty(R_2, \infty)} \lesssim 1.
\end{align}
	\begin{proof}
	From \cref{eq:eigenvalueestimate} we know that $L \gtrsim m^2$. Now, in view of \cref{rmk:volterra2} we write $\uhr$ as the solution to the Volterra equation
		\begin{align}
			\uhr = e^{-i(\omega - \omega_+ m)r^\ast} + \int_{-\infty}^{r^\ast} { K(r^\ast,y) } (1 + R_1 - y)\bar V(y) \uhr (y)\d y,
		\end{align}
		where the kernel is given by \begin{align}\label{eq:kernel}K(r^\ast,y) = \frac{1}{1 + R_1 - y} \frac{\sin( (\omega-\omega_+ m )(r^\ast -y) )}{\omega-\omega_+ m} \end{align} and $\bar V := V_\sharp + V_1 +  (\omega - \omega_+ m)^2$. For $y \in (-\infty, R_1)$, a direct computation shows 
		\begin{align} \nonumber
		(1 + R_1 - y) |\bar V (y)| \lesssim (1 + R_1 - y) L e^{{2\kappa_+} y}, \; \;\; \int_{-\infty}^{R_1}(1 + R_1 - y) |\bar V(y)|\d y \lesssim 1
		\end{align} 
		and 
		\begin{align}
		\sup_{y\leq r^\ast \leq R_1}|K(r^\ast,y)|\lesssim  1.  \end{align} 
		Standard estimates on Volterra integral equations (apply \cite[Chapter~6, Theorem 10.1]{olver} to the term $\uhr - e^{-i (\omega- \omega_+ m)r^\ast}$) yield
		\begin{align}\label{eq:uhratr11}
		\|\uhr \|_{L^\infty(-\infty,R_1)} & \lesssim 1\\
			\nonumber \|\uhr'\|_{L^\infty(-\infty, R_1)}& \lesssim 1+ |\omega - \omega_+ m| \lesssim m |.\label{eq:uhratr12}
		\end{align}
		Now, for the region $r^\ast \in [R_1,R_2]$ we approximate $\uhr$ with a WKB approximation. To do so we remark that for $r^\ast \in [R_1,R_2]$ we have \begin{align}-  V_\sharp \gtrsim 1 \text{ and }  \left|\frac{ V_\sharp'}{ V_\sharp}\right|, \left|\frac{ V_\sharp''}{ V_\sharp} \right| \lesssim  |\Delta|
		\end{align}
		and the error term satisfies $|V_1|\lesssim |\Delta|$. Thus, the error control function 
		\begin{align}
		F_{\uhr_2} (r^\ast):= \int_{R_1}^{r^\ast} \frac{1}{| V_\sharp|^{\frac 14}} \frac{\d^2}{\d y^2}\left(|  V_\sharp|^{-\frac 14}\right) - \frac{V_1}{|  V_\sharp|^{\frac 12}} \d y
		\end{align}
		is bounded as $\mathcal{V}_{R_1,R_2}(F_{\uhr_2} ) \lesssim 1$.  This allows us to apply \cite[Chapter~6, Theorem~2.2]{olver} to deduce that 
		\begin{align}\nonumber
		\uhr &= A_{\uhr} u_{\textup{WKB}_A} + B_\uhr u_{\textup{WKB}_B} \\ \nonumber
	&	= A_{\uhr} \frac{|V_\sharp (R_1)|^{\frac 14}}{| V_\sharp(r^\ast)|^{\frac 14} } \exp \left(-i  \int_{R_1}^{r^\ast} |  V_\sharp(y)|^{\frac 12} \d y  \right)  \left( 1 + \epsilon_{\uhr_A} (r^\ast) \right)\\ & + B_{\uhr} \frac{|V_\sharp (R_1)|^{\frac 14}}{| V_\sharp(r^\ast)|^{\frac 14} } \exp \left(i  \int_{R_1}^{r^\ast} |  V_\sharp(y)|^{\frac 12} \d y  \right)  \left( 1 + \epsilon_{\uhr_B} (r^\ast) \right),
		\end{align}
		for  \begin{align}&A_{\uhr} = \frac{\mathfrak W (u_{\textup{WKB}_B}, \uhr)}{\mathfrak W (u_{\textup{WKB}_B} , u_{\textup{WKB}_A})} \text{ and } B_\uhr =  \frac{\mathfrak W (u_{\textup{WKB}_A}, \uhr)}{\mathfrak W (u_{\textup{WKB}_A} , u_{\textup{WKB}_B})}.\end{align} Moreover, \begin{align}&\sup_{r^\ast \in [R_1,R_2]}|\epsilon_{\uhr_A}(r^\ast) |\lesssim 1, \\ & \sup_{r^\ast \in [R_1,R_2]}|\epsilon'_{\uhr_A}(r^\ast) |\lesssim \sup_{r^\ast \in [R_1,R_2]}|V_\sharp|^{\frac 12} \lesssim L^\frac 12 + |m|+ |\omega|, \\ & \epsilon_{\uhr_A} (R_1) =\epsilon'_{\uhr_A}(R_1) = 0,\end{align}
		and analogously for $\epsilon_{\uhr_B}$.
		Evaluating the Wronskians at $r^\ast = R_1$, we obtain \begin{align}|A_\uhr|, |B_\uhr| \lesssim 1\end{align} in view of \eqref{eq:uhratr11} and \eqref{eq:uhratr12}. 
				This shows that 
		\begin{align}\label{eq:38531}
		&	\|\uhr \|_{L^\infty (-\infty, R_2)} \lesssim 1,\\
					&\|\uhr'\|_{L^\infty (-\infty, R_2)} \lesssim |\omega| + |m| + L^\frac 12.\label{eq:38532}
		\end{align}
			To show the bound on $\partial_\omega \uhr$ we consider $\tilde \uhr$.  Then, 
			$ \tilde \uhr$ satisfies the Volterra equation
			\begin{align}
				\tilde \uhr = 1 +  \int_{-\infty}^{r^\ast}\frac{ \tilde K(r^\ast,y) }{1 + R_1 - y} (1 + R_1 - y)\bar V(y) \tilde \uhr(y) \d y,
			\end{align}
where $\tilde K (r^\ast,y) =  e^{i (\omega-\omega_+ m)(r^\ast -y)} \frac{\sin( (\omega-\omega_+ m)(r^\ast-y))}{\omega- \omega_+ m  } $. Completely analogous to  before, it follows that \begin{align}
\label{eq:bdsontildeu}
 \| \tilde \uhr \|_{L^\infty(-\infty, R_1)} \lesssim 1 \text{ and } \| \tilde \uhr'\|_{L^\infty(-\infty, R_1)} \lesssim 1.\end{align}
Now $\partial_\omega \tilde \uhr$ solves
\begin{align}\nonumber 
	\partial_{\omega} \tilde \uhr = & \int_{-\infty}^{r^\ast} \left(  {\partial_\omega\tilde K(r^\ast,y)     \bar V(y)+ \tilde K(r^\ast,y)   \partial_\omega \bar V(y) } \right) \tilde \uhr (y) \d y \\
	&+ \int_{-\infty}^{r^\ast} \frac{ \tilde K (r^\ast,y)}{1 +R_1 - y} (1+ R_1 -y)\bar V(y)  \partial_\omega   \tilde \uhr (y) \d y.
\end{align}		
	As $|\partial_\omega \lambda_{m\ell}(a \omega) | \lesssim |m|$ from \cref{lem:partialxlambda}, we conclude that 
	\begin{align}|\partial_\omega \bar V | \lesssim |\Delta| m \text{ and }|\partial_{\omega} \tilde K(r^\ast,y)| \lesssim (r^\ast -y)^2
	\end{align} such that
	\begin{align} \int_{-\infty}^{R_1}\left| \left(  {\partial_\omega\tilde K(r^\ast,y)     \bar V(y)+ \tilde K(r^\ast,y)   \partial_\omega \bar V(y) } \right) \tilde \uhr (y) \right| \d y\lesssim 1. 
	 \end{align}
		 Again, by standard bounds on Volterra integral equations \cite[Chapter~6, \S10]{olver} and using \eqref{eq:38531}, \eqref{eq:38532}, we obtain 
		 \begin{align}
		 	\| \partial_{\omega} \tilde \uhr \|_{L^\infty(-\infty, R_1)} \lesssim 1
		 \end{align}
		 and  
		 \begin{align}
		 		\| \partial_{\omega} \tilde \uhr ' \|_{L^\infty(-\infty, R_1)} \lesssim 1.
		 \end{align}
This shows \eqref{eq:partwuhr1}. Completely analogously we obtain \eqref{eq:partwuhr2} and \eqref{eq:partwuhr3}.	Now, we write
		\begin{align}
		\partial_\omega \uhr =\partial_\omega ( e^{-i(\omega-\omega_+ m)r^\ast} \tilde \uhr) =-i r^\ast   \uhr +  e^{-i(\omega-\omega_+ m)r^\ast} \partial_\omega\tilde \uhr
		\end{align}
		and
		\begin{align}
		\partial_\omega \uhr' =-i \uhr - i r^\ast \uhr' - i(\omega-\omega_+m)e^{-i(\omega-\omega_+m) r^\ast} \partial_\omega \tilde \uhr + e^{-i (\omega-\omega_+m)r^\ast} \partial_\omega \tilde \uhr'.
		\end{align}
		Evaluating this at $r^\ast = R_1$ yields
		\begin{align}|\partial_\omega \uhr |(R_1) \lesssim |R_1| \lesssim \log(L).\end{align} and	 
		 \begin{align}
| \partial_\omega \uhr'|(R_1) \lesssim \log(L) (|m|+|\omega| ).
		 \end{align}
This shows \eqref{eq:partwuhratw-1}.
		The proofs of \eqref{eq:partwuhratw-2} and \eqref{eq:partwuhratw-3} are completely analogous.
	\end{proof}
\end{lemma}
 
\begin{lemma}\label{prop:bdsontransandrefl}
The renormalized transmission and reflection coefficients satisfy 
\begin{align}
2|\mathfrak t| = 	|\mathfrak W[\uhr,\uchr] |\lesssim |m| + |\omega| + L^\frac 12 \label{eq:wronskianbound}\\
2	|\mathfrak r|=|\mathfrak W[\uhr,\uchl] |\lesssim |m| + |\omega| + L^\frac 12 \label{eq:wronskianbound2}
\end{align}
for $m$ sufficiently large and all frequencies $\omega,\ell$. Moreover,
\begin{align} \label{eq:derivativeonwronskian}
&\sup_{ \omega \in (\omega_-m -1, \omega_-m +1) }2	\left|\partial_\omega \mathfrak t \right| 
= \sup_{ \omega \in (\omega_-m -1, \omega_-m +1) }	\left|\partial_\omega \mathfrak W[\uhr, \uchr] \right| \lesssim (|m|  + L^{\frac 12}) \log(L)  \\ &
\sup_{ \omega \in (\omega_-m -1, \omega_-m +1) }2	\left|\partial_\omega \mathfrak r \right| 
= \sup_{ \omega \in (\omega_-m -1, \omega_-m +1) }	\left|\partial_\omega \mathfrak W[\uhr, \uchl] \right| \lesssim (|m | + L^{\frac 12}) \log(L) . \label{eq:derivativeonwronskian2}
\end{align}
and
\begin{align}\label{eq:uniforminr1r2}
&|\mathfrak W[\partial_\omega\uhr, \uchr] (r^\ast) |
+|\mathfrak W[\partial_\omega\uhr, \uchl] (r^\ast) |\lesssim \log(L) (|\omega| + |m| + L^{\frac 12}),\\
& |\mathfrak W[\uhr, \partial_\omega\uchr] (r^\ast) |
+|\mathfrak W[\uhr, \partial_\omega \uchl] (r^\ast) |\lesssim \log(L) (|\omega| + |m| + L^{\frac 12})\label{eq:uniforminr1r22}
\end{align}
 uniformly for $r^\ast \in [R_1,R_2]$ and $|\omega - \omega_- m |\leq 1$, $m$ sufficiently large.
\begin{proof}
The bounds \eqref{eq:wronskianbound} and \eqref{eq:wronskianbound2} follow directly from \cref{prop:uniformboundsonu1} and \cref{prop:u1linftybounds}. To show \eqref{eq:derivativeonwronskian} we assume that $|\omega - \omega_- m|\leq 1$ and   evaluate the Wronskian at $r^\ast =0$:
\begin{align}\nonumber
 \partial_\omega \mathfrak W[\uhr, \uchr] &= \partial_\omega \mathfrak W[\uhr, \uchr] (r^\ast=0) \\ & =  \mathfrak W[\partial_\omega\uhr, \uchr](r^\ast=0) + \mathfrak W[\uhr, \partial_\omega\uchr](r^\ast=0).
\end{align}
Hence, \eqref{eq:derivativeonwronskian} follows from \eqref{eq:uniforminr1r2}. To show \eqref{eq:uniforminr1r2}, we apply the fundamental theorem of calculus  for $R_1 \leq r^\ast \leq R_2$ and obtain
\begin{align}|\mathfrak W[\partial_\omega\uhr, \uchr](r^\ast) | \leq \int_{R_1}^{r^\ast}  |\partial_{r^\ast} \mathfrak W[\partial_\omega\uhr, \uchr]|  \d r^\ast + |\mathfrak W[\partial_\omega\uhr, \uchr] (R_1)|. \end{align}
A direct computation shows 
\begin{align}
\partial_{r^\ast} \mathfrak W[\partial_\omega\uhr, \uchr]  = -  \uhr \uchr \partial_\omega (V_\sharp + V_1).
\end{align}
Thus, in view of   \cref{prop:u1linftybounds} we obtain
\begin{align}
\sup_{r^\ast \in [R_1,R_2]}|\partial_{r^\ast} \mathfrak W[\partial_\omega\uhr, \uchr]  |\lesssim |m|   + |\omega|.
\end{align}
From the proof of \cref{prop:u1linftybounds}  we also have 
\begin{align}|\mathfrak W[\partial_\omega\uhr, \uchr] (R_1)|\lesssim \log (L) (|\omega| + |m| + L^{\frac 12}) \end{align} such that
\begin{align}
\sup_{r^\ast \in [R_1,R_2]}|\mathfrak W[\partial_\omega\uhr, \uchr]  (r^\ast)|\lesssim \log(L) (|\omega| + |m| + L^{\frac 12})
\end{align}
follows. Similarly, we obtain \begin{align}\sup_{r^\ast \in [R_1,R_2]}|\mathfrak W[\uhr, \partial_\omega\uchr]  (r^\ast)|\lesssim \log(L) (|\omega| + |m|)\end{align} leading to \eqref{eq:derivativeonwronskian} and \eqref{eq:uniforminr1r2}.
Completely analogously we obtain \eqref{eq:derivativeonwronskian2} as well as \eqref{eq:uniforminr1r22}.
\end{proof}
\end{lemma}
With the above lemma in hand we conclude 
\begin{lemma}\label{lem:lemmaonpartialwuhr}
	Let $m\in \mathbb N$ be sufficiently large and let $\epsilon>0$ be sufficiently small only depending on the black hole parameters. Then,
	\begin{align}
	&\sup_{|\omega - \omega_+m|\leq \epsilon }\| \partial_{\omega} \tilde \uhr \|_{L^\infty(-\infty, 0)} \lesssim L^\frac 12  \log(L),\\
	&\sup_{|\omega - \omega_-m|\leq \epsilon } \| \partial_{\omega} \tilde \uchr \|_{L^\infty(0, \infty)} \lesssim L^\frac 12 \log(L),	\\
	& \sup_{|\omega - \omega_-m|\leq \epsilon }  \| \partial_{\omega} \tilde \uchl \|_{L^\infty(0, \infty)} \lesssim L^\frac 12 \log(L).
	\end{align}
\begin{proof}	
We again only show the claim for  $\tilde \uhr$ as the other cases are completely analogous. Assume that $|\omega - \omega_+ m| \leq \epsilon$ for some $\epsilon>0$ sufficiently small. In view of \cref{prop:u1linftybounds} it suffices to consider the region $r^\ast \in [R_1,0]$.
Now, note that   
	\begin{align} \nonumber
	\partial_\omega \uhr =& \frac{1}{\mathfrak W [\uchr , \uchl]} \Big(  \uchl \int_{R_1}^{r^\ast} \uchr \uhr \partial_\omega (-V_\sharp - V_1)\\ \nonumber & \;\;\; - \uchr \int_{R_1}^{r^\ast} \uhr \uchl \partial_\omega (-V_\sharp - V_1) \Big) \\ &+ \frac{\mathfrak W[\partial_\omega \uhr, \uchl](R_1)}{\mathfrak W[\uchr, \uchl]} \uchr + \frac{\mathfrak W[\partial_\omega \uhr, \uchr] (R_1)}{\mathfrak W[\uchl, \uchr] } \uchl .
	\end{align}
	Hence, using \cref{prop:u1linftybounds},   \cref{prop:bdsontransandrefl}, \begin{align}\sup_{r^\ast \in [R_1,R_2]} |\partial_\omega (V_\sharp+V_1)|\lesssim |m|,
	 \end{align} 
	 as well as the lower bound $|\mathfrak W[\uchr,\uchl]| \gtrsim |m|$, we obtain 
	\begin{align}
		\sup_{r^\ast \in [R_1, 0]} |\partial_\omega \uhr | \lesssim L^\frac 12 \log (L). 
	\end{align}
	In view of $\tilde \uhr = e^{i(\omega - \omega_+ m)r^\ast} \uhr $ and the chain rule, the claim follows.
	\end{proof}
\end{lemma}
\begin{lemma}\label{lem:lowerboundont}
The renormalized transmission and reflection coefficients satisfy \begin{align}
&|\mathfrak t^{\omega_-}| \gtrsim |m | \text{ and } 
 |\mathfrak r^{\omega_-}| \gtrsim |m|
\end{align}
for all $m$ sufficiently large.
\begin{proof}Throughout the proof we assume that $\omega = \omega_- m$.
As $\uchr = \uchl$ for $\omega = \omega_- m$, it suffices to bound the Wronskian $|\mathfrak W[\uhr,\uchr]|$ from below. To do so, let $\mathfrak A$ and $\mathfrak B$ be the unique coefficients satisfying \begin{align}\uchr = \mathfrak A \uhr + \mathfrak B \uhl.\end{align} From $\uchr = \overline{\uchr}$ it follows that \begin{align}\uchr = 2\mathrm{Re}(\mathfrak A \uhr).\end{align} Now, for $\epsilon >0$ to be chosen later, define \begin{align}R_2^\epsilon:= \frac{1}{2|\kappa_-|} \log (L) + \frac{1}{\epsilon}.\end{align} 
Now, $\uchr -1$ is a solution to the Volterra equation
\begin{align}
	\uchr -1 =  \int_{r^\ast}^\infty \frac{y-r^\ast}{y - R_2^\epsilon } (y - R_2^\epsilon )  \tilde V(y) \left[(\uchr -1) + 1 \right] \d y,
\end{align}
where $\tilde V = V_1+V_\sharp (\omega=\omega_- m)$. We have 
\begin{align}
 \int_{R_2^\epsilon}^\infty (y-R_2^\epsilon) \tilde V (y) \lesssim L e^{- 2 |\kappa_-|R_2^\epsilon} \lesssim e^{- \frac{2 |\kappa_-|}{\epsilon}}. 
\end{align}
Using bounds on solutions to Volterra integral equations as before (see \cite[Chapter~6, \S10]{olver}), we obtain that \begin{align}
\| \uchr - 1\|_{L^\infty (R_2^\epsilon, \infty)} < \frac{1}{2} \end{align}
for $\epsilon>0$ sufficiently small enough.
Thus,
\begin{align}
	\frac{1}{2} < \uchr(R_2^\epsilon) = 2 \mathrm{Re}(\mathfrak A \uhr (R_2^\epsilon)) \lesssim 2 |\mathfrak A| \| \uhr \|_{L^\infty(-\infty, R_2^\epsilon)}.
\end{align}
Note that \eqref{eq:537} also holds if we replace $R_2$ by $R_2^\epsilon$ for some fixed value of $\epsilon>0$. Thus, we conclude that $ |\mathfrak B | = |\mathfrak A| \gtrsim 1$ which shows 
\begin{align}|\mathfrak W [\uhr, \uchr] |\gtrsim (\omega_- - \omega_+) |m |\gtrsim |m|.\end{align} This concludes the proof.
\end{proof}
\end{lemma}
\subsection{Scattering poles: Representation formula for \texorpdfstring{$\psi$}{psi} on the interior}
\begin{prop}\label{prop:represen}
	Let $\psi_0 \in C_c^\infty(\mathcal H_R)$ and assume that $\psi_0$ is only supported on azimuthal modes $m \geq m_0$ for some $m_0$ sufficiently large. Let $\tilde \psi \in C^\infty(\mathcal B)$ be the arising solution of \eqref{eq:wavekerr} with vanishing data on $\mathcal H_L \cup\mathcal B_\mathcal H$ and $\psi\restriction_{\mathcal H_R } = \psi_0$. Then,
	\begin{align}\label{eq:represenationpsi}
		\tilde\psi(v,r,\theta,\tilde \phi_+) = \frac{1}{\sqrt{2\pi (r^2+a^2)}} \sum_{m\ell}  \int_{\mathbb R} e^{- i \omega (v - r^\ast)} e^{i m (\tilde \phi_+ -\omega_+ r^\ast) } S_{m\ell}(a\omega, \cos\theta) \mathcal{F}_{\mathcal{H}} [\psi_0] \uhr  \d \omega ,
	\end{align}
	where $\uhr$ is defined in \eqref{eq:defnu1}
	and 
	\begin{align}
		\mathcal{F}_{\mathcal H} [\psi_0] (\omega,m,\ell) := \frac{\sqrt{r_+^2+a^2}}{\sqrt{2\pi}} \int_{\mathbb S^2} \int_{\mathbb R} \psi_0 (v,\theta,\tilde \phi_+) e^{i \omega v } e^{- i m \tilde \phi_+} S_{m\ell}(a \omega, \cos \theta) \d\sigma_{\mathbb S^2} \d v .
	\end{align}
	Moreover, in $\mathcal B$ we have
	\begin{align}\nonumber
		\tilde \psi(v,r,\theta,\tilde \phi_-) &= 
		  \frac{1}{\sqrt{2\pi ( r^2+a^2)}} \sum_{m\ell}\; \textup{p.v.}\int_{\mathbb R} e^{-i \omega (v-r^\ast)} e^{i m (\tilde \phi_- - \omega_- r^\ast)} \\ \nonumber & \;\;\;\; \;\;\;\; \;\;\;\; \;\;\;\; S_{m\ell} (a\omega, \cos\theta) \mathcal{F}_{\mathcal{H}}[\psi_0] \frac{\mathfrak t(\omega, m,\ell)}{\omega - \omega_- m}  \uchl \d \omega\\& + \frac{1}{\sqrt{2\pi ( r^2+a^2)}} \sum_{m\ell}\; \textup{p.v.}\int_{\mathbb R} e^{-i \omega (v-r^\ast)} e^{i m (\tilde \phi_- - \omega_- r^\ast)} \nonumber \\ & \;\;\;\; \;\;\;\;\;\;\;\; \;\;\;\; S_{m\ell} (a\omega, \cos\theta) \mathcal{F}_{\mathcal{H}}[\psi_0] \frac{\mathfrak r(\omega, m,\ell)}{\omega - \omega_- m}  \uchr \d \omega \label{eq:repre}
	\end{align}
	 	as well as 
		\begin{align}\nonumber
	\tilde	\psi(u,r,\theta, \phi_-^\ast) = & \frac{1}{\sqrt{2\pi ( r^2+a^2)}} \sum_{m\ell}\; \textup{p.v.}\int_{\mathbb R} e^{i \omega (u-r^\ast)} e^{i m ( \phi_-^\ast + \omega_- r^\ast)} \\ \nonumber & \;\;\;\; \;\;\;\; \;\;\;\; \;\;\;\;  S_{m\ell} (a\omega, \cos\theta) \mathcal{F}_{\mathcal{H}}[\psi_0] \frac{\mathfrak t(\omega, m,\ell)}{\omega - \omega_- m}  \uchl \d \omega\\& + \frac{1}{\sqrt{2\pi ( r^2+a^2)}} \sum_{m\ell}\; \textup{p.v.}\int_{\mathbb R} e^{i \omega (u-r^\ast)} e^{i m ( \phi_-^\ast + \omega_- r^\ast)}   \nonumber \\  & \;\;\;\; \;\;\;\; \;\;\;\; \;\;\;\; S_{m\ell} (a\omega, \cos\theta) \mathcal{F}_{\mathcal{H}}[\psi_0] \frac{\mathfrak r(\omega, m,\ell)}{\omega - \omega_- m}  \uchr \d \omega \label{eq:rep2}
		\end{align}
		for $r \leq \frac{r_+ + r_-}{2}$.
	 \begin{proof}
Note that $\mathcal{F}_{\mathcal H}[\psi_0] (\omega,m,\ell)$ is rapidly decaying in $\omega,m,\ell$ and smooth in $\omega$ which follows from the fact that $\psi_0 \in C_c^\infty(\mathcal{H}_R)$. Moreover, using  \cref{prop:uniformboundsonu1} and \cref{prop:u1linftybounds}, we have that the right hand side of \eqref{eq:represenationpsi} is a smooth solution to \eqref{eq:wavekerr} in the interior region $\mathcal{B}$. Now, note that the right hand side of \eqref{eq:represenationpsi} converges to $\psi_0$ as $r\to r_+$ for fixed $v$. Similarly, after a change of coordinates to $(u,r,\theta,\tilde \phi_+)$ we obtain $\tilde \psi $ converges to zero as $r\to r_+$ and $u$ fixed in view of the Riemann--Lebesgue lemma. Thus, \eqref{eq:represenationpsi} follows from the uniqueness of the characteristic problem. 
	 	
In order to show \eqref{eq:repre} we first write the right-hand side of \eqref{eq:represenationpsi} as a principal value integral and then use the definition of the reflection and transmission coefficients from \cref{defn:scatteringcoef} to replace  $\uhr$ with \begin{align}\uhr = \frac{\mathfrak t(\omega,m,\ell)}{\omega-\omega_- m } \uchl + \frac{\mathfrak r(\omega,m,\ell)}{\omega- \omega_- m } \uchr. \end{align}
In order to use linearity of the principal value to write $\tilde \psi$ as a sum of two terms as in \eqref{eq:repre}, it suffices to show that \begin{align} \sum_{m\ell}\textup{p.v.} \int_{\mathbb R} e^{- i \omega (v-r^\ast)} e^{i m (\tilde \phi_+ - \omega_+ r^\ast)} S_{m\ell}(a \omega,\cos\theta) \mathcal F_{\mathcal H}[\psi_0] \frac{\mathfrak t(\omega,m,\ell)}{\omega-\omega_- m} \uchl \d \omega
\label{eq:representationformulatocontrol}\end{align}
converges locally uniformly. Note that the other term with $\mathfrak t(\omega,m,\ell) \uchl$ replaced by $\mathfrak r (\omega,m,\ell) \uchr$ is treated completely analogously.

In the following we will be brief because in the proof of \cref{thm:rough}, where we have to quantitatively control terms of the form \eqref{eq:representationformulatocontrol},  we will indeed show stronger estimates and provide more details.  First, in view of the facts that $\mathcal{F}_{\mathcal H}[\psi_0] (\omega,m,\ell)$ is rapidly decaying  in $\omega, m,\ell$, that  $\| S_{m\ell}(a \omega, \cos \theta) \|_{L^2( (-\pi,\pi); \cos \theta \d \theta)} = 1$, and that  we have uniform (polynomial) bounds on  $|\mathfrak t (\omega,m,\ell)|$ and $\|\uchl\|_{L^\infty([R_1,+\infty))} $ (see \cref{prop:uniformboundsonu1},  \cref{prop:u1linftybounds}, \cref{prop:bdsontransandrefl}), it suffices to consider frequencies in the range $|a \omega- a\omega_- m|< \frac{1}{m}$. Now,  uniformly in $|a \omega- a\omega_- m|< \frac{1}{m}$, we have polynomial bounds in $\omega, m,\ell$ on  $\| \partial_\omega S_{m\ell}(a \omega, \cos \theta)\|_{L^2( (-\pi,\pi); \cos \theta \d \theta)}$,  $\| \partial_\omega \uchl\|_{L^\infty(0,+\infty)}$ and $|\partial_\omega \mathfrak t(\omega,m,\ell)|$ as shown in \cref{eq:smlbounds}, \cref{lem:lemmaonpartialwuhr} and  \cref{prop:bdsontransandrefl}, respectively. Moreover, again using the bound on $\| \partial_\omega S_{m\ell}(a \omega, \cos \theta)\|_{L^2}$ and the fact that $\psi_0$ is compactly supported, we also obtain that $\sup_{|a \omega -a \omega_- m|< \frac{1}{m}} |\partial_\omega \mathcal F_{\mathcal H}[\psi_0]|$ is rapidly decaying in $ m,\ell$. This shows that \eqref{eq:representationformulatocontrol}   converges locally uniformly for $r^\ast \geq 0$, $v \in \mathbb R$ with values in $L^2(\mathbb S^2)$. This shows that, after a change of coordinates, \eqref{eq:repre} and \eqref{eq:rep2} hold true pointwise for $r \leq \frac{r_- + r_+}{2}$, $v\in \mathbb R$ and in $L^2(\mathbb S^2)$. Finally, using standard elliptic estimates and the fundamental theorem of calculus, we also have polynomial bounds in $\omega, m,\ell$  on $\| S_{m\ell}(a \omega, \cos \theta)\|_{L^\infty(-\pi,\pi)}$, as well as polynomial bounds in $m,\ell$  on $ \sup_{|a \omega - a \omega_- m | < \frac 1 m }\| \partial_\omega S_{m\ell}(a \omega, \cos \theta)\|_{L^\infty(-\pi,\pi)}$. Thus, both terms on the right-hand side of  \eqref{eq:repre} are continuous and the equality  \eqref{eq:repre} holds pointwise. We obtain the analogous result for \eqref{eq:rep2}.
\end{proof}
\end{prop} 

 Before we  prove the blow-up result in \cref{sec:mainthmkerr}, we need one more final ingredient which is a consequence of the domain of dependence.

\begin{lemma}\label{lem:domain}
	Let $\tilde \psi \in C^\infty(\mathcal B)$ be a solution to \eqref{eq:wavekerr} arising from vanishing data on $\mathcal H_L\cup \mathcal B_{\mathcal H}$  and compatible smooth  data $\psi_0 \in C^\infty(\mathcal H_R)$.  Then, $\tilde \psi(u_0,r_0,\theta_0,\phi_{-}^\ast)$ only depends on $\psi_0 \restriction_{\{v \leq 2r^\ast(r_0) - u_0 + \tilde c \} }$, where  $\tilde c = \tilde c(\mathfrak p,l) >0$ is  a constant.
	\begin{proof}
	In coordinates $(v,r,\theta,\tilde \phi_-)$ (or equivalently in coordinates $(v,r,\theta,\tilde \phi_+)$) define the function $\tilde v:= v + f(r) $ on $\mathcal B$ and choose $f$ to satisfy 
	\begin{align}\frac{\d f}{\d r} = - \sqrt{\frac{a^2}{\Xi} \frac{1}{|\Delta|}}\end{align}
	with initial condition $f(r_+)=0$. This is well defined as $\frac{1}{\sqrt{|\Delta|}}$ is integrable at the event and Cauchy horizons. Now $f$ is non-negative and satisfies $\sup_{r\in (r_-, r_+)} f\leq \tilde c$ for a constant $\tilde c  >0$ only depending on the black hole parameters. A computation also shows that, uniformly on $\mathcal B$, we have \begin{align}g_\textup{KAdS}(\nabla \tilde v, \nabla \tilde v) = \frac{a^2 \sin^2\theta}{\Sigma \Delta_\theta} - \frac{a^2 }{\Sigma \Xi} < 0 \text{ and } g_{\textup{KAdS}}(\nabla \tilde v, -\nabla r) <0 .\end{align} This means that $\nabla \tilde v$ is a future-directed timelike vector field. Thus,  the level sets of the function $\tilde v$ are spacelike.
	
	Now, consider \begin{align}\label{eq:domainofdep} \tilde \psi(u_0,r_0,\theta_0,  \phi_{-}^\ast).  \end{align}
 Since $\nabla \tilde v$ is future directed and timelike, it follows from the domain of dependence   that  \eqref{eq:domainofdep} only depends on \begin{align} {\psi_0}\restriction_{\{ \tilde v(v,r_+) \leq \tilde v( v(r_0,u_0),r_0) \}} = {\psi_0}\restriction_{ \{ v  \leq 2 r^\ast(r_0) - u_0 + f(r_0)\} },
		\end{align}
		since $\tilde v( v(r_0,u_0),r_0) =  2 r^\ast(r_0) - u_0 + f(r_0) $ .
		This concludes the proof.
	\end{proof}
\end{lemma}
We will now finally turn to the proof of \cref{thm:rough}.
\section{Proof of \texorpdfstring{\cref{thm:rough}}{Theorem~1}: Small divisors lead to blow-up}\label{sec:mainthmkerr}
We recall that the cosmological constant $\Lambda <0$ (and thus $l = \sqrt{-3/\Lambda}>0$) was arbitrary but fixed as in \eqref{eq:defnofl}.
\begin{thmrough}
	\cref{con:2a} holds true.
	
More precisely, let the dimensionless  black hole parameters $(\mathfrak m, \mathfrak a) \in \mathscr P_\textup{Blow-up}$ be arbitrary but fixed as in \eqref{eq:fixedp}, where   $\mathscr P_\textup{Blow-up}$ is defined in \cref{eq:defnpblowup}.

	Let $\psi \in C^\infty(\mathcal{M}_{\mathrm{KAdS}}\setminus \mathcal{CH} )$ be the unique solution to \eqref{eq:wavekerr} arising from the smooth and compactly supported initial data   specified in \cref{defninitialdata} on Kerr--AdS with parameters $(M,a) =( \mathfrak m/ \sqrt{-\Lambda}, \mathfrak a /\sqrt{-\Lambda})$.  
	
	Then,  for each $u_0 \in \mathbb R$, the solution $\psi$ blows up at the Cauchy horizon $\mathcal{CH}_R$ as
\begin{align}\label{eq:mainresult}
 \lim_{r \to r_-} \| \psi(u_0, r) \|_{L^2(\mathbb S^2)}^2 = +\infty.
	\end{align}
Moreover,  $\mathscr P_{\textup{Blow-up}} \subset \mathscr P$	has the following properties:
	\begin{itemize}
	\item  $\mathscr{P}_{\textup{Blow-up}}$ is Baire-generic,
	\item  $\mathscr{P}_{\textup{Blow-up}}$ is  Lebesgue-exceptional ($\mathscr{P}_{\textup{Blow-up}}$ has zero Lebesgue measure),
	\item  $\mathscr{P}_{\textup{Blow-up}}$ has full packing dimension $\dim_P (\mathscr{P}_{\textup{Blow-up}}) =2$.
\end{itemize} 
\end{thmrough}
From our proof we will also obtain the following corollary which gives a genericity condition for compactly supported initial data which lead to blow-up as in \eqref{eq:mainresult}.
 \begin{cor}\label{cor:genericity}
Let the dimensionless  black hole parameters $(\mathfrak m, \mathfrak a) \in \mathscr P_\textup{Blow-up}$ be arbitrary but fixed as in \eqref{eq:fixedp}. Let $\tilde \Psi_0, \tilde \Psi_1 \in C_c^\infty(\Sigma_0)$ be arbitrary initial data satisfying the following genericity condition
 \begin{align} \label{eq:genericitycond}
 \sum_{i\in \mathbb N} |m_i e^{ \sqrt{ |m_i|}}|^2 |G(\tilde \Psi_0,\tilde \Psi_1,m_i,\ell_i)|^2 = +\infty,
 \end{align}
 where  $m_i,\ell_i$ are the subsequences in \eqref{eq:mili} associated to the non-Diophantine condition (i.e.\ the choice of $\mathfrak p \in \mathscr{P}_\textup{Blow-up}$)
 and
 \begin{align}
 \nonumber
 G(\tilde \Psi_0,\tilde \Psi_1,m,\ell):= &\int_{r_+}^{\infty}\int_{\mathbb S^2}	\bigg\{ \frac{\Sigma}{\sqrt{r^2+a^2}}  u_\infty(\omega=\omega_- m,m,\ell)    e^{-im\phi}    S_{m\ell}(a \omega_- m,\cos\theta) \\
 &\times \left( -2 \sqrt{-g^{tt}}  \tilde \Psi_1(r,\theta,\phi) - i \omega_- m g^{tt} \tilde \Psi_0(r,\theta,\phi)  \right)\bigg\} \d \sigma_{\mathbb S^2} \d r.
 \end{align}   
 Then, the arising  solution $\tilde \psi\in C^\infty(\mathcal M_{\textup{KAdS}}\setminus\mathcal{CH})$ to \eqref{eq:wavekerr} with $(\tilde \psi\restriction_{\Sigma_0}, n_{\Sigma_0} \tilde \psi\restriction_{\Sigma_0}) = (\tilde \Psi_0,\tilde \Psi_1)$, vanishing incoming data on $\mathcal H_L\cup \mathcal B_\mathcal H$ and with Dirichlet boundary conditions imposed at infinity,  blows up at  the Cauchy horizon $\mathcal{CH}_R$ for every $u_0 \in \mathbb R$ as
 \begin{align}
 \lim_{r \to r_-} \| \tilde \psi(u_0, r) \|_{L^2(\mathbb S^2)}^2 = +\infty.
 \end{align}
  \end{cor}
\cref{cor:genericity} will be an immediate consequence of the proof of \cref{thm:rough} and will be given thereafter. Also note that the initial data which we construct in \cref{sec:initialdata} do indeed satisfy the genericity condition of \eqref{eq:genericitycond} as shown in \cref{eq:defnah}.

	\begin{proof}[Proof of Theorem~3.1]
	The stated properties of $\mathscr P_{\textup{Blow-up}}$ on the Baire-genericity, the zero Lebesgue measure and the full packing dimension follow from \cref{thm:Baire}, \cref{prop:lebesgue} and \cref{prop:packingdimension}, respectively.  
		
We now turn to the proof of \eqref{eq:mainresult}. First, we write $\psi_0 := \psi\restriction_{\mathcal H_R}$ and note that
		\begin{align}
		D:=  \sum_{0 \leq  i + j \leq 4} \int_{\mathbb R \times \mathbb S^2} |\slashed \nabla^i K_+^j  \psi_0(v,\theta,\tilde \phi_+)|^2  \d\sigma_{\mathbb S^2} \d v  <\infty
		\end{align}
		in view of \cref{thm:wellposedanddecay} and \cref{eq:estimateonl2norm}.
Now, let $u_0 \in \mathbb R$ be fixed and let $r^\ast_n \to \infty $ be a sequence with $r_n^\ast>r_0^\ast$ for sufficiently large $r_0^\ast$. 
We will first prove
\begin{prop}\label{prop:mainpropinmainthm}
 For all $r_n^\ast \geq r_0$, we have that
		\begin{align}\label{eq:formulaforlimit}
		\| \psi(u_0, r^\ast_n) \|_{L^2(\mathbb S^2)}^2 =  \sum_{m\ell} \left|\pi  \frac{\mathfrak r^{\omega_-}(m,\ell) \tilde \uchl^{\omega_-}   (r_n^\ast,m,\ell) }{\sqrt{r_n^2 + a^2}}   a_{\mathcal H}^{R_n} (\omega = \omega_- m,m,\ell)\right|^2 + \textup{Err}(D),
		\end{align}
		where $|\textup{Err}(D)| \lesssim_{u_0} D$ uniformly for all $r_n^\ast \geq r_0^\ast$ and $R_n:= 2r_n^\ast - u_0 + \tilde c$. 	Also recall the definition of $a_{\mathcal H}^R$ in \eqref{eq:defnar}. Here we also use the notation $\tilde \uchl^{\omega_-}(r_n^\ast,m,\ell) := \tilde \uchl(r_n^\ast, \omega = \omega_- m,m,\ell)$. 
\end{prop} 
	 Once we have established \eqref{eq:formulaforlimit}, the blow-up result of \eqref{eq:mainresult} will be proved. We now turn to the proof of  \cref{prop:mainpropinmainthm}.	
\begin{proof}[Proof of \cref{prop:mainpropinmainthm}]
  	In view of the domain of dependence property as stated in \cref{lem:domain}, we have that  $\psi(u_0, r_n^\ast, \theta,  \phi_-^\ast)$ only depends on $\psi_0 \restriction_{\{ v \leq 2 r_n^\ast - u_0 + \tilde c\} }$. Consider now \begin{align}\psi_0^n(v,\theta,\tilde \phi_+):=\psi_0^{R_n}(v,\theta,\tilde \phi_+),\end{align}
		  where $  \psi_0^{R_n}(v,\theta,\tilde \phi_+) = \psi_0(v,\theta,\tilde \phi_+) \chi(R_n-v)$ is defined in \eqref{defn:psiRepsilon} with $R_n= 2r_n^\ast - u_0 + \tilde c$. 	Now, $\psi(u_0,r_n^\ast,\theta, \phi_-^\ast) $ only depends on $\psi_0^n$. 
		
		Using the representation formula \eqref{eq:rep2} in \cref{prop:represen} we write
		 \begin{align} \nonumber
			\psi(u_0,r_n^\ast,\theta,  \phi_-^\ast)   = \frac{1}{\sqrt{2\pi ( r_n^2+a^2)}} \sum_{m\ell}\; \textup{p.v.}\int_{\mathbb R} & e^{i \omega (u_0-r_n^\ast)} e^{i m (\phi_-^\ast + \omega_- r_n^\ast)} \\  \nonumber&\times  S_{m\ell} (a\omega, \cos\theta) \mathcal{F}_{\mathcal{H}}[\psi_0^n] \frac{\mathfrak t(\omega, m,\ell)}{\omega - \omega_- m}  \uchl \d \omega \\ \nonumber + \frac{1}{\sqrt{2\pi ( r_n^2+a^2)}} \sum_{m\ell}\; \textup{p.v.}\int_{\mathbb R} & e^{i \omega (u_0-r_n^\ast)} e^{i m (\phi_-^\ast+ \omega_- r_n^\ast)} \\ \nonumber &\times  S_{m\ell} (a\omega, \cos\theta) \mathcal{F}_{\mathcal{H}}[\psi_0^n] \frac{\mathfrak r(\omega, m,\ell)}{\omega - \omega_- m}  \uchr \d \omega\\
			& =: I + II . \label{eq:TERMI}
		\end{align}
		We   consider both terms individually and start with the term $I$.
		Moreover, we split the term $I$ into $| a\omega-a\omega_- m|< \frac 1m$ and $|a\omega - a\omega_- m| \geq \frac 1m$ and call the terms $I_{\textup{res}}$ and $I_\textup{non-res}$, respectively, such that $I= I_{\textup{res}} + I_\textup{non-res}$.
		First, we claim that the spherical $L^2$-norm of the term \begin{align}\nonumber
			I_\textup{non-res} = \frac{1}{\sqrt{2\pi ( r_n^2+a^2)}} \sum_{m\ell}\; \textup{p.v.}\int_{ |\omega_- m - \omega|\geq \frac{1}{am} } & e^{i \omega (u_0-r_n^\ast)} e^{i m (\phi_-^\ast + \omega_- r_n^\ast)}  \\ &\times  S_{m\ell} (a\omega, \cos\theta) \mathcal{F}_{\mathcal{H}}[\psi_0^n] \frac{\mathfrak t(\omega, m,\ell)}{\omega - \omega_- m}  \uchl \d \omega
\label{eq:nonresestimates}		\end{align}is controlled by $D$ uniformly as $r^\ast_n\to\infty$. 
		\begin{lemma}
		We have $ \| 	I_\textup{non-res}\|^2_{L^2(\mathbb S^2)} (r_n^\ast,u_0) \lesssim  D$ for all $r_n^\ast \geq  r_0^\ast$.
		\begin{proof}
Using $|\frac{1}{\omega - \omega_- m}| \leq am$ in the integrand of  \eqref{eq:nonresestimates} and $\int_0^{2\pi} e^{i(m-\tilde m ) \phi} \d \phi =2\pi \delta_{m\tilde m}$ we estimate 
 \begin{align}
\|I_\textup{non-res} \|_{L^2(\mathbb S^2)}^2\lesssim  \sum_{m} m^2
\int_{0}^{\pi} \left|\;\sum_{\ell\geq |m|} \int_{ \mathbb R}|S_{m\ell} (a\omega, \cos\theta) \mathcal{F}_{\mathcal{H}}[\psi_0^n]\,  \mathfrak t \,   \uchl| \d \omega \right|^2 \sin \theta \d \theta.
\end{align}
From the Cauchy--Schwarz inequality as well as \cref{prop:uniformboundsonu1},  \cref{prop:u1linftybounds} and \cref{prop:bdsontransandrefl}, we obtain
\begin{align}\nonumber
\|I_\textup{non-res} \|_{L^2(\mathbb S^2)}^2 \lesssim &\sum_{m}\Big[ \int_{0}^{\pi } \sum_{\tilde \ell\geq |m|}  \int_{\mathbb R} \frac{|S_{m\tilde \ell} (a \omega ,\cos \theta) |^2}{(1+\omega^2)(1+  \Lambda_{m \tilde \ell} ) }  \d \omega \sin \theta \d\theta  \\  &\times  \sum_{  \ell\geq |m|} \int_{\mathbb R} (1+\omega^2)(1+  \Lambda_{m  \ell}  )m^2(1+\omega^2+  \Lambda_{m  \ell}) |\mathcal F_{\mathcal H} [\psi_0^n] |^2\d \omega\Big] \nonumber \\
& \lesssim \sum_{m}   \sum_{  \ell\geq |m|} \int_{\mathbb R} (1+\omega^2)(1+  \Lambda_{m  \ell}  )m(1+\omega^2+  \Lambda_{m  \ell}) |\mathcal F_{\mathcal H} [\psi_0^n] |^2\d \omega \lesssim D,
\end{align}
where we have used that $\Lambda_{m\ell}\geq \Xi^2 \ell (\ell+1)$ such that $\sum_{\tilde \ell \geq |m|} \frac{1}{1+  \Lambda_{m\tilde \ell}  } \lesssim \frac{1}{m}$.
		\end{proof}
		\end{lemma}
		Now, we turn to the term $I_\textup{res}$:
		\begin{align}\nonumber
I_\textup{res} =	 \frac{1}{\sqrt{2\pi ( r_n^2+a^2)}} \sum_{m\ell}  \textup{p.v.}\int_{\omega_- m - \frac{1}{am}}^{\omega_- m + \frac{1}{am} }& e^{i \omega (u_0-r_n^\ast)} e^{i m (\phi_-^\ast + \omega_- r_n^\ast)} \\ & \times S_{m\ell} (a\omega, \cos\theta) \mathcal{F}_{\mathcal{H}}[\psi_0^n] \frac{\mathfrak t(\omega, m,\ell)}{\omega - \omega_- m}  \uchl \d \omega
		\end{align}
	and write $\uchl = e^{-i{ (\omega - \omega_- m)}r^\ast} \tilde \uchl $.
	 Then,
		\begin{align} \nonumber
I_\textup{res} =I_\textup{res}^a + I_{\textup{res}}^b = & \frac{1}{\sqrt{2\pi ( r_n^2+a^2)}} \sum_{m\ell} \textup{p.v.}\int_{\omega_-m -\frac{1}{am}}^{\omega_- m + \frac{1}{am}}\frac{ e^{-2 i (\omega - \omega_- m ) r_n^\ast }e^{i\omega  u_0}\mathcal F_{\mathcal H}[\psi_0^n] }{\omega-\omega_- m} \d \omega \\\nonumber & \;\;\;\;\; \times e^{im \phi_-^\ast} S_{m\ell}(a\omega_- m,\cos \theta)  \mathfrak t(\omega_-m,m,\ell)\tilde \uchl^{\omega_-} \\ \nonumber&+ \frac{1}{\sqrt{2\pi ( r_n^2+a^2)}} \sum_{m\ell} \int_{\omega_-m -\frac{1}{am}}^{\omega_- m + \frac{1}{am}} e^{-2 i (\omega - \omega_- m ) r_n^\ast } e^{i\omega u_0} \mathcal F_{\mathcal H}[\psi_0^n]  e^{im \phi_-^\ast} \\ & \;\;\;\;\; \times \Big[  S_{m\ell}(a\omega_- m,\cos \theta) \partial_\omega\left(  \mathfrak t(\omega,m,\ell)\tilde \uchl \right) (\tilde \xi) \nonumber \\  & \;\;\;\;\;\;\;  + \mathfrak t(\omega,m,\ell)\tilde \uchl  \frac{ S_{m\ell}(a\omega,\cos \theta) - S_{m\ell}(a\omega_- m,\cos \theta)  }{\omega - \omega_- m}  \Big ]\d \omega  \label{eq:ires}
		\end{align}
		for some $\tilde \xi = \tilde \xi ( \omega) \in (\omega_- m - \frac{1}{am} , \omega_- m + \frac{1}{am})$ in view of the mean value theorem.
		Again, we consider both terms $I_{\textup{res}}^a$ and $I_{\textup{res}}^b$ individually and begin with term $I_\textup{res}^b$. 
		\begin{lemma}\label{lem:inproof1}
					We have $ \| 	I_\textup{res}^b\|^2_{L^2(\mathbb S^2)} (r_n^\ast,u_0) \lesssim  D$ for all $r_n^\ast \geq r_0^\ast$.
					\begin{proof}
		We decompose the term $I_\textup{res}^b = I_\textup{res}^{b1} + I_\textup{res}^{b2}$ further into the two summands appearing in the $\omega$-integral. We will estimate each of them individually. We begin with $I_\textup{res}^{b1}$ and estimate 
		\begin{align}\nonumber& \| 	I_\textup{res}^{b1}\|^2_{L^2(\mathbb S^2)}
						  \lesssim \sum_{m}  \int_0^\pi \left|\sum_{\ell\geq |m|} \int_{\omega_- m - \frac{1}{am}}^{\omega_- m + \frac{1}{am}} |\mathcal F_{\mathcal H}[\psi_0^n] S_{m\ell} (a \omega_- m) \partial_\omega \left( \mathfrak t \tilde \uchl\right)  (\tilde \xi) | \d \omega \right|^2 \sin \theta \d \theta 
						  \\ \nonumber & \lesssim 
						\sum_{m}    \left(\sum_{\ell \geq |m|} \int_0^\pi \int_{\mathbb R} (1+\Lambda_{m\ell}^{3})|\mathcal F_{\mathcal H}[\psi_0^n] |^2  |  S_{m\ell } (a \omega_- m )|^2 \d \omega   \sin \theta \d \theta  \right)\\ \nonumber &\;\;\;\;\; \times  \left( \frac 1 m \sum_{\ell \geq |m|}\sup_{|\tilde \xi- a\omega_-m|< \frac 1m}\frac{ |  \partial_\omega( \mathfrak t \tilde \uchl)(\tilde \xi) |^2}{1+\Lambda_{m\ell}^{3}(\tilde \xi)} \right) 
						 \\ & \lesssim  \sum_{m} \left( \sum_{\ell \geq |m|} \int_{\mathbb R} (1+\Lambda_{m\ell}^{3})|\mathcal F_{\mathcal H}[\psi_0^n] |^2 \d \omega\right) \left(  \sum_{\ell \geq |m| } \frac{  \Lambda_{m\ell}^{2}( a\omega_-m)   \log^2( \Lambda_{m\ell}( a\omega_-m))  }{1+\Lambda_{m\ell}^{3}( a\omega_-m) }  \right) \lesssim D.
						\end{align}
						Here we have used \cref{prop:u1linftybounds},  \cref{prop:bdsontransandrefl}, \cref{lem:lemmaonpartialwuhr} and the fact that \begin{align}\Lambda_{\omega_- ,m\ell}  :=\Lambda_{m\ell}(a\omega_- m) \sim  \Lambda_{m\ell} (a\xi)  \end{align} for all $|\xi - \omega_- m|< \frac{1}{m}$ which in turn is a consequence of \cref{lem:partialxlambda}.
						
						We now control the second term $I_\textup{res}^{b2}$ and estimate
								\begin{align}\nonumber& \| 	I_\textup{res}^{b2}\|^2_{L^2(\mathbb S^2)}
						\lesssim \sum_{m}  \int_0^\pi \left|\sum_{\ell\geq |m|} \int_{\omega_- m - \frac{1}{am}}^{\omega_- m + \frac{1}{am}} \left|\mathcal F_{\mathcal H}[\psi_0^n]\frac{ S_{m\ell} (a \omega) - S_{m\ell}(a \omega_- m)}{\omega - \omega_- m} \mathfrak t \tilde \uchl \right| \d \omega \right|^2 \sin \theta \d \theta 
						\\ \nonumber & \lesssim 
						\sum_{m}    \left(\sum_{\ell \geq |m|} \int_{\mathbb R} (1+\Lambda_{m\ell}^{3})|\mathcal F_{\mathcal H}[\psi_0^n] |^2    \d \omega \right)  \\ \nonumber &\;\;\;\;\; \times  \left( \sum_{\ell \geq |m|}  \int_{\omega_- m - \frac{1}{am}}^{\omega_- m + \frac{1}{am}}  \frac{ |   \mathfrak t \tilde \uchl  |^2}{1+\Lambda_{m\ell}^{3}} \int_0^\pi \left| \frac{ S_{m\ell} (a \omega) - S_{m\ell}(a \omega_- m)}{\omega - \omega_- m} \right|^2 \sin \theta \d \theta \d \omega  \right)
						\\ & \lesssim  \sum_{m} \left( \sum_{\ell \geq |m|} \int_{\mathbb R} (1+\Lambda_{m\ell}^{3})|\mathcal F_{\mathcal H}[\psi_0^n] |^2 \d \omega\right)  \nonumber \\  &\;\;\;\;\; \times \left(  \sum_{\ell \geq |m| } \frac{  \Lambda_{\omega_- ,m\ell}     }{1+\Lambda_{\omega_- ,m\ell}^{3} } \int_{\omega_- m - \frac{1}{am}}^{\omega_- m + \frac{1}{am}} \sup_{|\tilde \xi - \omega_- m|\leq \frac 1m} \int_0^\pi |\partial_\omega S_{m\ell}|^2(a \tilde \xi)  \sin \theta \d \theta \d \omega \right) \lesssim D,
						\end{align}
				where we have used the mean value property for Fr\'echet derivatives, \cref{prop:u1linftybounds}, \cref{prop:bdsontransandrefl} and \cref{eq:smlbounds}.
					\end{proof}
		\end{lemma}

		Now, we proceed with $I_\textup{res}^a$, i.e.\ the first term in \eqref{eq:ires}. 
We begin by recalling the definition of $\mathcal{F}_{\mathcal H}[\psi_0^n]$:  \begin{align}\label{eq:replacesml}\mathcal{F}_{\mathcal H}[\psi_0^n] = \frac{\sqrt{r_+^2 + a^2}}{\sqrt{2\pi  } } \int_{\mathbb R} \int_{\mathbb S^2} e^{i\omega v} \psi_0^n(v,\theta,\tilde \phi_+) S_{m\ell}(a \omega, \cos \theta) e^{-im\tilde \phi_+} \d v \d \sigma_{\mathbb S^2}. \end{align}
Similar to \cref{lem:inproof1} we will replace the $S_{m\ell}(a\omega)$ appearing in \eqref{eq:replacesml} with $S_{m\ell}(a \omega_- m)$. In order to do so, we introduce 
 \begin{align}\nonumber
\hat I_\textup{res}^a := \frac{1}{\sqrt{2\pi ( r_n^2+a^2)}} &  \sum_{m\ell} \textup{p.v.}\int_{\omega_-m -\frac{1}{am}}^{\omega_- m + \frac{1}{am}} \frac{ e^{-2 i (\omega - \omega_- m ) r_n^\ast }e^{i\omega  u_0} \tilde {\mathcal F_{\mathcal H}}[\psi_0^n] }{\omega-\omega_- m} \d \omega \\  & \;\;\;\;\; \times e^{im\phi_-^\ast} S_{m\ell}(am\omega_-,\cos \theta)  \mathfrak t(\omega_-m,m,\ell)\tilde \uchl^{\omega_-},
\end{align}
\begin{align}\tilde {\mathcal{F}_{\mathcal H}}[\psi_0^n] := \frac{\sqrt{r_+^2 + a^2}}{\sqrt{2\pi}} \int_{\mathbb R} {\psi_0^n}_{m\ell} (v) e^{i \omega v}  \d v = \sqrt{r_+^2 + a^2} \mathfrak{F} [{\psi_0^n}_{m\ell}],\end{align}
and\footnote{Recall that $\mathfrak F$ denotes the standard Fourier transform $\mathfrak F[f](\xi) := \frac{1}{\sqrt{2\pi}} \int_{\mathbb R} f(x) e^{i \xi x}  \d x$.}
\begin{align}
{	\psi_0^n}_{m\ell} (v):= \int_{\mathbb S^2}\psi_0^n(v,\theta,\tilde \phi_+)    S_{m\ell}(a\omega_- m, \cos \theta) e^{-im\tilde \phi_+} \d \sigma_{\mathbb S^2}.
\end{align}
\begin{lemma}
 \begin{align}
\| \hat I_\textup{res}^a - I_\textup{res}^a \|^2_{L^2(\mathbb S^2)} \lesssim D,
\end{align}
\begin{proof}
	Similarly to the proof of \cref{lem:inproof1}, we write
		\begin{align}\label{eq:smalw}
			S_{m\ell}(a\omega) = S_{m\ell}(a\omega_-m) + (\omega - \omega_- m)\frac{ S_{m\ell}(a \omega) - S_{m\ell}(a\omega_-m) }{\omega - \omega_- m}.
		\end{align} 
		 for frequencies $|\omega - \omega_- m|\leq \frac{1}{am}$ in \eqref{eq:replacesml}. Then, using a Cauchy--Schwarz inequality on the sphere, $\sup_{|\xi - a \omega_- m |\leq  \frac{ 1}{m}} |\partial_\xi \Lambda_{m\ell}(\xi) |\lesssim |m|$, $\sup_{|\xi - a \omega_- m |\leq  \frac{ 1}{m}} |\partial_\xi P(\xi) |\lesssim  |m|$ (see \eqref{eq:partialxofp}), \cref{eq:smlbounds} as well as elliptic estimates, we control the error term as 
		 \begin{align}	\nonumber 
(1+   m^2 \Lambda^2_{\omega_- ,m\ell} )&	\left|	\frac{ 	\mathcal{F}_{\mathcal H} - \tilde {  \mathcal{F}}_{\mathcal H}}{\omega - \omega_- m} \right|^2\lesssim  m  \Big[ \int_{\mathbb S^2}  \left|\int_{\mathbb R} e^{i\omega v}   \psi_0^n(v,\theta,\tilde\phi_+) \d v\right|^2 \d \sigma_{\mathbb S^2}  \\ & +  \int_{\mathbb S^2}  \left|\int_{\mathbb R} e^{i\omega v}   \slashed\nabla^3\psi_0^n(v,\theta,\tilde\phi_+) \d v\right|^2 \d \sigma_{\mathbb S^2}\Big].
		 \end{align}
Now, from \cref{prop:u1linftybounds} and \cref{prop:bdsontransandrefl} we conclude after an application of the Cauchy--Schwarz inequality and Plancherel's theorem  that
\begin{align}	\| \hat I_\textup{res}^a - I_\textup{res}^a \|^2_{L^2(\mathbb S^2)}   \lesssim D.
\end{align}

\end{proof}
\end{lemma}
	
	 Note that  the function $\omega \mapsto \mathfrak F [{\psi_0^n}_{m\ell}] (\omega)$ is a $L^2(\mathbb Z_m \times \mathbb Z_{\ell \geq |m|})$-valued Schwartz function since $v \mapsto {\psi_0^n}_{m\ell}(v)$ is a $L^2(\mathbb Z_m \times \mathbb Z_{\ell \geq |m|})$-valued Schwartz function.
		We also define 
		\begin{align}\nonumber
			\tilde I_\textup{res}^a := \frac{1}{\sqrt{2\pi ( r_n^2+a^2)}} &  \sum_{m\ell} \textup{p.v.}\int_\mathbb R \frac{ e^{-2 i (\omega - \omega_- m ) r_n^\ast }e^{i\omega  u_0} \tilde {\mathcal F_{\mathcal H}}[\psi_0^n] }{\omega-\omega_- m} \d \omega \\  & \;\;\;\;\; \times e^{im \phi_-^\ast} S_{m\ell}(am\omega_-,\cos \theta)  \mathfrak t(\omega_-m,m,\ell)\tilde \uchl^{\omega_-}. \label{eq:tildeiares}
			\end{align}

		\begin{lemma}
			We have $ \| 	\hat I_\textup{res}^a  - \tilde 	 I_\textup{res}^a \|^2_{L^2(\mathbb S^2)} \lesssim D$ for all $r^\ast_n\geq r_0^\ast$.
			\begin{proof}We use that the spheroidal harmonics  $S_{m\ell}(a m \omega_-, \cos\theta) e^{im   \phi_-^\ast}$ form an orthonormal basis of $L^2(\mathbb S^2)$ to estimate
			\begin{align} \nonumber
				\|   \tilde I_{\textup{res}}^a - \hat I_\textup{res}^a \|_{L^2(\mathbb S^2)}^2 &\lesssim \sum_{m \ell}  |m|^2    \left|  \left( \int_{-\infty}^{\omega_- m - \frac{1}{am}}  + \int_{\omega_- m + \frac{1}{am}}^{+\infty}\right)  |   \tilde {\mathcal F_{\mathcal H}}[\psi_0^n] [\psi_0^n] | \d \omega  |\tilde \uchl^{\omega_-}  \mathfrak t^{\omega_-} |  \right|^2  \\
				&\lesssim \sum_{m \ell}  |m|^2     \Lambda_{m\ell}(a\omega_- m)  \left|\int_{\mathbb R}  |  \tilde {\mathcal F_{\mathcal H}}[\psi_0^n]  | \d \omega  |  \right|^2  \lesssim D,
			\end{align}
			where we used the Cauchy--Schwarz inequality in the last step.
			\end{proof}
		\end{lemma}
Now, we turn to  $\tilde I_\textup{res}^a $ as defined in \eqref{eq:tildeiares} and first only consider the $\omega$-integral 
\begin{align} \textup{Int}_\textup{res}^a := 
 \frac{1}{\sqrt{2\pi ( r_n^2+a^2)}}	\textup{p.v.} \int_\mathbb R \frac{ e^{-2 i (\omega - \omega_- m ) r_n^\ast }e^{i\omega  u_0} \tilde {\mathcal F_{\mathcal H}}[\psi_0^n] }{\omega-\omega_- m} \d \omega.
\end{align}
We have 
\begin{align}\nonumber
	\textup{Int}_\textup{res}^a & = \frac{\sqrt{r_+^2+a^2}}{\sqrt{r_n^2+a^2}} \frac{1}{\sqrt{2\pi}} 	\textup{p.v.} \int_{\mathbb R} \frac{\mathfrak F [{\psi_0^n}_{m\ell}(\cdot -u_0 +  2 r_n^\ast) e^{i \omega_- m \cdot } ]}{\omega} e^{2i \omega_- m r_n^\ast} \d \omega \\ \nonumber & = 
	 \frac{\sqrt{r_+^2+a^2}}{\sqrt{r_n^2+a^2}} \frac{1}{\sqrt{2\pi}} e^{2i \omega_- m r_n^\ast}	\textup{p.v.}\left( \frac{1}{\omega} \right) \Big[ \mathfrak F [{\psi_0^n}_{m\ell}(\cdot - u_0 +  2 r_n^\ast) e^{i \omega_- m \cdot }] \Big]
	 \\ & = \frac{\sqrt{r_+^2+a^2}}{\sqrt{r_n^2+a^2}} \frac{1}{\sqrt{2\pi}} e^{2i \omega_- m r_n^\ast} i \pi\,  \textup{sgn} \Big[ {\psi_0^n}_{m\ell}(\cdot - u_0 + 2 r_n^\ast) e^{i \omega_- m \cdot } \Big],
\end{align}
where $\textup{sgn}$ has to be understood as a Schwartz distribution. We have used that $\mathfrak F \left[ \textup{p.v.} \left( \frac{1}{\omega} \right) \right] = i \pi \textup{sgn}$ in the sense of distributions. Now, since $\psi_0$ is smooth, the function $v \mapsto  {\psi_0^n}_{m\ell} $ is a Schwartz function with values in the space of superpolynomially decaying sequences in $m,\ell$ as a subspace of $L^2(\mathbb Z_{m} \times \mathbb Z_{\ell \geq |m|})$. Particularly, this implies that \begin{align} \mathfrak t^{\omega_-} \tilde{\uchl}^{\omega_-} \textup{Int}_\textup{res}^a  \in L^2(\mathbb Z_m \times \mathbb Z_{\ell \geq |m|}; L^\infty(r_0^\ast,\infty)),\end{align} 
so we can project $\tilde I_\textup{res}^a$ on $e^{i m \phi_-^\ast} S_{m\ell}(am \omega_-,\cos\theta)$. Indeed, this yields
\begin{align}\nonumber
 \langle e^{im \tilde \phi_-^\ast } & S_{m\ell}(a m \omega_-, \cos \theta) ,\tilde I_\textup{res}^a\rangle_{L^2(\mathbb S^2 )}\\ & = \frac{\sqrt{r_+^2+a^2}}{\sqrt{r_n^2+a^2}} \frac{\mathfrak t^{\omega_-} \,  \tilde \uchl^{\omega_-} }{\sqrt{2\pi}} e^{2i \omega_- m r_n^\ast} i \pi\,  \textup{sgn} \Big[ {\psi_0^n}_{m\ell}(\cdot - u_0 +  2 r_n^\ast) e^{i \omega_- m \cdot } \Big].
 \label{eq:tildeiresa}
\end{align}
To summarize, we have decomposed $I$ as 
\begin{align}
	I = I_\textup{res} + I_\textup{non-res} = I_\textup{res} + I_\textup{non-res} = \tilde I_{\textup{res}}^a +  (  I_{\textup{res}}^a  - \hat I_{\textup{res}}^a ) +  (  \hat I_{\textup{res}}^a  - \tilde I_{\textup{res}}^a ) +I_{\textup{res}}^b +  I_\textup{non-res},
\end{align}
where $\tilde I_{\textup{res}}^a$ satisfies \eqref{eq:tildeiresa} and \begin{align}\|    (  I_{\textup{res}}^a  - \hat I_{\textup{res}}^a ) +  (  \hat I_{\textup{res}}^a  - \tilde I_{\textup{res}}^a ) +I_{\textup{res}}^b +  I_\textup{non-res}\|_{L^2(\mathbb S^2)} \lesssim D^\frac 12.\end{align} 

Completely analogous to the   analysis before, we also decompose $II$ as 
\begin{align}\nonumber
II = II_\textup{res} + II_\textup{non-res}& = II_\textup{res} + II_\textup{non-res} \\ &=  \tilde{II}_{\textup{res}}^a +  (  {II}_{\textup{res}}^a  - \hat {II}_{\textup{res}}^a ) +  (  \hat {II}_{\textup{res}}^a  - \tilde {II}_{\textup{res}}^a ) +{II}_{\textup{res}}^b +  {II}_\textup{non-res},
\end{align}
where \begin{align}\|  (  {II}_{\textup{res}}^a  - \hat {II}_{\textup{res}}^a ) +  (  \hat {II}_{\textup{res}}^a  - \tilde {II}_{\textup{res}}^a ) +{II}_{\textup{res}}^b +  {II}_\textup{non-res}\|_{L^2(\mathbb S^2)} \lesssim D^\frac 12\end{align}
and  $\tilde{II}_{\textup{res}}^a $ satisfies
\begin{align}
\langle e^{im \phi_-^\ast} S_{m\ell}(a m \omega_-, \cos \theta) ,\tilde{II}_\textup{res}^a\rangle_{L^2(\mathbb S^2 )} = \frac{\sqrt{r_+^2+a^2}}{\sqrt{r_n^2+a^2}} \frac{\mathfrak r^{\omega_-}     \tilde \uchr^{\omega_-} }{\sqrt{2\pi}}  i \pi\,  \textup{sgn} \Big[ {\psi_0^n}_{m\ell}(\cdot - u_0 ) e^{i \omega_- m \cdot } \Big].
\end{align}
Hence, using \begin{align}\mathfrak r^{\omega_-} = - \mathfrak t^{\omega_-} \text{ and }\tilde \uchr^{\omega_-} = \tilde \uchl^{\omega_-},\end{align} we obtain
\begin{align}\nonumber
	\langle e^{im \phi_-^\ast} S_{m\ell}(a m \omega_-,& \cos \theta) ,\tilde{I}_\textup{res}^a + \tilde{II}_\textup{res}^a\rangle_{L^2(\mathbb S^2 )}\\
	& = - i \pi  \frac{\sqrt{r_+^2 + a^2}}{\sqrt{r_n^2 + a^2}} \frac{\mathfrak r^{\omega_-} \tilde \uchl^{\omega_-}}{\sqrt{2\pi}}e^{i \omega_- m u_0} \int_{-u_0}^{2 r_n^\ast-u_0} {\psi_0^n}_{m\ell} (v) e^{i \omega_- m v} \d v.\label{eq:projectionofI+II}
\end{align}
	Now, by construction of $\psi_0^n$, we have that ${\psi_0^n}_{m\ell}(v) =0$ for $v \geq 2 r_n^\ast -u_0 +\tilde c$, where $\tilde c$ is a constant only depending on the black hole parameters. In particular, this implies that 
		\begin{align}
	\sum_{m\ell} \left| i \pi  \frac{\sqrt{r_+^2 + a^2}}{\sqrt{r_n^2 + a^2}} \frac{\mathfrak r^{\omega_-} \tilde \uchl^{\omega_-}}{\sqrt{2\pi}}e^{i \omega_- m u_0}\left( \int_{-\infty}^{-u_0}+ \int_{2r_n^\ast-u_0}^{+\infty}  \right) {\psi_0^n}_{m\ell} (v) e^{i \omega_- m v} \d v \right|^2   \lesssim_{u_0} D
		\end{align}
which allows us to---up to a term bounded by $D^\frac 12$---replace the integral in \eqref{eq:projectionofI+II} with an integral on the whole real line $v\in \mathbb R$.
Finally, from \cref{prop:rep} (more precisely \eqref{eq:representationformulainproof}), we obtain 
\begin{align}\nonumber
	\| \psi \|_{L^2(\mathbb S^2)}^2 (u_0, r_n^\ast) &= \sum_{m\ell} \left|\pi  \frac{\sqrt{r_+^2 + a^2}}{\sqrt{r_n^2 + a^2}} \frac{\mathfrak r^{\omega_-} \tilde \uchl^{\omega_-}}{\sqrt{2\pi}} \int_{\mathbb R} {\psi_0^n}_{m\ell} (v) e^{i \omega_- m v} \d v\right|^2 + \textup{Err}(D)\\ & =  \sum_{m\ell} \left|\pi  \frac{\mathfrak r^{\omega_-} \tilde \uchl^{\omega_-}}{\sqrt{r_n^2 + a^2}}   a_{\mathcal H}^{R_n} (\omega = \omega_- m)\right|^2 + \textup{Err}(D), \label{eq:applyfatouto}
\end{align}
where $|\textup{Err}(D)| \lesssim_{u_0} D$ uniformly for all $r_n^\ast \geq r_0^\ast$. We have established   formula \eqref{eq:formulaforlimit} which concludes the proof of \cref{prop:mainpropinmainthm}.\end{proof}
We will now finish off the proof of \cref{thm:rough}. 
From \cref{lem:HrepstoHr} we have that $a_{\mathcal H}^{R_n} \to a_{\mathcal H}$ pointwise for fixed $\omega, m, \ell$ as $R_n \to \infty$. We also have the pointwise limit \begin{align}\tilde \uchl^{\omega_-} \to 1 \text{ as } r_n^\ast \to \infty.\end{align} Hence, applying Fatou's lemma to \eqref{eq:applyfatouto} yields
\begin{align}
\liminf_{r^\ast_n \to \infty}	\| \psi \|_{L^2(\mathbb S^2)}^2 (u_0, r_n^\ast) \geq \frac{\pi^2}{r_-^2 + a^2} \sum_{m\ell} |\mathfrak r^{\omega_-} |^2 |a_{\mathcal{H}}(\omega=\omega_- m,m,\ell)|^2 - C_{u_0} D,
\end{align}
where   $C_{u_0} >0$ is a constant depending on $u_0$. Since \begin{align}|\mathfrak r^{\omega_-}| \gtrsim |m|
\end{align}
for all $m$ sufficiently large as shown in \cref{lem:lowerboundont}, we obtain 
\begin{align}\label{eq:thefinalformulaforblowup}
\liminf_{r^\ast_n \to \infty}	\| \psi \|_{L^2(\mathbb S^2)}^2 (u_0, r_n^\ast) \gtrsim \frac{\pi^2}{r_-^2 + a^2} \sum_{i\in \mathbb N } |m_i |^2 |a_{\mathcal{H}}(\omega=\omega_- m_i,m_i,\ell_i)|^2 - C_{u_0} D.
\end{align}
Finally, from \cref{eq:defnah} we have that
\begin{align}|a_{\mathcal H}(\omega = \omega_- m_i, m_i, \ell_i)|\gtrsim e^{\frac 12 \sqrt{m_i}}\label{eq:inifinitelymany}\end{align} for infinitely many $m_i$ such that we conclude
\begin{align}\lim_{r^\ast_n \to \infty}	\| \psi \|_{L^2(\mathbb S^2)}^2 (u_0, r_n^\ast) = + \infty.
\label{eq:thefinalformula2} \end{align}
Since the sequence $r_n^\ast\to \infty$ was arbitrary we obtain \eqref{eq:mainresult}. This concludes the proof of \cref{thm:rough}.
\end{proof}
\begin{proof}[Proof of \cref{cor:genericity}]
Let $\tilde \Psi_0, \tilde \Psi_1$ be initial data as in the statement of \cref{cor:genericity} and let $\tilde \psi$ the arising solution to \eqref{eq:wavekerr}.  In view of the fact that different  azimuthal modes $m$ are $L^2(\mathbb S^2)$-orthogonal in evolution, it suffices to show the blow-up for the modes $m=m_i$, where $m_i$ is the sequence in \eqref{eq:mili} associated to the non-Diophantine condition \eqref{eq:smallnessofwronskian}. Now, the proof of \cref{thm:rough} also carries over for the initial data  $\tilde \Psi_0, \tilde \Psi_1$ and in particularly the analog of \eqref{eq:thefinalformulaforblowup} holds true.  Recalling the definition of $a_\mathcal H$ in \cref{defn:amathcalhr}, the analog of \eqref{eq:thefinalformulaforblowup} for $\tilde \psi$  is
\begin{align}\nonumber
\liminf_{r^\ast_n \to \infty}	\| \tilde \psi \|_{L^2(\mathbb S^2)}^2 (u_0, r_n^\ast)& \gtrsim \frac{\pi^2}{r_-^2 + a^2} \sum_{i \in \mathbb N} \frac{|m_i |^2 |G(
\tilde \Psi_0,\tilde \Psi_1,m_i,\ell_i)|^2}{{2\pi}|\mathfrak W[\uhplus,u_\infty](\omega=\omega_- m_i,m_i,\ell_i)|^2}   - C_{u_0} D(\tilde \Psi_0, \tilde \Psi_1)\\
& \gtrsim\sum_{i \in\mathbb N}  |m_i e^{\sqrt{|m_i|}}|^2 |G(
\tilde \Psi_0,\tilde \Psi_1,m_i,\ell_i)|^2 -  C_{u_0} D(\tilde \Psi_0, \tilde \Psi_1) \end{align}
in view of the non-Diophantine condition $|\mathfrak W[\uhplus,u_\infty](\omega=\omega_- m_i,m_i,\ell_i)|< e^{- \sqrt{m_i}}$ as in  \eqref{eq:smallnessofwronskian}. 
Thus, if the data $\tilde \Psi_0, \tilde \Psi_1$ satisfy the genericity condition \eqref{eq:genericitycond}, then $\lim_{r\to r_-} \| \tilde \psi (u_0,r)\|_{L^2(\mathbb S^2)}^2  = + \infty$ for every $u_0 \in \mathbb R$.

\end{proof}
\appendix
\addtocontents{toc}{\protect\setcounter{tocdepth}{1}} 
\section{Appendix}
\subsection{Airy functions}\label{sec:airyfunctions}
We recall the definition of the Airy functions of first and second kind $\Ai$ and $\Bi$ as follows.\begin{definition}
	For $x\in \mathbb R$, we define $\Ai(x)$ and $\Bi(x)$ via the improper Riemann integrals \begin{align}
	&\Ai (x) := \frac{1}{\pi} \int_{0}^{\infty} \cos\left(\frac{t^3}{3} + xt\right) \d t,\\
	&\Bi(x) := \frac{1}{\pi}\int_{0}^{\infty} \left[\exp\left(- \frac{t^3}{3} + xt\right) + \sin\left(\frac{t^3}{3} + xt \right)\right] \d t.
	\end{align}
\end{definition} 
Equivalently, the Airy functions are the unique solutions of
\begin{align}
	u'' = xu
\end{align}
with 
\begin{align}&\Ai (0)=\frac{1}{3^{\frac 23} \Gamma(\frac 23)}, \Ai'(0) = \frac{-1}{3^{\frac 13} \Gamma(\frac 13)}, \\ &\Bi (0)=\frac{1}{3^{\frac 16} \Gamma(\frac 23)}, \Bi'(0) = \frac{3^{\frac 16}}{ \Gamma(\frac 13)} \end{align} such that $\mathfrak W_{x} (\Ai(x), \Bi(x)) = \frac 1 \pi$.  
Further, we  define the constant $c$ as the largest negative root of $\mathrm{Ai}(x) = \mathrm{Bi}(x)$. Then, we introduce the error-control functions 
\begin{align}
E_\Ai(x):=\begin{cases}
\left(	\mathrm{Bi}(x)/\mathrm{Ai}(x) \right)^{\frac 12} & x \geq c\\
1 &  x \leq c
\end{cases}\text{ and }
M_\Ai(x):= \begin{cases}
(2 \mathrm{Ai}(x) \mathrm{Bi}(x) )^{\frac 12} & x \geq c\\
( \mathrm{Ai}^2(x) + \mathrm{Bi}^2(x) )^{\frac 12} & x \leq c
\end{cases} \label{eq:definsofeaimai}
\end{align} and $E_\Ai^{-1}(x) := \frac{1}{E_\Ai(x)}$. 
From \cite[Chapter~11, \S2]{olver} we remark that $E_\Ai$ is a monotonically increasing function of $x$ which is never less than $1$ and moreover,
\begin{align}
 |\Ai (x) |\leq \frac{M_\Ai (x)}{E_\Ai (x)}  \text{ as well as } |\Bi(x) | \leq M_{\Ai}(x) E_{\Ai}(x).
\end{align}
Furthermore,  we have (see \cite[Chapter~11, \S2]{olver} )
\begin{align}
|M_{\Ai}(x) | \lesssim \frac{1}{\langle x \rangle^{\frac 14}} \label{eq:boundonmai}
\end{align}
for $x \in \mathbb R$. Similarly, we define 
\begin{align}
N_\Ai(x):= \begin{cases}\left( 
	\frac{
 \Ai'(x)^2 \Bi'(x)^2 + \Bi'(x)^2 \Ai(x)^2}{\Ai(x) \Bi(x)}\right)^{\frac 12 }  & x \geq c\\
	( \mathrm{Ai}'(x)^2 + \mathrm{Bi}'(x)^2 )^{\frac 12} & x \leq c,
\end{cases} \label{eq:definsofeainai}
\end{align}
which satisfies (see \cite[Chapter~11, \S2]{olver}) \begin{align}|N_{\Ai}(x) |\lesssim \langle x \rangle^{\frac 14}.
\label{eq:errorestimateonnai}\end{align} 

The Airy functions obey the following asymptotics.
\begin{lemma}[{\cite[Chapter~11, \S1,\S2]{olver}, \cite[\S9.7]{NIST:DLMF}}] \label{lem:asymptoticsforairy}
	For large $x>0$, the asymptotic behaviors of the Airy functions are \begin{align}\label{eq:airyasy}
	&\mathrm{Ai} (-x) = \frac{1}{\sqrt \pi x^{\frac 14} } \cos\left(\frac 23  x^{\frac 32} - \frac \pi 4\right) + \tilde \epsilon_{\Ai}(x),\;\;\; \Ai'(-x) = \frac{x^{\frac 14}}{\sqrt \pi} \sin\left(\frac 23 x^{\frac 32} - \frac 14 \pi\right) + \tilde \epsilon_{\Ai'}(x), \\ \label{eq:biairyasy}
	&\mathrm{Bi} (-x) = \frac{-1}{\sqrt \pi x^{\frac 14} } \sin\left(\frac 23  x^{\frac 32} - \frac \pi 4\right) + \tilde \epsilon_{\Ai}(x),\;\;\; \mathrm{Bi	}'(-x) = \frac{x^{\frac 14}}{\sqrt \pi} \cos\left(\frac 23 x^{\frac 32} - \frac 14 \pi\right) + \tilde \epsilon_{\Ai'}(x),\end{align}
	 where $|\tilde \epsilon_{\Ai}| \lesssim  x^{-\frac{7}{4}}$ and $|\tilde \epsilon_{\Ai'}| \lesssim x^{-\frac 54}$. In particular, we have 
	 \begin{align}
	 &	|\Ai(-x)|, |\Bi(-x)| \lesssim \frac{1}{1+x^{\frac 14}} \text{ and } 	|\Ai'(-x)|, |\Bi'(-x)| \lesssim 1+x^{\frac 14} 
	 \end{align}
	for $x\geq 0$. 
	  Moreover, for $x> 0$ we have 
	 \begin{align}
	 &0 \leq 	\Ai(x) \leq \frac{e^{-\frac{2}{3} x^{\frac 32} }}{2 \sqrt \pi x^{\frac 14}},  \\ 
	 & |\Ai'(x)|\leq \frac{x^{\frac 14} e^{-\frac{2}{3} x^{\frac 32}}}{2 \sqrt{\pi}} \left(1 + \frac{7}{ 48 x^{\frac 32}}\right),\\
	 	&0 \leq  \Bi(x)\leq \frac{e^{\frac 23 x^{\frac 32}}}{\sqrt \pi x^{\frac 14}} \left( 1+ \left(\chi_\Ai \left( \frac 76 \right) + 1\right) \frac{5}{48 x^{\frac 32}} \right), 
	 	\\ &0 \leq   \Bi'(x)\leq \frac{x^{\frac 14} e^{\frac 23 x^{\frac 32}}}{\sqrt \pi } \left( 1+ \left(\frac \pi 2+ 1\right) \frac{7}{48 x^{\frac 32}} \right),
	 \end{align}
	 where $\chi_\Ai (x) = \sqrt{\pi}\frac{ \Gamma(\frac{1}{2} x +1) }{\Gamma(\frac 12 x + \frac 12)}$.
\end{lemma}
\subsection{Parabolic cylinder functions}
\label{sec:appweight}
We define the parabolic cylinder functions $U$ and $\bar U$ in the following.  We refer to \cite[Section~5]{second-order} or \cite[Chapter~12]{NIST:DLMF} for more details.
\begin{definition}\label{defn:uubar}
For  $b\leq 0$ and $x\geq 0$ we define the parabolic cylinder functions
\begin{align}
U(b, x) =& \frac{\pi^{\frac 12} 2^{- \frac 14 (2b+1)} e^{-\frac 14 x^2}} {\Gamma(\frac{3}{4} + \frac 12 b)} {}_1{F}_1\left(\frac{1}{2} b + \frac{1}{4}; \frac 12;\frac 12 x^2\right) \nonumber\\ & - \frac{\pi^{\frac 12} 2^{- \frac 14(2b-1)}}{\Gamma(\frac 14 + \frac 12 b)}e^{- \frac 14 x^2} x {}_{1}{F}_1 \left( \frac 12 b + \frac 34; \frac 32; \frac 12 x^2\right),\\ \nonumber 
\bar U (b ,  x)   = &\pi^{- \frac 12} 2^{- \frac 14 (2b+1)} \Gamma\left(\frac 14 - \frac 12 b\right) \sin\left(\frac 34 \pi - \frac 12 b\pi\right) e^{- \frac 14 x^2} {}_1{F}_1 \left(\frac 12 b + \frac 14; \frac 12 ; \frac 12 x^2\right) \\ & - \pi^{- \frac 12} 2^{-\frac 14 ( 2b - 1) } \Gamma\left(\frac 34 - \frac 12 b\right) \sin\left(\frac 54 \pi - \frac 12 b \pi\right) e^{- \frac 14 x^2} x {}_1{F}_1 \left( \frac 12 b + \frac 34; \frac 32; \frac 12 x^2\right),
\end{align}
where ${}_1{F}_1(a;b;z):= \sum_{n=0}^\infty \frac{a^{(n)} z^n}{b^{(n)} n!}$ denotes the confluent hypergeometric function. Here, we use the notation $a^{(n)}:= a (a+1) (a+2) \cdots (a+n)$ for the rising factorial.
\end{definition} 
Remark that $\mathfrak W(U,\bar U) = \sqrt \frac{2}{\pi}\Gamma(\frac 12 - b)$ and that $U$ and $\bar U$ solve the equation 
\begin{align}
-u'' + \left(\frac 14 x^2 + b\right)  u =0.
\end{align}
Moreover, we have (e.g.\ \cite[Section 5.2]{second-order}) that
\begin{align}\label{eq:valuesatzero1}
&U(b,0) = \pi^{-\frac 12} 2^{-\frac 14(2b+1)} \Gamma\left(\frac 14 - \frac 12 b\right) \sin\left(\frac 1 4 \pi  - \frac 12 b \pi\right), \\\label{eq:valuesatzero2}
&U'(b,0) = - \pi^{-\frac 12} 2^{-\frac 14(2b-1) } \Gamma\left(\frac 34 - \frac 12 b\right) \sin\left(\frac 34 \pi - \frac 12 b \pi\right), \\\label{eq:valuesatzero3}
&\bar U(b,0) = \pi^{-\frac 12} 2^{-\frac 14(2b+1)} \Gamma\left(\frac 14 - \frac 12 b\right) \sin\left(\frac 3 4 \pi  - \frac 12 b \pi\right), \\\label{eq:valuesatzero4}
&\bar U'(b,0) = - \pi^{-\frac 12} 2^{-\frac 14(2b-1) } \Gamma\left(\frac 34 - \frac 12 b\right) \sin\left(\frac 54 \pi - \frac 12 b \pi\right).
\end{align}

We define auxiliary functions to control error terms in terms of parabolic cylinder functions. We first define $\rho(b)$ as the largest real root of the equation $\bar U(b,x) = U (b,x)$. Note that $\rho(b) \geq 0$ for $b \leq 0$.
\begin{definition}
For $b\leq 0$, we set 
\begin{align}
	E_U(b,x) = \begin{cases}
	1 &\textup{ for } 0\leq x \leq \rho(b) \\
	\sqrt{\frac{\bar U (b,x) }{U(b,x)}} & \textup{ for } x \geq \rho(b) 
	\end{cases}
\end{align}
\end{definition}
For fixed $b$, the function $E_U(b,x)$ is continuous and non-decreasing in $0\leq x <\infty$. We  denote $E_U^{-1}:= \frac{1}{E_U}$. 
\begin{definition}\label{defn:munu}
For $b\leq 0$, $x\geq 0$, we also define functions $M_U$  and $N_U$ by
\begin{align}
	&M_U(b,x):= \begin{cases}
	\sqrt{U^2 + \bar U^2} & \textup{ for } 0 \leq x \leq \rho(b)\\
	\sqrt{2 U \bar U } &\textup{ for } \rho(b) \leq x  .
	\end{cases}\\
	&N_U(b,x):= \begin{cases}
	\sqrt{U'^2 + \bar U'^2} & \textup{ for } 0 \leq x \leq \rho(b)\\
	\sqrt{\frac{ {U'}^2 \bar U^2 + {\bar U}^{\prime 2} U^2}{U \bar U} } &\textup{ for } \rho(b) \leq x  .
	\end{cases}
\end{align}
\end{definition}
\begin{definition}\label{defn:zetaU}
	We define the function $\zeta_U$ as
	\begin{align}
		  \zeta_U(t):= \begin{cases}-
		  \left( \frac{3}{2} \int_{t}^1 (1 - \tau^2  )^{\frac 12} \d \tau \right)^{\frac 23} & \text{ for }0 \leq  t \leq 1, \\ 
\left( \frac{3}{2} \int_{1}^t (\tau^2 - 1 )^{\frac 12} \d \tau \right)^{\frac 23} & \text{ for } t \geq 1.
		  \end{cases} 
	\end{align}
\end{definition}
Note that we   have (see e.g.\ \cite[\S5.8]{second-order})
\begin{align}
	&|U|\leq E^{-1}M_U, |\bar U|\leq E M_U \text{ and } |U \bar U|\leq M_U^2 \label{eq:estimatesonU}\\
	&|U'|\leq E^{-1}N_U, |\bar U'|\leq E N_U \text{ and } |U \bar U|\leq N_U^2 \label{eq:estimatesonU'}
\end{align}
for $x \geq 0$ and $b\leq 0$. 
\begin{prop} \label{prop:boundsonmu}
	The envelope function $M_U$ satisfies
	\begin{align}
	M^2_{U}\left(-\frac 12 \mu^2,\mu y \sqrt 2\right) \lesssim\frac{1}{\mu^{\frac 13}} \frac{1}{1+ |\zeta_U(y)|^{\frac 14}} \frac{1}{1+\mu^{\frac 23} |\zeta_U(y)|^{\frac 12}} \Gamma\left(\frac 12 + \frac 12 \mu^2\right)
	\end{align}
	uniformly in $\mu \geq 1$ and $y \geq 0$,
	and 
	\begin{align}
			M^2_{U}\left(-\frac 12 \mu^2,\mu y \sqrt 2\right) \lesssim \frac{1}{1+ \sqrt{\mu y}} \Gamma\left(\frac 12 + \frac 12 \mu^2\right)
	\end{align}
		uniformly in $0 \leq \mu \leq 1$ and $y \geq 0$.
		In particular, $M_U$ satisfies
		\begin{align}
	\left|M_{U}\left(-\frac 12 \mu^2,\mu y \sqrt 2\right) \right|^2 \lesssim \Gamma\left( \frac 12 + \frac 12 \mu^2 \right).
		\end{align}
\begin{proof}
	These estimates follow from \cite[Equation (5.23), (6.12) and Section 6.2]{second-order}.
\end{proof}
\end{prop}

\printbibliography[heading=bibintoc]
\end{document}